\documentclass[10pt,a4paper]{article}
\usepackage[utf8]{inputenc}
\usepackage{amsmath, amsthm}
\usepackage{amsfonts}
\usepackage{amssymb}
\usepackage{bbm}
\usepackage{url}

\usepackage{ulem}
\usepackage{hyperref}	
\usepackage{float}
\usepackage{xcolor}
\usepackage[left=1.9cm,right=1.9cm,top=1.9cm,bottom=1.9cm]{geometry}

\newcommand{\ep}{\varepsilon_0}
\newcommand{\pe}{p(e)}

\newcommand{\stkout}[1]{\ifmmode\text{\sout{\ensuremath{#1}}}\else\sout{#1}\fi}
\usepackage{enumerate}
\newcommand{\rd}{\partial}
\newcommand{\ls}{\lesssim}
\theoremstyle{plain}
\newtheorem{thm}{Theorem}[section]

\newtheorem{lem}[thm]{Lemma}
\newtheorem{prop}[thm]{Proposition}

\newtheorem{cor}[thm]{Corollary}

\newtheorem{conjecture}[thm]{Conjecture}

\theoremstyle{remark}
\newtheorem{rmk}{Remark}

\theoremstyle{definition}

\newcommand{\tnu}{\tilde{\nu}(e)}
\newcommand{\tp}{\breve{p}(e)}
\newcommand{\np}{p_{\nu}}
\newcommand{\qa}{q_{\alpha}(e)}
\newcommand{\ga}{g_{\alpha}(e)}

\newcommand{\pp}{\breve{p}}

\usepackage{color}

\addtocounter{tocdepth}{-2}
\usepackage{graphicx}
\usepackage{authblk}

\numberwithin{equation}{section}
\begin{document}
	
	\title{Decay of weakly charged solutions for the spherically symmetric Maxwell-Charged-Scalar-Field equations on a \\ Reissner--Nordstr\"{o}m exterior space-time }

	\author[a]{Maxime Van de Moortel  \thanks{E-mail : mcrv2@cam.ac.uk}}
	
	\affil[a]{University of Cambridge, Department of Pure Mathematics and Mathematical Statistics, Wilberforce Road, Cambridge CB3 0WA, United Kingdom}
	\maketitle
	\abstract
	
	We consider the Cauchy problem for the (non-linear) Maxwell-Charged-Scalar-Field equations with spherically symmetric initial data, on a sub-extremal  Reissner--Nordstr\"{o}m or Schwarzschild \color{black} exterior space-time. We prove that the solutions are bounded and decay at an inverse polynomial rate towards time-like infinity and along the black hole event horizon, provided the charge of the Maxwell equation is sufficiently small. 
	
	This condition is in particular satisfied for small data in energy space that enjoy a sufficient decay towards the asymptotically flat end.

	Some of the decay estimates we prove are arbitrarily close to the  conjectured optimal rate in the limit where the charge tends to  zero, according the heuristics present in the physics literature.
	
	Our result can also be interpreted as a first \color{black}step towards the 
	stability of Reissner--Nordstr\"{o}m black holes for the \textbf{gravity coupled} Einstein--Maxwell-Charged-Scalar-Field model. This problem is closely connected to the understanding of strong cosmic censorship and charged gravitational collapse in this setting.

	\setcounter{tocdepth}{1}

	\tableofcontents

	\section{Introduction}
	\paragraph{The model}	In this paper, we study the asymptotic behaviour of solutions to the Maxwell-Charged-Scalar-Field equations, sometimes referred to as massless Maxwell--Klein--Gordon, arising from spherically symmetric and asymptotically decaying initial data on a fixed sub-extremal Reissner--Nordstr\"{o}m exterior space-time : 	\begin{equation} \label{Maxwelleq}\nabla^{\mu} F_{\mu \nu}= iq_0 \frac{ (\phi \overline{D_{\nu}\phi} -\overline{\phi} D_{\nu}\phi)}{2} , \; F=dA ,
	\end{equation} \begin{equation} \label{5} g^{\mu \nu} D_{\mu} D_{\nu}\phi = 0,
	\end{equation}
	\begin{equation} \label{RN}
		g=-\Omega^{2}dt^{2}+\Omega^{-2}dr^{2}+r^{2}[ d\theta^{2}+\sin(\theta)^{2}d \varphi^{2}],
	\end{equation}	\begin{equation} \label{OmegaRN}
		\Omega^{2}=1-\frac{2M}{r}+\frac{\rho^2}{r^2},
	\end{equation} where $q_0 \geq 0$ is a constant called the charge \footnote{This charge $q_0$ is also the coupling constant between the Maxwell and the scalar field equations. This is not to be confused with the charge of the Maxwell equation or the parameter $\rho$. For a precise definition of all ``charges'', c.f.\ section \ref{chargenotations}.} of the scalar field $\phi$,  $M$, $\rho$ are respectively the mass and the charge of the Reissner--Nordstr\"{o}m black hole 	with $0 \leq |\rho| <M$ , $\nabla_{\mu}$ is the Levi--Civita connection and $D_{\mu}=\nabla_{\mu}+iq_0 A_{\mu}$ is the gauge derivative. Note that --- due to the interaction between the Maxwell field $F$ and the charged scalar field $\phi$ --- this system of equations is \textbf{non-linear} when $q_0 \neq 0$, the case of interest for this paper. This is in contrast to the uncharged case $q_0=0$, where \eqref{5} is then the linear wave equation. 
	
	\paragraph{Main results} Since global regularity is known for this system \footnote{It follows essentially from the global regularity of Yang-Mills equations on globally hyperbolic (3+1) Lorentzian manifolds, established in \cite{Shatah}. In spherical symmetry, this can also be deduced from the methods of \cite{Kommemi}.}, we focus on the asymptotic behaviour.
	
	The case of a charged scalar field on a black hole space-time that we consider is considerably different from the analogous problem on Minkowski space-time. While on the flat space-time, the charge of the Maxwell equation tends to $0$ towards time-like infinity, this is not expected to be the case on black hole space-times. This fact constitutes a major difference and renders the proof of decay harder, already in the spherically symmetric case and when the charge is small.
	
	We show that if the charge in the Maxwell equation and the scalar field energies are initially smaller than a constant depending on the black hole parameters $M$ and $\rho$, then
	
	\begin{enumerate}
		\item \label{intro1} the charge in the Maxwell equation is bounded and small on the Reissner--Nordstr\"{o}m exterior space-time.

		\item  \label{intro2}Boundedness of the scalar field energy holds.
		
		\item  \label{intro3}A local integrated energy decay estimate holds for the scalar field. 
		
	\end{enumerate}

	If now we relax the smallness hypothesis on the charge, requiring only that the initial charge is smaller than a numerical constant \footnote{More precisely, we use the sufficient bounds $q_0|e_0| < {\frac{1}{12}}$ for the weakest claimed decay and $q_0|e_0| < {0.04}$ for the improved one, where $e_0$ is the initial asymptotic charge, c.f.\ section \ref{chargenotations}.} and assume \ref{intro2} and \ref{intro3}, then
	\begin{enumerate}		
		\setcounter{enumi}{3}	
		\item \label{intro4} the energy \footnote{By this, we mean all the energies transverse or parallel to the event horizon or null infinity, or $L^2$ flux on any constant $r$ curve.} of the scalar field decays at an inverse polynomial rate, depending on the charge. 
		
		\item The scalar field enjoys point-wise decay estimates at an inverse polynomial rate consistent with \ref{intro4}.
		
	\end{enumerate}
	
	These results are stated in a simplified version in Theorem \ref{simplifiedthm} and later in a more precise way in Theorem \ref{maintheorem}, Theorem \ref{boundednesstheorem} and Theorem \ref{decaytheorem}.

	The decay rate of the energy --- and of some point-wise estimates that we derive --- has been conjectured to be optimal in \cite{HodPiran1}, in the limit when the asymptotic charge tends to zero, c.f.\ section \ref{conjecture}. Other point-wise estimates, notably along the black hole event horizon, are however not sharp in that sense.

	In the case of an uncharged scalar field $q_0=0$ (namely the wave equation on sub-extremal black holes), it is well-known that the long term asymptotics are governed by the so-called Price's law, first put forth heuristically by Price in \cite{Pricepaper} and later proved in \cite{Newrp}, \cite{Latetime}, \cite{PriceLaw}, \cite{Schlag}, \cite{Tataru2}, \cite{Tataru}. Generic solutions then decay in time at an universal inverse polynomial decay rate, in the sense that this rate does not depend on physical parameters or on the initial data. In contrast, the optimal decay rate for our charged case $q_0 \neq 0$  is conjectured to be \textbf{slower} and depending on the charge in the Maxwell equation, itself determined by the data.

	Roughly, this can be explained by the fact that in the uncharged case $q_0=0$, the equation effectively looks like the wave equation on Minkowski space in the presence of a potential decaying like  $r^{-3}$, see section \ref{conjecture}.
	
	This decay of the potential-like term is somehow more ``forgiving'' than in the 
	charged case $q_0 \neq 0$. In the latter case, the equation becomes similar to the wave equation in the presence of a potential decaying like  $r^{-2}$. As explained beautifully in pages 7 and 8 of \cite{LindbladKG}, while the system is sub-critical with respect to the conserved energy in dimensions (3+1), the new term coming from the charge induces a form of criticality with respect to decay at space-like infinity, i.e.\ a criticality with respect to $r$ weights. 
	
	One important consequence --- in the black hole case --- is that the sharp decay rate is expected to depend on the charge, c.f.\ section \ref{conjecture}. Interestingly in our paper, in order to deal with this criticality, we need to use the full non-linear structure of the system. Also, the criticality with respect to decay implies the absence of any ``extra convergence factor'' that facilitates the proof of bounds on long time intervals, in the language of \cite{LindbladKG}. This is in contrast to the uncharged case and requires sharpness in the estimates, as explained in \cite{LindbladKG}.
	



	\paragraph{Motivation} Our result can be viewed as a first step towards the understanding of the analogous Einstein--Maxwell-Charged-Scalar-Field model, where the Maxwell and Scalar Field equations are now coupled with gravity, c.f.\ equations \eqref{EMKG1}, \eqref{EMKG2}, \eqref{EMKG3}, \eqref{EMKG4}, \eqref{EMKG5} when $m^2=0$. In this setting, the asymptotic behaviour of the scalar field is important, in particular because it determines the black hole interior structure, c.f.\ section \ref{interior}. This is also closely related to the Strong Cosmic Censorship Conjecture, c.f.\ section \ref{interior}. 
	
	Before discussing the relevance of our result for the interior of black holes, let us mention that the Einstein--Maxwell-Charged-Scalar-Field system possesses a number of new features compared to its uncharged analogue $q_0=0$. One of them is the existence of one-ended \footnote{In spherical symmetry, the initial data is one-ended if it is diffeomorphic to $\mathbb{R}^3$ and two-ended if it is diffeomorphic to $\mathbb{R} \times \mathbb{S}^2$.} charged black holes solutions, which are of great interest to study the formation of a Cauchy horizon during gravitational collapse. 
	Unlike their uncharged analogue in spherical symmetry, charged black holes admit a Cauchy horizon, a feature which is also expected for the Einstein vacuum equations with no symmetry, c.f.\ \cite{KerrStab}. 
	While previously studied models in spherical symmetry only admit either uncharged one-ended black holes (c.f.\ \cite{Christo1}, \cite{Christo2}, \cite{Christo3}) or charged two-ended black holes (c.f.\ \cite{MihalisPHD}, \cite{Mihalis1}, \cite{PriceLaw}, \cite{JonathanStab}, \cite{JonathanStabExt}), the Einstein--Maxwell-Charged-Scalar-Field model admits \textbf{charged one-ended} black holes. This is because the coupling between the Maxwell-Field and the charged scalar field allows for a non-constant charge.
	

	We now return to the interior structure of black holes. Various previous works for different models have highlighted that the interior possesses both stability and instability \footnote{See also the recent \cite{GregJan} that investigates a very different kind of scalar field instability on Schwarzschild black hole, which is somehow stronger but specific to uncharged and non-rotating black holes.} features that strongly depend on the asymptotic behaviour of the scalar field on the event horizon c.f.\ \cite{MihalisPHD}, \cite{Mihalis1}, \cite{KerrStab}, \cite{JonathanInstab}, \cite{JonathanStab}, \cite{JonathanStabExt}, \cite{KerrInstab}.
	For the Einstein--Maxwell-Charged-Scalar-Field model, stability and instability results in the interior have been established in \cite{Moi}. Concerning stability, it is proven in \cite{Moi} that a scalar field decaying point-wise at a \textbf{strictly integrable} inverse-polynomial rate on the event horizon gives rise to a $C^0$ stable Cauchy horizon over which the metric is continuously extendible.  In the present paper, we establish such strictly integrable point-wise decay for the scalar field on a fixed black hole. If those bounds can be extrapolated to the Einstein--Maxwell-Charged-Scalar-Field case, then the hypotheses of \cite{Moi} are satisfied and the stability result applies. In particular, the continuous extendibility result of \cite{Moi} would then disprove the \underline{continuous} formulation of the so-called Strong Cosmic Censorship Conjecture as discussed in more details in section \ref{interior}.

	Another important aspect is the singularity structure of the black hole interior, which is related to the instability feature we stated earlier. This is relevant to a (weaker) $C^2$ formulation of the Strong Cosmic Censorship Conjecture we mentioned before, c.f.\ section \ref{interior}.
	More specifically, it is proven in \cite{Moi} that lower bounds for the energy on the event horizon imply the formation of a singular Cauchy horizon. It was proven in \cite{JonathanStab} for the uncharged analogue model  that the positive resolution of the $C^2$ strong cosmic censorship conjecture actually follows from this singular nature of the Cauchy horizon. For this, lower bounds on the event horizon must be proven for solutions of a Cauchy problem, c.f.\ \cite{JonathanInstab}, \cite{JonathanStabExt} for the uncharged model. 
	Even though in the present paper, we only focus on proving upper bounds for the Cauchy problem, the energy estimates that we carry out are conjectured to be \textbf{sharp} in the limit when the charge tends to zero, c.f.\ section \ref{conjecture}. Therefore, our work can also be seen as a first step towards the understanding of energy lower bounds, and consequently towards the resolution of strong cosmic censorship in spherical symmetry, for the Einstein--Maxwell-Charged-Scalar-Field model. While this philosophy has been successfully applied to uncharged fields on two-ended black holes \cite{JonathanStab}, \cite{JonathanStabExt}, the analogous question in the charged case, both for one-ended and two-ended black holes remains open.
	
	


	\paragraph{Outline}	The introduction is outlined as followed : first in section \ref{main} we describe a rough version of the main results, namely asymptotic decay estimates for small decaying data. Then in section \ref{review}, we review previous works on related topics. More specifically in section \ref{conjecture} we state the conjectured asymptotic behaviour of charged scalar fields on black holes.
	Then in section \ref{interior} we sum up the known results in the black hole interior for this model, as a motivation for the present study. Then in section \ref{chargedprevious} we review 
	previous results for the Maxwell-Charged-Scalar-Field equations on Minkowski space-time and for small energy data. Then section \ref{uncharged} deals with previous works on the wave equations on black holes space-times and strategies to reach Price's law optimal decay. 
	After in section \ref{methods} we explain the main ideas of the proof. This includes mostly the strategy to prove energy boundedness, integrated energy estimates and energy decay.
	Finally in section \ref{outline}, we outline the rest of the paper, section by section.

	\subsection{Simplified version of the main results} \label{main}
	
	We now give a first version of the main results. The formulation of this section is extremely simplified and a more precise version is available in section \ref{mainresult}. See also remark \ref{crucialremark} for important precisions.
	
	\begin{thm} \label{simplifiedthm}
		Consider  spherically symmetric regular data $(\phi_0,Q_0)$,  for the Maxwell-Charged-Scalar-Field equations on a sub-extremal Reissner-Nordstr\"{o}m exterior space-time of mass $M$ and charge $\rho$. Assume that $\phi_0$ and its derivatives decay sufficiently towards spatial infinity. \\

		If $Q_0$ and $\phi_0$ are small enough in appropriate norms then \textbf{energy boundedness} \eqref{energyintro} is true and an \textbf{integrated local energy} estimate \eqref{Morawetzestimateintro} holds. \\

		Also we can define the future asymptotic charge $e$  such that on any constant $r$ curve 
		
		$\gamma_{R_0} := \{ r=R_0\}$ for  $ r_+ \leq R_0 \leq +\infty$, we have $Q_{|\gamma_{R_0}}(t) \rightarrow e$ as t $\rightarrow  +\infty$.  \\

		Then we have \textbf{energy decay} : there exists $2<{\pe}<3$, with ${\pe} \rightarrow 3$ as $e \rightarrow 0$ and such that for all $R_0>r_+$, for all $u>0$  : 
		
		\begin{equation} \label{Sthm1}
			E(u)  + \int_{\gamma_{R_0} \cap [u,+\infty]}|\phi|^2  +  |D_v\phi|^2+ \int_{\mathcal{H}^+ \cap [v_R(u),+\infty] }|\phi|^2 +  |D_v\phi|^2 \lesssim u^{-{\pe}},
		\end{equation} where $\gamma_{R_0}:= \{ r=R_0\}$, $(u,v)$ are defined in section \ref{Coordinates}, $v_R(u)=u+R^*$ is defined in section \ref{foliations},  and $E$ 
		is defined in section \ref{energynotation}.

		We also have the following \textbf{point-wise decay}, for any $R_0>r_+$, for $u>0$, $v>0$:
		
		\begin{equation}   \label{Sthm2}
			r^{ \frac{1}{2}}|\phi|(u,v) + r^{ \frac{3}{2}}|D_v\phi|(u,v) \lesssim \left(\min \{u,v\}\right)^{-\frac{{\pe}}{2}},
		\end{equation}	\begin{equation}  \label{Sthm3}
			|\psi|_{|\mathcal{I}^+}(u) \lesssim  u^{\frac{1-{\pe}}{2}},
		\end{equation}		
		\begin{equation}  \label{SthmRS}
			|D_u\psi| \lesssim \Omega^2 \cdot  u^{\frac{-{\pe}}{2}},
		\end{equation}		 			\begin{equation}  \label{Sthm4}
			|D_v \psi|_{ |  \{ v \geq 2u+R^*\}}(u,v) \lesssim  v^{-\frac{1+{\pe}}{2}},
		\end{equation}
		\begin{equation}   \label{Sthm5} |Q - e| (u,v)	\lesssim	 u^{1-{\pe}}  1_{ \{ r \geq R_0\}}+v^{-{\pe}} 1_{ \{ r \leq R_0\}},
		\end{equation}	where $\psi:=r \phi$ denotes the radiation field and $Q$ is the Maxwell charge \footnote{The charge $Q$ from the Maxwell equation, should not be confused with the charge of the black hole $\rho$. Because the problem is not coupled with gravity, these are different objects, in contrast to the model of \cite{Moi}. For definitions, c.f.\ also section \ref{chargenotations}. } defined by $F_{u v} = \frac{2Q \Omega^2}{r^2}$.
	\end{thm}

	\begin{rmk} \label{crucialremark}
		In reality, this theorem ---which is a broad version of Theorem \ref{maintheorem} --- contains two different intermediate results. 
		
		The first one, Theorem \ref{boundednesstheorem}, proves simultaneously energy boundedness and the integrated local energy estimate, on condition that the charge $Q$ is \textbf{everywhere} smaller than a constant depending on the black hole parameters. In particular, this is the case if the $r$ weighted energies of the scalar field data and the initial charge are sufficiently small.

		The second one, Theorem \ref{decaytheorem}, proves the decay of the energy at a polynomial rate and subsequent point-wise estimates. The decay rate depends only on the dimensionless quantity $q_0e$ and is conjectured to be almost optimal when $q_0e \rightarrow 0$. This theorem only requires that the boundedness of the energy and the integrated local energy estimate are satisfied, together with the bound $q_0|e| < {0.04}$ (this explicit constant is obtained as the solution of an optimization problem: for further explanations c.f.\ the preamble of section \ref{sectionp=3} and Remarks \ref{p(e)remark} and \ref{proofsimple}). 
		In particular, this is the case if the limit of the initial charge $e_0$ verifies $q_0|e_0| < {0.04}$  and if the $r$ weighted energies of the scalar field data are sufficiently small. The charge smallness requirement   (which impacts  the decay rate) is explicit in these theorems and \underline{independent} of the black hole mass. This relates to the physical expectation that the decay rate depends only on the asymptotic charge, c.f.\ section \ref{conjecture}.

	\end{rmk}
	\begin{rmk}Note that, extracting from \eqref{Sthm1} we obtain a point-wise decay rate of the form $|\phi|_{|\mathcal{H}^+}(v) \lesssim v^{-s(e)}$ \textit{along a dyadic sequence}, with $s(e) \rightarrow 2$ as $q_0|e| \rightarrow 0$. In fact, both \eqref{Sthm1} and this point-wise estimate are conjectured to be sharp (at least for the $0$th order term of the Taylor expansion in $q_0e$), c.f.\ section \ref{conjecture}. While \eqref{Sthm2} is also conjectured to be sharp in a region of the form $\{ v \geq 2u+R^*\}$, it is not on the event horizon where \eqref{Sthm2} takes the form $|\phi|_{|\mathcal{H}^+}(v) \lesssim v^{-s'(e)}$ with $s'(e) \rightarrow \frac{3}{2}$ as $q_0|e| \rightarrow 0$. \color{black}
	\end{rmk}

	\begin{rmk}
		It is also possible to prove an alternative to \eqref{Sthm4}, c.f.\ Theorem \ref{decaytheoremmod} in Appendix \ref{moredecayappendix}. Essentially, if we require more point-wise decay of the initial data, we can prove that $r^2 \partial_v \psi$, $r^2 \partial_v(r^2 \partial_v \psi) $ etc.  admit a finite non-zero limit on $\mathcal{I}^+$, say in the gauge $A_v=0$. Such bounds are reminiscent of so-called peeling estimates.\\
	\end{rmk}

	Broadly speaking, the present paper contains three different ingredients,  which are new in the context of a charged scalar field on a black hole space-time\color{black}: a non-degenerate energy boundedness statement, an integrated local energy decay Morawetz estimate and a hierarchy of $r^p$ weighted estimates.
	
	The presence of an interaction between the Maxwell charge and the scalar field renders the problem non-linear, which makes the estimates very coupled. In particular, we need to prove  energy boundedness,  the Morawetz estimate, and the $r^p$ estimates hierarchy together.
	
	This coupling represents a major difficulty that we overcome proving that the charge is small. Even then however, the absorption of the interaction term requires great care. This is essentially due to the presence of a non-decaying quantity, the charge $Q$, that is not present for the uncharged problem or on Minkowski, c.f.\ sections \ref{chargedprevious}, \ref{uncharged} and to the criticality of the equations, c.f.\ section \ref{chargedprevious} and the introduction.

	
	\begin{rmk}
		
		Following the conjecture of \cite{HodPiran1}, one expects that outside spherical symmetry, the spherically symmetric mode dominates the late time behaviour like in the uncharged case, c.f.\ Remark \ref{conjectureSS}. One can hope that some of the main ideas of the present paper can be adapted to the case where no symmetry assumption is made.
		In particular, as it can be seen in the statement of Theorem \ref{decaytheorem}, the decay of the energy is \textbf{completely independent} of point-wise bounds which do not propagate easily without any symmetry assumption. 
	\end{rmk}

	\begin{rmk}
		One of the novelties of the present work is to give an asymptotic expansion of the decay rate in terms of $q_0e$, as $e \rightarrow 0$, c.f.\ the Taylor expansion \eqref{Taylorp}. This was not present in previous work  \cite{Bieri}, \cite{LindbladKG}, \cite{ShiwuKG} precisely because on Minkowski space-time $e=0$ so the long time effect of the charge is not as determinant.
	\end{rmk}
	
	
	\begin{rmk}
		It should be noted that everything said in the present paper also works for the case of a spherically symmetric charged scalar field on a Schwarzschild black hole, i.e.\ when $\rho=0$.
	\end{rmk}

	\subsection{Review of previous work and motivation} \label{review}
	
	The motivations to study the Maxwell-Charged-Scalar-Field model are multiple. 
	
	First charged scalar fields trigger a lot of interest in the Physics community, sometimes in connection with problems of Mathematical General Relativity, as an example see the recent \cite{Harvey}, which discusses strong cosmic censorship for cosmological space-times in the presence of a charged scalar field.
	
	Second, the difficulty of this problem, due to its criticality with respect to decay at space-like infinity, c.f.\ section \ref{chargedprevious} demands a certain robustness in the estimates. Such decay--critical problems are a common occurrence in many situations of interests in Geometry and Physics, such as the wave equation on metrics with conical singularities \cite{conical}, or the stability of the catenoid \cite{catenoid}. \color{black}  Therefore, our methods of proof may be re-used in different, potentially more complicated situations where traditional strategies are insufficient. 
	
	Finally, if the estimates of the paper can be transposed to the problem where the Maxwell and scalar fields are coupled with gravity (a problem which presents additional difficulties), this proves the stability of Reissner-Nordstr\"{o}m black holes against charged perturbations. 
	Additionally studying this model on black holes space-time can be considered as a first step towards understanding strong cosmic censorship and gravitational collapse for charged scalar fields in spherical symmetry, which may have different geometric characteristics from its uncharged analogue, see the introduction and section \ref{interior}. \\
	
	The goal of this section is to review previous works related to the problem of the present manuscript. This allows us to motivate the problem and to compare our results to what already exists in the literature.
	
	The latter is the sole object of section \ref{conjecture} where we express the conjectured asymptotic behaviour of charged scalar fields, obtained by heuristic considerations in \cite{HodPiran1}. 
	The former is the object of section \ref{interior}, which summarizes the conditions to apply the results of \cite{Moi}, as one of the main motivation for the study of the exterior we present. 
	
	The previous known results for this model are  essentially all proved on Minkowski space-time, although outside spherical symmetry. They either count crucially on conformal symmetries, a method that cannot be generalized easily to black hole space-times, or on the fact that the charge in the Maxwell equations has to tend to $0$ towards time-like infinity. This fact greatly simplifies the analysis on Minkowski but is not true on black hole space-times because the charge asymptotes a finite generically non-zero value $e$. These works are discussed in section \ref{chargedprevious}.
	
	Finally in section \ref{uncharged}, we review some of the numerous works on wave equations on black holes space-time. This is the uncharged analogue $q_0=0$ of the equations we study. This is the occasion to review the $r^p$ method which is central to our argument and its application to proving exact Price's law tail in \cite{Latetime}.
	
	
	\subsubsection{Conjectured asymptotic behaviour of weakly charged scalar fields on black hole space-times} \label{conjecture}
	
	We now state the expected asymptotics for a charged scalar field on the event horizon of asymptotic flat black hole space-times. The Physics literature is surprisingly scarce. We base this section on \cite{HodPiran1}, which is a heuristics-based work and the subsequent papers of the same authors. They state inverse polynomial decay estimates for the charged scalar field on Reissner-Nordstr\"{o}m space-time when the asymptotic charge of the Maxwell equation is arbitrarily close to $0$.
	
	It is argued that the limiting decay rate $\phi(r,t) \approx t^{-2+\epsilon(e)}$ with $\epsilon(e) \rightarrow 0$ as $q_0|e| \rightarrow 0$ (as opposed to $\phi(r,t) \approx t^{-3}$ in the uncharged case, by Price's law \cite{Pricepaper}, see next paragraph) \color{black} is a consequence of multiple scattering already present in flat space-time. Therefore, they suggest that the limit decay rate on black holes space-time ---when the charge tends to zero--- is the same as on Minkowski, which is consistent with the best decay rate on a constant $r$ curves found in \cite{LindbladKG}, c.f.\ section \ref{chargedprevious}. 
	
	
	This slow decay stands in contrast to the faster rate prescribed by Price's law for uncharged perturbations, c.f.\ Theorem \ref{PriceLawtheorem}. This is because for charged perturbations, the curvature term coming directly from the black hole metric decays faster than the term proportional to the charge of the scalar field. The work \cite{HodPiran1} was one of the first to notice, in the language of Physics, that charged scalar hairs decay slower than neutral ones.
	
	A heuristic argument  based on the wave equation \eqref{wavev} explains why the limit decay rate is $2$: to understand we state the charged scalar field equations in spherically symmetry on a $(M,\rho)$ Reissner-Nordstr\"{o}m background 
	
	$$	D_u(D_v \psi) = \frac{\Omega^2}{r^2} \psi \left(iq_0 Q -  \frac{2M}{r}+ \frac{2\rho^2}{r^2}\right) \approx \frac{iq_0 Q}{r^2} \psi +O(r^{-3}),\color{black}$$
	while its uncharged analogue when $q_0=0$ is
	
	$$	\partial_u(\partial_v \psi) = \frac{\Omega^2}{r^3}\psi \left(  -2M+ \frac{2\rho^2}{r}\right)\approx -\frac{2M }{r^3} \psi +O(r^{-4}).\color{black}$$
	
	If we expect that the radiation field $\psi$ and the charge $Q$ are bounded, we can infer that the charged equations may give a $t^{-2}$ decay of $\phi$ on a constant $r$ curve for compactly supported data, while the uncharged analogue gives $t^{-3}$ decay, because of the different $r$ power. This is because for asymptotically flat hyperbolic problems, we expect the $r$ decay  at spatial infinity to be translated into $t$ decay towards time-like infinity.
	
	It should also be noted that, since higher angular modes are expected to decay \textbf{faster} than the spherically symmetric average, the decay of scalar fields \underline{without any symmetry assumption} is expected to be the same. 
	
	Another interesting discovery made in \cite{HodPiran1} is the oscillatory behaviour of the scalar field on the event horizon. This is connected to the fact that the scalar field is complex, otherwise the dynamics of the Maxwell equation is trivial in spherical symmetry \color{black} as it can be seen in equation \eqref{Maxwelleq}. This makes the proof of decay more delicate than in the uncharged case, c.f.\ section \ref{uncharged}. The result of \cite{HodPiran1} can be summed up as follows :

	\begin{conjecture} [Asymptotic behaviour of weakly charged scalar fields, \cite{HodPiran1}] \label{Priceconj} Let $ (\phi,F)$ be a solution of the Maxwell-Charged-Scalar field system \textbf{with no particular symmetry assumption} on a Reissner-Nordstr\"{o}m space-time, and define the Maxwell charge $Q$ with the relationship $ F_{u v} = \frac{2Q \Omega^2}{r^2}$.
		
		Suppose that the data for $|\phi|$ is sufficiently decaying towards spatial infinity.
		
		Denote the asymptotic charge $e= \lim_{v \rightarrow +\infty} Q_{|\mathcal{H}^+}(v)$. 
		
		The gauge is fixed to be $$ A_u = A_v=\frac{-e}{2r},$$ where $(u,v)$ are standard Eddington--Finkelstein coordinates.
		\color{black}		
		Let $\epsilon>0$. Then there exists $\delta>0$, such if $q_0|e| <\delta$ we have \footnote{\eqref{Pricelaw} has to be understood as $\lim_{u \rightarrow +\infty} \phi(u,v)e^{ -iq_0e \frac{r^*}{r_+}} \sim_{v\rightarrow +\infty} \Gamma_0 \cdot v^{-2+\eta}$ and \eqref{Pricelaw2} has to be understood as $\lim_{v \rightarrow +\infty} \psi(u,v)v^{ iq_0e} \sim_{u\rightarrow +\infty} \Gamma'_0 \cdot  u^{-1+\eta+iq_0e}$. }on the event horizon of the black hole, parametrized by an advance time coordinate $v$, as defined in section \ref{Coordinates}: there exists $\psi _{|\mathcal{I}^+}(u)$ and $\psi _{|\mathcal{H}^+}(v)$ such that
		\begin{equation} \label{Pricelaw0}
			\lim_{v \rightarrow +\infty}	v^{ iq_0e } \psi(u,v)=	\psi _{|\mathcal{I}^+}(u) ,
		\end{equation}\begin{equation} \label{Pricelaw-1}
			\lim_{u \rightarrow +\infty}	 e^{ iq_0e \frac{u}{r_+}} \phi(u,v)=	\phi _{|\mathcal{H}^+}(v),
		\end{equation}	\color{black}
		\begin{equation} \label{Pricelaw2}
			\psi _{|\mathcal{I}^+}(u) \sim \Gamma'_0 \cdot u^{ iq_0e }\color{black}  \cdot u^{-1+\eta(q_0e)},
		\end{equation}			\begin{equation} \label{Pricelaw}
			\phi _{|\mathcal{H}^+}(v) \sim \Gamma_0 \cdot e^{ iq_0e \frac{v}{r_+}} \cdot  v^{-2+\eta(q_0e)},
		\end{equation}
		\begin{equation} \label{Pricelaw3}
			\phi_{|\gamma}(t) \sim \Gamma''_0 \cdot 	t^{ iq_0e} \cdot  v^{-2+\eta(q_0e)},\end{equation} as $v \rightarrow + \infty$ and	where $\Gamma_0$, $\Gamma'_0$, $\Gamma''_0$ are constants, $0<\eta(q_0e)< \epsilon$ and $\gamma$ is a far-away curve on which $\{ t \sim u \sim v \sim r \}$, defined in section \ref{section2}. Moreover the energy on the $V$ foliation of section \ref{foliations} behaves like 
		
		\begin{equation} \label{PricelawEnergy}
			E(u) \sim E_0 \cdot u^{-3+2 \eta(q_0e)},
		\end{equation}	where $E_0$ is a constant.
		
	\end{conjecture}
	
	\begin{rmk}
		Note that the decay of the charged scalar field depends on the dimension-less quantity $q_0e$ only.
	\end{rmk}
	\begin{rmk} \label{conjectureSS}
		Notice also that the conjecture of \cite{HodPiran1} does not imply any spherical symmetry assumption. It is actually argued that --- due to a better decay of the higher angular modes --- generic solutions decay as the same rate as spherically symmetric ones.
	\end{rmk}

	The present paper mostly solves this conjecture: in particular we are able to prove the upper bounds corresponding to \eqref{Pricelaw2}, \eqref{Pricelaw3} and \eqref{PricelawEnergy}. Moreover, we give a Taylor expansion $\eta(q_0e) = O(\sqrt{q_0|e|})$ (see also \eqref{Taylorp} for a more precise expansion).
	
	The upper bound corresponding to \eqref{Pricelaw} is more difficult to prove due to the degeneration of $r$ weights in the bounded $r$ region. We indeed prove
	
	$$ r^{ \frac{1}{2}}  |\phi|(u,v) \lesssim  (E(u))^{ \frac{1}{2}} \lesssim  u^{ -\frac{3}{2}+\eta(q_0e)}, $$
	
	which gives the optimal $t^{ -2+\eta(q_0e)}$ decay on $\gamma$ but only  $v^{ -\frac{3}{2}+\eta(q_0e)}$ on $\mathcal{H}^+$. While this issue can be circumvented for the uncharged scalar field using the better decay of the derivative, such a strategy is out of reach here, c.f.\ section \ref{uncharged}. Nonetheless, the strictly integral point-wise estimates on the event horizon are sufficient for future work on the interior of black holes, c.f.\ section \ref{interior}. \\
	
	The decay of charged scalar fields should be compared to the uncharged case prescribed by Price's law:

	\begin{thm} [Price's law, \cite{Latetime}, \cite{PriceLaw}, \cite{JonathanStabExt}, \cite{Tataru2}, \cite{Pricepaper}, \cite{Tataru}] \label{PriceLawtheorem}
		
		Let $\phi$ be a finite energy solution of the wave equation on a Reissner-Nordstr\"{o}m space-time.
		
		Then $\phi$ is bounded everywhere on the space-time and \textbf{generically} decays in time on the event horizon parametrized by the coordinate $v$ of section \ref{Coordinates} :
		
		\begin{equation} \label{pricelawMihalis}
			\phi_{|\mathcal{H}^+}(v)  \sim_{\infty} \Gamma_0  \cdot v^{-3},
		\end{equation}			where $\sim_{\infty}$ denotes the numerical equivalence relation as $v\rightarrow+\infty$ for two functions and their derivatives at any order and  $\Gamma_0 \neq 0$ is a constant. 
		
	\end{thm}

	\subsubsection{Towards the gravity coupled model in spherical symmetry} \label{interior}
	
	In this section, we mention some of the consequences of a future adaptation of this work to the case of the Einstein--Maxwell-Charged-Scalar-Field equations in spherical symmetry.
	
	The coupling of the Einstein equations with a scalar field has been well understood in prior works, starting from the proof of Price's law by  Dafermos and Rodnianski \cite{PriceLaw}, and culminating with the breakthrough of Luk and Oh \cite{JonathanStabExt} on Strong Cosmic Censorship for uncharged fields. In particular, many of techniques to capture decay known in the literature are robust to non-linear gravitational perturbations, especially in the spherically symmetric setting.
	Building up on the present article and the work on uncharged fields \cite{JonathanStabExt} provides a road map to proving decay for the Einstein--Maxwell--charged-scalar-field model. \color{black}
	
	If all carries through, this establishes in particular the \textbf{non-linear asymptotic stability of Reissner-Nordstr\"{o}m} space-time against small charged perturbations.
	
	Additionally, one can understand the \textbf{interior structure} of black holes for this model.
	To that effect, we briefly summarize the results of \cite{Moi} on the black hole interior for the Einstein--Maxwell-Charged-Scalar-Field equations in spherical symmetry. 
	In particular, if the results of the present paper can be extrapolated to the gravity coupled case, it would give interesting information on the structure of \textbf{one-ended black holes}, not modelled by the uncharged analogous equations. 
	
	The Einstein--Maxwell-Klein-Gordon equations, whose Einstein--Maxwell-Charged-Scalar-Field is a particular case where the constant $m^2=0$, can be formulated as 
	
	\begin{equation} \label{EMKG1} Ric_{\mu \nu}(g)- \frac{1}{2}R(g)g_{\mu \nu}= \mathbb{T}^{EM}_{\mu \nu}+  \mathbb{T}^{KG}_{\mu \nu} ,    \end{equation} 
	\begin{equation}  \label{EMKG2} \mathbb{T}^{EM}_{\mu \nu}=g^{\alpha \beta}F _{\alpha \nu}F_{\beta \mu }-\frac{1}{4}F^{\alpha \beta}F_{\alpha \beta}g_{\mu \nu},
	\end{equation}
	\begin{equation} \label{EMKG3}\mathbb{T}^{KG}_{\mu \nu}= \Re(D _{\mu}\phi \overline{D _{\nu}\phi}) -\frac{1}{2}(g^{\alpha \beta} D _{\alpha}\phi \overline{D _{\beta}\phi} + m ^{2}|\phi|^2  )g_{\mu \nu} , \end{equation} \begin{equation}  \label{EMKG4} \nabla^{\mu} F_{\mu \nu}=iq_{0} \frac{ (\phi \overline{D_{\nu}\phi} -\overline{\phi} D_{\nu}\phi)}{2} , \; F=dA ,
	\end{equation} \begin{equation}   \label{EMKG5}g^{\mu \nu} D_{\mu} D_{\nu}\phi = m ^{2} \phi .
	\end{equation}

	We then have the following result :

	\begin{thm}  [\cite{Moi}]  \label{mytheoremint}Let  $(M,g,F, \phi)$ a spherically symmetric solution of the Einstein--Maxwell-Klein-Gordon system. Suppose that for some $s > \frac{1}{2}$ the following holds, for a standard choice \footnote{This $v$ is the $v$ coordinate defined in section \ref{Coordinates}, although in the uncoupled case a gauge choice is necessary, c.f.\ \cite{Moi}.} of $v$ on the event horizon: 
		
		\begin{equation}  \label{pointwisehyp}
			|  \phi(0,v)| _{|\mathcal{H}^+} +| D_{v} \phi(0,v)| _{|\mathcal{H}^+} \lesssim v^{-s}.\end{equation}

		Then near time-like infinity, the solution remains regular \footnote{The Penrose diagram -locally near timelike infinity- of the resulting black hole solution is the same as Reissner--Nordstr\"{o}m's . }up to its Cauchy horizon.

		If in addition $s>1$ then the metric extends continuously across that Cauchy horizon.

		If moreover the following lower bound on the energy holds , 	for a $p$ such that  
		$ 0<2s-1 \leq p <  \min \{2s, 6s-3\}$ 
		
		\begin{equation}		\label{instabapparent1}	v^{-p} \lesssim {\frac{1}{2}\int_{v}^{+\infty} |D_v \phi|^2_{|\mathcal{H}^+}(v')dv'
				+   q_0 e\  [\frac{1}{r_-}-\frac{1}{r_+}]  \int_{v}^{+\infty} \Im(\bar{\phi}_{|\mathcal{H}^+}(v') D_v \phi_{|\mathcal{H}^+}(v')) dv' },\end{equation}

		{where $r_{\pm}(M,e) = M \pm \sqrt{M^2-e^2}$, with $(e,M)$ being the black hole charge and mass, with $0<|e|<M$,}	then the Cauchy horizon is $C^2$ singular \footnote{i.e.\ a $C^2$ invariant quantity blows up, namely $Ric(V,V)$ where $V$ is a null geodesic vector field transverse to the Cauchy horizon.}.
		Hence the metric is (locally) $C^2$ inextendible across the Cauchy horizon.

	\end{thm}
	
	Therefore, if the results of the present manuscript can be extended to the gravity coupled problem, it would imply the continuous extendibility of the metric, at least for small enough data.
	
	Moreover, if we assume that an energy boundedness statement and an integrated local energy estimate \footnote{Which we prove in the present paper for small enough data.} hold, then we can prove the continuous extendibility result for a larger class of initial data, namely for an initial asymptotic charge \footnote{We remind the reader that a definition of the initial asymptotic charge is present in section \ref{chargenotations}.} in the range $q_0|e_0| \in [0, {0.04})$ .
	
	This is relevant to the so-called Strong Cosmic Censorship Conjecture. Its $C^k$ formulation states that for generic admissible \footnote{For a precise formulation of Strong Cosmic Censorship, where the notion of admissible data is explained, we refer to \cite{JonathanStab}, \cite{JonathanStabExt}.} initial data, the maximal globally hyperbolic development is inextendible as a $C^k$ Lorentzian manifold. While its continuous formulation --- for $k=0$ --- is often conjectured, it has been disproved in the context of the Einstein--Maxwell-\textbf{Un}charged-Scalar-Field black holes in spherical symmetry, c.f.\ \cite{MihalisPHD}, \cite{Mihalis1}, \cite{PriceLaw} and more recently for the Vacuum Einstein equation with no symmetry assumption in \footnote{The authors of \cite{KerrStab} prove the continuous extendibility of the metric, assuming the widely-believed stability of Kerr black hole.} the seminal work \cite{KerrStab}. Roughly, the continuous formulation of Strong Cosmic Censorship is \underline{false} in these contexts because one expects dynamic rotating or charged black hole interiors to admit a Cauchy horizon over which tidal deformations are finite. Therefore, provided that the decay rate assumed in \eqref{pointwisehyp} can be proved for Cauchy data, the continuous extendibility result of Theorem \ref{mytheoremint} \underline{disproves} the continuous version of Strong Cosmic Censorship for the Einstein--Maxwell-\textbf{Charged-}Scalar-Field black holes in spherical symmetry.
	
	While the (strongest) continuous version of the conjecture is false, the (weaker) $C^2$ formulation of Strong Cosmic Censorship has been proven for Einstein--Maxwell-\textbf{Un}charged-Scalar-Field spherically symmetric black holes in the seminal works \cite{JonathanStab}, \cite{JonathanStabExt}. For the analogous  Einstein--Maxwell-\textbf{Charged}-Scalar-Field model, provided one can prove that \eqref{instabapparent1} holds for generic data, the result of Theorem \ref{mytheoremint} proves the $C^2$ inextendibility of the metric along a part of the Cauchy horizon, near time-like infinity. This should be thought of as a first step towards the proof of the $C^2$ formulation of Strong Cosmic Censorship in the charged case, c.f.\ \cite{Moi} for a more extended discussion.
	
	\begin{rmk} The case $k=2$, i.e.\ the $C^2$ inextendibility property of the metric is of particular physical interest. Indeed, classical solutions of the Einstein equations are considered in the $C^2$ class : therefore, a $C^2$ inextendibility property implies the impossibility to extend the metric as a classical solution of the Einstein equation. On the other hand, in principle the solution can still be extended as a \underline{weak} solution of the Einstein equation, for which we only require the Christoffel symbols to be in $L^2$.  An interesting but unexplored direction would be to prove the equivalent of Strong Cosmic Censorship in the Sobolev $H^1$ regularity class, that would also exclude extensions as weak solutions of the Einstein equation, in that sense.
	\end{rmk}

	For a discussion on other motivations to study the Einstein--Maxwell-Charged-Scalar-Field model in spherical symmetry, we refer to section 1.2.1 of \cite{Moi} and the pages 23-24, section 1.42 of \cite{KerrStab}.
	
	\subsubsection{Previous works on the Maxwell-Charged-Scalar-Field} \label{chargedprevious}
	
	We now turn to the study of the Maxwell-Charged-Scalar-Field model. It should be noted that, even in spherical symmetry, the equations are \textbf{non-linear}.

	Although this problem on Minkowski space-time has received a lot of attention in the last decade, it should be noted that the only quantitative and rigorous result for that model \footnote{More exactly, for the \textbf{Einstein}--Maxwell-Charged-Scalar-Field and in the context of spherical symmetry. } on black hole space-times was derived in \cite{Moi}, for the black hole interior (c.f.\ section \ref{interior}). 

	We also mention that \cite{Kommemi} contains many interesting preliminary results and geometric arguments for the Einstein--Maxwell--Klein--Gordon model in spherical symmetry, although no quantitative study is carried out, either in the interior or the exterior of the black hole.	In the rest of this section, we review previous works on flat space-time. \\
	
	The first step is to study global existence for variously regular data. This question has been extensively studied, starting with the global existence for smooth data by Eardley and Moncrief \cite{Eardley1}, \cite{Eardley2}. Various works have since made substantial progress in different directions \cite{Shatah}, \cite{KeelRoyTao}, \cite{KlainermanMachedon}, \cite{Krieger}, \cite{Krieger2}, \cite{Machedon}, \cite{OhTataru}, \cite{RodTao}.

	After global existence, the next natural question is to study the asymptotic behaviour of solutions, in particular decay. While this problem was pioneered by Shu \cite{Shu}, one of the first modern result is due to Lindblad and Sterbenz \cite{LindbladKG} who established point-wise inverse polynomial bounds for the scalar and the Maxwell Field, provided the data is small enough. More precisely they prove the following :

	\begin{thm}[Lindblad-Sterbenz, \cite{LindbladKG}] 
		
		Consider asymptotically flat energy Cauchy data, namely a scalar field/Maxwell form couple $(\phi_0,F_0)$ such that for some $\alpha>0$ and $s>0$, 
		
		$$\mathcal{E}^{EM+SF}:= \| r^{\alpha}\phi_0\|_{H^s(\mathbb{R}^3)}+ \| r^{\alpha}D_t \phi_0\|_{H^s(\mathbb{R}^3)}+ \| r^{\alpha}F_0\|_{H^s(\mathbb{R}^3)} < \infty.$$
		
		We also denote the asymptotic initial charge $e_0:= \lim_{r \rightarrow +\infty} \frac{r^2 F_{0}(\partial_u, \partial_v)}{2}$.
		
		For every $\epsilon>0$, there exists $\delta>0$ such that if 
		
		$$ \mathcal{E}^{EM+SF} < \delta $$
		
		then we have the following estimates, for $u$ and $v$ large enough 
		
		\begin{equation}
			|\phi|(u,v) \lesssim v^{-1} u^{-1+\epsilon},
		\end{equation}	\begin{equation}
			|\psi|(u,v) \lesssim  u^{-1+\epsilon},
		\end{equation}		 			\begin{equation}
			|D_v \psi|_{ |  \{ v \geq 2u +R^* \}}(u,v) \lesssim  v^{-2+\epsilon},
		\end{equation}
		\begin{equation} \label{chargeMink1}	 |Q|(u,v)	\lesssim |e_0| \cdot 1_{ \{ u \leq u_0(R)\}}	+ u^{ -1+\epsilon}.
		\end{equation}

	\end{thm}
	
	As explained in pages 9 and 10 of \cite{LindbladKG}, the Maxwell-Charged-Scalar-Field equations are \textbf{critical} with respect to decay at $r \rightarrow +\infty$. This makes the estimates very tight, as opposed for instance to the (uncharged) wave equation  equation on Schwarzschild/ Reissner--Nordstr\"{o}m (and even the Einstein-scalar-field system in spherical symmetry) which allow for more room\color{black}. This is due to the presence of a non-linear term that scales exactly like the dominant terms (see section\ref{conjecture}) in the energy while its sign cannot be controlled.
	This very fact relates to the dependence of the decay rate on the asymptotic charge $e$, as we explain in section \ref{main}. 
	
	To overcome this difficulty, the authors of \cite{LindbladKG} make use of a conformal energy, a fractional Morawetz estimate and some $L^2/ L^{\infty}$ Stricharz-type estimates. These arguments rely on the specific form of the Minkowski metric and are difficult to transpose to any black hole space-time. We also mention the works  \cite{Bieri} 
	and \cite{Kauffman}.

	Recently, this problem has been revisited by Yang \cite{ShiwuKG} using the modern $r^p$ method invented in \cite{RP} to establish decay estimates.
	While the proven decay was weaker than that of \cite{LindbladKG}, it made the decay of energy more explicit. More importantly, the proof requires weaker hypothesis. In particular, while the scalar field initial data need to be small, the Maxwell field is allowed to be large. This is summed up by :
	
	\begin{thm}[Yang, \cite{ShiwuKG}]
		
		Consider asymptotically flat energy Cauchy data, namely a scalar field/Maxwell form couple $(\phi_0,F_0)$ such that for some $\alpha>0$ and $s>0$, 
		
		$$\mathcal{E}^{SF}:= \| r^{\alpha}\phi_0\|_{H^s(\mathbb{R}^3)}+ \| r^{\alpha}D_t \phi_0\|_{H^s(\mathbb{R}^3)},$$
		$$ \mathcal{E}^{EM} := \| r^{\alpha}F_0\|_{H^s(\mathbb{R}^3)} < \infty.$$
		
		We also denote the asymptotic initial charge $e_0:= \lim_{r \rightarrow +\infty} \frac{r^2 F_{0}(\partial_u, \partial_v)}{2}$.
		
		For every $\epsilon>0$, there exists $\delta>0$ such that if 
		
		$$ \mathcal{E}^{SF} < \delta $$
		
		then we have the following estimates, for $u$ and $v$ large enough 
		
		\begin{equation}
			r^{ \frac{1}{2}}|\phi|(u,v) \lesssim u^{-1+\epsilon},
		\end{equation}	\begin{equation}
			|\psi|(u) \lesssim  u^{-\frac{1}{2}+\epsilon},
		\end{equation}		 			\begin{equation}
			|D_v \psi|_{ |  \{ v \geq 2u+R^*\}}(u,v) \lesssim  v^{-1+\epsilon},
		\end{equation}
		\begin{equation} \label{chargeMink2}	 |Q - e_0 \cdot 1_{ \{ u \leq u_0(R)\}}|(u,v)	\lesssim	u^{ -1+\epsilon},	
		\end{equation}
		\begin{equation}	E(u) \lesssim u^{-2+2\epsilon},
		\end{equation}
		
		where $E(u)$ is the energy of the scalar field, similar to the one defined in section \ref{energynotation}, and $u_0(R)$ is defined in section \ref{foliations}.
	\end{thm}

	While the $r^p$ method utilized by Yang is very robust, his work cannot be generalized easily to the case of black hole space-times. This is because on Minkowski space-time, the long range  effect of the charge manifests itself only in the exterior of a fixed forward light cone, as it can be seen in estimates \eqref{chargeMink1}, \eqref{chargeMink2}. In contrast to the charge on black hole space-times that always admits a \footnote{The future asymptotic charge $e$ , defined as the limit of $Q$ towards infinity on $ \mathcal{H}^+$,  $\mathcal{I}^+$ or any constant $r$ curve.} limit  $e$, the charge on Minkowski space-time tends to $0$ towards time-like infinity. As a result \footnote{In \cite{LindbladKG} and \cite{ShiwuKG}, it has been argued that it is not sufficient to study compactly supported data to understand how decaying data behave. This is because the main charge term cancels for the former and not for the latter. This fact is not any longer true on a black hole space-time.  }, studying compactly supported initial data on a black hole space-time is not a priori more difficult than studying data that decay sufficiently towards spatial infinity.
	The strategy employed in  \cite{ShiwuKG} is to count on the $u$ decay of the Maxwell term to absorb the interaction terms in the $r^p$ weighted estimate for the scalar field. However, because of the existence of a \textbf{non-zero} asymptotic charge, this fails on any black hole space-time and instead we need to rely on the smallness of the charge in the present paper.
	
	It should be noted that these works \cite{Bieri}, \cite{LindbladKG}, \cite{Shu}, \cite{ShiwuKG} are treating the Maxwell-Charged-Scalar-Field outside spherical symmetry, which makes the dynamics of the Maxwell field much richer than for the symmetric case considered in the present paper where the Maxwell form is reduced to the charge.
	
	We also mention the extremely recent \cite{YangYu} extending the results of \cite{ShiwuKG}, with better decay rates and remarkably for large Maxwell field and large scalar field.
	
	\subsubsection{Previous works on wave equations on black hole space-times and $r^p$ method} \label{uncharged}
	
	The wave equation on black hole space-times has been an extremely active field of research over the last fifteen years, c.f.\ \cite{Newrp}, \cite{Latetime} \cite{PriceLaw}, \cite{Redshift}, \cite{RP}, \cite{SlowKerr}, \cite{KerrDaf}, \cite{JonathanInstab}, \cite{JonathanStabExt}, \cite{Moschidisrp}, \cite{Tataru} to cite a few. This is the uncharged analogue --- when $q_0=0$  --- of the problem we study in the present paper. 
	
	It is related to one of the main open problems of General Relativity, the  question of black holes stability for the Einstein equations c.f.\ \cite{MihalisStabExt}, \cite{Andras}, \cite{Klainerman} for some recent remarkable advances in various directions. 
	
	Subsequently, the black holes interior structure could be inferred from the resolution of this problem, c.f.\ \cite{MihalisPHD}, \cite{Mihalis1}, \cite{KerrStab}, \cite{JonathanStab}, \cite{Moi}. In addition to the analyst curiosity to understand the wave equation in different contexts, these works also aim at exploring toy models. This may give valuable insight on the mentioned problems coming from Physics. This is also one of the goals of the present paper. 
	
	In this section, we review some results that are related to the decay of scalar fields on spherically symmetric space-times, which is the uncharged version of the model considered in the present manuscript. We are going to mention in particular the different uses of the new $r^p$ method, pioneered in \cite{RP}. \\

	After the broader discussion of section \ref{interior},	we would like to emphasize how the quantitative late time behaviour of scalar fields impacts the geometry of black holes.			
	
	This was first understood in \cite{MihalisPHD}, \cite{Mihalis1} in the context of Einstein--Maxwell-Uncharged-Scalar-Field in spherical symmetry. It is proved that Price's law of Theorem \ref{PriceLawtheorem} implies that generic black holes for this model possess a Cauchy horizon over which the metric is continuously extendible.
	Therefore, from \cite{PriceLaw}, the continuous formulation of Strong Cosmic Censorship conjecture is false for this model.
	
	This insight, provided by the toy model, gave a good indication about black holes that satisfy the Einstein equation without symmetry assumptions. This is best illustrated by the remarkable and recent work of Dafermos and Luk \cite{KerrStab}, where it is proven that the decay of energy-like quantities on the event horizon implies the formation of a $C^0$ regular Cauchy horizon, \textbf{with no symmetry assumption}.


	Once the first step --- namely understanding (almost) sharp upper bounds--- has been carried out, the next step is to understand lower bounds. This is the object of the work \cite{JonathanStabExt} for the Einstein--Maxwell-Uncharged-Scalar-Field in spherical symmetry, in which  $L^2$ lower bounds are proved on the event horizon. Then, in \cite{JonathanStab}, it is shown that these lower bounds propagate to the interior of the black hole. The result implies that generic black holes possess a $C^2$ singular \footnote{Precisely, there exists a geodesic vector field $\partial_V$ that is transverse to the Cauchy horizon and regular, so that $Ric(\partial_V,\partial_V)=\infty$.} Cauchy horizon. This means that the $C^2$ formulation of Strong Cosmic Censorship conjecture is true for this model. Note that the upper bounds of \cite{PriceLaw} are extremely useful in the instability proof of  \cite{JonathanStabExt}.

	The main intake of this short review is the idea that the fine geometry of the black hole is determined by the decaying quantities on the event horizon, which makes the study of black hole exteriors all the more important. The problem gathers both \footnote{This apparent paradox is resolved once one realizes that stability estimates are proven for a very weak norm, whereas instability estimates originate from a blow-up of stronger norms. c.f.\ \cite{JonathanStab} and \cite{Moi}.} stability and instability features (c.f.\ also \cite{Moi}) : stability of scalar fields implies the formation of a $C^0$ regular Cauchy horizon while instability ensures its $C^2$ singular nature. \\

	Now we start a short review of the $r^p$ method from \cite{RP}. This should be thought of as a new vector field method, which makes use of $r$ weights as opposed to conformal vector fields that were used more traditionally, like in \cite{LindbladKG}. The objective is to prove that energy decays in time. Nevertheless, it is a well-known fact that the energy on constant time slice is constant and \textbf{does not} decay. 
	The idea is to consider instead the energy on a $J$-shaped foliation $ \Sigma_{\tau} = \{r \leq R, \tau' = \tau \} \cup  \{ r \geq R, u=u_{R}(\tau) \}$ as depicted in \cite{RP} or \cite{ShiwuKG}. We can then establish a hierarchy of $r^p$ weighted energy to obtain time decay, using the pigeon-hole principle.
	
	In this article we are going to consider a $V$-shaped foliation instead $ \mathcal{V}_u = \{r \leq R, v=v_R(u) \} \cup  \{ r \geq R, u'=u \}$, c.f.\ Figure \ref{Fig1} and section \ref{foliations}. This is purely for the sake of simplicity and does not change anything. 
	
	\begin{thm} [$r^p$ method for $0 \leq p \leq 2$, Dafermos-Rodnianski, \cite{RP}]	Let $\phi$ be a finite energy solution of 	$$ \Box_{g} \phi = 0,$$	where $g$ is a Schwarzschild exterior metric. The following hierarchy of $r^p$ weighted estimates is true : for all $u_0(R) \leq u_1 < u_2$ :$$ \int_{u_1}^{u_2} E_{1}[\psi](u) du + E_{2}[\psi](u_2) \lesssim  E_{2}[\psi](u_1) , $$$$ \int_{u_1}^{u_2} E_{0}[\psi](u) du + E_{1}[\psi](u_2) \lesssim  E_{1}[\psi](u_1) ,$$where $E_{q}[\psi]$ is defined in section \ref{energynotation}. Therefore \footnote{We must then crucially make use of a Morawetz estimate and of the energy boundedness on Schwarzschild space-time. } the following estimates are true :  \begin{equation}	 r^{ \frac{1}{2}}|\phi|_{ | \{r \geq R \} }(u,v) \lesssim u^{-1},	 \end{equation}	\begin{equation}	 |\psi|(u,v) \lesssim  u^{-\frac{1}{2}},	 \end{equation}			 	\begin{equation}	 	E(u) [\phi]\lesssim  u^{-2},	 	\end{equation}		 where $E(u)=E(u) [\phi]$ is defined in section \ref{energynotation}.				\end{thm}
	
	The method was then subsequently extended to the case of n-dimensional Schwarzschild black holes for $n\geq3$ 
	by Schlue in \cite{Volker}. The main novelty is the existence of a better energy decay estimate for $\partial_t \phi$ which proves a better point-wise for $\phi$ as well : for all $\epsilon>0$
	
	\begin{equation}
		r^{ \frac{1}{2}}|\phi|_{ | \{r \geq R \} }(u,v) \lesssim u^{-\frac{3}{2}+\epsilon},
	\end{equation}	\begin{equation}
		|\psi|(u,v) \lesssim  u^{-1+\epsilon},
	\end{equation}		
	\begin{equation} \label{ameliore1}
		E(u) [\partial_t\phi]\lesssim  u^{-4+2\epsilon}. \end{equation}
	
	The work of Moschidis \cite{Moschidisrp}, which generalizes the $r^p$ method and point-wise decay estimates to a very general class of space-times, should also be mentioned.

	The $r^p$ hierarchy has been subsequently extended to $p<5$ in \cite{Newrp}. This has led to the proof of the (almost) optimal decay for the scalar field and its derivatives on Reissner--Nordstr\"{o}m space-time.

	\begin{thm} [$r^p$ method for $0 \leq p < 5 $, Angelopoulos-Aretakis-Gajic, \cite{Newrp}] \label{AAGthm1}
		
		Let $\phi$ be a compactly supported  solution of  

		$$ \Box_{g} \phi = 0,$$
		
		where $g$ is a sub-extremal Reissner--Nordstr\"{o}m exterior metric. 
		
		For all $0 \leq p <5$,	the following hierarchy of $r^p$ weighted estimates is true : for all $u_0(R) \leq u_1 < u_2$ : 
		
		$$ p\int_{u_1}^{u_2} E_{p-1}[\psi](u) du + E_{p}[\psi](u_2) \lesssim E_{p}[\psi](u_1) . $$
		
		For all $2 \leq q <6$,	the following hierarchy of $r^q$ weighted estimates is true : for all $u_0(R) \leq u_1 < u_2$ : 
		
		$$ \int_{u_1}^{u_2} E_{q-1}[\partial_t \psi](u) du  \lesssim E_{q}[\partial_r \psi](u)+E_{q-2}[\psi](u) + E(u)[\phi] + E(u)[\partial_t \phi].$$
		
		For all $6 \leq s <7$,	the following hierarchy of $r^s$ weighted estimates is true : for all $u_0(R) \leq u_1 < u_2$ :
		
		$$ \int_{u_1}^{u_2} E_{s-1}[\partial_r \psi](u) du  + E_{s}[\partial_r \psi](u_2) \lesssim E_{s}[\partial_r \psi](u)+E_{s-2}[ \psi](u) + E(u)[\phi] + E(u)[\partial_t \phi].$$ 
		
		Therefore for all $\epsilon>0$, we have the following energy decay

		\begin{equation}				E(u)[\phi] \lesssim v^{-5+\epsilon},
		\end{equation}
		\begin{equation}	\label{ameliore2}			E(u)[\partial_t \phi]\lesssim v^{-7+\epsilon}.
		\end{equation}

	\end{thm}
	
	\begin{rmk} This gives an alternative proof of estimate \ref{pricelawMihalis} for the linear problem , i.e.\ the wave equation on a  Reissner--Nordstr\"{o}m background.
		Actually, after commuting twice with $\partial_t$, one can also obtain the estimate $|\partial_v \phi|_{\mathcal{H}^+} \lesssim v^{-4+\epsilon}$. This is a better estimate than in \cite{PriceLaw}, although obtained only for the linear problem.
	\end{rmk}
	
	This strategy to prove the almost optimal energy decay is the first step towards understanding lower bounds. The second step, carried out in \cite{Latetime}, is to identify a conservation law that allows for precise estimates. 
	
	\begin{cor}  [Angelopoulos-Aretakis-Gajic, \cite{Latetime}]
		With the same hypothesis as for Theorem \ref{AAGthm1}, for every $r_0>r_+$, where $r_+$ is the radius of the black hole, there exists a constant $C>0$ and $\epsilon>0$ such that on a $\{ r= r_0\}$ curve : 
		
		$$ \phi(r_0,t) = \frac{C}{t^3} + O(t^{-3-\epsilon}).$$

	\end{cor}

	In the present paper and although we do not use any of the techniques of \cite{Newrp} and \cite{Latetime}, we intend to pursue the same program for the \textbf{non-linear} Maxwell-charged scalar field equations. In our case the decay mechanism is more complicated, in particular the decay rate is not universal and depends on the asymptotic charge $e$. 
	
	We find that the maximum $p$ for which we can derive a hierarchy of $r^p$ weighted energies (with no loss) is $2< p(e) <2+\sqrt{1-4q_0|e|}$, described in \eqref{Taylorp}. Since $p(e) \rightarrow 3$ as $e \rightarrow 0$, we reach the optimal energy decay rate predicted by \cite{HodPiran1} as $q_0e$ tends to $0$, at least for the energy. However, we cannot retrieve the optimal point-wise bound on the event horizon from this : our method only proves $|\phi(v)| \lesssim v^{-s}$, $s \rightarrow \frac{3}{2}$ as $q_0e$ tends to $0$. This is because it is not clear whether $E(u)[D_t \phi]$ enjoys a better $u$ decay than $E(u)[\phi]$, in contrast to the uncharged problem as it can be seen in equations \eqref{ameliore1}, \eqref{ameliore2}. 
	
	The physical explanation behind this phenomenon is the presence of an oscillatory term, because the scalar field is complex, as explained in \cite{HodPiran1}. This makes the analysis more difficult. Therefore, while in the uncharged problem $\phi$ and $\partial_t \phi$ decay like $v^{-3}$ and $v^{-4} $ respectively, it is not certain that an analogous property is true in the charged case.
	
	Notice however that the bounds we prove on the charge ---depending only on the energy--- and the bounds towards null infinity are expected to be almost sharp. Moreover, the point-wise bounds on the event horizon are always strictly integrable, which is crucial to study the interior of the black hole, see section \ref{interior}.
	

	\subsection{Methods of  proof} \label{methods}

	We now briefly discuss some of the main ideas involved in the proofs.

	\subsubsection{Smallness of the charge}

	During the whole paper, we require $q_0Q$ to be smaller than some constant.	This smallness originates from that of the initial data. More precisely, we prove that provided the asymptotic initial charge $e_0$ and of the $r$ weighted initial scalar field energies are small, then so is the charge, everywhere.
	
	We explain heuristically why this is the case.	Schematically, the Maxwell equation looks like 
	
	$|\partial Q|  \lesssim r^2|\phi| |D\phi|$. Then by using Cauchy-Schwarz with some Hardy inequality to control the zero order term in $|\phi|$, we see that the charge difference is roughly controlled by the $r$ weighted energy $\tilde{E}_{1}$ (c.f.\ section \ref{energynotation} for a precise definition), which is itself bounded by its initial value $\tilde{\mathcal{E}}_1$. Therefore, broadly $|Q-e_0| \lesssim \tilde{\mathcal{E}}_1$ so $Q$ is small if both $|e_0|$ and $\tilde{\mathcal{E}}_1$ are small.
	
	This issue relying essentially on the boundedness of various energy-like quantity by the initial data, an estimate of the form $\tilde{E}_{1}(u) \lesssim \tilde{\mathcal{E}}_1$ suffices.

	This is sensibly easier to prove than an estimate of the form $\tilde{E}_{1}(u_2) \lesssim \tilde{E}_{1}(u_1)  $ for all $u_1 <u_2$, which we prove later and is required to prove decay, with the use of a pigeon-hole like argument.
	
	This is why this charge smallness step is carried out first, as a preliminary estimate, before the much more precise versions later required to prove decay. This first step carried out in section \ref{CauchyCharac} should be thought of as the analogue of a boundedness proof.
	
	Doing so, we reduce the Cauchy problem to a characteristic initial value problem on the double null surface $ \{ u=u_0(R)\} \cup \{ v=v_0(R)\}$. Therefore, section \ref{CauchyCharac} is also the only part of the paper where estimates cover the whole space-time, including the region $\{ u \leq u_0(R)\}  \cup \{ v \leq v_0(R)\}$. In later sections dealing with decay, we use the results of this first part and only consider a characteristic initial value problem on the domain $\mathcal{D}(u_0(R),+\infty)$, the complement of $\{ u < u_0(R)\}  \cup \{ v < v_0(R)\}$.

	\subsubsection{Overview of decay estimates} \label{overviewdecay}
	
	Now we turn to the core of the present paper :  decay estimates. As explained earlier, they rely on \begin{enumerate}
		
		\item \label{1O} Degenerate energy boundedness, c.f.\ section \ref{degoverview}.
		\item \label{2O} An integrated local energy decay, also called a Morawetz estimate, see section \ref{morawetzoverview}.
		\item \label{3O} Non-degenerate energy boundedness, using the red-shift effect c.f.\ section \ref{RSoverline}.
		\item \label{4O} A hierarchy of $r^p$ weighted estimates, see section \ref{rpoverview} .
		\item \label{5O} A pigeon-hole principle like argument from which time decay can be retrieved for the un-weighted energy, using the last three estimates.
		\item \label{6O} Point-wise decay estimates, using crucially the energy decay, c.f. section \ref{pointoverview}.
	\end{enumerate}	
	
	Steps \ref{1O}, \ref{2O}, \ref{3O} and \ref{4O} are inter-connected and must be carried out all simultaneously, in contrast to the uncharged case (the wave equation) c.f.\ \cite{BH}. Step \ref{4O} and \ref{5O} are also connected and moreover rely crucially on the results of steps \ref{2O} and \ref{3O}. The last step \ref{6O} requires the results of \ref{4O} and \ref{5O} together with some additional point-wise decay hypothesis of the scalar field Cauchy data.

	The distinction between the degenerate and non-degenerate energy is due to the causal character of $\partial_t$, the Killing vector field which allows for energy conservation. While on Minkowski space-time $\partial_t$ is everywhere time-like, this is not true on black hole space-times since $\partial_t$ then becomes null on the event horizon. For this reason, the energy conserved by $\partial_t$ is called degenerate. To obtain the non-degenerate energy (the most natural to consider) the use of red-shift estimates is required, c.f.\ section \ref{RSoverline} for a more precise description.
	
	To prove time decay of the energy (the main objective of this paper) we use the $r^p$ method : from the boundedness of the $r^p$ weighted energies, one can roughly retrieve time decay $t^{-p}$ of the un-weighted energy.

	For steps \ref{1O} to \ref{4O}, we mainly use the vector field method, which is a robust technique to establish $L^2$ estimates with the use of geometry-inspired vector fields and the divergence theorem, c.f.\ Appendix \ref{appendix}. The principal difficulty,  when we apply the energy identity to a vector field $X$, is to absorb an interaction \footnote{This term is the source of the criticality with respect to $r$ decay that we describe in the introduction and in section \ref{chargedprevious}.} term between the scalar field and the electromagnetic part of the form 
	
	$$ \frac{q_0Q}{r^2} \Im \left( \phi ( X^v \overline{D_v \phi}-X^u \overline{D_u \phi}) \right) .$$ It comes from the second term of the identity $	\nabla^{\mu}( \mathbb{T}_{\mu \nu}^{SF} X^{\nu})= \mathbb{T}_{\mu \nu}^{SF} \Pi_{X}^{\mu \nu}+ F_{\mu \nu } X^{\mu} \mathcal{J}^{\nu}(\phi)$, c.f.\ \eqref{current}.

	This term must be absorbed by a controlled quantity to close the energy identities of steps \ref{1O} to \ref{4O}. However it has the same $r$ weight as the positive main term controlled by the energy. The strategy is then to apply Cauchy-Schwarz inequality to turn this interaction term into a product of $L^2$ norms. Thereafter we use the smallness of $q_0Q$ to absorb it. For a more precise description, c.f.\ for instance section \ref{rpoverview}.

	Because the main term controlled by the energy is proportional to $|D \phi|^2$, we also need to use Hardy-type inequalities throughout the paper, to absorb any term proportional to $|\phi|^2$. These estimates are proven in section \ref{section2}. More details on each step are provided in the subsequent sub-sections.

	\subsubsection{Degenerate energy boundedness} \label{degoverview}
	
	We define the degenerate energy of the scalar field on our $\mathcal{V}_u$ foliation by  	$$ E_{deg}(u)= \int_{u}^{+\infty} r^2 |D_u \phi|^2(u',v_R(u))du'+\int_{v_R(u)}^{+\infty} r^2 |D_v \phi|^2(u,v)dv, $$
	
	c.f.\ section \ref{foliations} and section \ref{energynotation}. We want to prove an estimate of the form $E_{deg}(u_2) \lesssim E_{deg}(u_1)$ for any $u_1 < u_2$. For this, we make use of the Killing vector field $T=\partial_t$ and notice that $\nabla^{\mu} ( \mathbb{T}_{\mu \nu} T^{\nu})=0$ where $\mathbb{T} = \mathbb{T}^{SF} +\mathbb{T}^{EM}$, c.f.\ section \ref{T} and section \ref{Pi}.
	
	While the analogous boundedness estimate in the uncharged case of the wave equation is trivial, even on Reissner-Nordstr\"{o}m space-time, this is in the charged case one of the technical hearts of the paper.
	
	Indeed, the method we described above now gives rise to an equation of the form 
	
	\begin{equation} 
		\begin{split}
			E_{deg}(u_2) +	\int_{v_R(u_1)}^{v_R(u_2)} r^2 |D_v \phi|^2_{|\mathcal{H}^+}(v)dv +  \int_{u_1}^{u_2} r^2 |D_u \phi|^2_{|\mathcal{I}^+}(u)du +\int_{u_2}^{+\infty}  \frac{2\Omega^2 Q^2}{r^2}(u,v_R(u_2)) du \\+ \int_{v_R(u_2)}^{+\infty}  \frac{2\Omega^2 Q^2}{r^2}(u_2,v) dv	=  E_{deg}(u_1)+\int_{u_1}^{+\infty}  \frac{2\Omega^2 Q^2}{r^2}(u,v_R(u_1)) du+ \int_{v_R(u_1)}^{+\infty} \frac{2\Omega^2 Q^2}{r^2}(u_1,v) dv .
		\end{split}
	\end{equation} 
	This is problematic, since the quadratic terms involving the charge \underline{do not decay}, $Q$ tending to a finite limit at infinity in the black hole case. This is of course in contrast to the Minkowski case where $Q$ tends to $0$ towards time-like infinity.	However, there is a hope that the \underline{difference} of such terms, e.g.\ a term like $\int_{v_R(u_2)}^{+\infty} \frac{\Omega^2 Q^2}{r^2}(u_2,v) dv - \int_{v_R(u_1)}^{+\infty} \frac{\Omega^2 Q^2}{r^2}(u_1,v) dv$ can be absorbed into the energy of the scalar field. 
	
	To prove this, we require the charge to be small and we have to estimate the difference carefully. More precisely, we need to transport some charge differences towards constant $r$ curves and then to use the Morawetz estimate of step \ref{2O}.
	
	\subsubsection{Integrated local energy decay} \label{morawetzoverview}
	It has been known since \cite{Morawetz} that an integrated local energy decay estimate --- now called Morawetz estimate --- is useful to prove time decay of the energy. This is classically an estimate roughly  of the form $\int_{space-time} r^{-1-\delta} \left(|\phi|^2 + r^2|D\phi|^2 \right) \lesssim E(u)$ where $E(u)$ is the energy coming from $\partial_t$ and $\delta>0$ can be taken arbitrarily small. Note that this is a global estimate with sub-optimal $r$ weights but involving all derivatives. We prove such an estimate in the step \ref{2O} but for $\delta>0$ that can be actually large. This does not make the later proof of decay harder, since our argument --- carried out in the $r$ bounded region $\{ r \leq R \}$ --- is unaffected by the value of $\delta$.
	
	The Morawetz estimate is probably --- together with the degenerate energy boundedness of step \ref{1O} --- one of the most delicate point of the present paper. This is because in the charged case, the customary use of the vector field method with $f(r) \partial_{r^*}$ now involves a supplementary term of the form $ error= \int_{space-time} q_0 Q \cdot f(r) \cdot \Im( \bar{\phi} D\phi)$ that was not present in the uncharged case. This creates additional difficulties : \begin{enumerate}
		\item \label{Mor1} the zero order term $A_0= \int_{space-time} r^{-1-\delta} |\phi|^2 $ cannot be controlled independently. This is in contrast to the uncharged case where one can first control  $\int_{space-time} r^{-1-\delta'} \cdot r^2|D\phi|^2$ for some $\delta'>\delta$ and then use this preliminary bound to finally control $\int_{space-time} r^{-1-\delta} |\phi|^2 $, c.f. \cite{JonathanStabExt}.
		
		\item \label{Mor2} The integrand of $error$ decays in $r$ at the same rate as the main controlled term, for any reasonable choice of $f(r)$. To absorb $error$ in the large $r$ region, we require $r \cdot f'(r)\gtrsim |f(r)|$ as $r \rightarrow +\infty$ . This is because, unlike on Minkowski space-time where $Q$ tends to $0$ towards time-like infinity, we can only rely on $|Q| \lesssim e$, where $e>0$ is a (small) constant. This roughly gives an estimate of the form 
		
		\begin{equation} \label{Moroverline} \int_{space-time} f'(r) \cdot  r^2|D\phi|^2  \lesssim A_0 + q_0 e \int_{space-time}   |f(r)| \cdot|\phi| \cdot  |D\phi| \lesssim A_0 + q_0 e \int_{space-time}   r|f(r)| \cdot  |D\phi|^2, \end{equation}
		
		where for the last inequality, we used Cauchy-Schwarz and the Hardy inequality roughly under the form $\int_{space-time}  r^{-1} |f(r)| \cdot|\phi|^2 \lesssim \int_{space-time}  r |f(r)| \cdot|D\phi|^2  $. The condition $r \cdot f'(r)\gtrsim |f(r)|$ together with the smallness of $q_0e$ then allows us to close the estimate, up to the zero order term $A_0$. This line of thought suggests that $f(r) \approx -r^{-\delta}$, $\delta>0$ is an appropriate choice.
		
		\item On a black hole space-time, the zero order term $A_0$ is harder to control than on Minkowski space. This is because \eqref{Moroverline} can actually be written in a more precise manner as 
		\begin{equation} \label{Moroverline2} \int_{space-time} f'(r) \cdot  r^2|D\phi|^2 + r^2 \Box_{g}( \Omega^2 \cdot r^{-1} f(r)) \cdot  |\phi|^2 \lesssim  q_0 e \int_{space-time}   r|f(r)| \cdot  |D\phi|^2 +E(u). \end{equation}
		
		For $f(r)=-r^{-\delta}$, we compute $r^2 \Box_{g}( \Omega^2 \cdot r^{-1} f(r)) = r^{-\delta-1}  P_{M,\rho}(r)$ where $P_{M,\rho}(r)$ is a second order polynomial \footnote{whose coefficients involve $M$ and $\rho$, the parameters of the Reissner--Nordstr\"{o}m black hole.} in $r^{-1}$ that is positive on $[r_+,r(\delta)] \cup [R(\delta),+\infty]$ and negative on $(r(\delta),R(\delta))$ for some 
		
		$r_+ < r(\delta) < R(\delta)$, hence this traditional approach is of little use for a charged scalar field on a black hole. An analogous computation on Minkowski gives a strictly positive constant polynomial $P_{0,0}(r)= (\delta+1)(\delta+4)$.
		
	\end{enumerate}
	
	To deal with these difficulties, we first need to prove an estimate for $A_0$ separately, more precisely on a constant-$r$ timelike curve, coupling the estimate to a $r^p$-weighted energy estimate; see Section~\ref{rpoverview}.
	\subsubsection{Non-degenerate energy boundedness and red-shift } \label{RSoverline}
	
	We now define the non-degenerate energy of the scalar field on our $\mathcal{V}_u$ foliation by  	\begin{equation}\label{RS.intro}
		E(u)=  \int_{u}^{+\infty} r^2 \frac{|D_u \phi|^2}{\Omega^2}(u',v_R(u))du'+\int_{v_R(u)}^{+\infty} r^2 |D_v \phi|^2(u,v)dv,\end{equation} 	c.f.\ section \ref{foliations} and section \ref{energynotation}. 
	
	This is called non-degenerate precisely because on a fixed $v$ line, $\Omega^{-2} \partial_u$ is a non-degenerate vector field across the event horizon $\{ u=+\infty\}$. Therefore, we expect (and prove in step \ref{6O}) a bound of the form $|D_u \phi| \lesssim \Omega^2$, consistent with the boundedness of the quantity $E(u)$. This also means that $D_u\phi=0$ on the event horizon, which explains why $\partial_u$ is degenerate and needs to be renormalized to obtain a finite limit.

	A natural energy boundedness statement would be of the form $E(u_2) \lesssim E(u_1)$ for any $u_1 < u_2$; however, we instead prove a statement of the form $\tilde{E}_{\ep}(u_2) \lesssim \tilde{E}_{\epsilon}(u_1)$, where, schematically, $\tilde{E}_{\ep}$ is the $r^{\ep}$-weighted energy  $$ \tilde{E}_{\ep}(u) = E(u) + \int r^{\ep} |D_v \psi|^2 dv,$$ for a small $\epsilon>0$ proportional to $q_0 |e|$ ($e$ being in Maxwell final charge).

	To prove this, we use a so-called red-shift estimate, pioneered in \cite{PriceLaw}, \cite{BH}, \cite{Redshift}. In our context, it boils down to using the vector field method with $X=\Omega^{-2} \partial_u$. While this is not the hardest part of the paper, we still need to use the Morawetz estimate of step \ref{2O} to conclude, unlike in the uncharged case. This is because in our case, we must absorb a bulk term coming from the charge (c.f. section \ref{overviewdecay}) into a controlled scalar field bulk term, while for the wave equation, no control of the bulk term is needed at this stage, c.f.\ \cite{BH}. In turn, the Morawetz estimate needs to be coupled to a $r^{\ep}$-weighted energy estimate, in order to control a zero order term; this feature, too, is absent in the context of an uncharged wave equation on a black hole \cite{Redshift} and explains why only an energy boundedness statement of the form $\tilde{E}_{\ep}(u_2)\ls  \tilde{E}_{\ep}(u_1)$ can be achieved (as opposed to the usual  $E(u_2)\ls  E(u_1)$ boundedness estimate,  whose validity is unclear  for this problem).

	\subsubsection{Energy decay and $r^p$ method } \label{rpoverview}

	To prove time decay of the energy, we use the $r^p$ method, pioneered in \cite{RP}. 
	
	The idea is to prove the boundedness of $r^p$ weighted energies 	$$ E_p(u):= \int_{v_{R}(u)}^{+\infty} r^{p} \ |D_v \psi|^2  (u,v)dv,$$
	
	for a certain range of $p$, ultimately responsible for time decay $t^{-p}$. Actually, we prove a hierarchical estimate of the form 
	\begin{equation}\label{rp.bound.intro}
		\int_{u_1}^{u_2} E_{p-1}(u) du + E_p(u_2) \lesssim E_p(u_1) \lesssim 1.
	\end{equation}
	This estimate is obtained by applying the vector field method in a region $\{ r \geq R \} $ with the vector field $r^p \partial_v$.
	
	Thereafter applying the mean-value theorem or a pigeon-hole like argument, we can retrieve an estimate of the form $E_{p-1}(u) \lesssim u^{-1}$ and eventually $E(u) \lesssim u^{-p}$.
	
	Now in the case of the Maxwell-Charged-Scalar-Field model, we also have to control an interaction term coming from the charge c.f.\ section \ref{overviewdecay}. We now have an estimate of the form
	
	$$ \int_{u_1}^{u_2} E_{p-1}(u) du + E_p(u_2) \lesssim E_p(u_1) + error,$$
	
	where $error \approx q_0e \int \int_{\{ r \geq R \}} r^{p-2} \Im(\bar{\psi}D_v \psi)  $ is the interaction term we mentioned earlier.
	
	Due to the form of this term, we need a Hardy inequality to absorb to $|\phi|$ into the energy term.
	
	More explicitly we roughly estimate $error$ using the following type of bounds : 
	
	$$ |  \int \int_{\{ r \geq R \}} r^{p-2} \Im(\bar{\psi} D_v\psi) | \lesssim  \left( \int \int_{\{ r \geq R \}} r^{p_1} |\psi|^2 \right)^{\frac{1}{2}} \left( \int \int_{\{ r \geq R \}} r^{p_2} |D_v\psi|^2 \right)^{\frac{1}{2}},$$
	
	where $p_1+p_2=2p-4$, simply using Cauchy-Schwarz. We then apply a Hardy inequality to find if $p_1 < -1$:	$$ | \int \int_{\{ r \geq R \}} r^{p-2} \Im(\bar{\psi} D_v\psi) | \lesssim |1+p_1|^{-1}  \left( \int \int_{\{ r \geq R \}} r^{p_1+2} |D_v\psi|^2 \right)^{\frac{1}{2}} \left( \int \int_{\{ r \geq R \}} r^{p_2} |D_v\psi|^2 \right)^{\frac{1}{2}}.$$
	
	Now because  $ \int \int_{\{ r \geq R \}} r^{p_2} |D_v\psi|^2 =\int_{u_1}^{u_2} E_{p_2}(u) du$, which is already controlled for $p_2=p-1$, the choice $(p_1,p_2)=(p-3,p-1)$ seems natural, c.f. section \ref{p<2}.
	
	We then roughly need to absorb a term $q_0 |e| \cdot |2-p|^{-1} \int_{u_1}^{u_2} E_{p-1}(u) du$ into $\int_{u_1}^{u_2} E_{p-1}(u) du$ which is essentially doable if $q_0|e|$ is small but requires $p \in [0,2-\epsilon q_0e]$ in particular $p<2${, leading to \eqref{rp.bound.intro}. In fact, one can prove a stronger estimate than \eqref{rp.bound.intro}, which is relevant to the Morawetz estimate: schematically \begin{equation}\label{0.order.rp}
			\int_{u_1}^{u_2} E_{p-1}(u) du + E_p(u_2) + \int_{r=R} |\phi|^2\lesssim_R E_p(u_1)+  \int_{r=R} |D_v\phi|^2,
		\end{equation} noting the presence of a zero-order term on the timelike boundary on the LHS, at the expense of a similar first order term on the RHS (see \cite{Dejaninverse} where this estimate was used). Coupling \eqref{0.order.rp} with \eqref{Moroverline} and the boundedness of the non-degenerate energy \eqref{RS.intro} (under the modified form $\tilde{E}_{\ep}(u_2)\ls  \tilde{E}_{\ep}(u_1)$ as we previously explained)  then allows to close an integrated local energy estimate of the schematic form \begin{equation}
			\int_{u_1}^{u_2} \int_{r\leq R} [|\phi|^2 + |D\phi|^2 ] dr du\ls_R \tilde{E}_{\ep}(u),
		\end{equation} for some small $\ep\in(0,1)$ proportional to $q_0 |e|$.}

	{Turning to energy decay estimates in time, }calling $p_{0}$ the maximal $p$ for which  {\eqref{rp.bound.intro} holds}, we then essentially prove for $u>0$ large : 
	
	\begin{equation} \label{p0} E_{p_{0}-1}(u) \lesssim  \frac{E_{p_{0}}(u)}{u}, \end{equation}	\begin{equation}  \label{p02} \int_{u_1}^{u_2} E_{p_{0}-1}(u) \lesssim  E_{p_{0}}(u_1). \color{black} \end{equation}
	
	We employ this strategy in section \ref{p<2}. While the estimates we get are necessary to ``start'' the argument, they are insufficient to reach the best possible decay advertised in the theorems. 
	
	This is why in section \ref{sectionp=3}, we adopt a completely different strategy. This time we chose $(p_1,p_2)=(p-4,p)$ for $p=p_0+1$. Then the error term we need to absorb is roughly of the form
	
	\begin{equation*}\begin{split}
			q_0 |e| \cdot |2-p_0|^{-1}  \left(\int_{u_1}^{u_2} E_{p-2}(u) du  \right)^{\frac{1}{2}} \left(\int_{u_1}^{u_2} E_{p}(u) du \right)^{\frac{1}{2}}\\ \lesssim q_0 |e| \cdot |3-p|^{-1}  \left(  E_{p-1}(u_1) \color{black}\right)^{\frac{1}{2}} \left(\int_{u_1}^{u_2} E_{p}(u) du \right)^{\frac{1}{2}}, 
		\end{split}
	\end{equation*}  where we used the equation \eqref{p02} and the fact that $p_0=p-1$.
	
	Now we need to absorb the right-hand side into  $\int_{u_1}^{u_2} E_{p-1}(u) du + E_p(u_2)$. This requires a Gr\"{o}nwall-like method, in which the $r^p$ weighted energy experiences a controlled $u$ growth : $E_p(u) \lesssim u^{2\epsilon}$. Eventually, we get  $E(u) \lesssim u^{-p+2\epsilon}$ for some small $\epsilon$.
	
	The most delicate part of the proof is to choose $2<p=p(e)<3$ and $0<\epsilon(e)$ so as to close the estimates on the one hand, and to maximise the decay rate on the other hand. This requires an optimisation procedure which is explained at the beginning of section \ref{sectionp=3} and in Remark \ref{proofsimple}.
	
	With this last argument, a $r^p$ weighted hierarchy is proven for  all $0 \leq p \leq p(e)$, where $p(e)>2$ and $p(e) \rightarrow 3$ as $e \rightarrow0$, which is a significant improvement with respect to the first method and allows us to claim a stronger time decay of the energy. 
	
	\color{black}

	\subsubsection{Point-wise bounds} \label{pointoverview}
	
	To prove point-wise bounds on the scalar field and the charge, we need two essential ingredients: the weighted energy decay of step \ref{5O} and a point-wise estimate $|D_v \psi|(u_0(R),v) \lesssim v^{-\omega}$ on a fixed light cone, for some $\omega>0$.
	
	The latter comes from the point-wise decay of the scalar field Cauchy data, that implies consistent point-wise bounds in the past of a fixed forward light cone, c.f.\ section \ref{CauchyCharac}. We then use the former to ``initiate'' some decay estimate for $\phi$. For this we essentially use Cauchy-Schwarz  under the form $r^{\frac{1}{2}} |\phi| \lesssim E(u) \lesssim u^{-p}$, for the maximum $p$ in the $r^p$ hierarchy.
	
	Point-wise bounds are then established in the rest of the space-time integrating \eqref{wavev} and \eqref{waveu} along constant $u$ and constant $v$ lines, after carefully splitting the space-time into different regions as follows: \begin{enumerate}
		\item a far away region --- including $\mathcal{I}^+ = \{ r=+\infty\}$ --- where $r \sim v$, which is somehow the easiest.
		
		\item An intermediate region  $ \{R \leq r \lesssim v \}$ where $R>r_+$ is a large constant.
		
		\item The bounded $r$ region $\{ r_+ \leq r \leq R \}$, which includes  $\mathcal{H}^+ = \{ r=r_+\}$ .
	\end{enumerate}
	
	The far away region is the one where $r$ weights are strong so the ``conversion'' between the $L^2$ and the $L^{\infty}$ occurs ``with no loss''. 
	
	The bounded $r$ region is also not so difficult due to a point-wise version of the red-shift effect : there is again no loss between the estimates on the curve $r=R$ and the event horizon $r=r_+$.
	
	However, in the intermediate region, the $r$ weights, strong on $\{r \sim v\}$ degenerate to a mere constant near $\{ r=R\}$. For this reason, the estimates imply a loss of $v^{\frac{1}{2}}$, which explains why the point-wise decay rate obtained on the event horizon is not expected to be optimal, while the rates on the energy and in the far away region are, at least in the limit $q_0|e| \rightarrow 0$. \\

	In every section and subsection, the precise strategy of the proof is discussed. We refer the reader to these paragraphs for more details.

	\subsection{Outline of the paper} \label{outline}
	
	The paper is outlined as follows : after introducing the equations, some notations and our foliation in section \ref{section2}, we announce in section \ref{mainresult} a more precise version of our results. 
	
	Because our techniques (with the $V$ shaped foliation of section \ref{section2}) only deal with characteristic initial value data, we explain in section \ref{CauchyCharac} how the Cauchy problem can be reduced to a characteristic initial value problem, with the correct assumptions. This is also the occasion to derive a priori smallness estimates on the charge, useful to close the harder estimates of the following sections. To avoid repetition and because the estimates are easier than the ones in the next sections, we postpone the proof to Appendix \ref{appendixCC}. 
	
	Then in section \ref{energyboundedsection}, we prove Theorem \ref{boundednesstheorem}. This section is divided as follows : first the subsection \ref{Morawetz} where the integrated local energy decay, redshift estimates are established with a help of a first $r^p$-weighted estimate (for small $p>0$), and modulo energy boundary terms. Finally we close together the energy boundedness and the Morawetz estimate in subsection \ref{energysection}, using crucially the results of the previous subsection.
	
	Then, we turn to the proof of the second main result of the present paper, Theorem \ref{decaytheorem}. Establishing energy decay is the object of section \ref{decay}. It is divided as follows : first in subsection \ref{preliminaries}, we carry out preliminary computations. We also re-prove explicitly the version of the $r^p$ method needed for our purpose, with the notations of the paper. Then a first energy decay estimate is proven, for $p<2$, in section \ref{p<2}, using the smallness of the future asymptotic charge under the form $q_0|e| < \frac{1}{4}$. 
	
	Finally, in section \ref{sectionp=3}, we prove a hierarchy of $r^p$ estimates for $2<p<3$, requiring now that $q_0|e| < {0.04}$ . The employed strategy differs radically from the one of section \ref{p<2} but we use crucially the  $r^p$ hierarchy proven in the previous section. This proves Theorem \ref{decaytheorem}.
	
	Section \ref{sectionp=3} is one of the technical hearts of the paper, and the part that allows eventually for strictly integrable bounds on the event horizon.
	
	The intake of section \ref{p<2} and \ref{sectionp=3} is that energy decay at a rate $p(e)=3-O( (q_0|e|)^{\frac{1}{2}})$ holds, provided energy boundedness is true. This rate $p(e)$ tends to $3$ as $q_0e$ tends to $0$, which is the expected optimal limit. It also approaches $2$ as $q_0|e|$ approaches the endpoint of the admissible range.\footnote{More precisely, there exists $a_{\max}$ with $a_{\max}>{0.04}$ such that $p(e) \rightarrow 2$ as $q_0|e| \rightarrow a_{\max}$, c.f.\ the proof of Lemma \ref{fLemma}.}

	We can then retrieve point-wise estimates from the energy decay in section \ref{pointwise}. We use methods that only work in spherical symmetry and which are optimal in a region near null infinity. However, the presence of $r$ weights that degenerate in the bounded $r$ region creates a $u^\frac{1}{2}$ loss in the decay that cannot be compensated otherwise, like it was \footnote{The reason is that in the uncharged case, $\partial_t \phi$ decays better than $\phi$ but in this charged case, $D\phi$ decays like $\phi$, due to oscillations. } in the uncharged case, c.f.\ Section \ref{uncharged} . This is why the optimal energy decay does not give rise to optimal point-wise bounds on the event horizon. However, since the decay rate of the energy is always superior to $2$, we retrieve \textbf{strictly integrable} point-wise decay rate for $\phi$ on the event horizon, for the full range of $q_0|e|$ we consider. 
	
	Finally, we prove a slight variant \footnote{The main intake is that, for slightly more decaying data than required in Theorem \ref{decaytheorem}, one can prove that $v^2 D_v \psi$ admits a bounded limit when v $\rightarrow+\infty$, for fixed $u$.} of Theorem \ref{decaytheorem} in Appendix \ref{moredecayappendix}, the proofs of Section \ref{CauchyCharac} are carried out in Appendix \ref{appendixCC} and a few computations are made explicit in Appendix \ref{appendix}.

	\subsection{Acknowledgements}
	
	I express my profound gratitude to my PhD advisor Jonathan Luk for proposing this problem, for his guidance and for numerous extremely stimulating discussions.
	
	I am grateful to Mihalis Dafermos for very helpful comments on the manuscript.
	
	I am grateful to Stefanos Aretakis, Dejan Gajic and Georgios Moschidis for insightful conversations related to this project.
	
	My special thanks go to Hayd\'{e}e Pacheco for drawing the Penrose diagrams.

	I gratefully acknowledge the financial support of the EPSRC, grant reference number EP/L016516/1.
	
	This work was completed while I was a visiting student in Stanford University and I gratefully acknowledge their financial support and their hospitality.
	
	This revised version corrects a mistake in the proof of the integrated local energy decay estimate; the author warmly thanks Yannis Angelopoulos, Christoph Kehle and Ryan Unger   for pointing  out the issue.

	\section{Equations, definitions, foliations and calculus preliminaries} \label{section2}

	\subsection{Coordinates and vector fields on Reissner--Nordstr\"{o}m's space-time } \label{Coordinates}
	
	The sub-extremal Reissner--Nordstr\"{o}m exterior metric can be written in coordinates $(t,r,\theta,\varphi)$
	
	\begin{equation*} 
		g=-\Omega^{2}dt^{2}+\Omega^{-2}dr^{2}+r^{2}[ d\theta^{2}+\sin(\theta)^{2}d \psi^{2}],
	\end{equation*}	\begin{equation*} 
		\Omega^{2}=1-\frac{2M}{r}+\frac{\rho^2}{r^2},
	\end{equation*}
	
	for some $0 \leq |\rho| < M$ and	where $(r,t,\theta, \varphi) \in (r_{+},+\infty) \times \mathbb{R}^+ \times  [0,\pi) \times [0, 2 \pi]$ .

	$r_+(M,\rho):=M + \sqrt{M^2-\rho^2}$ is one of the two positives roots of $\Omega^2(r)$ and represents the radius of the event horizon, defined to be 
	
	$$ \mathcal{H}^+ := \{ r=r_+, \hskip 2 mm  (t,\theta, \varphi) \in \mathbb{R}^+ \times  [0,\pi) \times [0, 2 \pi]  \}.$$

	Symmetrically we define future null infinity to be 
	
	$$ \mathcal{I}^+ := \{ r=+\infty, \hskip 2 mm (t,\theta, \varphi) \in \mathbb{R}^+ \times  [0,\pi) \times [0, 2 \pi]  \}.$$
	
	$	\mathcal{H}^+$ and $	\mathcal{I}^+$ are then null and complete hyper-surfaces for the Reissner--Nordstr\"{o}m metric.
	
	The other root 	$r_-(M,\rho):=M - \sqrt{M^2-\rho^2}$ corresponds to the locus of the Cauchy horizon, inside the black hole, c.f.\ \cite{Moi}.  Therefore $r_-$ does not play any role in the exterior.

	We want to built a double null coordinate system $(u,v)$ on spheres that can replace $(t,r)$.
	
	One possibility is to define a function $r^*(r)$, sometimes called tortoise radial coordinate, such that 
	
	$$ \frac{ dr^*}{dr} = \Omega^{-2}(r), $$
	
	$$ \lim_{r \rightarrow +\infty} \frac{ r^*}{r} = 1.$$
	
	Notice that the new coordinate $r^*(r)$ is an increasing function of $r$ that takes its values in $(-\infty,+\infty)$.
	
	We denote by $\partial_t$ and $\partial_{r^*}$ the corresponding vector fields in the $(r^*,t,\theta, \varphi)$ coordinate system. Notice that $\partial_t$ is a time-like \textbf{Killing vector field} for this metric.
	
	We can then define the functions $u(r^*,t)$ and $v(r^*,t)$ as
	
	$$ v = \frac{t+r^*}{2} , \hskip 5 mm u = \frac{t-r^*}{2}.$$
	
	Notice that $u$ takes its values in $(-\infty,+\infty)$ and $\{ u=+\infty\}  = \mathcal{H}^+$. 
	
	$v$ takes its value in $(-\infty,+\infty)$ and $\{ v=+\infty\}  = \mathcal{I}^+$.
	
	Spatial infinity $i^0$ is $\{ u=-\infty, v=+\infty \}$ and the bifurcation sphere is $\{ u=+\infty, v=-\infty \}$.
	
	Defining $\partial_u$ and $\partial_{v}$, the corresponding vector fields in the $(u,v,\theta, \varphi)$ coordinate system, we can check that $(\partial_u,\partial_v,\partial_{\theta},\partial_{\varphi})$ is a null frame for the Reissner--Nordstr\"{o}m metric.
	
	Notice that we have 
	
	$$ \partial_v =  \partial_t +\partial_{r^*}, \partial_u= \partial_t -\partial_{r^*}.$$
	
	The Reissner--Nordstr\"{o}m's metric can then be re-written as	$$ g =- 2\Omega^2 (du  \otimes dv+dv  \otimes du)+r^{2}[ d\theta^{2}+\sin(\theta)^{2}d \varphi^{2}].$$
	
	Note that in this coordinate system we also have 
	
	$$ \Omega^2(r) = \partial_v r = -\partial_u r.$$
	
	Then we define the quantity $2K$, the $\log$ derivative of $\Omega^2$ that plays a role in the present paper :

	\begin{equation} \label{K} 2K(r):=\frac{2}{r^2}(M- \frac{\rho^2}{r})= \partial_v \log(\Omega^2)=  -\partial_u \log(\Omega^2).\end{equation}
	
	We also define the surface gravity of the event horizon $2K_+:=2K(r_+)$ and that of the Cauchy horizon $2K_-:=2K(r_-)$. These definitions allow for an explicit and simple expression of $r^*$ as 
	
	$$ r^* = r + \frac{\log(r-r_+)}{2K_+}+\frac{\log(r-r_-)}{2K_-}.$$
	
	This implies that when $r \rightarrow r_+$ : 
	
	\begin{equation} \label{regularomega}
		\Omega^2(r) \sim  C_+^2  \cdot e^{2K_+ \cdot (v-u)},
	\end{equation}
	
	where $C_+=C_+(M,\rho)>0$ is defined by $C_+^2 = e^{-2K_+ r_+} \frac{ (r_+-r_-)^{1-\frac{K_+}{K_-}}}{r_+^2}$.


	\subsection{The spherically symmetric Maxwell-Charged-Scalar-Field equations in null coordinates} \label{eqcoord}
	
	We now want to express Maxwell-Charged-Scalar-Field system 
	\begin{equation} \label{Maxwell}\nabla^{\mu} F_{\mu \nu}= iq_0 \frac{ (\phi \overline{D_{\nu}\phi} -\overline{\phi} D_{\nu}\phi)}{2} , \; F=dA ,
	\end{equation} \begin{equation} \label{CSF} g^{\mu \nu} D_{\mu} D_{\nu}\phi = 0
	\end{equation}	
	
	in the $(u,v)$ coordinates of section \ref{Coordinates} and for spherically symmetric solutions. 
	
	Here $\phi$ represents a complex valued function while $F$ is a real-valued 2-form.
	
	In what follows, for a spherically symmetric solution $(\phi,F)$,  we are going to consider the projection of $(\phi,F)$ ---that we still denote $(\phi,F)$--- on the 2-dimensional manifold indexed by the null coordinate system $(u,v)$ and on which every point represents a sphere of radius $r(u,v) $ on Reissner--Nordstr\"{o}m space-time. Much more details can be found on the procedure in \cite{Moi}, section 2.2.

	It can be shown (c.f.\ \cite{Kommemi})  that in spherical symmetry, we can express the two-form $F$ as
	
	\begin{equation} \label{chargedef} F= \frac{2Q \Omega^2}{ r^{2}} du \wedge dv, \end{equation} where $Q$ is a scalar function that we call the charge of the Maxwell equation.
	
	\begin{rmk}
		Notice that this formula, that defines $Q$, differs from a multiplicative factor $4$ from the formula of \cite{extremeJonathan}, \cite{Kommemi} and \cite{Moi}. This is because, in the present paper, $\Omega^2$ is defined to be $\Omega^2 = 1-\frac{2M}{r}+\frac{\rho^2}{r^2}$. In the previous papers, $\Omega^2$ was actually defined as $\Omega^2 = 4(1-\frac{2M}{r}+\frac{\rho^2}{r^2})$, hence the difference.
	\end{rmk}

	$F=dA$ also allows us to chose a spherically symmetric potential 
	$A$ 1-form  written as : 
	
	$$ A = A_u du + A_v dv. $$
	
	We then define the gauge derivative $D$ as $D_{\mu} = \partial_{\mu}+ iq_0 A_{\mu}$. 
	
	The equations \eqref{Maxwell}, \eqref{CSF} are invariant under the following gauge transformation : 
	
	$$ \phi \rightarrow  \tilde{\phi}=e^{-i q_0 f } \phi $$
	$$ A \rightarrow  \tilde{A}=A+ d f ,$$	where $f$ is a smooth real-valued function. There is therefore a gauge freedom. However, all the estimates derived in this paper are essentially gauge invariant, so \textbf{we do not need to choose a gauge}.
	
	Notice that for the gauge derivative $\tilde{D}:=\nabla+\tilde{A}$, we have 
	
	$$ \tilde{D} \tilde{\phi} = e^{-iq_0f} D\phi.$$
	Therefore, as explained in \cite{extremeJonathan}, the quantities $|\phi|$ and $|D\phi|$ are gauge invariant and so are the energies defined in section \ref{energynotation}. It should be noted that gauge derivatives do not commute: indeed one can show that

	$$ [D_u,D_v] = iq_0 F_{u v} Id =\frac{2iq_0Q \Omega^2}{r^2} Id.$$
	
	We now express equation \eqref{CSF} in $(u,v)$ coordinates. For this, it is convenient to define the radiation field $\psi := r\phi$. We then find :

	\begin{equation} \label{wavev}
		D_u(D_v \psi) = \frac{\Omega^2}{r} \phi \left(iq_0 Q -  \frac{2M}{r}+ \frac{2\rho^2}{r^2}\right)
	\end{equation}
	\begin{equation} \label{waveu}
		D_v(D_u \psi) = \frac{\Omega^2}{r} \phi \left(-iq_0 Q -  \frac{2M}{r}+ \frac{2\rho^2}{r^2}\right)
	\end{equation}
	
	Maxwell's equation \eqref{Maxwell} can also be written in $(u,v) $ coordinates as
	
	\begin{equation} \label{chargeUEinstein}
		\partial_u Q = - q_0 r^2 \Im( \bar{\phi} D_u \phi)= - q_0 \Im( \bar{\psi} D_u \psi),  \color{black} 
	\end{equation}
	\begin{equation} \label{ChargeVEinstein}
		\partial_v Q =  q_0  r^2 \Im( \bar{\phi} D_v \phi)= q_0   \Im( \bar{\psi} D_v \psi)  \color{black}.
	\end{equation}
	
	Notice that in the spherically symmetric case, the Maxwell form is reduced to the charge $Q$.
	
	Then finally the existence of an electro-magnetic potential $A$ can be written as:

	\begin{equation} \label{potentialEinstein}
		\partial_u A_v-	\partial_v A_u = F_{u v} =  \frac{2Q\Omega^2}{r^2} .
	\end{equation}
	
	Finally we would like to finish this section by a Lemma stated in \cite{extremeJonathan} (Lemma 2.1) that says that ``gauge derivatives can be integrated normally ''. More precisely :
	
	\begin{lem} \label{integrationlemma}
		For every $u_1 \in (-\infty,+\infty)$ and $v_1 \in (-\infty,+\infty)$ and every function $f$,
		
		$$ |f(u,v)| \leq  |f(u_1,v)| + \int_{u_1}^{u} |D_u f|(u',v) du',$$
		
		$$ |f(u,v)| \leq  |f(u,v_1)| + \int_{v_1}^{v} |D_v f|(u,v') dv'.$$
	\end{lem} 
	
	We refer to \cite{extremeJonathan} for a proof, which is identical in the exterior case. 
	
	This lemma will be used implicitly throughout the paper.

	\subsection{Notations for different charges} \label{chargenotations}
	
	In this paper, we make use of several quantities that we call ``charge'' although they may not be related. This section is present to clarify the notations and the relationships between these different quantities. 
	
	First we work on a Reissner--Nordstr\"{o}m space-time of parameters $(M,\rho)$. 	The \underline{Reissner--Nordstr\"{o}m charge} $\rho$ is defined by the expression of the metric \eqref{RN}, \eqref{OmegaRN}. 
	
	We assume the sub-extremality condition $0 \leq |\rho|<M$, where $M$ is the Reissner--Nordstr\"{o}m mass.
	
	Note that $\rho$ is here \textbf{constant}, as a parameter of the black hole that does not vary dynamically because we study the \textbf{gravity uncoupled} problem.	Note also that the value of $\rho$ \textbf{plays no role} in the paper, provided $|\rho|<M$.  \\
	
	In this paper we consider the Maxwell-\textbf{Charged}-Scalar-Field model. This means that the \textbf{charged} scalar field interacts with the Maxwell field. 
	
	The interaction constant $q_0 \geq0$ is called the \underline{charge of the scalar field}, as appearing explicitly in the system \eqref{Maxwell}, \eqref{CSF}. 
	
	We require the condition $q_0 \geq 0$ uniquely for aesthetic reasons and this is purely conventional. Everything said in this paper also works if $q_0<0$, after replacing $q_0|e|$ by $|q_0e|$ and $q_0|e_0|$ by $|q_0e_0|$.
	
	Note that we consider $q_0$ as a universal constant, not as a parameter. Note also that $q_0$ has the dimension of the inverse of a charge : $q_0Q$ is a dimensionless quantity. 
	
	Note also that the uncharged case $q_0=0$ corresponds to the well-known linear wave equation. \\
	
	Now we define the \underline{charge of the Maxwell equation} $Q$, defined explicitly from the Maxwell Field form $F$ by \eqref{chargedef}, itself appearing in \eqref{Maxwell}. 
	
	$Q$ is a scalar \textbf{dynamic} function on the space-time ---in contrast to the uncharged case $q_0=0$ where $Q$ is forced to be a constant --- that determines completely the Maxwell Field in spherical symmetry.
	
	Note also that the Maxwell equations can be completely written in terms of $Q$, c.f.\ equations \eqref{chargeUEinstein}, \eqref{ChargeVEinstein}. \\
	
	Because we consider a Cauchy initial value problem, we also consider  the \underline{initial charge of the Maxwell equation} $Q_0$ defined as $Q_0 = Q_{ |\Sigma_0}$ where $\Sigma_0$ is the initial space-like Cauchy surface. 
	
	$Q_0$ is then one part of the initial data $(\phi_0,Q_0)$. \\
	
	Then we define the \underline{initial asymptotic charge} $e_0$ as 
	
	$$ e_0 = \lim_{r \rightarrow+\infty} Q_0(r),$$
	when it exists. This is the limit value of the Maxwell charge at spatial infinity. \\
	
	Finally we define the \underline{future asymptotic charge} $e$ as 
	
	$$ e = \lim_{t \rightarrow+\infty} Q(t,r),$$
	for all $r \in [r_+,+\infty]$, 	when it exists. It can be proven \footnote{The proof is made in section \ref{remaining}. } that the limit $e$ does not depend on $r$.  This is the limit value of the Maxwell charge at time-like and null infinity. \\

	\subsection{Foliations, domains and curves} \label{foliations}
	
	In this section we define the foliation over which we control the energy, c.f.\ section \ref{energynotation}. This is a $V$-shaped foliation, similar to the one of \cite{JonathanStabExt} but different from the $J$-shaped foliation of \cite{RP} and subsequent works.

	For any $r_+ \leq  R_1$ we define the curve $\gamma_{R_1} = \{ r=R_1\} $.
	
	For any $u$, we denote $v_{R_1}(u)$ the only $v$ such that $(u,v_{R_1}(u)) \in \gamma_{R_1}$. Such a $v$ is given explicitly by the formula $v-u=R_1^*$. Similarly for every $v$,  we introduce $u_{R_1}(v)$. \\

	Then for some $R>r_+$ large enough to be chosen in course of the proof, we considered the curve $\gamma_R$.
	
	We also denote $(u_0(R),v_0(R))$, the coordinates of the intersection of $\gamma_R$ and $\{t=0\}$ : 
	
	$v_0(R)=-u_0(R)=\frac{R^*}{2}$.
	
	We then define the foliation $\mathcal{V}$ :  for every $u \geq u_0(R)$ 
	
	$$ \mathcal{V}_u = \left( \{ v_R(u) \}\times [u,+\infty ] \right)  \cup  \left([v_R(u),+\infty ]\times \{u \}\right).$$
	\begin{figure}[H]
		
		\begin{center}
			
			\includegraphics[width=107 mm, height=65 mm]{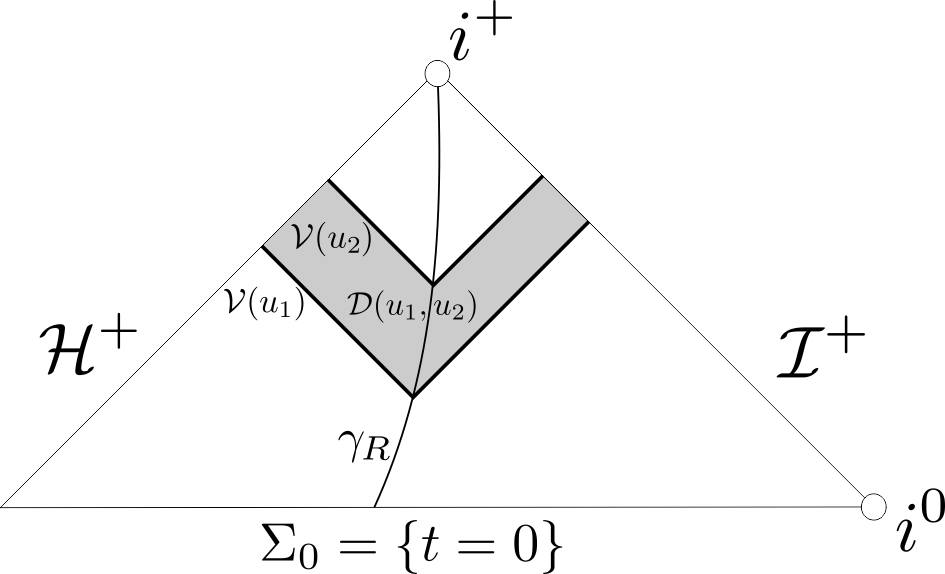}
			
		\end{center}
		
		\caption{Illustration of the foliation $\mathcal{V}_u$, $u\geq u_0(R)$}
		\label{Fig1}
	\end{figure}
	This foliation is composed from a constant $v$ segment in the region $\{ r \leq R \}$ joining a constant $u$ segment in the region $\{ r \geq R \}$. It is illustrated in Figure \ref{Fig1}.
	
	Notice that the foliation does not cover the regions $\{ u \leq u_0(R) \}$ and $\{ v \leq v_0(R) \}$ as it can be seen in Figure \ref{Fig3}. This is why we need section \ref{CauchyCharac} to connect the energy on this foliation, and in particular on $\mathcal{V}_{ u_0(R)}$ to the energy on the initial space-like hyper-surface $\mathcal{E}$. \\

	We now defined the domain $\mathcal{D}(u_1,u_2)$, for all $u_0(R)\leq u_1<u_2$ as 
	
	\begin{equation} \label{domain}
		\mathcal{D}(u_1,u_2) := \cup_{u_1 \leq u \leq u_2} \mathcal{V}_u . 
	\end{equation}
	Numerous $L^2$ identities of this paper are going to be integrated either on $	\mathcal{D}(u_1,u_2)$, or 	$\mathcal{D}(u_1,u_2) \cap \{ r \leq R\}$ or 	$\mathcal{D}(u_1,u_2) \cap \{ r \geq R\}$. \\

	We also define the initial space-like hyper-surface $\Sigma_0 =\{ t=0\} = \{ v = -u\}$ on which we set the Cauchy data $(\phi_0,Q_0)$ c.f.\ section \ref{CauchyCharac}. 
	
	We are going to make use of the following notations $u_0(v)=-v$ , $v_0(u)=-u$, $r_0(u)$ the radius corresponding to $(r_0(u))^*=-2u$ and $r_0(v)$ the radius corresponding to $(r_0(v))^*=2v$,  when there is no ambiguity between $u$ and $v$. \\

	Finally we will also need a curve close enough to null infinity, in particular to retrieve point-wise bounds. 
	
	We define $ \gamma = \{ r^* = \frac{v}{2}+R^* \}$. Notice that on this curve, $r \sim u \sim 	 \frac{v}{2} $ as $v \rightarrow +\infty$ .
	
	For any $u$  , we then define $v_{\gamma}(u)$, the only $v$ such that $(u,v_{\gamma}(u)) \in \gamma$. Similarly, we introduce $u_{\gamma}(v)$.

	\begin{rmk} Note that because $\gamma$ and $\gamma_{R_1}$ are time-like curves, talking of their future domain is not very interesting. Instead, notably in section \ref{pointwise}, we are going to make use of the domain ``to the right'' of $\gamma$ : $\{ r^* \geq \frac{v}{2}+R^* \} $ and the domain ``to the right'' of $\gamma_{R_1}$ : $\{  r \geq R_1 \}$.
	\end{rmk}

	\begin{figure} 
		
		\begin{center}
			
			\includegraphics[width=107 mm, height=65 mm]{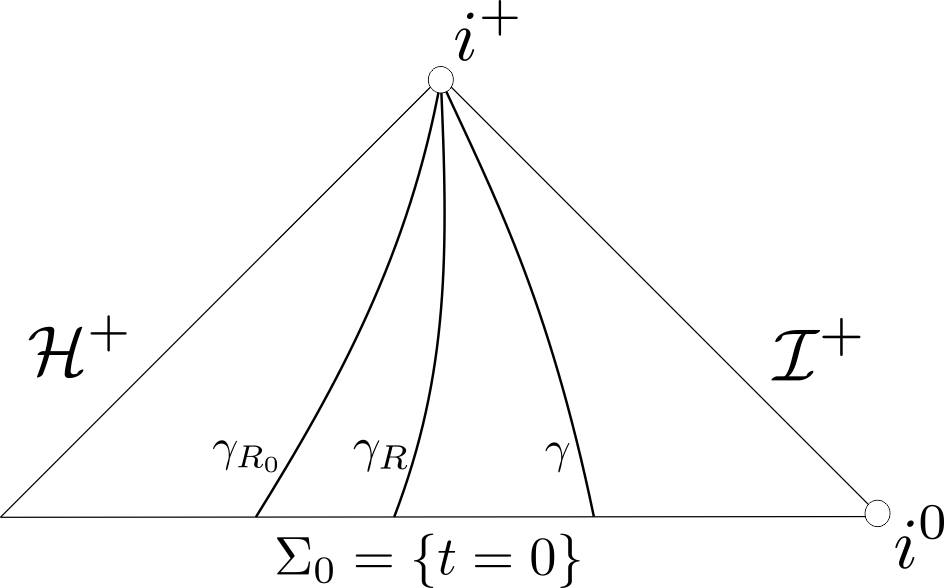}
			
		\end{center}
		
		\caption{Curves $\gamma_{R_0}$ for $r_+<R_0=R_0(M,\rho)<R$, $\gamma_R$ and $\gamma$}
		\label{Fig2}
	\end{figure}
	\subsection{Energy notations} \label{energynotation}
	
	In this section, we gather the definitions of all the energies used in the paper. Those definitions are however repeated just before being used for the first time.
	
	The main definition concerns the non-degenerate energy, for all $u \geq u_0(R)$ : 
	
	$$ E(u)=E_R(u)=  \int_{u}^{+\infty} r^2 \frac{|D_u \phi|^2}{\Omega^2}(u',v_R(u))du'+\int_{v_R(u)}^{+\infty} r^2 |D_v \phi|^2(u,v)dv. $$
	
	By default, $E(u)$ is the energy related to the $\gamma_R$-based foliation $\mathcal{V}$ we mention in section \ref{Coordinates}. Sometimes however, we may talk of $E_{R_1}(u)$ \underline{ for a $R_1$ that is different from $R$}, in section \ref{CauchyCharac}. \\
	
	We also define the degenerate energy : 
	
	$$ E_{deg}(u)=  \int_{u}^{+\infty} r^2 |D_u \phi|^2(u',v_R(u))du'+\int_{v_R(u)}^{+\infty} r^2 |D_v \phi|^2(u,v)dv. $$
	
	
	It will be convenient to use for all $u_0(R) \leq u_1 < u_2$ :
	
	$$ E^+(u_1,u_2):= E(u_2)+E(u_1)+ 	\int_{v_R(u_1)}^{v_R(u_2)} r^2 |D_v \phi|^2_{|\mathcal{H}^+}(v)dv 
	+\int_{u_1}^{u_2} r^2 |D_u \phi|^2_{|\mathcal{I}^+}(u)du, $$

	$$ E_{deg}^+(u_1,u_2):= E_{deg}(u_2)+E_{deg}(u_1)+ 	\int_{v_R(u_1)}^{v_R(u_2)} r^2 |D_v \phi|^2_{|\mathcal{H}^+}(v)dv 
	+\int_{u_1}^{u_2} r^2 |D_u \phi|^2_{|\mathcal{I}^+}(u)du. $$

	For the $r^p$ hierarchy we are also going to use for all $u_0(R) \leq u_1 < u_2$ : 
	
	$$ E_p[\psi](u):= \int_{v_{R}(u)}^{+\infty} r^{p} \ |D_v \psi|^2  (u,v)dv,$$
	
	$$ \tilde{E}_p(u) := E_p[\psi](u) + E(u). $$
	
	{In the boundedness section, whenever $\delta_Q$ is an $L^{\infty}$
		bound for $Q$, the multiplier argument produces the weight
		${\ep}_Q:=8q_0\delta_Q$. In the decay section we instead fix
		$$ \ep:=8q_0|e|. $$
		Proposition \ref{CharacProp2} allows ${\ep}_Q$ to be chosen arbitrarily close
		to $\ep$ from above. In particular
		$\ep=O(q_0|e|)$.}

	Finally we define the initial  non-degenerate energy on the initial slice $\Sigma_0 := \{ t=0\}$
	
	\begin{equation} \label{initialenergy} \mathcal{E} = \int_{\Sigma_0}  \frac{r^2  |D_u \phi_0|^2+ r^2 |D_v \phi_0|^2}{\Omega^2} dr^*.
	\end{equation}  
	
	\begin{rmk} \label{regularomegaremark}Notice that this energy is proportional to the regular $H^1$ norm for $\phi$ on $\Sigma_0$. Indeed the vector field $\Omega^{-1} \partial_t = \frac{ \partial_u +\partial_v }{2\Omega}$ --- time-like and unitary --- is regular across the bifurcation sphere. But from \eqref{regularomega}, we see that towards the future or the past event horizon $\{ r=r_+\}$, $\Omega^2 = (C_+)^2 \cdot  e^{2K_+ v} e^{-2K_+ u}$. Therefore, for on any constant $v$ slice transverse to the future event horizon, $(C_+)^{-1} e^{2K_+ u} \partial_u$ is regular. Symmetrically on any constant $u$ slice transverse to the past event horizon, $(C_+)^{-1} e^{-2K_+ v} \partial_v$ is regular. We then see that $(2C_+)^{-1} \left( e^{2K_+ u} \partial_u +e^{-2K_+ v} \partial_v \right)$ is then also a regular vector field across the bifurcation sphere, actually proportional to $\Omega^{-1} \partial_t$ because on $\Sigma_0$, $v \equiv -u $.
		
		Notice that consistently, to prove point-wise bounds, we are going assume 	$ |D_u \phi_0|(r) \lesssim e^{-2K_+u} \sim \Omega(u,v_0(u))$ --- c.f.\ Hypothesis \ref{regderivRS} --- as opposed to the too strong but somehow naively intuitive hypothesis 	$ |D_u \phi_0|(r) \lesssim  \Omega^2$. 
		
		Actually, the $L^2$ bound is ``morally'' slightly stronger than the point-wise one. Indeed $\Omega^{-2}|D_u \phi_0|^2 dr^* = \Omega^{-4}|D_u \phi_0|^2 dr $ is integrable if $\Omega^{-4}|D_u \phi_0|^2 \lesssim (r-r_+)^{-1+\epsilon}$ for any $\epsilon>0$, i.e\ 
		$ |D_u \phi_0|(u,v_0(u)) \lesssim  e^{-2K_+(1+\epsilon)u}$.
		\\ 
		
	\end{rmk}
	We also define the initial $r^p$ weighted energy for $\psi_0:=r\phi_0$ : 
	
	\begin{equation} \mathcal{E}_p :=   \int_{\Sigma_0} \left( r^p |D_v \psi_0|^2 +r^p |D_u \psi_0|^2 +\Omega^2 r^p  |\phi_0|^2 \right)  dr^*.
	\end{equation}
	
	Note that $\mathcal{E}_p$ includes a $0$ order term $r^p |\phi_0|^2$, of the ``same homogeneity'' as $r^p |D_v \psi_0|^2 $. Notice also that 
	
	$$ \int_{\Sigma_0} r^p |D_{r^*} \psi_0|^2 dr^* \leq \frac{ \mathcal{E}_p}{2}.$$
	
	Finally we regroup the two former definitions into 
	
	$$\tilde{\mathcal{E}}_p = \mathcal{E}_p+\mathcal{E}.$$
	
	We also want to recall that the Stress-Energy momentum tensors are defined as: 	\begin{equation} \mathbb{T}^{SF}_{\mu \nu}=  \Re(D _{\mu}\phi \overline{D _{\nu}\phi}) -\frac{1}{2}(g^{\alpha \beta} D _{\alpha}\phi \overline{D _{\beta}\phi} )g_{\mu \nu}, \end{equation} 
	
	\begin{equation} \mathbb{T}^{EM}_{\mu \nu}=g^{\alpha \beta}F _{\alpha \nu}F_{\beta \mu }-\frac{1}{4}F^{\alpha \beta}F_{\alpha \beta}g_{\mu \nu}.
	\end{equation}
	
	For an expression in $(u,v)$ coordinates and explicit computations, c.f.\ Appendix \ref{appendix}.
	
	\subsection{Hardy-type inequalities}
	
	In this paper, we use Hardy-type inequalities numerous times. As explained in section \ref{outline}, this is mainly due to the necessity to absorb the $0 $ order term that appears in the interaction term with the Maxwell field. In this section, we are going to state and prove the  different Hardy-type inequalities that we use throughout the paper. The only objective is to prove the estimates exactly the way they are later used in the paper.
	
	\begin{lem}
		For $R>0$ as in the definition of the foliation and all $u_0(R)\leq u_1<u_2$, the following estimates are true, for any $q \in [0,2)$, and any $r_+<R_1<R$ : 
		\begin{equation} \label{Hardy1}
			\int_{u_i}^{+\infty} \Omega^2 \cdot 2K \cdot  |\phi|^2(u,v_R(u_i))du \leq 4 \int_{u_i}^{+\infty} \Omega^2 \frac{ r^2|D_u\phi|^2(u,v_R(u_i))}{r^2 \cdot 2K}du + 2 \Omega^2(R)  \cdot |\phi|^2 (u_i,v_R(u_i)).
		\end{equation}
		\begin{equation} \label{Hardy2}
			R|\phi|^2(u_i,v_R(u_i))  \leq \Omega^{-2}(R) \int_{v_R(u_i)}^{+\infty} r^2|D_v \phi|^2(u_i,v)dv
		\end{equation}
		\begin{equation} \label{Hardy3}
			\int_{u_i}^{u_{R_1}(v_R(u_i))} \Omega^2 |\phi|^2(u,v_R(u_i))du \leq \frac{4}{\Omega^2(R_1)}\int_{u_i}^{+\infty} r^2|D_u \phi|^2(u,v_R(u_i))du +  2R|\phi|^2(u_i,v_R(u_i)). 
		\end{equation}
		\begin{equation} \label{Hardy4}
			\int_{v_R(u_i)}^{+\infty}  |\phi|^2(u_i,v) dv\leq \frac{4}{\Omega^4(R)}\int_{v_R(u_i)}^{+\infty}  r^2|D_v\phi|^2(u_i,v)dv.
		\end{equation}
		\begin{equation} \label{Hardy5}
			\left(\int_{\mathcal{D}(u_1,u_2)\cap \{r \geq R\}}r^{q-3} \Omega^2  | \psi|^2 du dv\right)^{\frac{1}{2}} \leq               \frac{2}{(2-q)\Omega(R)} \left(\int_{\mathcal{D}(u_1,u_2)\cap \{r \geq R\}}r^{q-1}  |D_v \psi|^2 du dv  \right)^{\frac{1}{2}}  +                       \left( \frac{R^{q-2}}{2-q} \int_{u_1}^{u_2}|\psi|^2 (u,v_R(u)) du \right)^{\frac{1}{2}}.  
		\end{equation}
	\end{lem} 
	
	\begin{proof} We start by \eqref{Hardy1} : after noticing that $\Omega^2 \cdot 2K=  (-\partial_u\Omega^2 )$ we integrate by parts and get :  
		
		$$		\int_{u_i}^{+\infty} \Omega^2 \cdot 2K \cdot  |\phi|^2(u,v_R(u_i))du = 	\int_{u_i}^{+\infty} (-\partial_u\Omega^2 )  |\phi|^2(u,v_R(u_i))du \leq 2 	\int_{u_i}^{+\infty} \Omega^2 \Re(\bar{\phi} D_u \phi )(u,v_R(u_i))du+\Omega^2(R)  \cdot |\phi|^2 (u_i,v_R(u_i)).$$
		
		Then using Cauchy-Schwarz, we obtain $$		\int_{u_i}^{+\infty} \Omega^2 \cdot 2K \cdot  |\phi|^2(u,v_R(u_i))du\leq 2 	(\int_{u_i}^{+\infty}  \Omega^2 \cdot 2K \cdot  |\phi|^2(u,v_R(u_i))du )^{\frac{1}{2}}(\int_{u_i}^{+\infty} \Omega^2 \frac{ |D_u\phi|^2(u,v_R(u_i))}{ 2K}du)^{\frac{1}{2}}+\Omega^2(R)  \cdot |\phi|^2 (u_i,v_R(u_i)).$$

		\color{black}
		This is an inequality of the form $A^2 \leq 2AB + C^2$, where
		
		$$ A^2 =  \int_{u_i}^{+\infty} \Omega^2 \cdot 2K \cdot  |\phi|^2(u,v_R(u_i))du , \hskip 2 mm B^2  = \int_{u_i}^{+\infty} \Omega^2 \frac{ |D_u\phi|^2(u,v_R(u_i))}{ 2K}du, \hskip 2 mm C^2 = \Omega^2(R)  \cdot |\phi|^2 (u_i,v_R(u_i)).  $$
		
		Solving $P(X)=X^2 - 2XB - C^2=0$ as a second order polynomial  equation in $X>0$ ($B$ and $C$ being fixed parameters), we get $X=B + \sqrt{B^2+C^2}$, hence since $P(A) \leq 0$ we get \eqref{Hardy1} under the following form: 
		
		$$ A^2 \leq \left( B + \sqrt{B^2+C^2}\right)^2  \leq 4B^2 +2C^2,$$ where in the last inequality we used $(a+b)^2 \leq 2a^2+b^2$. \color{black}
		
		Now we prove \eqref{Hardy2}. Since $\lim_{v \rightarrow +\infty}   \phi(u,v) =0$, we \footnote{This comes from the finiteness of $\mathcal{E}$, c.f.\ the proof of Proposition \ref{characProp1}.} can write
		
		$$ \phi(u,v) = - \int_{v}^{+\infty} e^{\int_{v}^{v'}iq_0 A_v} D_v\phi(u,v')dv',$$
		
		where we used the fact that $\partial_v( e^{\int_{v_0(R)}^{v}iq_0A_v} \phi)= e^{\int_{v_0(R)}^{v}iq_0A_v} D_v\phi$.
		
		Now using Cauchy-Schwarz we have 
		\begin{equation*}
			|\phi|(u,v) \leq \left( \int_{v}^{+\infty} \Omega^2 r^{-2}(u,v')dv' \right)^{\frac{1}{2}}  \left( \int_{v}^{+\infty} \Omega^{-2} r^{2} |D_v \phi|^2 (u,v')dv'\right)^{\frac{1}{2}}.
		\end{equation*}
		
		After squaring, this directly implies \eqref{Hardy2} under the form
		
		\begin{equation*}
			r |\phi|^2(u,v) \leq  \Omega^{-2}(u,v) \int_{v}^{+\infty} r^{2} |D_v \phi|^2 (u,v')dv'.
		\end{equation*}
		
		Then we turn to \eqref{Hardy3} : after an integration by parts we get 
		
		$$	\int_{u_i}^{u_{R_1}(v_R(u_i))} \Omega^2 |\phi|^2(u,v_R(u_i))du 	 \leq 2 	\int_{u_i}^{+\infty} r\Re(\bar{\phi} D_u \phi )(u,v_R(u_i))du+ R  |\phi|^2 (u_i,v_R(u_i)).$$
		
		Then using Cauchy-Schwarz, we find an inequality of the form $A^2 \leq 2AB + C^2$, where
		
		$$ A^2 =  \int_{u_i}^{u_{R_1}(v_R(u_i))} \Omega^2 |\phi|^2(u,v_R(u_i))du  , \hskip 2 mm B^2  = \int_{u_i}^{u_{R_1}(v_R(u_i))} r^2 \frac{|D_u\phi|^2}{\Omega^2}(u,v_R(u_i))du ,   \hskip 2 mm C^2 = R  |\phi|^2 (u_i,v_R(u_i)). $$
		
		Similarly to \eqref{Hardy1}, this concludes the proof of \eqref{Hardy3}, after we notice that $\Omega^{-2} \leq \Omega^{-2}(R_1)$.

		We now turn to \eqref{Hardy4} : after notice that $\Omega^2(R) \leq  \partial_v r$ we integrate by parts and get 
		
		$$\int_{v_R(u_i)}^{+\infty}  |\phi|^2(u_i,v) dv \leq \Omega^{-2}(R) \int_{v_R(u_i)}^{+\infty}  \partial_v r|\phi|^2(u_i,v) dv \leq -2 \Omega^{-2}(R)\int_{v_R(u_i)}^{+\infty}  r\Re(\bar{\phi} D_v \phi )(u_i,v) dv.$$

		Notice that the sum of the boundary terms is negative because \footnote{This comes from the finiteness of $\tilde{\mathcal{E}}_{1+\eta}$, c.f.\ the proof of Proposition \ref{characProp1}.} $\lim_{r \rightarrow +\infty}	r |\phi|^2 = 0$. Then Cauchy-Schwarz inequality directly gives \eqref{Hardy4}. 
		
		Finally we prove \eqref{Hardy5} : after noticing that $r^{q-3} \Omega^2 = - (2-q)^{-1}\partial_v (r^{q-2})$ and an integration by parts in $v$ we get, using that $q <2$

		$$\int_{\mathcal{D}(u_1,u_2)\cap \{r \geq R\}}r^{q-3} \Omega^2  | \psi|^2 du dv \leq               \frac{2}{(2-q)} \int_{\mathcal{D}(u_1,u_2)\cap \{r \geq R\}}r^{q-2}  \Re( \bar{\psi} D_v \psi)du dv   +                        \frac{R^{q-2}}{2-q} \int_{u_1}^{u_2}|\psi|^2 (u,v_R(u)) du .  $$
		
		Then using Cauchy-Schwarz, we end up having an inequality of the form $A^2 \leq 2AB + C^2$, where
		
		$$ A^2 =  \int_{\mathcal{D}(u_1,u_2)\cap \{r \geq R\}}r^{q-3} \Omega^2  | \psi|^2 du dv ,  \hskip 2 mm B^2  = \frac{1}{(2-q)^2} \int_{\mathcal{D}(u_1,u_2)\cap \{r \geq R\}}r^{q-1} |D_v \psi|^2 du dv,  \hskip 2 mm  C^2 = \frac{R^{q-2}}{2-q} \int_{u_1}^{u_2}|\psi|^2 (u,v_R(u)) du . $$
		
		Then, as seen for the proof of \eqref{Hardy1}, we have $A \leq B + \sqrt{B^2+C^2}$. This implies that $A \leq 2B+C$, which is exactly \eqref{Hardy5}. We used the fact that for all $a,b\geq0$, $\sqrt{a+b} \leq \sqrt{a}+\sqrt{b}.$ This concludes the proof.

	\end{proof}

	\section{Precise statement of the main result} \label{mainresult}
	
	\subsection{The main result}
	
	We now state the results of the present paper, with precise hypothesis. The main result is an almost-optimal energy and point-wise decay statement, for \textbf{small} and point-wise decaying initial data. 
	
	It should be noted that all the required hypotheses are only needed to prove all the claims. 
	
	For example, to prove only \textbf{energy decay, point-wise decay rates are not necessary}, and one can assume weaker $r^p$ weighted initial energy boundedness, c.f.\ the theorems of next sub-sections.
	
	\begin{thm} [Energy and point-wise decay for small and decaying data] \label{maintheorem}
		Consider spherically symmetric Cauchy data $(\phi_0,Q_0)$ on the surface $\Sigma_0=\{t=0\}$  --- on a Reissner--Nordstr\"{o}m exterior space-time of mass $M$ and charge $\rho$ with $0 \leq |\rho| <M$ --- that satisfy the constraint equation \eqref{constraint}.
		
		Assume that the data satisfy the following regularity hypotheses : 
		
		\begin{enumerate}[i]
			
			\item \label{H1} The initial energy is finite : $\mathcal{E} < \infty$, 	where $\mathcal{E}$ is defined in section \ref{energynotation}.
			
			
			
			\item  \label{H2} $Q_0 \in L^{\infty}(\Sigma_0)$ and admits a limit $e_0 \in \mathbb{R}$ as $r \rightarrow +\infty$.
			
			\item \label{H4} We have the following data smallness assumption, for some $\delta>0$ and $\eta>0$ : 
			
			$$ \| Q_0 \|_{L^{\infty}(\Sigma_0)} + \tilde{\mathcal{E}}_{1+\eta} < \delta.$$
			
			
			
			\item  \label{H3} \label{regderivRS} Defining $\psi_0:=r \phi_0$, assume $\psi_0 \in C^{1}(\Sigma_0)$ enjoys \footnote{By Hypothesis \ref{H4}, one can in particular assume $q_0|e_0| < \frac{1}{4}$, without loss of generality.} point-wise decay : 
			
			for all $\epsilon_0>0$, there exist $C_0=C_0(\epsilon_0)>0$, such that
			
			$$ r |D_v \psi_0|(r)+ |\psi_0|(r) \leq C_0 \cdot  r^{\frac{-1-\sqrt{1-4q_0|e_0|}}{2}+\epsilon_0}.$$
			$$ |D_u \phi_0|(r) \leq C_0 \cdot \Omega.$$

		\end{enumerate}
		
		Then, there exists $\delta_0=\delta_0(M,\rho)>0$, $C=C(M,\rho)>0$ such that if $\delta<\delta_0$, the following are true : 
		
		\begin{enumerate}
			
			\item \label{R1} The charge $Q$ admits a future asymptotic charge : there exists $e \in \mathbb{R}$ such that  for every $R_1 > r_+$ :
			
			$$	\lim_{ t \rightarrow + \infty} Q_{ |\mathcal{H}^+}(t)= 	\lim_{ t \rightarrow + \infty} Q_{ |\mathcal{I}^+}(t)= 	\lim_{ t \rightarrow + \infty} Q_{ |\gamma_{R_1} }(t) = e.$$
			
			Moreover \textbf{the charge is small} :  on the whole-space-time
			
			$$ |Q(u,v)-e_0| \leq C \cdot \left(  \| Q_0 \|_{L^{\infty}(\Sigma_0)} + \tilde{\mathcal{E}}_{1+\eta} \right),$$
			
			in particular $|e-e_0| \leq C \cdot \delta$ and $|e| \leq (C+1) \cdot \delta$.
			
			
			\item \label{R2} \textbf{Non-degenerate energy boundedness} \footnote{Here energy boundedness is stated for all $u \geq u_0(R)$ for convenience, because we have a specific foliation but the argument can be slightly modified to prove the boundedness of the analogous quantity when $u < u_0(R)$. } holds :  {there exists $0<\ep\leq C|q_0|\delta$  such that}	 for all $  u_0(R)\leq u_1 < u_2$  : 	\begin{equation} \label{energyintro}
				\begin{split}
					\int_{v_R(u_1)}^{v_R(u_2)} r^2 |D_v \phi|^2_{| \mathcal{H}^+}(v)dv 
					+\int_{u_1}^{u_2} r^2 |D_u \phi|^2_{|\mathcal{I}^+}(u)du 
					+   { \tilde{E}_{{\ep}}(u_2)} \leq C \cdot  \tilde{E}_{{\ep}}(u_1) \leq C^2 \cdot \tilde{\mathcal{E}}_{{\ep}},
				\end{split}
			\end{equation}
			where $\tilde{E}_{\ep}$ and $\tilde{\mathcal{E}}_{\ep}$ are defined in Section~\ref{energynotation}.
			
			\item \label{R3} \textbf{Integrated local decay estimate} holds: there exists $\sigma>1$ such that for all $u_0(R)\leq u_1<u_2$:
			\begin{equation} \label{Morawetzestimateintro}
				\int_{\mathcal{D}(u_1,u_2)}\frac{r^2\Omega^2|D_v\phi|^2+r^2\Omega^{-2}|D_u\phi|^2+\Omega^2|\phi|^2}{r^{\sigma}}\,du\,dv
				\leq C\tilde E_{{\ep}}(u_1)\leq C^2\tilde{\mathcal{E}}_{{\ep}}.
			\end{equation}

			\item \label{R4} \textbf{Weighted $r^p$ energies are bounded and decay}: there
			exists $2<p(e)<3$, $p(e) \rightarrow 3$ as $e \rightarrow 0$ obeying the
			asymptotics of \eqref{Taylorp}, such that for all {$\ep \leq p \leq p(e)+\ep$},
			there exists ${R_{low}}={R_{low}}(e,M,\rho)>r_+$ such that if
			$R>{R_{low}}$, there exists $C'=C'(M,\rho,R,C_0)>0$ such that for all
			$u \geq u_0(R)$ :
			
			$$ \tilde{E}_p(u) \leq C' \cdot |u|^{ {p-\ep-p(e)}},$$
			
			where $\tilde{E}_p$ is defined in section \ref{energynotation}.
			{For $0 \leq p \leq \ep$ one has $\tilde{E}_p(u) \leq C\tilde{E}_{\ep}(u)
				\leq C'\cdot |u|^{-p(e)}$.}

			\item  \label{R5} In particular \textbf{the non-degenerate energy decays in time} : for all $u \geq u_0(R)$ :
			
			\begin{equation}\label{thm0} E(u) \leq C'_0\cdot |u|^{ -{\pe} }. \end{equation}
			
			We also have the decay of the scalar field $ L^2$ flux along constant $r$ curves and the event horizon :  for all $R_0>r_+$, there exists $\tilde{C}=\tilde{C}(R_0,M,\rho,R)>0$ such that for all $v \geq v_0(R)$ :

			\begin{equation} \label{thm1}
				\int_{v}^{+\infty} \left[|\phi|^2 (u_{R_0}(v'), v') +  |D_v\phi|^2(u_{R_0}(v'), v')\right] {dv'} \leq \tilde{C}  \cdot v^{-{\pe}}, \end{equation}		
			\begin{equation} \label{thm12} 
				\int_{v}^{+\infty} \left[|\phi|_{|\mathcal{H}^+}^2 (v') +  |D_v\phi|_{|\mathcal{H}^+}^2 (v')\right] dv'\leq C'_0 \cdot v^{-{\pe} \color{black}}. \end{equation}

			\item \label{R6} 	Finally, we obtain \textbf{point-wise decay estimates} : 
			there exists ${R_{low}}={R_{low}}(e,M,\rho)>r_+$ such that if $R>{R_{low}}$ then there exists $C'=C'(e,R,M,\rho)>0$ such that for all $u>0$, $v>0$ :	\begin{equation}   \label{thm2}
				r^{ \frac{1}{2}}|\phi|(u,v) + r^{ \frac{3}{2}}|D_v\phi|(u,v) \leq C' \cdot  \left(\min \{u,v\}\right)^{-\frac{{\pe}}{2}\color{black}},
			\end{equation}	\begin{equation}  \label{thmRS}
				|D_u\psi| \leq C' \cdot  \Omega^2 \cdot  \left(\min \{u,v\}\right)^{-\frac{{\pe}}{2}},
			\end{equation}	\begin{equation}  \label{thm3}
				|\psi|_{|\mathcal{I}^+}(u) \leq C' \cdot  u^{\frac{1-{\pe}}{2}},
			\end{equation}		 	
			\begin{equation}  \label{thm4}
				|D_v \psi|_{ |  \{ v \geq 2u+R^*\}}(u,v) \leq C' \cdot  v^{-\frac{1+{\pe}}{2}},
			\end{equation}
			\begin{equation}   \label{thm5} |Q - e| (u,v)	\leq C' \cdot \left(	 u^{1-{\pe}} 1_{ \{ r \geq 2r_+\}}+v^{-{\pe}} 1_{ \{ r \leq 2r_+\}}\right),
			\end{equation}
			
			\begin{rmk}
				Notice that having spherically symmetric initial data implies that the whole solution is spherically symmetric, by a standard propagation of symmetries argument.
			\end{rmk}
			\color{black}

			\begin{rmk}
				As explained in sections \ref{main} and \ref{conjecture}, most of the estimates are expected to be (almost) optimal \footnote{\eqref{thm0} is also expected to be optimal in the limit $e \rightarrow 0$, because of its component in the large $r$ region that decays slower than the bounded $r$ term, c.f.\ section \ref{energynotation} to see the precise terms .}, in the limit $e \rightarrow 0$, including  \eqref{thm1}, \eqref{thm12}, \eqref{thm3}, \eqref{thm5}. While \eqref{thm2} is expected to be optimal in a region near null infinity $\{ v \geq 2u+R^* \}$, it is not optimal on any constant $r$ curve or on  the event horizon for reasons explained in sections \ref{conjecture} and \ref{uncharged}. However, the $L^2$ bound on the event horizon from \eqref{thm1} gives point-wise bounds $ |\phi|(v_n)+|D_v\phi|(v_n) \lesssim v_n^{-\frac{1+p(e)}{2}}$ along a \textbf{dyadic sequence} $(v_n)$ that are expected to be \textbf{almost optimal} in the limit $e \rightarrow 0$, c.f.\ section \ref{conjecture}.
			\end{rmk}
			
			\begin{rmk} \label{RSstatementremark}
				Notice that we stated \eqref{thmRS} for $v>0$, in particular far away from $-\infty$. What happens near the bifurcation sphere is more subtle, as explained in Remark \ref{regularomegaremark} and Remark \ref{omegaRS}. Indeed we can prove that for any $V_0 \in \mathbb{R}$, 	$|D_u\psi|(u,v) \lesssim  \Omega^2 $ in $\{ v \geq V_0\}$.  The constant in the inequality blows up as $V_0 \rightarrow -\infty$ however, remarkably one can still prove that $|D_u\psi|(u,v) \lesssim  e^{-2K_+ u} $ on the whole space-time, in conformity with our hypothesis \ref{regderivRS} that the regular derivative \footnote{Thus, the "regular vector field across the bifurcation sphere" is $\Omega^{-1} \partial_t$ \textbf{on the initial surface} $\Sigma_0$.} of the scalar field  $\Omega^{-1}D_t\phi_0 \sim e^{2K_+ u} D_t\phi_0$ is initially bounded. Notice that at fixed $v=V_0$, $\Omega^{-2} \partial_u$ is a regular vector field transverse to the event horizon and it degenerates when approaching the bifurcation sphere $v=-\infty$.
			\end{rmk}

		\end{enumerate}
		
	\end{thm}
	
	This unconditional result is issued as a combination of several other statements, that we believe to be of independent interest. These are the object of the following sub-sections.

	\subsection{Energy boundedness and integrated local energy decay for small charge and scalar field energy data, no point-wise decay assumption}
	
	First we want to emphasize that energy boundedness and integrated local energy decay have nothing to do with point-wise property of the data. This is the content of the following short boundedness theorem : 
	
	\begin{thm}[Boundedness of the energy and integrated local energy decay for small data]  \label{boundednesstheorem}
		In the context of Theorem \ref{maintheorem},	assume hypothesis \ref{H1}, \ref{H2} and \ref{H4} for some $\eta>0$ and $\delta<\delta_0(M,\rho)$ sufficiently small.
		
		Suppose also that  $\lim_{r\rightarrow+\infty} \phi_0(r)=0$.

		Then statements \ref{R1}, \ref{R2} and \ref{R3} are true : the charge stays small, the energy is bounded and local integrated energy decay holds.
	\end{thm}
	\begin{rmk}
		Notice that the energy boundedness statement \eqref{energyintro} actually consists of two estimates : the first inequality is more subtle \footnote{This is because we want to have a right-hand-side that depends only on the scalar field and \textbf{not on the charge}.} to prove than the second, but \textbf{necessary} to later prove decay. The only difficulty of the second inequality is to prove the boundedness of the \textbf{non-degenerate} energy \footnote{The boundedness of the degenerate energy is easy, using the vector field $\partial_t$, if we do not care of initial charge terms.} using the red-shift effect, which is more difficult in this context than for the uncharged case, c.f.\ section \ref{RS}.
	\end{rmk}

	Then there are  a few  conditions we would like to relax in Theorem \ref{maintheorem}, at the cost of other assumptions. This is the object of the next section. 
	
	\subsection{Energy time decay without arbitrary charge smallness or point-wise decay assumption, conditional on energy boundedness}
	In this section, we try to relax some of the assumptions made in Theorem \ref{maintheorem}.

	It is interesting to obtain energy decay \textbf{without} relying on the point-wise decay of the initial data. This is because ultimately we aim to get rid of symmetry assumptions, and outside of spherical symmetry it is notorious that point-wise bounds are harder to propagate. We can prove indeed that energy decay holds with no point-wise decay assumptions of the initial data c.f.\ Theorem \ref{decaytheorem}. \color{black}

	Another important physical fact is that the decay rate should only depend on the future asymptotic charge of the Maxwell equation $e$, and in particular not on the mass of the black hole \footnote{Because we study the gravity uncoupled problem, it should not depend on the charge of the black hole $\rho$ either. However, in the gravity coupled problem $\rho=e$ so the decay should depend on the charge of the black hole this time.}. The first theorem that we state does not do justice to this fact, for the charge is required to be smaller than a constant \textbf{depending on the parameters of the black holes}, in particular on the mass. This is however due to the difficulty to prove energy boundedness estimates in the presence of a large charge. We overcome it using the red-shift effect and the argument requires such a smallness assumption on the initial charge.
	
	In the following result, \underline{assuming the boundedness of the energy and an integrated local energy decay estimate}, we can retrieve energy decay for $q_0|e|$ smaller than a \textbf{universal} numeric constant \footnote{We get some decay estimates already for $q_0|e|< {\frac{1}{12}}$ and an improved, almost optimal decay estimate for $q_0|e|<{0.04}$.}. \\

	We now state the theorem, in which we assume energy boundedness \eqref{energyintro} and an integrated local energy decay estimate \eqref{Morawetzestimateintro}, but \textbf{no point-wise decay} of the data. Smallness of the scalar field initial energy is still required but the initial charge $e_0$ is only required to satisfy $q_0|e_0| <  {\frac{1}{12}}$, which makes in principle the class of admissible initial data larger.

	The following theorem includes two statements. If
	$q_0 |e_0| \leq {0.04}$, we set $\tilde p(e)=p(e)\in(2,3)$ and obtain
	the almost-optimal decay rate $\tilde p(e)$. If
	${0.04}<q_0 |e_0| < {\frac{1}{12}}$, then for every choice
	$$ {1<\tilde p(e)<1+\sqrt{1-4q_0|e|}-\ep} $$
	we obtain the corresponding weaker energy-decay rate $\tilde p(e)$.

	\begin{thm} [Almost optimal decay for a small \underline{scalar field} energy and larger $q_0|e|$] \label{decaytheorem}

		In the same context as for Theorem \ref{maintheorem}, we make the different following assumptions : 
		
		\begin{enumerate}[i]
			
			\item The initial energy is finite : $\mathcal{E} < \infty$,  where $\mathcal{E}$ is defined in section \ref{energynotation}.
			
			\item $Q_0(r)$ admits a limit $e_0 \in \mathbb{R}$ as $r \rightarrow+\infty$.
			\item  The initial asymptotic charge is smaller than an universal constant :  $$q_0|e_0| < {\frac{1}{12}}.$$
			
			\item  We exclude constant solutions by the condition 	$\lim_{r\rightarrow+\infty} \phi_0(r)=0$.
			
			\item We have the finiteness condition : for every $0 \leq p < 2+ \sqrt{1-4q_0|e_0|}$,
			
			$$\tilde{ \mathcal{E}}_{p} < +\infty.$$
			
			\item  We have the smallness condition : for some $\eta>0$ and some $\delta>0$, 
			$$\tilde{ \mathcal{E}}_{1+\eta} < \delta.$$
			
			\item Energy boundedness \eqref{energyintro} and integrated local decay \eqref{Morawetzestimateintro} hold.

		\end{enumerate}	
		
		Then there exists $C=C(M,\rho)>0$, $\delta_0=\delta_0(e_0,M,\rho)>0$ such that if $\delta<\delta_0$  we have the following
		
		\begin{enumerate}
			
			\item \label{D1}  The charge is bounded and  there exists a future asymptotic charge $e \in \mathbb{R}$ such that  for every $R_1 > r_+$ 
			
			$$	\lim_{ t \rightarrow + \infty} Q_{ |\mathcal{H}^+}(t)= 	\lim_{ t \rightarrow + \infty} Q_{ |\mathcal{I}^+}(t)= 	\lim_{ t \rightarrow + \infty} Q_{ |\gamma_{R_1} }(t) = e.$$
			
			Moreover \textbf{the charge is close to its initial asymptotic value} : on the whole-space-time
			
			$$ |Q(u,v)-e_0| \leq C \cdot \tilde{\mathcal{E}}_{1+\eta} ,$$
			
			in particular $|e-e_0| \leq C \cdot \delta$ and therefore if $\delta$ is small enough, $q_0|e| <  {\frac{1}{12}}$.

			\item Moreover, \textbf{boundedness of $r^p$ weighted energies} and \textbf{energy decay} hold:
			
			for every $\ep\leq p\leq \tilde p(e)+\ep$, statements \ref{R4} and
			\ref{R5} of Theorem \ref{maintheorem} are true with $\tilde p(e)$ replacing
			$p(e)$. We recall that $\tilde{p}(e)>2$ if $q_0|e| \leq
			{0.04}$ and $\tilde{p}(e) \rightarrow 3$ as $e \rightarrow 0$.\color{black}

			\item If additionally, one assumes initial point-wise decay under the form : 
			
			for all $\epsilon_0>0$, there exist $C_0=C_0(\epsilon_0)>0$, such that
			
			$$ r |D_v \psi_0|(r)+ |\psi_0|(r) \leq C_0 \cdot  r^{\frac{-1-\sqrt{1-4q_0|e_0|}}{2}+\epsilon_0}.$$
			$$ |D_u \phi_0|(r) \leq C_0 \cdot \Omega,$$
			
			then one can retrieve \textbf{point-wise decay estimates}: statement \ref{R6} of Theorem \ref{maintheorem} is true, with $\tilde{p}(e)$ replacing $p(e)$ in the range ${0.04} < q_0 |e|<  {\frac{1}{12}}$.\color{black}
			
		\end{enumerate}
	\end{thm}

	\begin{rmk}\label{Qe}
		Finally, note that the requirements $q_0|e_0| < {\frac{1}{12}}$ or $q_0|e_0| < {0.04}$  should be understood --- together with the initial scalar field energy smallness--- as $q_0|Q| < {\frac{1}{12}}$, respectively $q_0|Q| <{0.04}$ everywhere.
		
		Thus, it is also equivalent to state  $q_0|e| < \frac{1}{12}$, respectively $q_0|e| < {0.04}\color{black}$, which is what we do in section \ref{decay}.
	\end{rmk}

	\begin{rmk}
		Note that in both theorems, the only smallness initial energy condition is on $\tilde{\mathcal{E}}_{1+\eta}$, to control the variations of $Q$. In particular, \textbf{no smallness condition} is imposed on the initial $r^p$ energies for $p>1+\eta$, even though we require them to be finite. This is related to the fact that the only smallness needed for this problem is that of the charge $Q$ but not of the scalar field.
	\end{rmk}
	
	\begin{rmk}
		Another version of Theorem \ref{decaytheorem} is proven in Appendix \ref{moredecayappendix}. While this different version requires more point-wise decay of the initial data, it also proves more $v$ decay for $D_v \psi$, at the expense of growing $u$ weights. In particular, we can prove that in \footnote{The result that we prove is actually gauge invariant but we state it here in the gauge $A_v=0$ for simplicity. } the gauge $A_v=0$, denoting $X=r^2 \partial_v$, $ X^n \psi(u,v)$ admits a finite limit when $v \rightarrow +\infty$, for fixed $u$. Such a result is reminiscent of the so-called ``peeling estimates'' governing the rate decay in $r$ of the curvature towards null infinity for the Einstein equations.\color{black}
	\end{rmk}

	\section{Reduction of the Cauchy problem to a characteristic problem and global control of the charge} \label{CauchyCharac}
	
	In this section, we explain how the Cauchy problem can be reduced to a characteristic problem, with suitable hypothesis. This is the step that we described earlier as ``boundedness in terms of initial data'', which is simpler than boundedness with respect to be past values that we prove in later sections.
	
	The main object is to show how the smallness of the initial charge propagates to the whole space-time, providing the initial energy of the scalar field is also small.
	
	The results split into four parts : in the first Proposition \ref{characProp1}, we show how the smallness of the scalar field energy \textbf{and the smallness of the initial charge} imply energy boundedness and integrated local energy decay, because the charge stays small everywhere. This is the boundedness part of Theorem \ref{boundednesstheorem}.
	
	Then in Proposition \ref{CharacProp2}, we show that if the two latter hold, then, if the initial energy is small enough, the charge stays close to its limit value at spatial infinity  $e_0$ on the whole space-time, provided  $q_0|e_0| < \frac{1}{4}$, a bound which is \textbf{independent of the black holes parameters}. This part is related to Theorems \ref{decaytheorem}, more precisely to statement \ref{D1}.
	
	In Proposition \ref{propagationdecay}, we show how \underline{certain} initial data point-wise decay rates can be extended towards a fixed forward light cone. This is useful for point-wise decay rates in Theorem  \ref{decaytheorem} or as an alternative  the finiteness assumption of higher order $r^p$ weighted initial energy, like in the statement of Theorem \ref{simplifiedthm}, although this is technically a stronger statement and that \textbf{initial point-wise decay is not needed for energy decay}.
	
	
	Finally in Proposition \ref{lastCharac}, we prove that --- provided the initial $r^p$ weighted energies of large order are finite --- they are still finite on a constant $u$ hyper-surface transverse to null infinity.
	
	To do this, we make use of a $r^p u^{s}$ weighted energy estimate, for $s>0$. 
	
	In the next sections, whose goal is to prove energy $u$ decay, the strategy is in contrast to what we do here : we will aim at estimating the energy in terms of its past values. This is sensibly more difficult than to bound the energy in terms of the data.

	This is why in this section, we only state the results, while their proofs are postponed to appendix \ref{appendixCC}. This allows us to focus on the main difficulty of the paper --- the energy decay --- and to avoid repeating very similar arguments.

	We denote $\Sigma_0 = \{ t=0\}$ the 3-dimensional Riemannian manifold on which we set the Cauchy data $(\phi_0, D_t\phi_0, Q_0)$ satisfying the following constraint equation 
	
	\begin{equation} \label{constraint}
		\partial_{r^*} Q_0 = q_0 r^2 \Im( \bar{\phi_0} D_t \phi_0).
	\end{equation} 

	We first show energy boundedness and global smallness of the charge, on condition that the data is small.

	\begin{prop} \label{characProp1} Suppose that there exists $p>1$ such that $\tilde{\mathcal{E}}_p < \infty$ and that $Q_0 \in L^{\infty}(\Sigma_0)$. 
		
		Assume also that $\lim_{r \rightarrow +\infty} \phi_0(r)=0$.
		
		We denote $Q_0^{\infty} = \| Q_0 \|_{ L^{\infty}(\Sigma_0)}$.
		
		There exists $r_+<R_0=R_0(M,\rho)$ large enough, $\delta=\delta(M,\rho)>0$ small enough, {$0<\ep\leq C|q_0|\delta$,} and $C=C(M,\rho)>0$ so that for all $R_1>R_0$, and if 
		
		$$ Q_0^{\infty}+  \tilde{\mathcal{E}}_p < \delta,$$
		
		then for all $u \geq u_{0}(R_1)$ :
		
		\begin{equation} \label{Charac1}
			E_{R_1}(u) \leq C \cdot \tilde{\mathcal{E}}_{{\ep}}.
		\end{equation}
		Also for all $v \leq v_{R_1}(u)$ : 
		\begin{equation}  \label{Charac4}
			\int_{\tilde{u}(v)}^{+\infty}  \frac{ r^2|D_u \phi|^2}{\Omega^2} (u',v)du'  \leq  C \cdot \tilde{\mathcal{E}}_{{\ep}} ,	\end{equation} 	where $\tilde{u}(v)=u_0(v)$ if $v \leq  v_0(R_1)$ and $\tilde{u}(v)=u_{R_1}(v)$ if $v \geq v_0(R_1)$.

		Moreover for all  $(u,v)$ in the space-time  :
		\begin{equation} \label{Charac3}
			|Q|(u,v)	\leq  C \cdot  \left( Q_0^{\infty} + \tilde{\mathcal{E}}_p \right).
		\end{equation}
		
		Finally there exists $C_1=C_1(R_1,M,\rho)>0$ such that for all $u \geq u_{0}(R_1)$ :
		\begin{equation}  \label{Charac2}
			\int_{         \{ v \leq v_{R_1}(u)\}      \cap \{ r \leq R_1 \}   } \left( |D_{t}\phi|^2(u',v)+|D_{r^*}\phi|^2(u',v)+|\phi|^2(u',v) \right) \Omega^2 du dv \leq C_1 \cdot \tilde{\mathcal{E}}_{{\ep}},
		\end{equation}

	\end{prop}
	
	\begin{rmk}
		As a by product of our analysis , one can show that
		
		\begin{enumerate}
			\item $Q_0$ admits a limit $e_0 \in \mathbb{R}$ when $r \rightarrow+\infty$. This just comes from $\tilde{\mathcal{E}}_p <\infty$ , for $p>1$.
			
			\item The future asymptotic charge exists : there exists $e \in \mathbb{R}$ such that for every $R_1 > r_+$ 
			
			$$		\lim_{ t \rightarrow + \infty} Q_{ |\gamma_{R_1} }(t) = e.$$
			
			
			\item The asymptotic charges are small :  $\max \{ |e| , |e_0|\} \leq C \cdot  \left( Q_0^{\infty} + \tilde{\mathcal{E}}_p \right)<C \cdot \delta $.
		\end{enumerate}
	\end{rmk}
	
	\begin{rmk}
		
		Notice also that no qualitative strong decay is required on the data for the statement \footnote{Even though the finiteness of the energy gives already some mild decay towards spatial infinity.}. The condition $\lim_{r \rightarrow +\infty} \phi_0(r)=0$ is present simply to exclude constant solutions that do not decay. \\
	\end{rmk}
	
	\begin{rmk}
		Notice that \eqref{Charac4} is stated for \textbf{all} $v \leq v_R(u)$, in particular $v$ can be arbitrarily close to $-\infty$. This is consistent with $e^{2K_+u}D_u \phi  \in L^{\infty}$, as we request initially in our hypothesis \ref{regderivRS} and prove everywhere on the space-time in Lemma \ref{Sigma0RSlemma}, c.f.\ also Remark \ref{omegaRS}. This is because --- by \eqref{regularomega} --- if $e^{2K_+u}|D_u \phi| \lesssim 1$ :
		
		$$ \int_{u_0(v)}^{+\infty} \frac{|D_u\phi|^2}{\Omega^2} du \sim e^{-2K_+ v } \int_{u_0(v)}^{+\infty} e^{-2K_+u} |e^{2K_+u}D_u \phi|^2 du \lesssim e^{-2K_+ v } \int_{u_0(v)}^{+\infty} e^{-2K_+u} du \lesssim 1.$$
		This subtlety is related to the degeneration \color{black} of the vector field $\Omega^{-2} \partial_u$ as $v$ tends to $-\infty$, c.f.\ Remark \ref{regularomegaremark}.
	\end{rmk}
	
	Now, we want to prove the following fact : in the far away region, the only smallness condition required on the charge is $q_0|Q| <\frac{1}{4}$. Provided energy boundedness and the integrated local energy decay hold, one can prove that decay of the energy follows {under the stronger condition $q_0|Q| <\frac{1}{12}$}, in the spirit of Theorem \ref{decaytheorem}.
	
	For this, we want to show that, if  $q_0|e_0| <\frac{1}{4}$, then for small enough initial energies, $|Q-e_0|$ is small : 
	
	\begin{prop} \label{CharacProp2}  Suppose that there exists $1<p<2$ such that $\tilde{\mathcal{E}}_p < \infty$.
		
		It follows that there exists $e_0 \in \mathbb{R}$ such that 
		
		$$ \lim_{ r \rightarrow + \infty} Q_0(r) = e_0.$$
		
		Without loss of generality one can assume that $1<p < 1+\sqrt{1-4q_0|e_0|}$.

		Assume also  that $\lim_{r \rightarrow +\infty} \phi_0(r)=0$.
		
		Now assume \eqref{Charac1}, \eqref{Charac4} for $R_1=R$  : there exists $\bar{C}=\bar{C}(M,\rho)>0$ such that for all $u \geq u_0(R)$ and for all $v \leq v_{R}(u)$ : 
		
		\begin{equation}   E(u)=E_R(u) \leq \bar{C} \cdot\tilde{\mathcal{E}}_{{\ep}}.	\end{equation}
		
		\begin{equation}
			\int_{\tilde{u}(v)}^{+\infty}  \frac{ r^2|D_u \phi|^2}{\Omega^2} (u',v)du'  \leq  \bar{C} \cdot\tilde{\mathcal{E}}_{{\ep}},	\end{equation}
		
		where $\tilde{u}(v)=u_0(v)$ if $v \leq  v_0(R)$ and $\tilde{u}(v)=u_R(v)$ if $ \ v_0(R) \leq v  \leq v_R(u)$.
		
		Assume also \eqref{Charac2} : there exists $\bar{R}_0=\bar{R}_0(M,\rho)>r_+$ such that for all $\bar{R}_1>\bar{R}_0$,

		there exists $\bar{C}_1=  \bar{C}_1(\bar{R}_1, M,\rho)>0$ such that 
		
		\begin{equation}
			\int_{         \{ u' \leq u\}      \cap \{ r \leq \bar{R}_1 \}   } \left( |D_{t}\phi|^2(u',v)+|D_{r^*}\phi|^2(u',v)+|\phi|^2(u',v) \right) \Omega^2 du dv \leq \bar{C}_1 \cdot\tilde{\mathcal{E}}_{{\ep}}.
		\end{equation}

		Make also the following smallness hypothesis : for some $\delta>0$ : 	
		
		$$  \tilde{\mathcal{E}}_p < \delta,$$
		$$q_0 |e_0| < \frac{1}{4}.$$
		
		There exists $\delta_0=\delta_0(e_0,M,\rho)>0$ and $C=C(M,\rho)>0$ such that if $\delta<\delta_0$ 	then for all $(u,v)$:
		\begin{equation}	|Q(u,v)-e_0|	\leq  C \cdot   \tilde{\mathcal{E}}_p ,
		\end{equation}
		\begin{equation}
			q_0|Q|(u,v)	<  \frac{1}{4}.
		\end{equation}
		
		Moreover,  there  exists 
		$C'=C'(e_0,p,M,\rho)>0$  such that 
		for all $u \geq u_0(R)$:
		
		\begin{equation} 
			E_{p}[\psi](u) \leq C' \cdot \tilde{\mathcal{E}}_p.
		\end{equation}

	\end{prop}

	\begin{rmk}
		Note that in this context, we also have \begin{equation}
			|e-e_0|	\leq  C \cdot   \tilde{\mathcal{E}}_p ,
		\end{equation}
		so the difference between the initial charge and the asymptotic charge is arbitrarily small, when the initial energies are small. This fact will be used extensively, in particular in the statement of the theorems. \\ \end{rmk}

	
	
	
	In spherical symmetry, some point-wise decay rates propagate, at least in the past of a forward light cone. This is important to derive point-wise decay estimate, because in the proofs of section \ref{pointwise}, an initial decay estimate is needed for $D_v \psi$ on the forward light cone $\{ u=u_0(R)\}$.

	This is also useful to obtain the \textbf{finiteness} \footnote{Note that we do not require the \textbf{boundedness} of such energy by the initial data. Indeed, for the decay, only the finiteness of these energies is required but \underline{not their smallness}, unlike for the smaller $p$ energies which ensure that the charge $Q$ is small. }of the $r^p$ weighted energies, for larger $p$. 
	
	We now prove point-wise bounds in the past of a forward light cone in the following proposition:
	
	\begin{prop} \label{propagationdecay} In the conditions of Proposition \ref{CharacProp2}, assume moreover that there exists $\omega \geq 0$ and $ C_0>0$ such that
		
		$$ r|D_v \psi_0| + |\psi_0| \leq C_0 \cdot r^{-\omega}.$$
		
		Then in the following cases \begin{itemize}
			\item $\omega=1+\theta$ with $\theta>\frac{q_0|e_0|}{4}$
			\item $\omega=\frac{1}{2}+\beta$ with $\beta \in ( - \frac{\sqrt{1-4q_0|e_0|}}{2},\frac{\sqrt{1-4q_0|e_0|}}{2})$, if $q_0|e_0| < \frac{1}{4}$,
		\end{itemize}
		
		there exists $\delta=\delta(e_0,\omega,M,\rho)>0$ and ${R_{low}}={R_{low}}(\omega,e_0,M,\rho)>r_+$ such that if $\tilde{\mathcal{E}}_p<\delta$ and $R>{R_{low}}$ then the decay is propagated: there exists $C'_0=C'_0(C_0,\omega,R,M,\rho,e_0)>0$ such that for all $u \leq u_0(R)$: 
		
		\begin{equation}  \label{ext1}|D_v \psi|(u,v) \leq C'_0 \cdot r^{-1-\omega'}, \end{equation}
		\begin{equation}  \label{ext2} |\psi|(u,v) \leq C'_0 \cdot |u|^{-\omega},\end{equation}
		where $\omega' = \min \{\omega,1\}$.
		
		In that case, for every $0<p<2\omega'+1$, we have the finiteness of the $r^p$ weighted energy on $\mathcal{V}_{u_0(R)}$
		
		$$ E_p[\psi](u_0(R)) < \infty.$$

	\end{prop}
	
	
	\begin{rmk}
		Notice that, even for very decaying data, one cannot obtain a better $r$ decay for $D_v \psi$ than $r^{-2}$. This is in contrast to the uncharged case $q_0=0$ where decay rate was $r^{-3}$ in spherical symmetry, due to the sub-criticality with respect to $r$ decay of the uncharged wave equation, c.f.\ the discussion in section \ref{conjecture}.
	\end{rmk}
	
	\begin{rmk}
		Making an hypothesis on point-wise decay can be thought of as an alternative to circumvent the assumption that $\tilde{\mathcal{E}}_{3-\epsilon} < \infty$ to prove that $ E_{3-\epsilon}[\psi](u_0(R)) < \infty$, like we do in the next proposition. This is why we do not need to assume any  higher order initial $r^p$ weighted energy boundedness for Theorem \ref{maintheorem}: the result we need is already included in the point-wise assumption, that is stronger.
	\end{rmk}

	Finally, we prove that higher order $r^p$ weighted energies boundedness holds on a characteristic constant $u$ surface transverse to null infinity, provided it holds on the initial surface. This proves the boundedness of higher $r^p$ weighted energies, for $p$ close to $3$, in order to close the argument of section \ref{sectionp=3}. 
	
	This also allows us to avoid making initial point-wise decay assumptions on the data, starting with ``weaker'' weighted $L^2$ boundedness hypotheses.
	
	The proof mainly makes use of a $r^p |u|^s$ weighted estimate in the past of a forward light cone.
	
	\begin{prop} \label{lastCharac}
		Suppose that for all $0 \leq p<2+ \sqrt{1-4q_0|e_0|}$,  $\tilde{\mathcal{E}}_{p}<\infty$.
		
		We also assume the other hypotheses of Theorem \ref{decaytheorem}.

		Then for all $0 \leq p<2+ \sqrt{1-4q_0|e_0|}$,  there \footnote{The dependence of $\delta_p$ on $p$ only exists as $p$ approaches $2+\sqrt{1-4q_0|e_0|}$.} exists $\delta_p=\delta_p(e_0,p,M,\rho)>0$,  such that if $\delta<\delta_p$ then  for all $u \leq u_0(R)$: 
		
		$$ E_{p}[\psi](u)<\infty.$$
	\end{prop}

	\section{Energy estimates} \label{energyboundedsection}
	
	In the former section, we explained how the global smallness of the charge can be monitored from assumptions on the initial data, and how various energies on characteristic surfaces or domains are bounded by the data on a space-like initial surface $\Sigma_0$.
	
	Taking this for granted, we turn to the proof of energy decay. We will assume that the charge is suitably small --- according to the needs --- and that the energies on the initial characteristic surface $\mathcal{V}_{u_0(R)}$ are finite. 
	
	The goal of this section \ref{energyboundedsection} and the next section \ref{decay} is to prove $u$ decay for a \underline{characteristic initial value problem} on the domain $\mathcal{D}(u_0(R),+\infty)$ --- c.f.\ Figure \ref{Fig3} --- with data on $\mathcal{V}_{u_0(R)}$ , \textbf{assuming that the charge is sufficiently small everywhere}. \\

	As explained in the introduction proving decay requires three estimates, similarly to the wave equation. 
	
	The first is a robust energy boundedness statement, which takes the  red-shift effect into account. This allows us in principle to prove point-wise boundedness of the scalar field, following the philosophy of \cite{BH}. 
	
	The second is an integrated local energy estimate, also called Morawetz estimate in reference to the seminal work \cite{Morawetz}, which is a global estimate on all derivatives but with a sub-optimal $r$ weight at infinity.
	
	Proving and closing these two estimates in the goal of the present section\footnote{However, we emphasize that, contrary to the wave equation, one needs to already use a $r^p$-weighted estimate to close the integrated local energy decay (and energy boundedness), c.f. \cite{Dejaninverse}}.
	
	Finally the last ingredient which is going to be developed in section \ref{decay} is a $r^p$ weighted estimate, which gives  inverse-polynomial time decay of the energy over the new foliation.

	The main difficulty, compared to the wave equation, is that these estimates are all very coupled. In particular the energy boundedness statement is not established independently and requires the use of the Morawetz estimate \footnote{The converse is always true even for the wave equation: the Morawetz estimate cannot close without the boundedness of the energy. However, for the wave equation the boundedness of the energy is already a closed independent estimate.}, because of the charge terms that cannot be absorbed easily otherwise. These terms are moreover ``critical'' with respect to $r$ decay, more precisely they possess the same $r$ weight as the positive terms controlled by the energy while their sign is not controllable. 
	
	Moreover, whenever an estimate is proven, a term of the form $q_0 Q\Im(\bar{\phi}D\phi)$ arises and these terms necessitate a control of both the zero order term and the derivative at the same time to be absorbed, even with a very small $q_0 Q$. This is already very demanding in the large $r$ region where the estimates are very tight while in the bounded $r$ region, we need to use to use the smallness of $q_0Q$ crucially to absorb the $r$ weights.

	We start by proving a Morawetz estimate, which bounds bulk terms with sub-optimal weights but \textbf{does not control} the boundary terms, proportional to the degenerate energy; it also does not control a zero-order bulk term involving $|\phi|^2$.
	
	Then we prove a Red-Shift estimate, which gives a good control of the regular derivative of the scalar field $\Omega^{-2} D_u \phi$. 
	
	The control of the  zero-order bulk term involving $|\phi|^2$ is eventually provided by a $r^p$-weighted estimate, for small $p>0$.

	Finally, we control the boundary terms using the Killing vector field $\partial_t$.

	The difficulty in this last estimate is the presence of the electromagnetic terms which need to be absorbed in the energy of the scalar field. We require the full strength of the Morawetz and the Red-Shift estimate, together with the smallness of $q_0Q$, to overcome this issue.
	
	Before starting, we recall a few notations from section \ref{energynotation}. We defined the non-degenerate energy $E(u)$ as
	
	$$ E(u)=  \int_{u}^{+\infty} r^2 \frac{|D_u \phi|^2}{\Omega^2}(u',v_R(u))du'+\int_{v_R(u)}^{+\infty} r^2 |D_v \phi|^2(u,v)dv. $$
	
	It will also be convenient to have a short notation $E^+(u_1,u_2)$ for the sum of all the boundary terms having a role in the energy identity: for all $u_0(R) \leq u_1 \leq u_2 $:
	
	$$ E^+(u_1,u_2):= E(u_2)+E(u_1)+ 	\int_{v_R(u_1)}^{v_R(u_2)} r^2 |D_v \phi|^2_{|\mathcal{H}^+}(v)dv 
	+\int_{u_1}^{u_2} r^2 |D_u \phi|^2_{|\mathcal{I}^+}(u)du. $$
	
	It will also be useful for the Red-Shift estimate to have an equivalent notation for the degenerate energy 
	
	$$ E_{deg}(u)=  \int_{u}^{+\infty} r^2 |D_u \phi|^2(u',v_R(u))du'+\int_{v_R(u)}^{+\infty} r^2 |D_v \phi|^2(u,v)dv, $$
	
	and for the sum of all scalar field terms appearing when one contracts $\mathbb{T}_{\mu \nu}^{SF}$ with the vector field $\partial_t$:

	$$ E_{deg}^+(u_1,u_2):= E_{deg}(u_2)+E_{deg}(u_1)+ 	\int_{v_R(u_1)}^{v_R(u_2)} r^2 |D_v \phi|^2_{|\mathcal{H}^+}(v)dv 
	+\int_{u_1}^{u_2} r^2 |D_u \phi|^2_{|\mathcal{I}^+}(u)du. $$

	\subsection{An integrated local energy  {and red-shift} estimate} \label{Morawetz}
	
	{The goal of the following estimate is to control the first derivatives term {$|D\phi|^2$} and the zero-order term $|\phi|^2$. We proceed in three steps, using the vector fields $X_{\alpha}=-r^{-\alpha}\partial_{r^*}$ (Morawetz), $X_{RS}=r^{-\alpha}\Omega^{-2}\partial_u$ (red-shift), and $X_p=r^p\partial_v$ ($r^p$-weighted) {for $p \in (0,1)$ small proportional to $q_0 |e|$}.}
	
	{More precisely, we prove the following:} 
	
	\begin{prop}  \label{PropositionM} {Assume that $\|Q \|_{L^{\infty}(\mathcal{D}(u_0(R),+\infty))} < \delta$, for $|q_0| \delta$ {sufficiently small (depending on $M$ and $\rho$ only)}, $\alpha>4$ and that there exists $ R_0(M,\rho) >r_+$ such that $R \geq R_0$, then there exists $D(M,\rho)>0${, $C(R)>0$} such that for all $ u_0(R) \leq  u_1 < u_2$ we have: }
		\begin{equation*}
			\begin{split}
				&{\int_{\mathcal D(u_1,u_2)}r^{1-\alpha}\left(\frac{|D_u\phi|^2}{\Omega^2}+\Omega^2|D_v\phi|^2\right)\,du\,dv
					+\int_{\mathcal D(u_1,u_2)}\Omega^2r^{-1-\alpha}|\phi|^2\,du\,dv}\\
				&\quad+{\int_{u_2}^{+\infty}\frac{r^{2-\alpha}|D_u\phi|^2}{\Omega^2}(u',v_R(u_2))\,du'}
				+{\int_{v_R(u_2)}^{+\infty}r^{\ep}|D_v\psi|^2(u_2,v)\,dv}\\
				&\leq D\left(E_{deg}^+(u_1,u_2)
				+{\int_{u_1}^{+\infty}\frac{r^{2-\alpha}|D_u\phi|^2}{\Omega^2}(u',v_R(u_1))\,du'
					+\int_{v_R(u_1)}^{+\infty}r^{\ep}|D_v\psi|^2(u_1,v)\,dv}\right){+C(R) E_{deg}(u_1)}.
			\end{split}
		\end{equation*}
		{Here $\ep=8q_0\delta>0$.}
	\end{prop}
	
	{Throughout this section, we will always assume that there exists $\delta>0$ such that for all $(t,r) \in \mathcal{D}(u_1,u_2)$  \begin{equation}
			|Q|(t,r) \leq \delta
	\end{equation}}

	\subsubsection{{An integrated estimate controlling first-order terms}} {In this section, we use the Morawetz vector field $X_{\alpha}=-r^{-\alpha} \partial_{r^*}$}
	{\begin{lem}\label{lem.ILED.easy} The following integrated local energy decay estimate (with a right-hand-side) holds: there exists $C(M,\rho)>0$ such that for $\alpha>1$
			\begin{equation}\label{ILED.derivatives}\begin{split}
					&\frac{\alpha}{8}\iint_{\mathcal{D}(u_1,u_2)} \Omega^2\left( r^{-\alpha+1} \left[  |D_{u} \phi|^2 + |D_v\phi|^2\right] \right) du dv \leq  [q_0 \delta]^2 \iint_{\mathcal{D}(u_1,u_2)} \Omega^2 |\phi|^2  r^{-\alpha-1} du dv +C E_{deg}^{+}(u_1,u_2).\end{split}
			\end{equation}
		\end{lem}
	}

	\begin{proof}
		
		We start by a computation, based on  identities of section \ref{T} and  \ref{Pi} and on \eqref{current2} {with $f(r) = r^{-\alpha}$.} { Integrating  \eqref{current2} on $\mathcal{D}(u_1,u_2)$ and integrating gives \begin{equation}\begin{split}
					&\iint_{\mathcal{D}(u_1,u_2)} \Omega^2 r^{-\alpha+1} \left[  [\alpha+\frac{1}{2}] |D_{r^{*}} \phi|^2 + [\alpha-\frac{1}{2}]  |D_t\phi|^2\right]  du dv=\iint_{\mathcal{D}(u_1,u_2)} \Omega^2 q_0 Q \Im(\bar{\phi} D_t \phi) r^{-\alpha} du dv +\tilde{E},\end{split}
			\end{equation} {where $\tilde{E}$ accounts for the boundary terms}. Then, by the Cauchy-Schwarz inequality, \begin{equation*}
				\bigl|	\iint_{\mathcal{D}(u_1,u_2)} \Omega^2 q_0 Q \Im(\bar{\phi} D_t \phi) r^{-\alpha} du dv\bigr| \leq \frac{\alpha-\frac{1}{2}}{2} \iint_{\mathcal{D}(u_1,u_2)} \Omega^2 | D_t \phi|^2 r^{-\alpha+1} du dv +\frac{[q_0 \delta]^2}{2\alpha-1} \iint_{\mathcal{D}(u_1,u_2)} \Omega^2 |\phi|^2 r^{-\alpha-1} du dv
			\end{equation*} so that, using the fact that $\alpha>1$ by assumption  \begin{equation*}\begin{split}
					&\frac{\alpha}{4}\iint_{\mathcal{D}(u_1,u_2)} \Omega^2 r^{-\alpha+1} \left[  |D_{r^{*}} \phi|^2 + |D_t\phi|^2\right]  du dv \leq  \frac{[q_0 \delta]^2}{\alpha} \iint_{\mathcal{D}(u_1,u_2)} \Omega^2 |\phi|^2  r^{-\alpha-1} du dv +\tilde{E}.\end{split}
		\end{equation*} }

		{Finally, since $\alpha>0$, we notice } that for some $C=C(M,\rho)>0$, 
		
		$$  \tilde{E}  \leq C \cdot E_{deg}^+(u_1,u_2),$$
		{which concludes the proof.}
	\end{proof}

	\subsubsection{{A Red-shift estimate}} \label{RS}
	
	In this section, we are going to prove that the non-degenerate energy near the event horizon is bounded by the degenerate energy. This echoes with the red-shift estimates pioneered in \cite{BH} and \cite{Redshift}. The main difference here is that we cannot obtain  non-degenerate energy boundedness as an independent statement, due to the presence of the charge term. 
	
	This red-shift estimate is proved using the vector field $\Omega^{-2}{r^{-\alpha}} \partial_u$, which is regular across the event horizon {and decays in $r$ at the same rate as the vector field $r^{-\alpha} [\rd_v -\rd_u]$ used for the Morawetz estimate}.		
	
	We prove the following {lemma}: 
	\begin{lem} \label{RESProp} 
		
		{There exist constants $C_{B}(M,\rho)>0$ and $C_0(M,\rho)>0$ such that the following inequality holds for all $\alpha\geq 4$:}
		
		{	\begin{equation*}  \begin{split}
					C_{B} \left[	\int_{\mathcal{D}(u_1,u_2)  } {r^{-\alpha+1}}\frac{ |D_u \phi|^2 }{\Omega^2}  dudv + \int_{\mathcal{D}(u_1,u_2)  } {r^{-\alpha+1}}\Omega^2|D_v \phi|^2   dudv\right]+ \int_{{u_2}}^{+\infty} \frac{r^{2-\alpha} |D_u \phi|^2}{\Omega^2} (u',v_R(u_2))du'\\ \leq \int_{{u_1}}^{+\infty} \frac{r^{2-\alpha} |D_u \phi|^2}{\Omega^2} (u',v_R(u_1))du'+{C(M,\rho)\left(  (q_0 \delta)^2 \int_{\mathcal{D}(u_1,u_2)  } r^{-\alpha-1} \Omega^2 | \phi|^2 du dv +  E^{+}_{deg}(u_1,u_2) \right)}.
				\end{split}
		\end{equation*}}
	\end{lem}
	
	\begin{proof}

		We are going to consider the vector field $X_{RS} =  \frac{{r^{-\alpha}}}{\Omega^2} \partial_u$. Then we get 
		
		\begin{equation*}
			\nabla^{\mu}J_{\mu}^{X_{RS}}(\phi) =\Omega^{-2}{r^{-\alpha}} [\frac{4K(r)}{\Omega^{2}}{+\alpha r^{-1}}]  |D_u \phi|^2 {-\frac{r^{-1-\alpha}}{\Omega^2} \Re(D_u \phi \overline{D_v\phi})}  + \frac{q_0 Q}{r^2 \Omega^2} {r^{-\alpha}}\Im (\bar{\phi} D_u\phi),
		\end{equation*}
		
		We now integrate this identity, multiplied \footnote{$dvol$ is defined in appendix \ref{appendix}.} by $dvol$, on {$\mathcal{D}(u_1,u_2)$}{:} 
		
		\begin{equation} \begin{split}
				2\int_{\mathcal{D}(u_1,u_2)  } {r^{-\alpha}}\left([\frac{4K(r) r^2}{\Omega^{2}} {+\alpha r}] |D_u \phi|^2 {-r \Re(D_u \phi \overline{D_v\phi})} + q_0 Q \Im (\bar{\phi} D_u\phi) \right) dudv \\ + \int_{{u_2}}^{+\infty}  \left( \frac{r^{2{-\alpha}} |D_u \phi|^2}{\Omega^2} \right) (u',v_R(u_2))du'     = \int_{{u_1}}^{+\infty}  \left( \frac{r^{2{-\alpha}} |D_u \phi|^2}{\Omega^2} \right) (u',v_R(u_1))du'.\end{split}
		\end{equation}

		{By the Cauchy-Schwarz inequality, we obtain from the above
			
			\begin{equation*} \begin{split}
					\int_{\mathcal{D}(u_1,u_2)  } {r^{-\alpha+1}}[4K(r) r {+\alpha \Omega^2 }] \frac{|D_u \phi|^2}{\Omega^2}   dudv +\int_{{u_2}}^{+\infty}   \frac{r^{2-\alpha} |D_u \phi|^2}{\Omega^2}   (u',v_R(u_2))du' \leq \int_{{u_1}}^{+\infty} \frac{r^{2-\alpha} |D_u \phi|^2}{\Omega^2} (u',v_R(u_1))du'\\ + 4 \left[	\int_{\mathcal{D}(u_1,u_2)  } r^{-\alpha+1} \Omega^2[4K(r) r+\alpha\Omega^2]^{-1} |D_v \phi|^2 du dv + (q_0 \delta)^2 \int_{\mathcal{D}(u_1,u_2)  } r^{-\alpha-1} \Omega^2[4K(r) r +\alpha\Omega^2]^{-1} | \phi|^2 du dv  \right]. \end{split}
			\end{equation*}
		}
		{		Note that, for $\alpha\geq 4$, $r\rightarrow 4K(r)r + \alpha \Omega^2(r)$ is an increasing function, therefore \begin{equation}
				4K(r)r + \alpha \Omega^2(r) \geq 4K_+ r_+
			\end{equation} so 
			\begin{equation*} \begin{split}
					4K_+ r_+	\int_{\mathcal{D}(u_1,u_2)  } {r^{-\alpha+1}}\frac{ |D_u \phi|^2 }{\Omega^2}  dudv +\int_{{u_2}}^{+\infty}   \frac{r^{2-\alpha} |D_u \phi|^2}{\Omega^2}   (u',v_R(u_2))du' \leq \int_{{u_1}}^{+\infty} \frac{r^{2-\alpha} |D_u \phi|^2}{\Omega^2} (u',v_R(u_1))du'\\ + [K_+ r_+]^{-1} \left[	\int_{\mathcal{D}(u_1,u_2)  } r^{-\alpha+1} \Omega^2|D_v \phi|^2 du dv + (q_0 \delta)^2 \int_{\mathcal{D}(u_1,u_2)  } r^{-\alpha-1} \Omega^2 | \phi|^2 du dv  \right]. \end{split}
			\end{equation*}
		}

		Thus, we obtain, 
		coupling with \eqref{ILED.derivatives} from Lemma~\ref{lem.ILED.easy} {	\begin{equation*} \begin{split}
					2K_+ r_+ 	\int_{\mathcal{D}(u_1,u_2)  } {r^{-\alpha+1}}\frac{ |D_u \phi|^2 }{\Omega^2}  dudv + \frac{2}{K_+ r_+} \int_{\mathcal{D}(u_1,u_2)  } {r^{-\alpha+1}}\Omega^2|D_v \phi|^2   dudv+ \int_{{u_2}}^{+\infty} \frac{r^{2-\alpha} |D_u \phi|^2}{\Omega^2} (u',v_R(u_2))du'\\ \leq \int_{{u_1}}^{+\infty} \frac{r^{2-\alpha} |D_u \phi|^2}{\Omega^2} (u',v_R(u_1))du' + {C(M,\rho)\left(  (q_0 \delta)^2 \int_{\mathcal{D}(u_1,u_2)  } r^{-\alpha-1} \Omega^2 | \phi|^2 du dv 
						+  E^{+}_{deg}(u_1,u_2) \right)} . \end{split}
		\end{equation*}}

		This concludes the proof of Lemma~\ref{RESProp}.

	\end{proof}
	
	{\subsubsection{$r^p$-weighted energy estimate}}
	
	We now {prove a $r^p$-weighted estimate for small $p\in [\ep,2-\ep]$, which allows us to close the red-shift and Morawetz estimate of the earlier sections.}
	
	\begin{lem}\label{lemma.rp.first} Assuming $0<q_0 \delta< \frac{1}{64}$, there exists $R_0(M,\rho)>r_+$, such that for all $R{>} R_0${, there exists $C(R)>0$ such that} for  $p  =8 q_0\delta$, the following estimate holds 
		for all $u_0(R) \leq u_1 < u_2$ 
		\begin{equation}\label{est.lemma.rp.first} \begin{split}
				\int_{\mathcal{D}(u_1,u_2) \cap \{r \geq {R_0}\}} {\Omega^2(r)}[ q_0 \delta\ r^{p-1}  |D_v \psi|^2 {+ M r^{p-4}|\psi|^2 }] du dv  + \int_{v_{R}(u_2)}^{+\infty} r^{p} \ |D_v \psi|^2  (u_2,v)dv\\ {+\frac{{R_0}^{p}}{8}}\int_{{v_R(u_1)}}^{{v_R(u_2)}}|\phi|^2 (u_{{R_0}}(v),v)dv {\leq} \int_{v_{R}(u_1)}^{+\infty}r^{p} \ |D_v \psi|^2   (u_1,v)dv \\ {+{R_0}^{p+2}}\int_{{v_R(u_1)}}^{{v_R(u_2)}}|D_v\phi|^2 (u_{{R_0}}(v),v) dv{+C(R) E_{deg}(u_1)}.  \end{split} \end{equation}  
	\end{lem}

	\begin{proof}{Let $R_0(M,\rho) > 32 r_+$.} We multiply \eqref{wavev} by $2r^p \overline{D_v \psi}$ and take the real part. We get:

		\begin{equation}\begin{split}
				&	{		\rd_u( r^{p} |D_v \psi|^2) + p \Omega^2 r^{p-1} |D_v \psi|^2 + \rd_v(\Omega^2 r^{p-3} [2M-\frac{2\rho^2}{r}]|\psi|^2)}\\&+ \Omega^2 \left[2M(3-p)-2(4-p)\frac{\rho^2{+2M^2}}{r}{+\frac{6M[5-p]\rho^2}{r^2}-\frac{2[6-p]\rho^4}{r^3}}\right] r^{p-4} |\psi|^2  = 2q_0 Q \Omega^2 r^{p-2}\Im(\bar{\psi} D_v \psi).\end{split}
		\end{equation} 
		We integrate this identity on $\mathcal{D}(u_1,u_2) \cap \{r \geq {R_0}\}$ and after a few integrations, we get 
		
		\begin{equation*} \begin{split}
				&	\int_{\mathcal{D}(u_1,u_2) \cap \{r \geq {R_0}\}} {\Omega^2(r)}\left( pr^{p-1}  |D_v \psi|^2 + \left(2M(3-p)-2(4-p)\frac{\rho^2{+2M^2}}{r}{+\frac{6M[5-p]\rho^2}{r^2}-\frac{2[6-p]\rho^4}{r^3}}\right)r^{p-4}|\psi|^2 \right) du dv \\  &+ \int_{v_{{R}}(u_2)}^{+\infty} r^{p} \ |D_v \psi|^2  (u_2,v)dv {+}\int_{v_R(u_1)}^{v_R(u_2)}\left(-\Omega^2({R_0})(2M-\frac{2\rho^2}{{R_0}}){R_0}^{p-3}|\psi|^2 (u_{R_0}(v),v)+ {{R_0}^{p}|D_v\psi|^2 (u_{R_0}(v),v)}\right)dv \\ &+  \int_{u_1}^{u_2}\Omega^2\left(2M-\frac{2\rho^2}{r}\right)r^{p-3}|\psi|^2_{\mathcal{I}^+} (u)du   {+\int_{u_2}^{u_{R_0}(v_R(u_2))} \Omega^2 r^{p-3} [2M-\frac{2\rho^2}{r}]|\psi|^2(u,v_R(u_2)) du}\\& {=} {\int_{u_1}^{u_{R_0}(v_R(u_1))} \Omega^2 r^{p-3} [2M-\frac{2\rho^2}{r}]|\psi|^2(u,v_R(u_1)) du+} \int_{\mathcal{D}(u_1,u_2)\cap \{r \geq {R_0}\}} 2q_0 Q \Omega^2 r^{p-2} \Im(\psi \overline{D_v \psi}) du dv + \int_{v_{R}(u_1)}^{+\infty}r^{p} \ |D_v \psi|^2   (u_1,v)dv.  \end{split} \end{equation*}
		
		Now, note that, since $D_v \psi = \Omega^2 \phi + r D_v \phi$, we have \begin{equation*}
			r^p |D_v \psi|^2 = r^{p+2} |D_v \phi|^2 + \Omega^4 r^{p} |\phi|^2 + 2 r^{p+1}\Omega^2 \Re(\bar{\phi} D_v \phi)
		\end{equation*} so that\footnote{{To the best of the author's knowledge, this observation was first recorded in \cite{Dejaninverse} in the context of inverse square potentials.}}, using the Cauchy-Schwarz inequality, we obtain {for all $ v_R(u_1)\leq v\leq v_R(u_2)$} \begin{equation*}\begin{split}
				&-\Omega^2({R_0})(2M-\frac{2\rho^2}{{R_0}}){R_0}^{p-3}|\psi|^2 (u_{R_0}(v),v)+ {{R_0}^{p}|D_v\psi|^2} (u_{R_0}(v),v)\\ & ={R_0}^{p+2} |D_v \phi|^2(u_{R_0}(v),v)+ \Omega^2({R_0}){R_0}^{p} \left[ 1 -\frac{4M}{{R_0}} + \frac{3\rho^2}{{R_0}^2} \right]|\phi|^2(u_{R_0}(v),v) +  2 {R_0}^{p+1} \Omega^2({R_0}) \Re(\bar{\phi} D_v \phi)(u_{R_0}(v),v) \\ & \geq  \Omega^2({R_0}){R_0}^{p} \left[ \frac{1}{2} -\frac{4M}{{R_0}} + \frac{3\rho^2}{{R_0}^2} \right]|\phi|^2(u_{R_0}(v),v) -[2\Omega^2({R_0})-1] {R_0}^{p+2} |D_v \phi|^2(u_{R_0}(v),v). \end{split}
		\end{equation*} Then, {since $p\in(0,2)$, $2M-\frac{2\rho^2}{r}\geq 0$ and $R_0 > 32r_+$ gives that $\Omega^2(R_0)(\frac{1}{2} -\frac{4M}{{R_0}} + \frac{3\rho^2}{{R_0}^2}) \geq \frac{1}{4}$ and $2M(3-p)-2(4-p)\frac{\rho^2{+2M^2}}{r}{+\frac{6M[5-p]\rho^2}{r^2}-\frac{2[6-p]\rho^4}{r^3}} \geq M$,} we find that 	\begin{equation*} \begin{split}
				\int_{\mathcal{D}(u_1,u_2) \cap \{r \geq  {R_0}\}} {\Omega^2(r)}[ pr^{p-1}  |D_v \psi|^2 + {M r^{p-4}|\psi|^2} ]  du dv  + \int_{v_{R}(u_2)}^{+\infty} r^{p} \ |D_v \psi|^2  (u_2,v)dv {+\frac{{R_0}^{p}}{4}}\int_{v_R(u_1)}^{v_R(u_2)}|\phi|^2 (u_{R_0}(v),v)dv \\ {\leq}  \int_{\mathcal{D}(u_1,u_2)\cap \{r \geq {R_0}\}} 2q_0 Q \Omega^2 r^{p-2} \Im(\psi \overline{D_v \psi}) du dv + \int_{v_{R}(u_1)}^{+\infty}r^{p} \ |D_v \psi|^2   (u_1,v)dv\\ {+{R_0}^{p+2}}\int_{{v_R(u_1)}}^{{v_R(u_2)}}|D_v\phi|^2 (u_{R_0}(v),v)dv+ { 2M\int_{u_1}^{u_{R_0}(v_R(u_1))} \Omega^2 r^{p-1} |\phi|^2(u,v_R(u_1)) du}.  \end{split} \end{equation*}
		
		\noindent Then, by Hardy inequality in $v$, we get for $0<p<2$ \begin{equation*}\begin{split}
				&\int_{\mathcal{D}(u_1,u_2)\cap \{r \geq {R_0}\}} \Omega^2 r^{p-3}|\psi|^2 du dv \leq \frac{2{R^{p}_0}}{2-p} \int_{v_R(u_1)}^{v_R(u_2)}|\phi|^2(u_{R_0}(v),v)dv+  {\frac{2}{2-p} \int_{u_1}^{u_{R_0}(v_R(u_1))}r^{p}|\phi|^2(u,v_{R}(u_1))du}\\ & +  \frac{4}{\Omega^2({R_0})[2-p]^2} \int_{\mathcal{D}(u_1,u_2)\cap \{r \geq {R_0}\}} \Omega^2 r^{p-1}|D_v\psi|^2 du dv,\end{split}
		\end{equation*} thus by Cauchy-Schwarz and combining with the above we get, for any $\epsilon>0$ (to be chosen later)  \begin{equation*} \begin{split}
				\int_{\mathcal{D}(u_1,u_2) \cap \{r \geq {R_0}\}} {\Omega^2(r)} pr^{p-1}  |D_v \psi|^2   du dv  + \int_{v_{R}(u_2)}^{+\infty} r^{p} \ |D_v \psi|^2  (u_2,v)dv {+\frac{{R_0}^{p}}{4}}\int_{{v_R(u_1)}}^{{v_R(u_2)}}|\phi|^2 {(u_{R_0}(v),v)dv} \\  {\leq}  2q_0 \delta \left( [\frac{2}{\Omega({R_0})[2-p]}+ \frac{\epsilon}{2}] \int_{\mathcal{D}(u_1,u_2)\cap \{r \geq {R_0}\}} \Omega^2 r^{p-1}|D_v \psi|^2 du dv + \frac{{R_0}^p}{2[2-p]\epsilon}\int_{{v_R(u_1)}}^{{v_R(u_2)}}|\phi|^2 (u_{R_0}(v),v)\right)\\+\int_{v_{R}(u_1)}^{+\infty}r^{p} \ |D_v \psi|^2   (u_1,v)dv {+{R_0}^{p+2}}\int_{u_1}^{u_2}|D_v\phi|^2 (u,v_R(u))du{+ \int_{u_1}^{u_{R_0}(v_R(u_1))} \Omega^2  [\frac{2M}{r} + \frac{q_0 \delta}{[2-p]\epsilon} ]r^p|\phi|^2(u,v_R(u_1)) du},  \end{split} \end{equation*} so choosing $\epsilon = \frac{8 q_0 \delta}{\Omega({R_0})[2-p]}$ gives
		\begin{equation*} \begin{split}
				\int_{\mathcal{D}(u_1,u_2) \cap \{r \geq {R_0}\}} {\Omega^2(r)} pr^{p-1}  |D_v \psi|^2   du dv  + \int_{v_{R}(u_2)}^{+\infty} r^{p} \ |D_v \psi|^2  (u_2,v)dv {+\frac{R_0^{p}}{8}}\int_{v_R(u_1)}^{v_R(u_2)}|\phi|^2 (u_{R_0}(v),v)dv\\ \leq   \frac{12 |q_0| \delta }{\Omega({R_0})[2-p]} \int_{\mathcal{D}(u_1,u_2)\cap \{r \geq {R_0}\}} \Omega^2 r^{p-1}|D_v \psi|^2 du dv +\int_{v_{R}(u_1)}^{+\infty}r^{p} \ |D_v \psi|^2   (u_1,v)dv\\ {+R_0^{p+2}}\int_{v_R(u_1)}^{v_R(u_2)}|D_v\phi|^2 (u_{R_0}(v),v)dv {+ \frac{1}{4}\int_{u_1}^{u_{R_0}(v_R(u_1))} \Omega^2 r^p|\phi|^2(u,v_R(u_1)) du}.  \end{split} \end{equation*}  
		{Now, by Hardy inequality in $u$, we can write \begin{equation*}\begin{split}
					\int_{u_1}^{u_{R_0}(v_R(u_1))} \Omega^2 r^p|\phi|^2(u,v_R(u_1)) du \lesssim R^{p+1} |\phi|^2(u_1,v_R(u_1)) + \Omega^{-2}(R_0) \int_{u_1}^{u_{R_0}(v_R(u_1))} \Omega^2 r^{p+2}|D_u\phi|^2(u,v_R(u_1))\lesssim R^p E_{deg}(u_1),
				\end{split}
			\end{equation*} where in the last inequality, we have used the inequality $\Omega r^{1/2}|\phi|(u,v) \ls [\int_{v}^{+\infty} r^2 |D_v \phi|^2(u,v')dv']^{1/2}$ applied to $(u,v)= (u_1,v_R(u_1))$, which can be easily proven using Cauchy-Schwarz.} This concludes the proof of the lemma.
		
	\end{proof}
	\begin{rmk}
		{Note that it is important in the above proof that $R_0$ is taken to be larger than a constant only depending on $M$ and $\rho$, while $R>R_0$ is left free; this is because in Section~\ref{decay}, $R$ will be taken larger than a constant depending on $q_0\delta$ so the difference between $R_0$ and $R$ is important in the above.}
	\end{rmk}
	\subsubsection{Conclusion of the proof of Proposition~\ref{PropositionM}}
	
	{	Next, we prove Proposition~\ref{PropositionM}: by combining Lemma~\ref{lem.ILED.easy} and Lemma~\ref{RESProp}, we have \begin{equation}\label{final.proof.M}  \begin{split}
				&	C_{B} \left[	\int_{\mathcal{D}(u_1,u_2)  } {r^{-\alpha+1}}\frac{ |D_u \phi|^2 }{\Omega^2}  dudv + \int_{\mathcal{D}(u_1,u_2)  } {r^{-\alpha+1}}\Omega^2|D_v \phi|^2   dudv\right]+ \int_{{u_2}}^{+\infty} \frac{r^{2-\alpha} |D_u \phi|^2}{\Omega^2} (u',v_R(u_2))du'\\& \leq \int_{{u_1}}^{+\infty} \frac{r^{2-\alpha} |D_u \phi|^2}{\Omega^2} (u',v_R(u_1))du' + C E^{+}_{deg}(u_1,u_2){
					+C(M,\rho)(q_0 \delta)^2\int_{\mathcal{D}(u_1,u_2)}  \Omega^2 r^{-\alpha-1} |\phi|^2 du dv}. 
			\end{split}
		\end{equation}

		{To address the zero order bulk term in the above inequality, we distinguish the regions $\{r\leq R_0\}$ and $\{r \geq R_0\}$, where $R_0(M,\rho)>0$ is defined in the proof of Lemma~\ref{RESProp}. Note that by Hardy inequality in $u$: \begin{equation*}\begin{split}
					&\int_{\mathcal{D}(u_1,u_2) \cap \{r\leq R_0\}}  \Omega^2 r^{-\alpha-1} |\phi|^2 du dv \leq r_+^{-\alpha-1} \int_{\mathcal{D}(u_1,u_2)\cap \{r\leq R_0\}}  \Omega^2  |\phi|^2 du dv\\ & \leq r_+^{-\alpha-1} \left( 2  R_0 \int_{v_R(u_1)}^{v_R(u_2)} |\phi|^2(u_{R_0}(v),v) dv+ 4  R_0^2 \int_{\mathcal{D}(u_1,u_2)\cap \{r\leq R_0\}}  \frac{|D_u \phi|^2}{\Omega^2}  du dv\right) \\ &\leq  r_+^{-\alpha-1} \left( 2 R_0  \int_{v_R(u_1)}^{v_R(u_2)} |\phi|^2(u_{R_0}(v),v) dv+ 4  R_0^{\alpha+1} \int_{\mathcal{D}(u_1,u_2)\cap \{r\leq R_0\}}  \frac{r^{-\alpha+1}|D_u \phi|^2}{\Omega^2}  du dv\right). \end{split}
			\end{equation*} Now, in the region $\{r\geq R_0\}$, we simply write  \begin{equation*}\begin{split}
					\int_{\mathcal{D}(u_1,u_2) \cap \{r\geq R_0\}}  \Omega^2 r^{-\alpha-1} |\phi|^2 du dv \leq R_0^{-\alpha+1-p} \int_{\mathcal{D}(u_1,u_2) \cap  \{r\geq R_0\}}  \Omega^2 r^{p-2} |\phi|^2 du dv.\end{split}
			\end{equation*} Thus,  for some $C'(M,\rho)>0$
			
			\begin{equation*} \begin{split}
					&	[C_{B}- C' (q_0 \delta)^2]	\int_{\mathcal{D}(u_1,u_2)  } {r^{-\alpha+1}}\frac{ |D_u \phi|^2 }{\Omega^2}  dudv + C_B\int_{\mathcal{D}(u_1,u_2)  } {r^{-\alpha+1}}\Omega^2|D_v \phi|^2   dudv+ \int_{{u_2}}^{+\infty} \frac{r^{2-\alpha} |D_u \phi|^2}{\Omega^2} (u',v_R(u_2))du'\\&   \leq \int_{{u_1}}^{+\infty} \frac{r^{2-\alpha} |D_u \phi|^2}{\Omega^2} (u',v_R(u_1))du'  + C E^{+}_{deg}(u_1,u_2)
					+C' (q_0 \delta)^2\left[\int_{v_R(u_1)}^{v_R(u_2)} |\phi|^2(u_{R_0}(v),v) dv+ \int_{\mathcal{D}(u_1,u_2) \cap  \{r\geq R_0\}}  \Omega^2 r^{p-4} |\psi|^2 du dv\right].  
				\end{split}
			\end{equation*} 
			Then, multiply  \eqref{est.lemma.rp.first} from Lemma~\ref{lemma.rp.first} by a small constant $\eta$ and add it to the above estimate,  giving
			\begin{equation*} \begin{split}
					&	[C_{B}- C' (q_0 \delta)^2]	\int_{\mathcal{D}(u_1,u_2)  } {r^{-\alpha+1}}\frac{ |D_u \phi|^2 }{\Omega^2}  dudv + C_B\int_{\mathcal{D}(u_1,u_2)  } {r^{-\alpha+1}}\Omega^2|D_v \phi|^2   dudv+ \int_{{u_2}}^{+\infty} \frac{r^{2-\alpha} |D_u \phi|^2}{\Omega^2} (u',v_R(u_2))du'\\&  +[\eta-C'(q_0 \delta)^2]M\int_{\mathcal{D}(u_1,u_2)\cap \{r\geq R_0\}}  \Omega^2 r^{p-4} |\psi|^2 du dv+C'(q_0 \delta)^2\int_{\mathcal{D}(u_1,u_2)\cap \{r\leq R_0\}}  \Omega^2 |\phi|^2 du dv  \\ & +[\eta -C' (q_0 \delta)^2]{R_0}^{p}\int_{{v_R(u_1)}}^{{v_R(u_2)}}|\phi|^2 (u_{{R_0}}(v),v) dv  {+ {\eta}\int_{v_{R}(u_2)}^{+\infty}r^{p} \ |D_v \psi|^2   (u_2,v)dv} \leq \int_{{u_1}}^{+\infty} \frac{r^{2-\alpha} |D_u \phi|^2}{\Omega^2} (u',v_R(u_1))du' \\ &  + C E^{+}_{deg}(u_1,u_2)
						+ {\eta}\int_{v_{R}(u_1)}^{+\infty}r^{p} \ |D_v \psi|^2   (u_1,v)dv  {+ \eta {R_0}^{p+2} }\int_{{v_R(u_1)}}^{{v_R(u_2)}}|D_v\phi|^2 (u_{{R_0}}(v),v) dv{+C(R) E_{deg}(u_1)}. 
				\end{split}
			\end{equation*}
			
		}
		Now, fix $\alpha>4$ and $R_0(M,\rho)$: by the mean-value theorem, there exists $\tilde{R}_0 \in [R_0,2R_0]$ (assuming $R> 2R_0$ with no loss of generality) such that \begin{equation*}\begin{split}
				&R_0^{-\alpha+2}\int_{{v_R(u_1)}}^{{v_R(u_2)}}|D_v\phi|^2 (u_{{R_0}}(v),v) dv= R_0^{-1}  	\int_{\mathcal{D}(u_1,u_2)  \cap \{R_0 \leq r \leq 2R_0\}} r^{-\alpha+2}\Omega^2 |D_v \phi|^2 du dv\\ & \leq  2\int_{\mathcal{D}(u_1,u_2) } r^{-\alpha+1}\Omega^2 |D_v \phi|^2 du dv\end{split}
		\end{equation*}
		
		Applying \eqref{final.proof.M} to $R_0=\tilde{R}_0$ thus gives, for some constant $C''(M,\rho)>0$:
		\begin{equation*} \begin{split}
				&	[C_{B}- C' (q_0 \delta)^2]	\int_{\mathcal{D}(u_1,u_2)  } {r^{-\alpha+1}}\frac{ |D_u \phi|^2 }{\Omega^2}  dudv + [C_B-{\eta C''}]\int_{\mathcal{D}(u_1,u_2)  } {r^{-\alpha+1}}\Omega^2|D_v \phi|^2   dudv+ \int_{{u_2}}^{+\infty} \frac{r^{2-\alpha} |D_u \phi|^2}{\Omega^2} (u',v_R(u_2))du'\\& {+\eta}\int_{v_{R}(u_2)}^{+\infty}r^{p} \ |D_v \psi|^2   (u_2,v)dv +[\eta-C'(q_0 \delta)^2]M\int_{\mathcal{D}(u_1,u_2)\cap \{r\geq R_0\}}  \Omega^2 r^{p-4} |\psi|^2 du dv  +[\eta {R_0}^{p}-C' (q_0 \delta)^2]\int_{{v_R(u_1)}}^{{v_R(u_2)}}|\phi|^2 (u_{{R_0}}(v),v) dv \\ & \leq \int_{{u_1}}^{+\infty} \frac{r^{2-\alpha} |D_u \phi|^2}{\Omega^2} (u',v_R(u_1))du'  + C E^{+}_{deg}(u_1,u_2)
					+ \eta \int_{v_{R}(u_1)}^{+\infty}r^{p} \ |D_v \psi|^2   (u_1,v)dv  {+C(R) E_{deg}(u_1)}.  
			\end{split}
		\end{equation*} which concludes the proof{, after taking $2C' (q_0 \delta)^2<\eta<\frac{C_B}{2C''}$ (and thus also $q_0 \delta$) small enough.}
	
	\subsection{Boundedness of the  energy} \label{energysection}
	
	This estimate is by far the most delicate in this section. Even though for a charged scalar, a positive conserved quantity ---arising from the coupling of the scalar field and the electromagnetic tensors--- is still available via the vector field $\partial_t$, it is of little direct use because one of its component, involving the charge, does not decay in time since the charge approaches a constant value at infinity.

	This problem does not exist on Minkowski space-time ---where the charge is constrained to tend to $0$ towards time-infinity--- or in the uncharged case $q_0=0$ where the conservation laws are uncoupled and the use of $\partial_t$ gives an estimate only in terms of scalar fields quantities. 
	
	Our strategy to deal with this charge term is to absorb the charge difference into a term involving the scalar field. More precisely, the key point is to control the \textbf{fluctuations of the charge} by the energy of the scalar field. This idea already appeared in \cite{extremeJonathan}, in the context of the black hole interior.
	
	In the exterior, we do this in four different steps.
	
	In the first one, we take care of the domain $\mathcal{D}(u_1,u_2) \cap \{r \geq R\}$ where we integrate $Q^2$ towards $\gamma_R$. This estimate leaves a term $R^{-1} \cdot \left(Q^2(u_2,v_R(u_2))-Q^2(u_1,v_R(u_1)) \right)$ on $\gamma_{R}$ to be controlled.
	
	In the second step, we transport in $u$ the $R^{-1} \cdot \left(Q^2(u_2,v_R(u_2))-Q^2(u_1,v_R(u_1)) \right)$ term towards the curve $\gamma_{R_0}$ for some fixed $r_+<R_0 < R$, $R_0=R_0(M,\rho)$. This term can be controlled by the $r^2|D_u \phi|^2$ part of the energy on $\{ R_0 \leq r \leq R \}$ and by the $r^2|D_v \phi|^2$ part of the energy on 
	$\{ R \leq r \}$, using a Hardy-type inequality in each region. Every-time a small term $q_0|Q|$ multiplies the terms controlled by the energy, which is why they can be absorbed.

	This estimate leaves this time a term 	$R^{-1} \cdot \left(Q^2 (u_{R_0}(v_R(u_2)),v_R(u_2))-Q^2 (u_{R_0}(v_R(u_1)),v_R(u_1))\right)$.
	
	In the third step, we take care of the domain $\mathcal{D}(u_1,u_2) \cap \{r \leq R\}$ where we integrate $Q^2$ towards $\gamma_{R_0}$. This leaves a charge difference term on $\gamma_{R_0}$:  $ (\frac{1}{r_+}-\frac{1}{R})\left(Q^2 (u_{R_0}(v_R(u_2)),v_R(u_2))-Q^2 (u_{R_0}(v_R(u_1)),v_R(u_1))\right)$.

	{In the fourth step, we control the charge difference $Q^2(u_{R_0}(v_R(u_2)),v_R(u_2))-Q^2(u_{R_0}(v_R(u_1)),v_R(u_1))$ using the Morawetz estimate of section \ref{Morawetz}.} The key point is that $R_0$ is chosen ``in between'' $r_+$ and $R$, so that it is not too close to the event horizon and not too close either from infinity. Since the argument loses a lot of $\Omega^{-2}$ and $r$ weights , it is fortunate that these losses are bounded (they only depend on $\rho$ and $M$) and can therefore be absorbed using the smallness of $q_0|Q|$. The argument does not depend on the smallness of $R$, which can still be taken arbitrarily large for the following sections, in particular section \ref{decay}.
	
	The different curves are summed up by the Penrose diagram of Figure \ref{Fig2}.

	The use of the vector field $\partial_t$ on the domain $\mathcal{D}(u_1,u_2)$ gives rise to an energy identity, summed up by the following lemma: 
	
	\begin{lem}
		For all $u_0(R)\leq u_1<u_2$ we have the following energy identity: 
		\begin{equation} \label{energyidentity}
			\begin{split}
				\int_{v_R(u_1)}^{v_R(u_2)} r^2 |D_v \phi|^2_{|\mathcal{H}^+}(v)dv +   \int_{u_2}^{+\infty} \left(r^2 |D_u \phi|^2(u,v_R(u_2)) + \frac{2\Omega^2 Q^2}{r^2}(u,v_R(u_2)) \right)du \\+ \int_{v_R(u_2)}^{+\infty} \left(r^2 |D_v \phi|^2(u_2,v)   + \frac{2\Omega^2 Q^2}{r^2}(u_2,v) \right)dv
				+\int_{u_1}^{u_2} r^2 |D_u \phi|^2_{|\mathcal{I}^+}(u)du \\
				=  \int_{u_1}^{+\infty} \left(r^2 |D_u \phi|^2(u,v_R(u_1)) + \frac{2\Omega^2 Q^2}{r^2}(u,v_R(u_1)) \right)du+ \int_{v_R(u_1)}^{+\infty} \left(r^2 |D_v \phi|^2(u_1,v)   + \frac{2\Omega^2 Q^2}{r^2}(u_1,v) \right)dv.
			\end{split}
		\end{equation}
		
		With the notations of this section, it can also be written as
		
		\begin{equation} \label{energyidentity2}
			\begin{split}
				\int_{v_R(u_1)}^{v_R(u_2)} r^2 |D_v \phi|^2_{|\mathcal{H}^+}(v)dv +   \int_{u_2}^{+\infty}  \frac{2\Omega^2 Q^2}{r^2}(u,v_R(u_2)) du + \int_{v_R(u_2)}^{+\infty}  
				\frac{2\Omega^2 Q^2}{r^2}(u_2,v)dv
				+\int_{u_1}^{u_2} r^2 |D_u \phi|^2_{|\mathcal{I}^+}(u)du + E_{deg}(u_2) \\
				=  \int_{u_1}^{+\infty} \frac{2\Omega^2 Q^2}{r^2}(u,v_R(u_1)) du+ \int_{v_R(u_1)}^{+\infty}    \frac{2\Omega^2 Q^2}{r^2}(u_1,v)dv +E_{deg}(u_1).
			\end{split}
		\end{equation}
	\end{lem}
	
	\begin{proof}
		The proof is an elementary computation based on the fact that $\nabla^{\mu}(\mathbb{T}^{EM}_{\mu \nu}+  \mathbb{T}^{SF}_{\mu \nu}) =0$ and that $\partial_t$ is a Killing vector field of the Reissner--Nordstr\"{o}m metric, c.f.\ Appendix \ref{appendix} for more details. We omit the (easy) argument in which one writes the estimate on a truncated domain $\mathcal{D}(u_1,u_2) \cap \{ v \leq v_0 \}$ and sends $v_0$ toward $+\infty$ to obtain the claimed boundary terms on $\mathcal{I}^+$. We are going to repeat this omission in what follows.
		
	\end{proof}
	
	The goal of this section is to prove the following: 
	
	\begin{prop} \label{Energy} There exists $C=C(M,\rho,{R}) >0$, $\delta=\delta(M,\rho)>0$  such that 
		
		if $\|Q \|_{L^{\infty}(\mathcal{D}(u_0(R),+\infty))} < \delta$, and {with $\ep=8q_0\delta$} {sufficiently small (depending only on $M$ and $\rho$)}, we have for all $u_0(R) \leq u_1 < u_2$:
		
		\begin{equation*} 
			\begin{split}
				\int_{v_R(u_1)}^{v_R(u_2)} r^2 |D_v \phi|^2_{| \mathcal{H}^+}(v)dv +   \int_{u_2}^{+\infty} \frac{ r^2|D_u \phi|^2}{\Omega^2}(u,v_R(u_2))du+ \int_{v_R(u_2)}^{+\infty}{\left[ r^2 |D_v \phi|^2(u_2,v)+r^{\ep}|D_v \psi|^2(u_2,v)\right]}dv
				\\+\int_{u_1}^{u_2} r^2 |D_u \phi|^2_{|\mathcal{I}^+}(u)du 
				\leq C \cdot \left(  \int_{u_1}^{+\infty} \frac{ r^2|D_u \phi|^2}{\Omega^2}(u,v_R(u_1)) du+ \int_{v_R(u_1)}^{+\infty}[ r^2 |D_v \phi|^2(u_1,v)  {+ r^{\ep} |D_v \psi|^2(u_1,v) }]dv \right).
			\end{split}
		\end{equation*}
		
		With the notations of this section, it can also be written as
		
		\begin{equation} \label{energyabsorbed2}
			\begin{split}
				\int_{v_R(u_1)}^{v_R(u_2)} r^2 |D_v \phi|^2_{| \mathcal{H}^+}(v)dv 
				+\int_{u_1}^{u_2} r^2 |D_u \phi|^2_{|\mathcal{I}^+}(u)du +  \tilde{E}_{{\ep}}(u_2) 
				\leq C \cdot  \tilde{E}_{{\ep}}(u_1).
			\end{split}
		\end{equation}

	\end{prop}
	In short, the strategy is to bound all the terms involving $Q^2$ by 
	$C'(M,\rho) \cdot q_0  |Q| \cdot E^+(u_1,u_2)$ for a constant $C'(M,\rho)>0$ that only depends on the black hole parameters. Then we take $|Q|$ to be small enough so that $C'(M,\rho) \cdot q_0|Q| <1$, in order to absorb those terms into the others.
	\subsubsection{Step 1: the region $\mathcal{D}(u_1,u_2) \cap \{r \geq R\}$ } \label{step1S}
	
	\begin{lem} \label{step1}
		\begin{equation}
			\begin{split}
				\int_{v_R(u_2)}^{+\infty} \frac{\Omega^2 Q^2}{r^2}(u_2,v)dv-	\int_{v_R(u_1)}^{+\infty} \frac{\Omega^2 Q^2}{r^2}(u_1,v)dv  \leq {\frac{Q^2(u_2,v_R(u_2))-Q^2(u_1,v_R(u_1))}{R}}\\+ 4q_0 \left(\sup_{\mathcal{D}(u_1,u_2) } |Q| \right) \Omega^{-2}(R) \left[\int_{v_R(u_2)}^{+\infty} r^2|D_v \phi|^2(u_2,v)dv+\int_{v_R(u_1)}^{+\infty} r^2|D_v \phi|^2(u_1,v)dv \right]
			\end{split}
		\end{equation}
	\end{lem}
	
	\begin{proof}
		We start to prove the identity:  	
		
		\begin{equation} \label{step1ID}
			\int_{v_R(u_i)}^{+\infty} \frac{\Omega^2 Q^2}{r^2}(u_i,v)dv  = \frac{ Q^2(u_i,v_R(u_i))}{R} + 2 \int_{v_R(u_i)}^{+\infty} q_0 Q r\Im(\phi\overline{D_v \phi})(u_i,v)dv
		\end{equation}
		
		For this we write using \eqref{ChargeVEinstein}, for all  $v_R(u_i) \leq v$: 
		
		$$ Q^2(u_i,v)=Q^2(u_i,v_R(u_i))+2\int_{v_R(u_i)}^{v} q_0 Q r^2 \Im(\phi\overline{D_v \phi})(u_i,v')dv'.$$
		
		Then for the first term, we just need to notice that because $\Omega^2 = \partial_v r$:
		
		$$ \int_{v_R(u_i)}^{+\infty} \frac{\Omega^2 }{r^2}(u_i,v)dv =  \frac{ 1}{R}.$$
		
		For the second term we integrate by parts as: \begin{equation*}\begin{split}
				\int_{v_R(u_i)}^{+\infty} \frac{\Omega^2}{r^2}(u_i,v)  \left( \int_{v_R(u_i)}^{v} q_0 Qr^2 \Im(\phi\overline{D_v \phi})(u_i,v')dv'\right) dv \\ =  \int_{v_R(u_i)}^{+\infty} (-\partial_v(\frac{1}{r})(u_i,v))  \left( \int_{v_R(u_i)}^{v} q_0 Qr^2 \Im(\phi\overline{D_v \phi})(u_i,v')dv'\right) dv \color{black}= \int_{v_R(u_i)}^{+\infty}  q_0 Q r  \Im(\phi\overline{D_v \phi})(u_i,v)dv.
			\end{split}
		\end{equation*}

		The boundary terms canceled because $ Q^2$ is bounded and $r(u_i,v) \rightarrow+\infty$ as $v \rightarrow+\infty$\color{black}: 		
		$$ \lim_{ v \rightarrow +\infty} \frac{1}{r(u_i,v)} \int_{v_R(u_i)}^{v} 2q_0Q r^2 \Im(\bar{\phi} D_v \phi)(u_i,v')dv' =  \frac{Q^2(u_i,v)-Q^2(u_i,v_R(u_i))}{r(u_i,v)}  =0.$$ 
		
		Therefore \eqref{step1ID} is proven.
		
		Thereafter we use Hardy's inequality \eqref{Hardy4} under the form 
		
		$$\int_{v_R(u_i)}^{+\infty}  |\phi|^2(u_i,v) dv\leq \frac{4}{\Omega^4(R)}\int_{v_R(u_i)}^{+\infty}  r^2|D_v\phi|^2(u_i,v)dv.$$
		
		We can then use Cauchy-Schwarz to get that 
		
		$$\int_{v_R(u_i)}^{+\infty}  q_0 Q r  \Im(\bar{\phi} D_v \phi)(u_i,v)dv \leq  2q_0 \left(\sup_{\mathcal{D}(u_1,u_2) } |Q| \right) \Omega^{-2}(R) \int_{v_R(u_i)}^{+\infty}  r^2|D_v\phi|^2(u_i,v)dv,$$
		
		which, summing for $i=1,2$ and combining the estimate with \eqref{step1ID}\color{black}, concludes the proof of Lemma \ref{step1}.

	\end{proof}

	\subsubsection{Step 2: transport in $u$ of $Q^2_{|\gamma_{R}}$ to $Q^2_{|\gamma_{R_{0}}}$} \label{step2S}
	
	\begin{lem} \label{step2} For all $r_+<R_0<R$ we have
		\begin{equation} \begin{split}	\frac{ Q^2(u_2,v_R(u_2))-Q^2(u_1,v_R(u_1))}{R} \leq \frac{ Q^2 (u_{R_0}(v_R(u_2)),v_R(u_2))-Q^2 (u_{R_0}(v_R(u_1)),v_R(u_1))}{R}\\+ 5 q_0 \left(\sup_{\mathcal{D}(u_1,u_2) } |Q| \right)  E^+(u_1,u_2).
			\end{split}
		\end{equation} 
		
	\end{lem}
	\begin{proof}
		
		We start by an identity coming directly from \eqref{chargeUEinstein}: 
		
		$$ Q^2(u_i,v_R(u_i))=Q^2 (u_{R_0}(v_R(u_i)),v_R(u_i))+2 \int_{u_i}^{u_{R_0}(v_R(u_i))} q_0 Q r^2 \Im(\bar{\phi} D_u \phi)(u,v_R(u_i))du.$$
		
		Using Hardy inequality \eqref{Hardy3} in $u$ we get:
		
		\begin{equation*}
			\int_{u_i}^{u_{R_0}(v_R(u_i))} \Omega^2 |\phi|^2(u,v_R(u_i))du \leq 4\int_{u_i}^{u_{R_0}(v_R(u_i))}\frac{ r^2|D_u \phi|^2}{\Omega^2}(u,v_R(u_i))du +  2R|\phi|^2(u_i,v_R(u_i)). 
		\end{equation*}	
		
		The last term of the right-hand side can be bounded, using Hardy's inequality \eqref{Hardy2} in $v$: 
		
		$$ R  \Omega^{2}(R)|\phi|^2(u_i,v_R(u_i))  \leq \int_{v_R(u_i)}^{+\infty} r^2|D_v \phi|^2(u_i,v)dv.$$
		
		We then use Cauchy-Schwarz, taking advantage of the fact that $r \leq R$: 
		
		$$ 2|\int_{u_i}^{u_{R_0}(v_R(u_i))} q_0 Q r^2 \Im(\bar{\phi} D_u \phi)(u,v_R(u_i))du| \leq R \cdot q_0 \left(\sup_{\mathcal{D}(u_1,u_2) } |Q| \right)   \int_{u_i}^{u_{R_0}(v_R(u_i))}\left( \frac{ r^2|D_u \phi|^2}{\Omega^2} + \Omega^2|\phi|^2  \right)  (u,v_{R}(u_i))du.$$
		
		Then we combine with the estimates proven formerly to get
		
		\begin{equation}
			\begin{split}
				2|\int_{u_i}^{u_{R_0}(v_R(u_i))} q_0 Q r^2 \Im(\bar{\phi} D_u \phi)(u,v_R(u_i))du| \\ \leq  R \cdot q_0 \left(\sup_{\mathcal{D}(u_1,u_2) } |Q| \right)   \left( 5\int_{u_i}^{u_{R_0}(v_R(u_i))} \frac{r^2|D_u \phi|^2}{\Omega^2}(u,v_{R}(u_i))du +2\Omega^{-2}(R) \int_{v_R(u_i)}^{+\infty} r^2|D_v \phi|^2(u_i,v)dv \right). \end{split}	\end{equation}
		
		Then if we take $R$ large enough so that $2\Omega^{-2}(R)<5$ , this last estimate proves Lemma \ref{step2}.

	\end{proof}

	\subsubsection{Step 3: the region $\mathcal{D}(u_1,u_2) \cap \{r \leq R\}$} \label{step3S}
	
	\begin{lem} \label{step3}  For all $r_+ < R_0 < R$ we have 
		
		\begin{equation}
			\begin{split}
				\int_{u_2}^{+\infty} \frac{\Omega^2 Q^2}{r^2}(u,v_R(u_2))du-		\int_{u_1}^{+\infty} \frac{\Omega^2 Q^2}{r^2}(u,v_R(u_1))du  \leq \\ \frac{ Q^2 (u_{R_0}(v_R(u_2)),v_R(u_2))-Q^2 (u_{R_0}(v_R(u_1)),v_R(u_1))}{r_+}(1-\frac{ r_+}{R})+\frac{5}{2} q_0 \left(\sup_{\mathcal{D}(u_1,u_2) } |Q| \right) E^+(u_1,u_2).
			\end{split}
		\end{equation}
	\end{lem}
	
	\begin{proof}
		We start to prove an identity for all $r_+ < R_0 < R$:  		\begin{equation*} \begin{split}
				\int_{u_i}^{+\infty} \frac{\Omega^2 Q^2}{r^2}(u,v_R(u_i))du  = \frac{ Q^2 (u_{R_0}(v_R(u_i)),v_R(u_i))}{r_+}(1-\frac{ r_+}{R}) \\ +2  \int_{u_{R_0}(v_R(u_i))}^{+\infty} q_0 Q r(1-\frac{r}{r_+})\Im(\bar{\phi}D_u \phi)(u,v_R(u_i))du+2  \int_{u_i}^{u_{R_0}(v_R(u_i))} q_0 Q r(1-\frac{ r}{R})\Im(\bar{\phi}D_u \phi)(u,v_R(u_i))du
			\end{split}
		\end{equation*}
		
		After dividing $[u_i,+\infty]$ into $[u_i,u_{R_0}(v_R(u_i))]$ and $[u_{R_0}(v_R(u_i)),+\infty]$, the method of proof is similar to that of the estimate \ref{step1ID}, except that we integrate $Q^2$ in $u$ towards $\gamma_{R_0}$ this time.
		
		Taking $R_0 < 2r_+$ and $R$ large enough so that $\frac{R_0}{r_+}-1<1-\frac{r_+}{R}$, we can write that 
		
		\begin{equation*} \label{step3estimate} \begin{split} |\int_{u_{R_0}(v_R(u_i))}^{+\infty} q_0 Q r(1-\frac{r}{r_+})\Im(\bar{\phi}D_u \phi)(u,v_R(u_i))du+ \int_{u_i}^{u_{R_0}(v_R(u_i))} q_0 Q r(1-\frac{ r}{R})\Im(\bar{\phi}D_u \phi)(u,v_R(u_i))du| \\ \leq (1-\frac{r_+}{R}) q_0 \left(\sup_{\mathcal{D}(u_1,u_2) } |Q| \right)\int_{u_i}^{+\infty} r |\Im(\bar{\phi}D_u \phi)|(u,v_R(u_i))du.\end{split}
		\end{equation*}
		
		Using Hardy inequality \eqref{Hardy3} in $u$ we get:
		
		\begin{equation*}
			\int_{u_i}^{+\infty} \Omega^2 |\phi|^2(u,v_R(u_i))du \leq 4\int_{u_i}^{+\infty}\frac{ r^2|D_u \phi|^2}{\Omega^2}(u,v_R(u_i))du +  2R|\phi|^2(u_i,v_R(u_i)). 
		\end{equation*}	
		
		The last term of the right-hand side can be bounded, using Hardy's inequality \eqref{Hardy2} in $v$: 
		
		$$ R  \Omega^{2}(R)|\phi|^2(u_i,v_R(u_i))  \leq \int_{v_R(u_i)}^{+\infty} r^2|D_v \phi|^2(u_i,v)dv.$$

		Combining all the former estimates and using Cauchy-Schwarz, we finally get: 
		
		\begin{equation}
			\int_{u_i}^{+\infty}  r|\Im(\bar{\phi}D_u \phi)|(u,v_R(u_i))du \leq \frac{5}{2} \int_{u_i}^{+\infty} \frac{ r^2|D_u \phi|^2}{\Omega^2}(u,v_R(u_i))du +2\Omega^{-2}(R) \int_{v_R(u_i)}^{+\infty} r^2|D_v \phi|^2(u_i,v)dv \leq \frac{5}{2} E(u_i),
		\end{equation} where we took $R$ large enough so that $2\Omega^{-2}(R) < \frac{5}{2}$.

		We then combine with {the estimates of Section~\ref{step2S}}, which proves Lemma \ref{step3}.

	\end{proof}

	\subsubsection{Step 4: control of the $Q^2_{|\gamma_{R_{0}}}$ terms by the Morawetz estimate }  \label{step4S}
	
	\begin{lem} \label{step4} For all $r_+<R'_0<R$ there exists $r_+<R_0<R'_0$ and $C_4=C_4(R'_0,M,\rho)>0$ such that
		\begin{equation}
			Q^2 (u_{R_0}(v_R(u_2)),v_R(u_2))-Q^2 (u_{R_0}(v_R(u_1)),v_R(u_1)) \leq C_4 q_0 \left(\sup_{\mathcal{D}(u_1,u_2) } |Q| \right) \cdot \left[E^+(u_1,u_2){+ \tilde{E}_{\ep}(u_1)}\right]{+C(R) E_{deg}(u_1)}.
		\end{equation}
	\end{lem}
	
	\begin{proof}
		For this proof, we work in $(t,r^*)$ coordinates. We define $t_i= 2v_{R}(u_i)-R_0^*$ such that the space-time point $(t_i,R_0^*)$ corresponds to the space-time point $(u_{R_0}(v_R(u_i)),v_R(u_i))$.
		
		Using \eqref{chargeUEinstein}, \eqref{ChargeVEinstein} we see that on  $\gamma_{R_{0}}$:
		
		$$Q^2 (u_{R_0}(v_R(u_2)),v_R(u_2))-Q^2 (u_{R_0}(v_R(u_1)),v_R(u_1))= Q^2 (t_2,R_0^*)-Q^2 (t_1,R_0^*)=2R_0^2 \int_{t_1}^{t_2} q_0 Q \Im(\bar{\phi} D_{r^*}\phi)(t,R_0^*) dt.$$
		
		This estimate is valid for any $r_+<R_0<R$. 
		
		Now fix some $r_+<R'_0<R$. Using the mean-value theorem in $r$, we see that there exists a $r_+<\tilde{R}_0<R'_0$ such that: 
		
		$$ \int_{r_+}^{R'_0} \left[\int_{t_1}^{t_2} q_0 Q \Im(\bar{\phi} D_{r^*}\phi)(t,r^*) dt \right] dr = (R'_0-r_+)\int_{t_1}^{t_2} q_0 Q \Im(\bar{\phi} D_{r^*}\phi)(t,(\tilde{R}_0)^*) dt. $$
		
		From now on, we choose $R_0$ to be exactly that $\tilde{R}_0$ and we take $R'_0 < 2r_+$.
		
		Therefore, since $R_0<R'_0$ and $R'_0 < 2r_+$ we have 
		
		$$Q^2 (u_{R_0}(v_R(u_2)),v_R(u_2))-Q^2 (u_{R_0}(v_R(u_1)),v_R(u_1))\leq  4 R'_0\int_{r_+}^{R'_0} \left[\int_{t_1}^{t_2} q_0 Q \Im(\bar{\phi} D_{r^*}\phi)(t,r^*) dt \right] dr.$$
		
		Then Lemma \ref{step4} follows from a direct application of Cauchy-Schwarz and the Morawetz estimate {of Proposition~\ref{PropositionM}}.
	\end{proof}

	\subsubsection{Completion of the proof of Proposition \ref{Energy} }
	
	\begin{proof}
		
		Now we choose $R'_0 = 2 r_+ $ and take the $R_0>r_+$ of Lemma \ref{step4}. We now apply Lemma \ref{step2} and Lemma \ref{step3} with that $R_0$. Coupled with Lemma \ref{step1} and identity \eqref{energyidentity} we find that there exists $C_5=C_5(M,\rho)>0$ such that 
		
		\begin{equation*}
			\begin{split}
				\int_{v_R(u_1)}^{v_R(u_2)} r^2 |D_v \phi|^2_{| \mathcal{H}^+}(v)dv +   \int_{u_2}^{+\infty} r^2|D_u \phi|^2(u,v_R(u_2))du+ \int_{v_R(u_2)}^{+\infty} r^2 |D_v \phi|^2(u_2,v)dv
				+\int_{u_1}^{u_2} r^2 |D_u \phi|^2_{|\mathcal{I}^+}(u)du 
				\\ \leq  C_5  q_0 \left(\sup_{\mathcal{D}(u_1,u_2) } |Q| \right) \left[ E^+(u_1,u_2){+\tilde{E}_{\ep}(u_1)}\right]+{C(R)E_{deg}(u_1)}.
			\end{split}
		\end{equation*}

		{Then, to the above estimate, we add a multiple of Proposition~\ref{PropositionM} to obtain the non-degenerate red-shift term on the left-hand side. Note that we obtain control of an ingoing boundary term of the form $$ \int_{u_2}^{+\infty} \frac{r^2|D_u \phi|^2}{\Omega^2} [\Omega^2 + r^{-\alpha}](u,v_R(u_2))du,$$ and noting that there exists $C_+(M,\rho)>0$ such that $[\Omega^2 + r^{-\alpha}]\geq C_+$, we see that the weighted red-shift boundary term controls the required non-degenerate boundary energy.}	Now choosing $|Q|$ small enough, we can absorb the $E^+(u_1,u_2)$ term on the left-hand-side, which proves Proposition \ref{Energy}.

	\end{proof}
	
	As a by-product of the boundedness of the energy, we finally close the Morawetz estimate of section \ref{Morawetz}. We indeed proved that:

	\begin{cor}\label{cor.M} {Assume that $\|Q \|_{L^{\infty}(\mathcal{D}(u_0(R),+\infty))} < \delta$, for $|q_0| \delta$ {sufficiently small (only depending on $M$ and $\rho$)}, $\alpha>4$ and that there exists $ R_0(M,\rho) >r_+$ such that $R \geq R_0$, then there exists $D(M,\rho,{R})>0$ 
		} such that for all $ u_0(R) \leq  u_1 < u_2$ we have: 
		\begin{equation}\label{Morawetzestimate3}
			{\int_{\mathcal D(u_1,u_2)}r^{1-\alpha}\left(\frac{|D_u\phi|^2}{\Omega^2}+\Omega^2|D_v\phi|^2\right)\,du\,dv
				+\int_{\mathcal D(u_1,u_2)}\Omega^2r^{-1-\alpha}|\phi|^2\,du\,dv
				+\tilde E_{\ep}(u_2)\leq D\tilde E_{\ep}(u_1)}
	\end{equation}
\end{cor}

\section{Decay of the energy} \label{decay}

We now establish the time decay of the energy. We use the $r^p$ method developed in \cite{RP}. The main idea is that the boundedness of $r^p$ weighted energies $\tilde{E}_p$ {(already introduced in Section~\ref{energyboundedsection})} can be converted into time decay for the standard energy $E$ on the $V$-shaped foliation $\mathcal{V}_u$ we introduced in section \ref{foliations}. 

The papers \cite{Newrp}, \cite{RP}, \cite{Moschidisrp} and \cite{Volker} all treat the linear wave equation on black hole space-times, in different contexts. Compared to these works, the main difference and difficulty for the charged scalar field model is the presence of a \textbf{non-linear} term, coming from the interaction with the Maxwell field.

The whole objective of this section is to find a way absorb this interaction term, for various values of $p$ and to prove the boundedness of the $r^p$ weighted energy.

We are going to assume the energy boundedness and the Morawetz estimates of the former section. As before, to treat the interaction term,  we rely on the smallness of the charge $Q$ , however it is not as drastic as in the former section: indeed, $q_0|Q|$ must be smaller than a \textbf{universal} constant, independent on the black holes parameters. This is one of the key ingredients of the proof of Theorem \ref{decaytheorem}.

Notice also that, thanks to the charge a priori estimates from section \ref{CauchyCharac}, provided that the initial energy of the scalar field is sufficiently small, it is equivalent to talk about the smallness of $Q$, of $e$ or of $e_0$, c.f.\ Remark \ref{Qe}. This remark will be used implicitly throughout this section. \\

The first step is establish a $r^p$ hierarchy for $p<2$. For this, we use a Hardy type inequality and we absorb the interaction term into the $\int_{u_1}^{u_2} E_{p-1}[\psi](u)du$ term \footnote{For a definition of $E_{p}[\psi](u)$, c.f.\ section \ref{energynotation} or section \ref{preliminaries}.}of the left-hand-side. The smallness of the charges $Q$ and $e$ is essential and limits the maximal $p$ to $p_{max}<1+\sqrt{1-4q_0|e|}$, which tends to $2$ as $e$ tends to $0$.

This method cannot be extended for $p>2$ because for a scalar field $\phi$, $r\phi$ admits a finite generically non-zero limit $\psi$ --- the radiation field --- towards $\mathcal{I}^+$. This fact imposes that the $r^{p-3}$ weight in the Hardy estimate is strictly inferior to $r^{-1}$, hence $p<2$ is necessary to apply the same argument, even when $e \rightarrow0$.

Still, this first step already gives some decay of the energy that will be crucial for the second step.

The second step uses a different strategy: this time, while we still use the Hardy inequality, we apply it to a weaker $r$ weight. This is because the Hardy inequality proceeds with a maximal $r^{-1}$ weight. This allows us to take $p$ larger, up to almost $3$ as $e$ tends to $0$. On the other hand, we now have to absorb the interaction term into $E_{p}[\psi](u)$, in contrast to what was done in the first step.

This is done in three other steps: first we prove that a differential inequality involving $E_{p}[\psi](u)$ holds, where the error terms are multiplied by a constant $\nu(e)$ depending \textit{essentially} only on the charge $e$.

The next step is to prove that $\nu(e)$ is small indeed: for this, we mostly require the smallness of $q_0|e|$, after various optimisation procedures. 

Then, we integrate the differential inequality \`{a} la Gr\"{o}nwall, making use of the smallness of $\nu(e)$ to prove that $E_p(u)$ enjoys a small control growth in $u$, of the order $u^{2\epsilon(e)}$.

Finally we use the pigeon-hole principle to get $u$ decay of $\tilde{E}_{p-1}(u)$, of the order $u^{-1+2\epsilon(e)}$.  \color{black}

It is interesting to notice that the procedure does \textbf{not} close the boundedness of the $r^p$ weighted energy for the largest $p$ that we consider but allows for a small $u$ growth, due to the use of a Gr\"{o}nwall-type argument. 

More details on this crucial and delicate step can be found in the introduction of section \ref{sectionp=3}.

This procedure imposes to take $p=3-\nu(q_0e)$, with $\nu(q_0e) =O((q_0|e|)^{\frac{1}{2}}) $ as $e \rightarrow 0$  and gives the boundedness of $r^p$ weighted energies for such a $p$. \color{black}

This finally proves that the standard energy decays at a rate that tends to $3$ when $e$ tends to $0$.

More precisely, the goal of this section is to prove the following result:

\begin{prop} \label{Energydecayproposition} Suppose that the energy boundedness \eqref{energyabsorbed2} and the Morawetz estimates \eqref{Morawetzestimate3} are valid. 
	Assume that $q_0|e| \leq {0.04}$.
	Assume also that for  all $   0 \leq   p' < 2+\sqrt{1-4q_0|e|}$, $E_{p'}(u_0(R))<\infty$.

	Then there exists $2<p(e)<2+\sqrt{1-4q_0|e|}$, with $p(e) \rightarrow 3$
	as $e \rightarrow0$ and $\delta=\delta(e,M,\rho)>0$ such that if
	$$\|Q-e \|_{L^{\infty}(\mathcal{D}(u_0(R),+\infty))} < \delta,$$
	
	Then for all {$\ep \leq p \leq p(e)+\ep$}, there exist
	${R_{low}}={R_{low}}(p,M,\rho,e)>r_+$ such that if $R>{R_{low}}$,
	there exists $D=D(M,\rho,R,e)>0$ such that for all $u>1$:
	
	\begin{equation}
		\tilde{E}_p(u) \leq D \cdot u^{{p-\ep-p(e)}},
	\end{equation}
	in particular,
	\begin{equation} \label{zeroenergydecayest}
		E(u) \leq 	\tilde{E}_{\ep}(u)\leq  D \cdot u^{-p(e)}.
	\end{equation}

	\color{black}
\end{prop}

\subsection{Preliminaries on the decay of the energy} \label{preliminaries}

In this section, we are going to establish a few preliminary results on the energy decay. 
The main goal is to understand how the boundedness of $r^p$ weighted energies implies the time decay of the un-weighted energy.

The estimates are so tight that they can be closed only with the smallness of the charge. Therefore it is very important to monitor the constants as carefully as possible. 

As we will see, the lowest order term, which should be controlled together with the $r^p$ weighted energy, is bounded by a large constant and the smallness of the charge cannot make it smaller. The key point is that this term actually enjoys additional $u$ decay so that the large constant can be absorbed for large $u$.

This motivates the following definitions, some of which have already been encountered by the reader: 

$$ E(u)=  \int_{u}^{+\infty} r^2 \frac{|D_u \phi|^2}{\Omega^2}(u',v_R(u))du'+\int_{v_R(u)}^{+\infty} r^2 |D_v \phi|^2(u,v)dv, $$

$$ E^+(u_1,u_2):= E(u_2)+E(u_1)+ 	\int_{v_R(u_1)}^{v_R(u_2)} r^2 |D_v \phi|^2_{|\mathcal{H}^+}(v)dv 
+\int_{u_1}^{u_2} r^2 |D_u \phi|^2_{|\mathcal{I}^+}(u)du, $$
{In what follows, we repeatedly use the trivial inequality} $E_p[\psi](u) \leq \tilde{E}_p(u)$.

$$ E_p[\psi](u):= \int_{v_{R}(u)}^{+\infty} r^{p} \ |D_v \psi|^2  (u,v)dv,$$

$$ \tilde{E}_p(u):= E_p[\psi](u) + E(u). $$

It should be noted that $\tilde{E}_0(u)$ and $E(u)$ are comparable.

We now establish a preliminary lemma, that reduces the problem to understanding the interaction term. 

\begin{lem} \label{rpidentity} There exists $C_1=C_1(M,\rho,R)>0$ such that for every $0<p<3$, { the following estimate holds}  for all $u_0(R) \leq u_1 < u_2$ 
	\begin{equation} \label{lemmapreliminary}
		p\int_{u_1}^{u_2} E_{p-1}[\psi](u)du + \tilde{E}_{p}(u_2) \leq \left[ 1+P_0(R) \right] \left(  E_{p}[\psi](u_1) +  |\int_{\mathcal{D}(u_1,u_2)\cap \{r \geq R\}} [2q_0 |Q|+{P_0(r)]} \Omega^2 r^{p-2} |\psi| |D_v \psi| du dv| \right) + C_1 \cdot  \tilde{E}_{\ep}(u_1),
	\end{equation} where $P_0(r)$ is a polynomial in $r$ that behaves like $O(r^{-1})$ as $r$ tends to $+\infty$ and with coefficients depending only on $M$ and $\rho$.
\end{lem}

\begin{proof}
	
		Proceeding exactly as in the proof of Lemma~\ref{lemma.rp.first}, multiplication of \eqref{wavev} by $2r^p\overline{D_v\psi}$, followed by integration over $\mathcal{D}(u_1,u_2)\cap\{r\geq R\}$, gives
		\begin{equation*}
			\begin{split}
				&\int_{\mathcal{D}(u_1,u_2)\cap\{r\geq R\}}\!\Omega^2\left(pr^{p-1}|D_v\psi|^2+\left[2M(3-p)-2(4-p)\frac{\rho^2{+2M^2}}{r}{+\frac{6M[5-p]\rho^2}{r^2}-\frac{2[6-p]\rho^4}{r^3}}\right]r^{p-4}|\psi|^2\right)\,du\,dv\\
				&\quad+E_p[\psi](u_2)
				\leq E_p[\psi](u_1)+\left|\int_{\mathcal{D}(u_1,u_2)\cap\{r\geq R\}}2q_0Q\Omega^2r^{p-2}\Im(\psi\overline{D_v\psi})\,du\,dv\right|+\int_{u_1}^{u_2}\Omega^2(R)\left(2M-\frac{2\rho^2}{R}\right)R^{p-3}|\psi|^2(u,v_R(u))\,du.
			\end{split}
		\end{equation*}
		
		{Since $0\leq p<3$, we can ignore the LHS bulk term involving $[2M(3-p)r^{p-4}+6M [5-p]\rho^2 r^{p-6}]|\psi|^2$, and put the LHS term bulk term involving $[2(4-p)\frac{\rho^2+2M^2}{r}+ \frac{2[6-p]\rho^4}{r^3}]r^{p-4}|\psi|^2$ on the RHS, giving 
			\begin{equation*}
				\begin{split}
					&p\int_{\mathcal{D}(u_1,u_2)\cap\{r\geq R\}}\!\Omega^2 r^{p-1}|D_v\psi|^2\,du\,dv
					+E_p[\psi](u_2) \leq   \int_{\mathcal{D}(u_1,u_2)\cap\{r\geq R\}}\!\Omega^2 \left([8\rho^2 {+16M^2}]r^{p-5}+{12 \rho^4 r^{p-7}}\right)|\psi|^2 du dv\\
					&+ E_p[\psi](u_1)+\left|\int_{\mathcal{D}(u_1,u_2)\cap\{r\geq R\}}2q_0Q\Omega^2r^{p-2}\Im(\psi\overline{D_v\psi})\,du\,dv\right|+\int_{u_1}^{u_2}\Omega^2(R)\left(2M-\frac{2\rho^2}{R}\right)R^{p-3}|\psi|^2(u,v_R(u))\,du.
				\end{split}
			\end{equation*} Then, by integration by parts: \begin{equation*}
				\int_{\mathcal{D}(u_1,u_2)\cap\{r\geq R\}}\!\Omega^2r^{p-5}|\psi|^2 du dv \leq \frac{R^{p-4}}{4-p} \int_{u_1}^{u_2}|\psi|^2(u,v_R(u)) du+\frac{2}{ [4-p]}  \int_{\mathcal{D}(u_1,u_2)\cap\{r\geq R\}}\!r^{p-4}|D_v\psi||\psi| du dv,
			\end{equation*} so using $[4-p]^{-1}\leq 1$, we get (also using a similar estimate to address to $	\int_{\mathcal{D}(u_1,u_2)\cap\{r\geq R\}}\!\Omega^2r^{p-7}|\psi|^2 du dv $ term) \begin{equation*}
				\begin{split}
					&p\int_{\mathcal{D}(u_1,u_2)\cap\{r\geq R\}}\!\Omega^2 r^{p-1}|D_v\psi|^2\,du\,dv
					+E_p[\psi](u_2)\leq 
					E_p[\psi](u_1)  \\ & +  \left|\int_{\mathcal{D}(u_1,u_2)\cap\{r\geq R\}}[2q_0|Q|+\frac{8\rho^2{+16M^2}}{r^2}{+\frac{24\rho^4}{r^4}}]\Omega^2r^{p-2}|\psi||D_v\psi|,du\,dv\right|
					+2M\int_{u_1}^{u_2}\Omega^2(R)R^{p-3}|\psi|^2(u,v_R(u))\,du.
				\end{split}
		\end{equation*}
		
		The boundary term on $\gamma_R$ is controlled by Corollary~\ref{cor.M} using an averaging argument in $R$ similar to the one of Section~\ref{energysection}:
		\[
		2M\int_{u_1}^{u_2}\Omega^2(R)R^{p-3}|\psi|^2(u,v_R(u))\,du
		\leq C(M,\rho,R)\,\tilde E_{\ep}(u_1),
		\]
		where $\ep$ is the fixed weight appearing in \eqref{energyabsorbed2}. Adding \eqref{energyabsorbed2} yields \eqref{lemmapreliminary}.
	}
\end{proof}

We are now going to explain how to  prove \color{black}the boundedness of the $r^p$ weighted energies $\tilde{E}_p(u)$ implies the time decay of the standard energy $E(u)$. This is the core of the new method invented in \cite{RP}.  We write a \textbf{refinement} of the classical argument, which will be important in section \ref{sectionp=3}. { Here and below, $\ep=8q_0 |e|$ corresponds to the fixed weight appearing in the boundedness estimate \eqref{energyabsorbed2}.}

\begin{lem} \label{convert}

		Suppose that $1+\ep<p<2$, and that there exists $C>0$ such that, for every
		$q\in[1+\ep,p]$ and all $u_0(R)\leq u_1<u_2$,
		\begin{equation} \label{lemmadecay}
			\int_{u_1}^{u_2}E_{q-1}[\psi](u)\,du+\tilde E_q(u_2)\leq C\,\tilde E_q(u_1).
		\end{equation}
		Then, for every $k\in\mathbb N$ and $u>1$, there exists $C'=C'(M,\rho,C,R,s,E(u_0(R)),k)>0$ such that
		\begin{align}
			\tilde E_s(u)&\leq C'\left(u^{-k}+\frac{\sup_{2^{-2k-1}u\leq u'\leq u}\tilde E_p(u')}{u^{p-\ep}}\right), &&0\leq s\leq\ep, \label{convert-low}\\
			\tilde E_s(u)&\leq C'\left(u^{-k\frac{p-s}{p-\ep}}+\frac{\sup_{2^{-2k-1}u\leq u'\leq u}\tilde E_p(u')}{u^{p-s}}\right), &&\ep\leq s\leq p. \label{convert-high}
		\end{align}
	
\end{lem}

\begin{rmk}
	{In the case $s=0$, this lemma says that, up to an arbitrarily fast decaying polynomial term and a supremum, the energy decays like $u^{-(p-\ep)}\tilde E_p$; it recovers the usual $u^{-p}$ rate when $\ep=0$.} This formulation that was not present in the original article \cite{RP} will be useful to obtain the almost optimal energy decay in section \ref{sectionp=3}.
\end{rmk}

\begin{rmk}
	The lemma is stated for $1+{\ep}<p<2$ because we apply it to $p<1+ \sqrt{1-4q_0|e|}$. Of course, a similar result still holds without that restriction, using a similar method. In the present paper, we will need more than this lemma in section \ref{sectionp=3}, where a major improvement of the method will become necessary.
\end{rmk}

\begin{proof}
	
	During this proof, we make use of the following notation: $A \lesssim B$ if there exists a constant 
	
	$C_0=C_0(M,\rho,C,R,s,\tilde{E}_p(u_0(R)),k)>0$ such that $A \leq C_0 B$.

	We take $(u_n)_{n \in \mathbb{N}}$ to be a dyadic sequence, i.e.\ $u_{n+1}=2 u_n$ and $u_0 >1$.
	
	We first apply the mean-value theorem on $[u_n, u_{n+1}]$: there exists $u_n <\bar{u}_n < u_{n+1}$ such that 
	
	$$ E_{p-1}[\psi](\bar{u}_n) = \frac{\int_{u_n}^{u_{n+1}} E_{p-1}[\psi](u)du}{u_n},$$ where we used the fact that  $u_{n+1}-u_n = u_n$.
	
	Applying \eqref{lemmadecay} gives 
	
	\begin{equation}  \label{lemmaeq1} E_{p-1}[\psi](\bar{u}_n) \lesssim  \frac{ \tilde{E}_{p}(u_n)}{u_n}. \end{equation}
	
	Now notice that because $1{+\ep}<p<2$, we can  apply H\"{o}lder's inequality under the form: 
	
	$$ {E}_{1{+\ep}}[\psi](u) \leq \left(E_{p-1}[\psi](u) \right)^{p-1{-\ep}} \left(E_{p}[\psi](u) \right)^{2-p{+\ep}}, $$ where we used the fact that $1{+\ep}= (p-1{-\ep})(p-1)+ (2-p{+\ep})p$.
	
	Then we apply \eqref{lemmaeq1} and the precedent inequality for $u=\bar{u}_n$: 
	
	\begin{equation} \label{lemmaeq2} {E}_{1{+\ep}}[\psi](\bar{u}_n) \lesssim \frac{ \sup_{u_n \leq u \leq u_{n+1}} \tilde{E}_{p}(u)}{(u_n)^{p-1{-\ep}}}. \end{equation}
	
	We now use \eqref{lemmadecay} and \eqref{lemmaeq2} to get: 
	
	$$ \int_{u_n}^{u_{n+1}}E_{{\ep}}[\psi](u)du \lesssim \tilde{E}_{1{+\ep}} (u_n)   \lesssim   \tilde{E}_{1{+\ep}} (\bar{u}_{n-1})\lesssim \tilde{E}_{{\ep}}(u_{n-1})+ \frac{ \sup_{u_{n-1} \leq u \leq u_{n}} \tilde{E}_{p}(u)}{(u_n)^{p-1{-\ep}}}.$$
	
	We used that $u_n\sim u_{n-1} \leq \bar{u}_{n-1} \leq u_n$ and for the last inequality, we also used the energy boundedness {of Corollary~\ref{cor.M}}  under the form $\tilde{E}_{{\ep}}(\bar{u}_{n-1}) \lesssim \tilde{E}_{{\ep}}(u_{n-1})$.
	
	Then, following the method developed in \cite{RP}, we add a multiple of the Morawetz estimate \eqref{Morawetzestimate3} to cover for the $ \{ r \leq R \}$ energy bulk terms. After a standard computation, {similarly} to that done in \cite{RP} we get 
	
	$$ \int_{u_n}^{u_{n+1}}\tilde{E}_{{\ep}}(u)du \lesssim \tilde{E}_{{\ep}}(u_{n-1})+ \frac{ \sup_{u_{n-1} \leq u \leq u_{n}} \tilde{E}_{p}(u)}{(u_n)^{p-1-{\ep}}}.$$
	
	Now using the mean-value theorem again on $[u_{n-1},u_{n}]$: there exists $\tilde{u}_{n-1} \in (u_{n-1},u_{n})$ such that 
	
	$$ \tilde{E}_{{\ep}}(\tilde{u}_{n-1}) = \frac{ \int_{u_{n-1}}^{u_{n}}\tilde{E}_{{\ep}}(u)du}{ u_{n}-u_{n-1}} \lesssim \frac{\tilde{E}_{{\ep}}(u_{n-2})}{u_n}+ \frac{ \sup_{u_{n-2} \leq u \leq u_{n-1}} \tilde{E}_{p}(u)}{(u_n)^{p-{\ep}}},$$

	Making use of the energy boundedness from  {Proposition~\ref{Energy}} finally gives 
	
	\begin{equation} \label{iteratelemma}
		\tilde{E}_{{\ep}}(u_n) \lesssim \frac{\tilde{E}_{{\ep}}(u_{n-2})}{u_n}+ \frac{ \sup_{u_{n-2} \leq u \leq u_{n-1}} \tilde{E}_{p}(u)}{(u_n)^{p-{\ep}}}.
	\end{equation}
	
	Now take an integer $k  \geq 1$ and iterate \eqref{iteratelemma} $k$ times:

	$$ \tilde{E}_{{\ep}}(u_n) \lesssim \frac{\tilde{E}_{{\ep}}(u_{n-2k})}{(u_n)^k}+ \sum_{i=1}^k\frac{ \sup_{u_{n-2i} \leq u \leq u_{n-2i+1}} \tilde{E}_{p}(u)}{(u_n)^{p+i-1-{\ep}}}.$$
	
	Now we can use the energy boundedness to bound $\tilde{E}_{{\ep}}(u_{n-2k})$: there exists a constant $\tilde{C}=\tilde{C}(M,\rho)>0$ such that 
	
	$$ \tilde{E}_{{\ep}}(u_{n-2k}) \leq C \cdot \tilde{E}_{{\ep}}(u_0(R)).$$
	
	We then have 
	
	$$ \tilde{E}_{{\ep}}(u_n) \lesssim (u_n)^{-k}+\frac{ \sup_{u_{n-2k} \leq u \leq u_{n-1}} \tilde{E}_{p}(u)}{(u_n)^{p-{\ep}}}.$$

	For $u \in (u_n, u_{n+1})$, we get
	
	\begin{equation} \label{eqfinalelemma}
		\tilde{E}_{{\ep}}(u) \lesssim u^{-k} + \frac{\sup_{ 2^{-2k-1} u \leq u' \leq u} \tilde{E}_{p}(u')}{u^{p-{\ep}}},
	\end{equation} where we used the fact that $\frac{u}{2} \leq u_n \leq u$.

		Estimate \eqref{eqfinalelemma} is precisely \eqref{convert-low} at $s=\ep$. If $0\leq s\leq\ep$, then on $\{r\geq R\}$ one has $r^s\leq C(R,\ep)r^{\ep}$, while the unweighted part of $\tilde E_s$ is independent of $s$. Hence $\tilde E_s\lesssim\tilde E_{\ep}$, and \eqref{convert-low} follows.
		
		For $\ep\leq s\leq p$, H\"older interpolation gives
		\[
		E_s[\psi](u)\leq \left(E_{\ep}[\psi](u)\right)^{\frac{p-s}{p-\ep}}\left(E_p[\psi](u)\right)^{\frac{s-\ep}{p-\ep}}
		\lesssim \left(\tilde E_{\ep}(u)\right)^{\frac{p-s}{p-\ep}}\left(\tilde E_p(u)\right)^{\frac{s-\ep}{p-\ep}}.
		\]
		Using \eqref{eqfinalelemma} and the boundedness of $\tilde E_p$ yields
		\[
		\tilde E_s(u)\lesssim u^{-k\frac{p-s}{p-\ep}}+\frac{\sup_{2^{-2k-1}u\leq u'\leq u}\tilde E_p(u')}{u^{p-s}},
		\]
		which is \eqref{convert-high} and concludes the proof.

\end{proof}

\subsection{Application for the $r^p$ method for $p<2$.} \label{p<2}

In this section, we establish the boundedness of the $r^p$ weighted energy for $p$ slightly inferior to $2$. The argument relies on the absorption of the interaction term into the $	\int_{u_1}^{u_2} E_{p-1}[\psi](u)du$ term. To do this, we make use of a Hardy-type inequality, coupled with the Morawetz estimate to treat the boundary terms on the curve $\{ r=R\} $. 

As a corollary, we establish the $u^{-p}$ decay of the un-weighted energy for the range of $p$ considered. This step is crucial to establish the almost-optimal decay claimed in section \ref{sectionp=3}.

The maximal value of $q_0|Q|$ authorized in this section is $\frac{1}{4}$.

We start by the main result of this section, which is the computation that illustrates how the interaction term can be absorbed into the left-hand-side.

\begin{prop}[$r^p$ weighted  energy boundedness] \label{p<2proposition}
	Assume that $q_0|e| < \frac{1}{4}$.
	
	For all $\eta_0>0$, there exists ${R_{low}}={R_{low}}(M,\rho,\eta_0,e)>r_+$, $\delta=\delta(M,\rho,\eta_0,e)>0$ such that 
	
	if $\|Q-e \|_{L^{\infty}(\mathcal{D}(u_0(R),+\infty))} < \delta$,	then  for all $ 1-\sqrt{1-4q_0|e|} < p < 1+\sqrt{1-4q_0|e|}$ and $ R>{R_{low}}$, there exists $C_2=C_2(M,\rho,R,\eta_0,p,e)>0$  such that for all $u_0(R)\leq u_1<u_2$, we have 		
	\begin{equation}
		\left(p-\frac{4q_0|e|}{2-p} -\eta_0 \right)	\int_{u_1}^{u_2} E_{p-1}[\psi](u)du + \tilde{E}_{p}(u_2) \leq (1+\eta_0) \cdot E_{p}[\psi](u_1 )+ C_2 \cdot \tilde{E}_{{\ep}}(u_1).
	\end{equation}

\end{prop}

\begin{proof}

	We apply Cauchy-Schwarz to control the charge term as: 
	
	$$ |\int_{\mathcal{D}(u_1,u_2)\cap \{r \geq R\}}  \Omega^2 r^{p-2} \Im(\psi \overline{D_v \psi}) du dv | \leq  \left(\int_{\mathcal{D}(u_1,u_2)\cap \{r \geq R\}}r^{p-1}  |D_v \psi|^2 du dv  \right)^{\frac{1}{2}} \left(\int_{\mathcal{D}(u_1,u_2)\cap \{r \geq R\}}r^{p-3} \Omega^4  | \psi|^2 du dv\right)^{\frac{1}{2}}. $$
	
	Using a version of Hardy's inequality \eqref{Hardy5} and since $\Omega^4 \leq \Omega^2$\color{black}, we then establish that 
	
	$$      \left(\int_{\mathcal{D}(u_1,u_2)\cap \{r \geq R\}}r^{p-3} \Omega^{4\color{black}}  | \psi|^2 du dv\right)^{\frac{1}{2}} \leq               \frac{2}{(2-p)\Omega(R)} \left(\int_{\mathcal{D}(u_1,u_2)\cap \{r \geq R\}}r^{p-1}  |D_v \psi|^2 du dv  \right)^{\frac{1}{2}}  +                       \left( \frac{R^{p-2}}{2-p} \int_{u_1}^{u_2}|\psi|^2 (u,v_R(u)) du \right)^{\frac{1}{2}}.  $$
	
	Then, using an averaging procedure in $r$, similar to the one in Lemma \ref{step4}, we find that for all $R>0$ , there exists $\frac{R}{2}<\tilde{R}<R$.
	
	$$  \int_{u_1}^{u_2}|\psi|^2 (u,v_{\tilde{R}}(u)) du \leq  \frac{2}{R} \int_{\mathcal{D}(u_1,u_2)\cap \{\frac{R}{2} \leq r \leq R\} } |\psi|^{2} du dv. $$ 
	
	We are going to abuse notations and still call $R$ this $\tilde{R}$. Using the Morawetz estimate of {Corollary~\ref{cor.M}}, we see that there exists $C=C(M,\rho,R)>0$ such that 
	
	$$	  \frac{R^{p-2}}{2-p} \int_{u_1}^{u_2}|\psi|^2 (u,v_R(u)) du \leq \frac{C}{2-p} \cdot \tilde{E}_{{\ep}}(u_1) .$$
	
	Putting everything together we proved that for a constant $D(M,\rho,R,p)>0$: $$ |\int_{\mathcal{D}(u_1,u_2)\cap \{r \geq R\}}  \Omega^2 r^{p-2} \Im(\psi \overline{D_v \psi}) du dv | \leq \frac{2}{(2-p)\Omega(R)} \int_{u_1}^{u_2} E_{p-1}[\psi](u)du+ D\left(\int_{u_1}^{u_2} E_{p-1}[\psi](u)du  \right)^{\frac{1}{2}} \left( \tilde{E}_{{\ep}}(u_1)\right)^{\frac{1}{2}}.$$ 
	
	Then we combine this estimate with Lemma \ref{rpidentity} to obtain, for all $\eta_0>0$ and using Cauchy--Schwarz on the $D\left(\int_{u_1}^{u_2} E_{p-1}[\psi](u)du  \right)^{\frac{1}{2}} \left( \tilde{E}_{{\ep}}(u_1)\right)^{\frac{1}{2}}$ term:
	$$	\left(p-\frac{\eta_0}{2}- (\sup_{\mathcal{D}(u_1,u_2) } |Q|)\frac{4q_0 \left[ 1+P_0(R) \right]}{(2-p)\Omega(R)}\right)\int_{u_1}^{u_2} E_{p-1}[\psi](u)du + \tilde{E}_{p}(u_2) \leq \left[ 1+P_0(R) \right] E_{p}[\psi](u_1) + (C_1+ [q_0 \sup_{\mathcal{D}(u_1,u_2) } |Q|] \cdot\frac{D^2}{\eta_0} )   \tilde{E}_{{\ep}}(u_1).$$

	Then, since $P_0(R)=O(R^{-1})$ when $R$ is large, we can take ${R_{low}}$ large enough and $\delta$ small enough so that: 
	\color{black}
	$$|(\sup_{\mathcal{D}(u_1,u_2) } |Q|)\frac{4q_0 \left[ 1+P_0(R) \right]}{(2-p)\Omega(R)}-\frac{4 q_0 |e| } {(2-p)}|\leq \frac{\eta_0}{	2	\color{black}},$$ $$1+P_0(R) \leq 1+\eta_0.$$ 	More precisely, it suffices that $ {R_{low}} > \frac{ C(M,\rho) \cdot  q_0 |e|}{(1-\sqrt{1-4q_0|e|}) \cdot \eta_0}$ and $\delta< 2(1-\sqrt{1-4q_0|e|}) \cdot \eta_0$, for some $C(M,\rho)>0$. Notice that we established in particular that: 
	$$  p- \eta_0- \frac{4q_0 |e|}{2-p} \leq p-\frac{\eta_0}{2}- (\sup_{\mathcal{D}(u_1,u_2) } |Q|)\frac{4q_0 \left[ 1+P_0(R) \right]}{(2-p)\Omega(R)}.$$

	This concludes the proof of Proposition \ref{p<2proposition}, after combining all the estimates. \color{black}

\end{proof}	

To finish this section, we establish the best decay of the energy that we can attain so far. Note that the constants do not need to be monitored as sharply as in section \ref{sectionp=3}.

\begin{cor} \label{energydecaycorollary} Suppose that
	$q_0|e| < {\frac{1}{12}}${, so that the range $1+\ep<p'<1+\sqrt{1-4q_0|e|}$
		below is non-empty}.
	
	Then for all $1{+\ep}<p'<1+\sqrt{1-4q_0|e|}$, for all $0 \leq s \leq p'$ and for all $k \in \mathbb{N}$ large enough, there exists $C_0=C_0(M,\rho,e,R,p',s,k,E(u_0(R)))>0$
	and  $C'_0=C'_0(M,\rho,e,R,p',s,k,\tilde{E}_{p'}(u_0(R)))>0$ such that for all $u>1$,  we have the following energy decay

	\begin{equation}\begin{split} \label{energydecayp<2}
			&{		\tilde E_s(u)\leq C'\left(u^{-k}+\frac{\sup_{2^{-2k-1}u\leq u'\leq u}\tilde E_{p'}(u')}{u^{p'-\ep}}\right) \leq \frac{C_0'}{u^{p'-\ep}},\ 0\leq s\leq\ep, }\\ &
			{	\tilde E_s(u)\leq C'\left(u^{-k\frac{p'-s}{p'-\ep}}+\frac{\sup_{2^{-2k-1}u\leq u'\leq u}\tilde E_{p'}(u')}{u^{p'-s}}\right) \leq \frac{C_0'}{u^{p'-s}},\ \ep\leq s\leq p'.}
		\end{split}.
	\end{equation}

\end{cor}

\begin{proof}
	
	
	We make use of the result of Proposition \ref{p<2proposition} and take $\eta_0$ sufficiently small so that there exists
	
	$D_0=D_0(M,\rho,p',e,R)$ so that for all $u_0(R) \leq u_1<u_2$:
	
	\begin{equation} 
		\int_{u_1}^{u_2}E_{p'-1}[\psi](u)du+\tilde{E}_{p'}(u_2) \leq D_0 \cdot  \tilde{E}_{p'}(u_1).
	\end{equation}

	Thus the hypothesis of Lemma \ref{convert} are satisfied for the chosen $p'$ range. 
	
	This proves the first inequality of \eqref{energydecayp<2}.
	
	The second one solely relies on the boundedness of the $r^p$ weighted energy:

	$$ \sup_{ 2^{-2k-1} u \leq u' \leq u} \tilde{E}_{p'}(u') \leq D_0 \tilde{E}_{p'}(u_0(R)). $$
	
	We also need to take $k>p'$ so that $k(1-\frac{s-{\ep}}{p'-{\ep}})> p'-s-{\ep}$.
	
	This concludes the proof of Corollary \ref{energydecaycorollary}.
	
\end{proof} 
\subsection{Extension of the $r^p$ method to $p<3$.} \label{sectionp=3}

In this section, we establish the improved decay of the energy, at a rate $u^{-p(e)}$, for some $2<p(e)<2+\sqrt{1-4q_0|e|}<3$ and provided that  $0 \leq q_0 |e| < {0.04}$. We prove that $p(e) \rightarrow3$ as $e \rightarrow 0$, which is the (limit) optimal rate as $e \rightarrow 0$, and we also produce a first order asymptotic expansion in $q_0|e|$. Unfortunately, $p(e)$ cannot be made explicit easily, as it is obtained as a solution of an optimisation problem.

Since $p(e)>2$, we obtain  an integrable point-wise decay on the event horizon, crucial for the interior study, c.f.\ section \ref{interior}.

The strategy employed to absorb the interaction term differs radically from that of section \ref{p<2}, where the maximal $p$ was strictly inferior to $2$.

Indeed, we now aim at absorbing this term inside the 	$\tilde{E}_{p}(u_2)$ term. To do this, we have to solve an ordinary differential inequality and, making use of the relevant smallness, prove a small growth $\tilde{E}_{p(e)+2\epsilon{+\ep}}(u) \lesssim u^{2\epsilon}$. 

An optimisation problem arises, as the estimates close more easily when $0<\epsilon<\frac{1}{2}$ is close to $\frac{1}{2}$. On the other hand, taking $\epsilon$ too large deteriorates the decay rate so there is a trade-off in the choice of $\epsilon$.

In the core of the proof we use a Gr\"{o}nwall like argument (although it involves a square-root, which is not standard) to handle an estimate of the form $\tilde{E}_{\pp}(u) \lesssim \nu(e) \cdot (\tilde{E}_{{\pp}-1}(u))^{\frac{1}{2}} \cdot  ( \int_{u_1}^{u}\tilde{E}_{{\pp}}(u))^{\frac{1}{2}}$, where $\nu(e)>0$ is a small constant, related the $\epsilon$ we mentioned. Ultimately, we will obtain a small growth of the $r^p$ weighted energy by this method: $\tilde{E}_{\pp}(u) \lesssim u^{2\epsilon}$, for $u$ large, where $\tp=p(e)+2\epsilon(e)+{\ep}$.

The first part of the argument is to prove that $\nu(e)$ is indeed small: to do this, we must impose that $q_0 |e|$ is small enough, hence the restriction $q_0 |e| < {0.04}$. This part involves an optimisation c.f. Lemma \ref{fLemma}.

The second part is to establish an approximate version of the estimate $\tilde{E}_{\pp}(u_2) \lesssim \nu(e) \cdot (\tilde{E}_{{\pp}-1}(u_1))^{\frac{1}{2}} \cdot  ( \int_{u_1}^{u_2}\tilde{E}_{{\pp}}(u)du)^{\frac{1}{2}}$. The most crucial argument is of a non-linear nature, and is best summed up by the alternative of 	Lemma \ref{alternatives}. In short, according to how $\left( \int_{u_1}^{u_2} \tilde{E}_{{\pp}}[\psi](u)du  \right)^{\frac{1}{2}} $ compares with some initial energies, we either have \eqref{alternative1} or \eqref{alternative2}.

If we worked directly on the interval $[1,u]$ for $u > 1$ large and apply Lemma \ref{alternatives}, we would obtain the
boundedness of $\tilde{E}_p$ immediately \textbf{only if} alternative 1 applies, i.e.\ if \eqref{alternative1} is true. On the other hand, in the event that \eqref{alternative2}
holds instead (alternative 2), we obtain a linear growth in $u$, which is disastrous.

This is why we work with a $\lambda$-adic sequence, $\lambda > 1$, to obtain the decay of the energy by a pigeon-hole
argument. In this case, in the event that \eqref{alternative2} holds on $[u_0 \cdot  \lambda^n, u_0 \cdot  \lambda^{n+1} ]$ (alternative 2), we do obtain the (almost \footnote{\eqref{alternative2} contains a $u^{2\epsilon}$ growing weight  because we already introduced a growth  in the induction hypothesis, as seen in \eqref{HREplemma}. This growth, however, cannot be avoided since it is possible that only \eqref{alternative1} (alternative 1) applies.}) boundedness of $\tilde{E}_{\pp}$.

The traditional use of the $ r^p $ method makes use of $\lambda=2$, namely dyadic sequence. However, the presence
of the two alternatives \eqref{alternative1} or \eqref{alternative2} renders necessary to choose  $\lambda$ appropriately. This is because, if \eqref{alternative1} (alternative 1)	applies on every $\lambda$-adic intervals between $1$ and $u$, then we obtain some $u^{2\epsilon_0}$ growth, where $\epsilon_0$ depends on $\lambda$. Since this possibility cannot be excluded, $\lambda$ has to be chosen in accordance with the $\epsilon$ resulting from the optimisation we mentioned earlier (see Remark \ref{proofsimple} for more details).

Because the argument is of a non-linear nature, we proceed by induction, as the induction hypothesis can
be fed into Lemma \ref{alternatives} to obtain an estimate which is sufficient to close the induction step. Upon completion
of the induction, we show that for some   $2 < \tp$ and $0 < \epsilon(e)$, $\tilde{E}_{\tp-1}\lesssim u^{-1+2\epsilon(e)}${, $\tilde{E}_{\tp}\lesssim u^{2\epsilon(e)}$} and we eventually 	obtain the energy decay $E(u) {\ls \tilde{E}_{\ep}(u)}
\lesssim u^{-\tp+2\epsilon(e)+{\ep}} $, making use of Corollary \ref{energydecaycorollary}.

\begin{prop} \label{p=3}
	
	Assume that $q_0|e|  < {0.04}$.

	Assume also that for  all $   0 \leq   p' < 2+\sqrt{1-4q_0|e|}$, $E_{p'}(u_0(R))<\infty$ and that the energy boundedness \eqref{energyabsorbed2} and the Morawetz estimate \eqref{Morawetzestimate3} hold.

	Then there exists $ 2 <   p(e) < 2+\sqrt{1-4q_0|e|}{-\ep}$
	and $D=D(M,\rho,R,e)>0$ such that for all $u>1$:
	
	\begin{equation}
		\tilde{E}_{\ep}(u)\leq \frac{D}{u^{{p(e)}}},
	\end{equation}
	\begin{equation}
		\tilde{E}_{{p(e)+\ep-1}}(u) \leq \frac{D}{u},
	\end{equation}
	\begin{equation}
		\tilde{E}_{{p(e)+\ep}}(u) \leq D.
	\end{equation}
	Here $p(e):= \tp - 2\epsilon(e) - \ep$, where $\tp$ and $\epsilon(e)$ are
		given by Lemma \ref{fLemma} and $\ep = 8q_0|e|$ as in Section
		\ref{energynotation}.
	Moreover $p(e) \rightarrow 3$ as $e\rightarrow 0$, and  $p(e)$ has the following
	Taylor expansion when $e \rightarrow 0$: \begin{equation}
		\label{Taylorp}
		p(e)=3-2\exp(\frac{1}{2})\cdot\sqrt{\frac{\sqrt{6}}{3}} \cdot (q_0 |e|)^{\frac{1}{2}}+O(q_0 |e|).\end{equation}
\end{prop}
\begin{rmk}\label{p(e)remark}
	The point of this proposition is two-fold: first, we want to find the maximal $q_0 |e|$  such that there exists a $s>2$ with $E(u) \lesssim u^{-s}$, implying an integrable decay for the scalar field on the event horizon $|\phi_{\mathcal{H}^+}|(v) \lesssim v^{-\frac{s}{2}}$. As it turns out, this cannot be managed on the full range $q_0 |e| \in (0,\frac{1}{4})$, as our techniques use the smallness of $q_0 |e|$: this is why we impose $q_0 |e|< {0.04}$.  We also want to prove the conjectured limit rate $s=3$ predicted by \cite{HodPiran1} when the charge is small, i.e.\ that there exists $0<\zeta(e)=o(1)$ when $|e| \rightarrow 0$ such that $E(u) \lesssim u^{-3+\zeta(e)}$. This is provided by \eqref{Taylorp} and $\zeta(e)=2 \exp(\frac{1}{2})\cdot\sqrt{\frac{\sqrt{6}}{3}} \cdot (q_0 |e|)^{\frac{1}{2}}+O(q_0 |e|)$. To achieve this result, the smallness of $q_0|e|$ is exploited again. In the former section, the $r^p$ weighted energy was bounded for a maximal $p$ which was $p_{max}=1+\sqrt{1-4q_0|e|}$. In this section, we try to increase the maximal $p$ to $p+1$. However, this is impossible: a loss is occurred which is why $p(e)<p_{max}+1= 2+\sqrt{1-4q_0|e|}$. Observe that the asymptotic expansion of $2+\sqrt{1-4q_0|e|}$ as $e \rightarrow 0$ is 	$2+\sqrt{1-4q_0|e|}=3- 2\cdot q_0 |e|+O((q_0 |e|)^2)$, to be compared with \eqref{Taylorp}, where we "lost" a $(q_0|e|)^{\frac{1}{2}}$ power in the expansion.
\end{rmk}
\begin{rmk} \label{proofsimple}
	Because the proof relies on many parameters to choose with respect to each others, we give a brief summary of the argument and a list of all parameters that we use for the benefit on the reader. The main goal is to close by induction on $k\in \mathbb{N}$ the following two estimates: \begin{equation} \label{proofsimpleinduction}
		\tilde E_{\breve{p}-1}(\tilde{u}_k) \leq \frac{\Delta}  {\tilde{u}_k^{1-2\epsilon}}, \hskip 5 mm \tilde E_{\breve{p}}(\tilde{u}_k) \leq \frac{\Delta \cdot \lambda \cdot \breve p}{\lambda-1} \cdot \tilde{u}_k^{2\epsilon}
	\end{equation}  for a large constant $\Delta>0$ and for $\tilde{u}_k= \lambda^k U_0$, $\lambda>1$,
	and {$\breve{p}-2\epsilon-\ep>2$}. Once we close this induction we obtain $$ E(u) {\ls \tilde{E}_{\ep}(u)}
	\lesssim u^{-\tp+2\epsilon(e)+{\ep}}, $$ for all $u \geq U_0$ which is the desired estimate, hence {$p(e)=\breve{p}-2\epsilon-\ep>2$} is the rate from Proposition \ref{p=3}.
	
	If \eqref{alternative2} (alternative 2) of Lemma \ref{alternatives} applies, one crucial smallness condition \eqref{fcondition} is necessary to absorb the non-linearity into the left-hand-side, involving $f(p,e)$ defined in \eqref{fdef} and $z_{\epsilon}(x)$ defined in \eqref{zdef}.
	
	We now make a list of all the quantities that will play a role in the proof and how we choose them:
	
	\begin{itemize}
		\item {$2+2\epsilon+\ep<\breve{p}(e)<3$}, the largest $q$ for which we can
		estimate the $r^q$ weighted estimate as $E_{\breve{p}(e)}(u) \lesssim u^{2\epsilon}$.
		One of the two unknowns in the problem. $\breve{p}(e)$ is obtained as the largest
		number satisfying \eqref{fcondition} as \begin{equation}
			\breve{p}(e) \approx 3- \exp(1)\cdot\frac{q_0 |e|}{\epsilon \sqrt{6}}{+ O( q_0|e|)}.
		\end{equation}
		
		\item $0<\epsilon(e)<\frac{1}{2}$ defined so that
		{$p(e)+\ep=\breve{p}(e)-2\epsilon$} is the largest ${q}$ for which
		$E_{{q}}(u)$ is bounded. This is the second unknown in the problem. $\epsilon(e)$ is
		chosen to minimize $\frac{\exp(1)\cdot q_0 |e|}{\epsilon \sqrt{6}} +2\epsilon$
		hence we also obtain $p(e)$ as such when $q_0e\rightarrow0$ \begin{equation} \label{proofsimpletaylor}
			\epsilon(e) \approx \frac{\exp(\frac{1}{2})}{2}\sqrt{\frac{\sqrt{6}}{3}} \cdot (q_0 |e|)^{\frac{1}{2}}+O(q_0 |e|),\hskip 8 mm 	p(e) = \breve{p}(e)-2\epsilon(e){-\ep} \approx 3-2\exp(\frac{1}{2})\cdot\sqrt{\frac{\sqrt{6}}{3}} \cdot (q_0 |e|)^{\frac{1}{2}}+O(q_0 |e|).
		\end{equation}

		\item The small quantity $\eta_0(R)=O(R^{-1})$ as $R\rightarrow+\infty$, previously imposed from Proposition \ref{p<2proposition}. We choose $R$ large enough so that $\eta_0(R)< \tilde{\eta}_0(e,\epsilon(e))$ and $\tilde{\eta}_0(e,\epsilon(e))$ only depends on $e$ as we already fixed $\epsilon=\epsilon(e)$.

		\item $\beta$, a free choice from Lemma \ref{alternatives}. $\beta(\epsilon)$ is chosen to minimize the function $\beta \rightarrow z_{\epsilon}(\beta)$ defined in \eqref{zdef}. $\beta$ can be expressed explicitly in terms of $\epsilon$ using $W_{-1}$, the $(-1)$ branch of the Lambert function see \eqref{betaexplicit}.
		
		\item $\lambda>1$, a free choice as we use $\lambda$-adic	sequences in the mean-value argument. We will choose $\lambda(\beta(\epsilon),\epsilon,\eta_0(e))$ according to the explicit formula \eqref{lambdacondition}. This choice is a compromise between a small $\lambda-1$ which is better when \eqref{alternative1} (alternative 1) applies (since it involves a geometric loss of $\lambda$) but worse when  \eqref{alternative2} (alternative 2) applies since it degrades in the constant in \eqref{proofsimpleinduction}.
		
		\item $\tilde{u}_0=U_0>1$, where $\{u \geq U_0\}$ is defined as the region on which decay is proven. $U_0$ must be large enough so that the \textit{lower order terms}, i.e.\ the terms that can be treated using the previously proven decay in section \ref{p<2}, are small. Since $\eta_0$, $\beta$, $\epsilon$, $\breve{p}$ and $\lambda$ have already been fixed at this stage and \underline{only depend on $e$} we can fix $U_0= \tilde{U}_0(e)$ to be large enough to only beat $e$-dependent constants.
		
		\item $\Delta>0$ the large constant appearing in \eqref{proofsimpleinduction}. Mostly $\Delta$ must be chosen to be compatible with the initialization of the induction i.e.\ \eqref{proofsimpleinduction} when $k=0$. Since $U_0(e)$ has already been fixed, it is enough to take $\Delta > \tilde{\Delta}(e)$ again to be large enough to only beat $e$-dependent constants.
		
	\end{itemize}
	The heart of the optimization, crucial to obtain the smallness condition \eqref{fcondition}, is obtained in Lemma \ref{fLemma}. Note that Lemma \ref{fLemma} consists of three results \begin{enumerate}
		\item \label{1}There exists $\delta_0>0$ such that, if $q_0 |e|< \delta_0$ then there exists $(\breve{p}(e),\epsilon(e))$ which satisfies \eqref{fcondition} and $$p(e)=\breve{p}(e)-2\epsilon(e){-\ep}>2.$$
		
		\item \label{1.1} $\breve{p}(e)$, $\epsilon(e)$ and $p(e)=\breve{p}(e)-2\epsilon(e){-\ep}$ satisfy the Taylor expansion prescribed by \eqref{proofsimpletaylor} as $e \rightarrow 0$.
		
		\item \label{1.3}The value $\delta_0={0.04}$ is sufficient for statement \ref{1} of Theorem~\ref{decaytheorem} (and consequently also statement \ref{1.1}) to apply.
		
	\end{enumerate}
	Statement \ref{1.3} requires another optimization and graphing an explicit function $\epsilon \rightarrow w(\epsilon)$ 
	depending on $z( \beta(\epsilon))$, itself being expressed explicitly in terms of $\epsilon$ via the $(-1)$ branch of the Lambert function.

\end{rmk}
\begin{proof}

	Like in section \ref{p<2}, we apply Cauchy-Schwarz to control the charge term but we distribute the $r$ weights differently as: 
	
	$$ |\int_{\mathcal{D}(u_1,u_2)\cap \{r \geq R\}}  \Omega^2 r^{{\pp}-2} \Im(\psi \overline{D_v \psi}) du dv | \leq  \left(\int_{\mathcal{D}(u_1,u_2)\cap \{r \geq R\}}r^{{\pp}}  |D_v \psi|^2 du dv  \right)^{\frac{1}{2}} \left(\int_{\mathcal{D}(u_1,u_2)\cap \{r \geq R\}}r^{{\pp}-4} \Omega^4  | \psi|^2 du dv\right)^{\frac{1}{2}}. $$

	Using a version of Hardy's inequality \eqref{Hardy5}, we then establish that 
	
	$$      \left(\int_{\mathcal{D}(u_1,u_2)\cap \{r \geq R\}}r^{{\pp}-4} \Omega^2  | \psi|^2 du dv\right)^{\frac{1}{2}} \leq               \frac{2}{(3-{\pp})\Omega(R)} \left(\int_{\mathcal{D}(u_1,u_2)\cap \{r \geq R\}}r^{{\pp}-2}  |D_v \psi|^2 du dv  \right)^{\frac{1}{2}}  +                       \left( \frac{R^{{\pp}-3}}{3 -{\pp}} \int_{u_1}^{u_2}|\psi|^2 (u,v_R(u)) du \right)^{\frac{1}{2}}.  $$
	
	Combining both inequalities we see that 
	
	\begin{equation} \label{p<3Identity} \begin{split}
			|\int_{\mathcal{D}(u_1,u_2)\cap \{r \geq R\}}  \Omega^2 r^{{\pp}-2} \Im(\psi \overline{D_v \psi}) du dv | \leq  \frac{2}{(3-{\pp})\Omega(R)} \left( \int_{u_1}^{u_2} E_{{\pp}-2}[\psi](u)du  \right)^{\frac{1}{2}} \left(\int_{u_1}^{u_2} E_{{\pp}}[\psi](u)du \right)^{\frac{1}{2}}\\+\left(\int_{u_1}^{u_2} E_{{\pp}}[\psi](u)du \right)^{\frac{1}{2}}\left( \frac{R^{{\pp}-3}}{3 -{\pp}} \int_{u_1}^{u_2}|\psi|^2 (u,v_R(u)) du \right)^{\frac{1}{2}}.
		\end{split}
	\end{equation}
	
	The main contribution of the right-hand-side is the first term.		Using the results of section \ref{p<2}, we see that 
	there exists $C'=C'(M,\rho,R,{\pp},e)>0$ such that \footnote{It is very important that $C'$, which depends on the large constant $R$, is multiplying $\left(  E(u_1) \right)^{\frac{1}{2}} $, which enjoys a faster decay in $u_1$, so that, for large $u_1$, the large constant $C'$ can be absorbed.} for all $u_0(R) \leq u_1 < u_2$: 
	
	$$  \left( \int_{u_1}^{u_2} E_{{\pp}-2}[\psi](u)du  \right)^{\frac{1}{2}} \leq (1+\eta_0) \cdot ({\pp}-1-\frac{4 q_0 |e|}{3-{\pp}})^{-\frac{1}{2}} \cdot  \left[   \left(E_{{\pp}-1}[\psi](u_1)\right)^{\frac{1}{2}}+ C' \cdot  \left(  \tilde{E}_{{\ep}}(u_1) \right)^{\frac{1}{2}} \right],$$	where $\eta_0>0$ is arbitrarily small. Note that we used the fact that for all $\alpha$, $\beta >0$, $\sqrt{\alpha + \beta} \leq \sqrt{\alpha}+\sqrt{\beta}$.

	Similarly to the method employed in section \ref{p<2}, we deal with the second term in the right-hand-side of \eqref{p<3Identity} using the Morawetz estimate of section \ref{Morawetz}. Therefore there exists $D=D(M,\rho,{\pp},e,R)>0$ such that
	
	\begin{equation} \label{p<3Identity2} \begin{split}
			|\int_{\mathcal{D}(u_1,u_2)\cap \{r \geq R\}}  \Omega^2 r^{{\pp}-2} \Im(\psi \overline{D_v \psi}) du dv | \\ \leq  \left[    \frac{(1+\eta_0) }{\Omega(R)}\cdot \frac{2}{3-{\pp}}\cdot ({\pp}-1-\frac{4 q_0 |e|}{3-{\pp}})^{-\frac{1}{2}} \cdot  \left(E_{{\pp}-1}[\psi](u_1)\right)^{\frac{1}{2}}+ D \cdot \left( \tilde{E}_{{\ep}}(u_1) \right)^{\frac{1}{2}} \right]\cdot \left(\int_{u_1}^{u_2} E_{{\pp}}[\psi](u)du \right)^{\frac{1}{2}}.
		\end{split}
	\end{equation}

	Again, the main contribution in the right-hand-side is the first term. Actually we will see that the ${\tilde{E}_{\ep}(u_1)}$ term enjoys a better decay in $u$ than $E_{{\pp}-1}[\psi](u_1)$ term because in this section ${\pp}>1{+\ep}$. This is how we will absorb the second term of the right-hand-side into the first term, even though $D$ can be arbitrarily large.

	We then combine the bound from Lemma \ref{rpidentity} with the bound on the charge term \eqref{p<3Identity2} to get for any $\eta_0>0$ small enough: 
	
	\begin{equation*}  \begin{split}
			{\pp}\int_{u_1}^{u_2}E_{{\pp}-1}[\psi](u)du + \tilde{E}_{{\pp}}(u_2) \leq (1+\eta_0) \cdot  E_{{\pp}}[\psi](u_1 )  + \tilde{f}({\pp},e,\eta_0) \cdot\left(E_{{\pp}-1}[\psi](u_1) \right)^{\frac{1}{2}} \left(\int_{u_1}^{u_2} E_{{\pp}}[\psi](u)du \right)^{\frac{1}{2}}  \\ + \tilde{D}  \cdot  \left[    \left(\tilde{E}_{{\ep}}(u_1)\right)^{\frac{1}{2}} \cdot \left(\int_{u_1}^{u_2} E_{{\pp}}[\psi](u)du \right)^{\frac{1}{2}} + \tilde{E}_{{\ep}}(u_1) \right],
		\end{split}		
	\end{equation*}	where $\tilde{D} =\tilde{D}(M,\rho,{\pp},e,R)>0$ and we define the functions $f({\pp},e)$ and $\tilde{f}({\pp},e,\eta_0)$ as \begin{equation} \label{fdef}
		\tilde{f}({\pp},e,\eta_0):=(1+\eta_0) \cdot f({\pp},e)=(1+\eta_0) \cdot 4 q_0 |e| \cdot (3-{\pp})^{-1}  \cdot ({\pp}-1-\frac{4 q_0 |e|}{3-{\pp}})^{-\frac{1}{2}}
	\end{equation}  and we took $R>R'_0(M,\rho,\eta_0)$ so $P_0(r)$ (defined in the statement of Lemma \ref{rpidentity}) satisfies $ |1+P_0(r)| < (1+\eta_0)$. 
	
	In contrast with the strategy adopted in section \ref{p<2}, we now aim at absorbing the error term into the $E_{{\pp}}[\psi](u_2)$ term of the left-hand-side.
	
	Therefore we are going to drop temporarily the first term of the left-hand-side and consider the differential functional inequality, with $u_2$ being the variable, $u_1$ the constant and $\int_{u_1}^{u_2} E_{{\pp}}[\psi](u)du$ the unknown function 
	
	\begin{equation} \begin{split} \label{diffeq}
			E_{{\pp}}[\psi](u_2) \leq   \left[ \tilde{f}({\pp},e,\eta_0) \cdot  \left(E_{{\pp}-1}[\psi](u_1) \right)^{\frac{1}{2}}+   \tilde{D} \cdot \left(\tilde{E}_{{\ep}}(u_1)\right)^{\frac{1}{2}}  \right] \cdot \left(\int_{u_1}^{u_2} E_{{\pp}}[\psi](u)du \right)^{\frac{1}{2}} + (1+\eta_0) \cdot E_{{\pp}}[\psi](u_1 ) + \tilde{D} \cdot \tilde{E}_{{\ep}}(u_1).
		\end{split}		
	\end{equation}
	
	As usual, the main contribution of the right-hand-side is the first term, the others being treated as errors.
	To deal with those error terms, we will require a small integration lemma:
	
	\begin{lem} \label{integration}
		
		\begin{equation}
			\int_{u_1}^{u_2} \frac{du}{a \sqrt{u}+b} = \frac{2}{a} \left[ \sqrt{u}- \frac{b}{a} \log(\sqrt{u}+ \frac{b}{a})\right]_{u_1}^{u_2},			\end{equation} where for every function $f$ we define $\left[ f(u)\right]_{u_1}^{u_2}:=f(u_2)-f(u_1).$	
	\end{lem}
	
	\begin{proof}
		The proof is elementary, using the change of variable $x=\sqrt{u}$. We leave the details to the reader.
	\end{proof}

	Now, we are going to integrate the differential equation we obtained earlier. Two behaviours are possible: integrated decay or boundedness. We show that only one of those behaviour can occur on any given interval, which is the object of the following "lemma of two alternatives":

	\begin{lem} \label{alternatives}
		Assume that for some $1 \leq  u_1$, $\Delta>0$, $\epsilon>0$, $2+{\ep}<{\pp}<2+\sqrt{1-4q_0|e|}$, we have \begin{equation} \label{HREplemma}
			E_{{\pp}-1}[\psi](u_1) \leq \frac{\Delta}{u_1^{1-2 \epsilon}}.
		\end{equation}
		
		Then for any $\beta>0$, $u_2 \geq u_1$ we either have (what we later call the alternative 1: boundedness)
		
		\begin{equation} \begin{split} \label{alternative1}
				{\pp}\int_{u_1}^{u_2} E_{{\pp}-1}[\psi](u)du + \tilde{E}_{{\pp}}(u_2) \leq (1+\beta) \cdot \left((1+\eta_0) \cdot E_{{\pp}}[\psi](u_1)+\frac{D_{0} }{(u_1)^{{\pp}-1-{\ep}}}\right),
			\end{split}		
		\end{equation}
		
		or we have (what we later call the alternative 2: integrated decay)
		
		\begin{equation} \begin{split} \label{alternative2}
				{\pp}\int_{u_1}^{u_2} E_{{\pp}-1}[\psi](u)du + \tilde{E}_{{\pp}}(u_2)\\ \leq \frac{\left(  \sqrt{\Delta} \cdot \tilde{f}({\pp},e,\eta_0) + D'_0 \cdot  (u_1)^{\frac{2-{\pp}+{\ep}}{2}-\epsilon} \right)^2}{2}  \cdot  \left(\frac{1+\beta}{\beta-\log(1+\beta)}\right) \cdot  \left(\frac{u_2-u_1}{u_1^{1-2\epsilon}}\right),
			\end{split}		
		\end{equation}	where we recall $\tilde{f}$ from \eqref{fdef}, $D_{0}=\tilde{D} \cdot C'_0$,  $D'_0:= \tilde{D} \cdot \sqrt{C'_0}$ and $C'_0$ is the constant in Corollary \ref{energydecaycorollary}.
	\end{lem}
	
	\begin{proof}

		We now combine \eqref{HREplemma} and the estimate
		$\tilde E_{\ep}(u)\lesssim u^{1-{\pp}+{\ep}}$ from Corollary
		\ref{energydecaycorollary}, applied with $p'={\pp}-1$, with
		\eqref{diffeq} to get:

		\begin{equation} \begin{split} \label{diffeq2}\begin{split}
					E_{{\pp}}[\psi](u_2) \leq {\pp}\int_{u_1}^{u_2} E_{{\pp}-1}[\psi](u)du +	E_{{\pp}}[\psi](u_2)  \leq  a  \left(\int_{u_1}^{u_2} E_{{\pp}}[\psi](u)du \right)^{\frac{1}{2}}+b,\end{split}
			\end{split}		
		\end{equation}
		
		with $a=   \left(\tilde{f}({\pp},e,\eta_0) \cdot \sqrt{\Delta} \cdot  (u_1)^{-\frac{1}{2}+\epsilon}+  D'_0\cdot  (u_1)^{-\frac{{\pp}-1-{\ep}}{2}} \right)>0$ and $b=(1+\eta_0) \cdot  E_{{\pp}}[\psi](u_1 )+ \frac{D_{0}}{(u_1)^{{\pp}-1-{\ep}}}>0 $. 
		
		Now, using Lemma \ref{integration} we see that for all $0 < u_1<u_2$: 
		
		\begin{equation} \label{ab}
			\left(\int_{u_1}^{u_2} E_{{\pp}}[\psi](u)du \right)^{\frac{1}{2}} \leq \frac{a}{2}(u_2-u_1)+\frac{b}{a}\log \left(1+\frac{a}{b}(\int_{u_1}^{u_2} E_{{\pp}}[\psi](u)du )^{\frac{1}{2}}\right)
		\end{equation} 
		
		Now for all $\beta>0$ we have the following alternative:

		Either $$\left(\int_{u_1}^{u_2} E_{{\pp}}[\psi](u)du \right)^{\frac{1}{2}} \leq \beta \cdot \frac{b}{a},$$
		
		in which case, combining with \eqref{diffeq2} we find 
		
		$$	{\pp}\int_{u_1}^{u_2} E_{{\pp}-1}[\psi](u)du+ \tilde{E}_{{\pp}}(u_2) \leq (1+\beta)  \cdot b,$$
		
		which is \eqref{alternative1}.
		
		Or $$\left(\int_{u_1}^{u_2} E_{{\pp}}[\psi](u)du \right)^{\frac{1}{2}} \geq \beta \frac{b}{a},$$
		which can also be written as $b \leq \frac{a}{\beta} 	\cdot	\left(\int_{u_1}^{u_2} E_{{\pp}}[\psi](u)du \right)^{\frac{1}{2}}$. {In that case, \eqref{diffeq2} gives}
		
		\begin{equation} \label{lemmaestimate}
			{\pp}\int_{u_1}^{u_2} E_{{\pp}-1}[\psi](u)du + \tilde{E}_{{\pp}}(u_2)\leq (1+\frac{1}{\beta})\cdot a  \left(\int_{u_1}^{u_2} E_{{\pp}}[\psi](u)du \right)^{\frac{1}{2}}.
		\end{equation}
		
		Now because the function $\frac{\log(1+x)}{x}$ is decreasing, we also have, since $\beta \leq \frac{a}{b} 	\cdot	\left(\int_{u_1}^{u_2} E_{{\pp}}[\psi](u)du \right)^{\frac{1}{2}}$: $$\frac{b}{a}\log \left(1+\frac{a}{b}(\int_{u_1}^{u_2} E_{{\pp}}[\psi](u)du )^{\frac{1}{2}}\right) \leq \frac{\log(1+\beta)}{\beta}\left(\int_{u_1}^{u_2} E_{{\pp}}[\psi](u)du \right)^{\frac{1}{2}}.$$
		
		Hence from \eqref{ab} we have
		
		$$	\left(\int_{u_1}^{u_2} E_{{\pp}}[\psi](u)du \right)^{\frac{1}{2}} \leq \frac{a}{2(1-\frac{\log(1+\beta)}{\beta})}(u_2-u_1).$$
		
		The combination with \eqref{lemmaestimate} gives $$ 	{\pp}\int_{u_1}^{u_2} E_{{\pp}-1}[\psi](u)du + \tilde{E}_{{\pp}}(u_2)\leq \frac{a^2 \cdot (1+\beta)}{2(\beta-\log(1+\beta)} \cdot (u_2-u_1). $$ which is exactly \eqref{alternative2}, after replacing $a$ by its definition.
		
	\end{proof}

	We define the following function \begin{equation} \label{zdef}
		z_{\epsilon}(x)=z(x)=\frac{(1+x)^{\frac{1}{2\epsilon}}}{x-\log(1+x)} 
	\end{equation} which will play a major role, in particular when \eqref{alternative2} (alternative 2) holds, as we are going to see later. Assuming $0< \epsilon< \frac{1}{2}$, it can be shown that function $z$ admits a unique minimum on $(0,+\infty)$ that we denote $\beta(\epsilon)$: this is the $\beta$ value to which we will later apply Lemma \ref{alternatives}. Since $\beta(\epsilon)$ solves the equation of a critical point $z'(\beta)_{ |\beta=\beta(\epsilon)}=0$ we get the implicit formula \begin{equation} \label{betaimplicit}
		{(1-2\epsilon)}	 \cdot  \beta(\epsilon)=\log(1+ \beta(\epsilon)).
	\end{equation} 
	
	Using Taylor expansions, it can be shown that $\beta(\epsilon) \rightarrow 0$, $z(\beta(\epsilon)) \rightarrow +\infty$ as  $\epsilon \rightarrow 0 $ and more precisely \begin{equation} \label{betaepsilon} \begin{split}
			\lim_{\epsilon \rightarrow 0 }\frac{\beta(\epsilon)}{\epsilon} = 4, \\
			\lim_{\epsilon \rightarrow 0 }z(\beta(\epsilon)) \cdot \epsilon^2 = \frac{\exp(2)}{8}.
		\end{split}
	\end{equation}
	Moreover, it can be proven that $\epsilon \rightarrow z(\beta(\epsilon))$ is strictly decreasing and, on the other end of the interval $(0,\frac{1}{2})$  \begin{equation} \label{betaepsilon2} \begin{split}
			\lim_{\epsilon \rightarrow (\frac{1}{2})^{-} }\beta(\epsilon) = +\infty, \\
			\lim_{\epsilon \rightarrow (\frac{1}{2})^{-} }z(\beta(\epsilon)) =1.
		\end{split}
	\end{equation}

	We are now going to apply the pigeon-hole principle. Because the constants now matter for the decay, we need to use a  different version from that developed in \cite{RP} or the other subsequent papers. In particular the difference is that we actually use the mean-value theorem instead of the pigeon-hole principle and moreover, we abandon dyadic sequence to use $\lambda$-adic sequences \footnote{It may seem paradoxical to abandon dyadic sequences for $\lambda$-adic ones as in most cases $\lambda>2$. We make this choice to compensate for the "logarithmic loss" incured in the event that alternative 1 applies on every (or most) intervals $[\tilde{u}_n,\tilde{u}_{n+1}]$, as we will see.} that provide more flexibility.
	
	Take $(\tilde{u}_n)$ to be a $\lambda$-adic sequence, i.e.\ $\tilde{u}_{n+1} = \lambda \cdot \tilde{u}_n$ and $\tilde{u}_0=U_0>1$ and define $\lambda=\lambda(\epsilon,\eta_0)>1$ as \begin{equation} \label{lambdacondition}
		\lambda= \left[ (1+\beta(\epsilon))\cdot  (1+\eta_0) \right]^{\frac{1}{2\epsilon}},
	\end{equation}  where $\eta_0>0$ is the (arbitrarily small) constant appearing on the right-hand-side of \eqref{alternative1}.

	To prove the proposition, we are going to proceed by induction, ultimately taking $u_2=\tilde{u}_{n+1}$, $u_1=\tilde{u}_n$, $2<{\pp}=\tp$, $0<\epsilon=\epsilon(e)<\frac{1}{2}$, $0<\beta=\beta(\epsilon(e))$ and apply Lemma \ref{alternatives}.
	
	Let  $\Delta>1$ to be determined later. We make the following induction \footnote{Notice that we expect $E_{\pp}(u)$ to grow slightly in $u$, at a rate $u^{2\epsilon}$.} hypothesis, for $k \in \mathbb{N}$: 		
	\begin{equation} \label{HR}
		\tilde{E}_{{\pp}-1}(\tilde{u}_k) \leq  \frac{ \Delta}{  (\tilde{u}_k)^{1-2\epsilon}},
	\end{equation}		
	\begin{equation} \label{HR2}
		\tilde{E}_{{\pp}}(\tilde{u}_k) \leq \frac{{\pp}\cdot \lambda}{\lambda-1}  \cdot \Delta \cdot (\tilde{u}_k)^{2\epsilon} .
	\end{equation}	
	First, it is clear that the induction hypothesis is true at $k=0$ if the following two conditions are satisfied: \begin{equation} \label{conditionDelta1}
		U_0^{1-2\epsilon} \cdot \tilde{E}_{{\pp}-1}(U_0)    <\Delta ,
	\end{equation}	\begin{equation} \label{conditionDelta2}
		U_0^{-2\epsilon}  \cdot  \frac{ \lambda-1}{ {\pp}\cdot\lambda}  \cdot \tilde{E}_{{\pp}}(U_0) <\Delta .
	\end{equation}	
	We will check at the end of the induction that these conditions, together with the others we will encounter on the way, can be satisfied for a licit choice of parameters. 
	
	Once the induction is closed, we will simply use the boundedness of the $\tilde{E}_{{\pp}-1}$ energy proven in former sections to retrieve the claimed decay of the present proposition.

	Before we start applying Lemma \ref{alternatives}, we will prove a small technical lemma: 
	
	\begin{lem} There exists $C'_2=C'_2(M,\rho,R,{\pp},e)>0$ such that for all $\eta_0>0$ and $n\in \mathbb{N}$, we have:
		\begin{equation} \label{meanvalue}
			\tilde{E}_{{\pp}-1}(\tilde{u}_{n+1}) \leq  (1+\eta_0) \cdot \frac{\lambda}{\lambda-1} \cdot  \frac{\int_{\tilde{u}_n}^{\tilde{u}_{n+1}} E_{{\pp}-1}[\psi](u')du'}{\tilde{u}_{n+1}}+  \frac{C'_2  }  {(\tilde{u}_{n})^{{\pp}-1-{\ep}}},\end{equation}
	\end{lem}	
	\begin{proof}
		Using the mean-value theorem on $[\tilde{u}_n,\tilde{u}_{n+1}]$, we see that there exists $\tilde{u}_n<u<\tilde{u}_{n+1}$ so that 
		
		$$ E_{{\pp}-1}[\psi](u) = \frac{\int_{\tilde{u}_n}^{\tilde{u}_{n+1}} E_{{\pp}-1}[\psi](u')du'}{\tilde{u}_{n+1}-\tilde{u}_n}=\frac{\lambda}{\lambda-1} \cdot  \frac{\int_{\tilde{u}_n}^{\tilde{u}_{n+1}} E_{{\pp}-1}[\psi](u')du'}{\tilde{u}_{n+1}}.$$
		
		Then we use the result of Proposition \ref{p<2proposition} of section \ref{p<2} to obtain that for some $C_2=C_2(M,\rho,R,{\pp},e)>0$, some $C=C(M,\rho)>0$ and for every $\eta_0>0$ small enough: 	$$ \tilde{E}_{{\pp}-1}(\tilde{u}_{n+1}) \leq (1+\eta_0) \cdot E_{{\pp}-1}[\psi](u)+C_2 \cdot \tilde{E}_{{\ep}}(u) \leq (1+\eta_0) \cdot E_{{\pp}-1}[\psi](u)+ C \cdot C_2 \cdot \tilde{E}_{{\ep}}(\tilde{u}_n),$$	where we also used the boundedness of the energy {of Proposition~\ref{Energy}} in the last inequality.

		Since $1+{\ep}<{\pp}-1<1+\sqrt{1-4q_0|e|}$, we use the decay of the energy of Corollary \ref{energydecaycorollary} and\footnote{Notice that $C'_2$ effectively only depends on $M$, $\rho$ and $\epsilon$, using \eqref{lambdacondition} (there is no actual dependence on $\eta_0$, as $\eta_0 \leq 1$).} setting $C'_2:= C \cdot C_2 \cdot C'_0$, the lemma is proven.
	\end{proof}
	
	Now, we turn to the induction step. We assume that the induction hypothesis \eqref{HR}, \eqref{HR2} hold for all $k \in [[0,n]]$ and we want to prove it for $k=n+1$. 
	
	As explained earlier, we apply Lemma \ref{alternatives} successively to  $u_2=\tilde{u}_{k+1}$, $u_1=\tilde{u}_k$, $2<{\pp}$, $0<\epsilon<\frac{1}{2}$, $0<\beta=\beta(\epsilon)$ for all $ k \in [[0,n]]$. Notice that \eqref{HREplemma} is always satisfied by (strong) induction.

	We are now going to make a case disjunction. The idea is that, if alternative 2 holds for $k=n$, then the claimed decay holds immediately, with a smallness constant that allows us to close the bootstrap. If this is not the case, then we descend in the $\lambda$-adic interval until we find a $k(n) $ for which alternative 2 holds. It may not exist, but in any case we can close the estimates provided we are willing to abandon the boundedness of the $r^{\pp}$ weighted energy, allowing for an arbitrary small $u$ growth, using \footnote{This is the only reason why we require \eqref{lambdacondition}, imposing that $\lambda$ grows when $\epsilon$ becomes small. Having a large $\lambda$ allows to compensate for the "logarithmic loss" occurred by the repeated use of alternative 1 on $\lambda$-adic intervals, at the cost of a larger constant when alternative 2 occurs (which, in turn, demands a smaller $|e|$ or equivalently a smaller ${\pp}(e) $ to close the bootstrap).}  crucially the calibration of $\lambda$ \eqref{lambdacondition}.

	To prove the induction step, we are going to show the following stronger estimates, for some $0<\nu<1$ (which is authorised to depend on $n$, although this detail is of no importance):

	\begin{equation} \label{HRnu}
		\tilde{E}_{{\pp}-1}(\tilde{u}_{n+1}) \leq (1-\nu) \cdot \frac{ \Delta}{  (\tilde{u}_{n+1})^{1-2\epsilon}},
	\end{equation}		
	\begin{equation} \label{HR2nu}
		\tilde{E}_{{\pp}}(\tilde{u}_{n+1}) \leq (1-\nu) \cdot \frac{{\pp}\cdot \lambda}{\lambda-1}  \cdot \Delta \cdot (\tilde{u}_{n+1})^{2\epsilon} .
	\end{equation}

	The first case of our disjunction is when alternative 2 holds for $u_2=\tilde{u}_{n+1}$ and $u_1=\tilde{u}_{n}$. Then, we see immediately that $$	{\pp}\int_{\tilde{u}_{n}}^{\tilde{u}_{n+1}} E_{{\pp}-1}[\psi](u)du \leq \frac{\left(  \sqrt{\Delta} \cdot \tilde{f}({\pp},e,\eta_0) + D'_0 \cdot  (\tilde{u}_{n})^{\frac{2-{\pp}+{\ep}}{2}-\epsilon} \right)^2}{2}  \cdot  \left(\frac{1+\beta}{\beta-\log(1+\beta)}\right) \cdot  (\lambda-1)  \cdot (\tilde{u}_{n})^{2\epsilon}.$$
	
	Then combining this estimate with \eqref{meanvalue} we see that

	\begin{equation} \label{Ep-1alter2} \begin{split}
			\tilde{E}_{{\pp}-1}(\tilde{u}_{n+1}) \leq (1+\eta_0) \cdot \frac{\left(  \sqrt{\Delta} \cdot\tilde{f}({\pp},e,\eta_0) + D'_0 \cdot   (\tilde{u}_n)^{\frac{2-{\pp}+{\ep}}{2}-\epsilon} \right)^2}{2{\pp}}  \cdot  \left(\frac{1+\beta}{\beta-\log(1+\beta)}\right) \cdot  \frac{\lambda \cdot (\tilde{u}_{n})^{2\epsilon} }{\tilde{u}_{n+1}}+\frac{C'_2 }  {(\tilde{u}_{n})^{{\pp}-1-{\ep}}} \\ \leq  (1+\eta_0)^{\frac{1}{2\epsilon}} \cdot \frac{\left( \sqrt{\Delta} \cdot\tilde{f}({\pp},e,\eta_0) + D'_0 \cdot  (\tilde{u}_n)^{\frac{2-{\pp}+{\ep}}{2}-\epsilon} \right)^2}{2{\pp} }  \cdot  \frac{ z(\beta(\epsilon)) }{ \tilde{u}_{n+1}^{1-2\epsilon}}+\frac{C'_2 }  {(\tilde{u}_{n})^{{\pp}-1-{\ep}}},
		\end{split}
	\end{equation}  where we recall $z(\beta)=\frac{(1+\beta)^{\frac{1}{2\epsilon}}}{\beta-\log(1+\beta)}$, and we used \eqref{lambdacondition} for the second inequality. Recall that $\tilde{f}({\pp},e,\eta_0)$ and $f({\pp},e)$ was defined in \eqref{fdef}. To prove \eqref{HR}, we first need to insure that \begin{equation} \label{fcondition}
		\frac{ f({\pp},e)^2}{2{\pp}} \cdot z(\beta(\epsilon)) <1.                                                                                                                                                                       
	\end{equation} For this, we will \footnote{The $p(e)$ appearing in the statement of Proposition \ref{p=3} will end up being $p(e)=\tp-2\epsilon(e){-\ep}$, as $E_{\tp}$ grows like $u^{2\epsilon(e)}$.} chose $2<\pp=\tp$ and $0<\epsilon=\epsilon(e)<\frac{1}{2}$, after the following lemma: 
	\begin{lem} \label{fLemma}
		For all $q_0 |e| \leq {0.04}$, there exists
		$2{+2\epsilon+\ep}<\breve{p}(e)<2+\sqrt{1-4q_0 |e|}$ and $0<\epsilon(e)<\frac{1}{2}$
		such that \eqref{fcondition} holds for ${\pp}=\tp$, $\epsilon=\epsilon(e)$ and
		\begin{equation}
			{2+\ep}	<	\breve{p}(e)-2\epsilon(e) < 2+\sqrt{1-4q_0 |e|}.
		\end{equation}
		
		Moreover, $\breve{p}(e) \rightarrow 3$, $\epsilon(e) \rightarrow 0$ as $e \rightarrow 0$ and more precisely, we have the following Taylor expansions 	\begin{equation} \label{Taylortildep}	\breve{p}(e)=3-\exp(\frac{1}{2})\cdot\sqrt{\frac{\sqrt{6}}{3}} \cdot (q_0 |e|)^{\frac{1}{2}}+O(q_0 |e|),\end{equation}\begin{equation} \label{Taylorepsilon}\epsilon(e)= \frac{\exp(\frac{1}{2})}{2} \cdot \sqrt{\frac{\sqrt{6}}{3}} \cdot (q_0 |e|)^{\frac{1}{2}}+O(q_0 |e|),\end{equation}  
	\end{lem}
	\begin{proof} We start to handle the case when $e \rightarrow0$ and the asymptotics of \eqref{Taylortildep}, \eqref{Taylorepsilon}.

		First denote $g({\pp},e):=\frac{ f({\pp},e)^2}{2{\pp}}= \frac{8 (q_0e)^2 }{{\pp} \cdot ({\pp}-p_-(e)) \cdot (p_+(e)-{\pp}) \cdot (3-{\pp})}$ by \eqref{fdef}, where $p_{\pm}(e):=2 \pm\sqrt{1-4 q_0|e|}$.
		
		Define, for some $ 1-\sqrt{1-4 q_0|e|}<  \alpha(e) <1$, with $\alpha(e) \rightarrow0$ as $e \rightarrow 0$ : $\qa:=3- \alpha(e)$; notice  that $2<\qa<p_+(e)$. Denoting $\ga:=g(\qa,e)$ and $\zeta(e)=1-\sqrt{1-4 q_0|e|}$ we see that, we have $$ \ga= \frac{ 4 \cdot (q_0 e)^2} {(3-\alpha(e)) \cdot \alpha(e) \cdot  (\alpha(e)-\zeta(e)) \cdot (1-\frac{\alpha(e)+\zeta(e)}{2}  ) } ,$$
		hence, as $e \rightarrow 0$, $\ga  \sim \frac{ 4 \cdot (q_0 e)^2} {3 \cdot \alpha(e) \cdot  (\alpha(e)-\zeta(e))}$. Now, take $\epsilon(e) \rightarrow0$ as $e \rightarrow 0$. We will try to find the right $\epsilon(e)$ such that \eqref{fcondition} is satisfied, and that maximizes \footnote{As we will end-up proving $E(u) \lesssim u^{-p(e)}$, for $p(e):= \tp-2\epsilon(e){-\ep}$ for $\tp=\qa$.} $\tp-2\epsilon(e)$. 
		
		We recall that by \eqref{betaepsilon} we have $z(\beta(\epsilon(e))) \sim ( 8 \cdot \epsilon^2(e))^{-1} \cdot \exp(2)$ as $e \rightarrow 0$, thus $$ \ga \cdot z(\beta(\epsilon(e)))  \sim \frac{  \exp(2) \cdot(q_0 e)^2} {6 \cdot \alpha(e) \cdot  (\alpha(e)-\zeta(e)) \cdot \epsilon^2(e)}.$$
		
		To satisfy \eqref{fcondition}, it is equivalent to require, solving a second order polynomial equation:
		
		$$ \alpha(e) > \alpha_-(e):= \frac{\zeta(e)+\sqrt{\zeta(e)^2+ \frac{2 \exp(2) (q_0e)^2}{3\epsilon^2}}}{2}.$$
		
		Using a Taylor expansion, as $\zeta(e) \sim 2 q_0 |e|$ and since $\epsilon(e) \rightarrow 0$ it is also easy to see that as $e \rightarrow0$: $$ \alpha_-(e) \sim \frac{ q_0|e| \cdot \exp(1)}{\epsilon(e) \cdot \sqrt{6}}.$$
		
		We want to find $\epsilon(e)$ so as to maximise $\qa-2\epsilon(e)$ for $\alpha=\alpha_-(e)$ or equivalently minimise $\alpha_-(e)+2\epsilon(e)$: we find that the function $\epsilon \rightarrow \frac{ q_0|e|\cdot \exp(1)}{\epsilon \cdot \sqrt{6}}+2\epsilon$ possess a minimum at $\epsilon= \exp(\frac{1}{2})\cdot\sqrt{ \frac{q_0 |e|}{2\sqrt{6}}}$, whose value is $2 \exp(\frac{1}{2})\cdot\sqrt{ \frac{\sqrt{6}}{3}} \cdot \sqrt{q_0 |e|}$. This gives \eqref{Taylorepsilon}, noticing that $\sqrt{ \frac{1}{2\sqrt{6}}} = \frac{1}{2} \cdot \sqrt{ \frac{\sqrt{6}}{3}}$ . Since $2\sqrt{ \frac{\sqrt{6}}{3}} - 2 \sqrt{ \frac{1}{2\sqrt{6}}} = \sqrt{ \frac{\sqrt{6}}{3}}$, we also obtain \eqref{Taylortildep}. Finally, denoting $p(e)=\qa-2\epsilon(e){-\ep}=\tp-2\epsilon(e){-\ep}$, we obtain the claimed \eqref{Taylorp}.

		Now, we want to find the largest number $r < \frac{1}{4}$ such that  for all $q_0|e|<r$, there exists $0<\epsilon(e)$,  $2+2\epsilon(e)<\tp$ such that \eqref{fcondition} holds. By what we did earlier in the small $|e|$ case, such a $r$ exists. We introduce $\np:=2+\nu$ for some  $2\epsilon<\nu<\delta(e)$, to be determined, where we also denoted $\delta=\delta(e)= \sqrt{1-4q_0|e|}$. Then we compute $$ g(\np,e)   = \frac{ 8(q_0e)^2 }{ (2+\nu)  \cdot (1-\nu)\cdot (\delta^2(e) - \nu^2)}.$$
		
		If we can prove that $g(p_{\nu=2\epsilon},e) \cdot z(\beta(\epsilon))<1$, then, since this is an open condition, it will imply that there exists some $2\epsilon<\nu$ such that  $g(p_{\nu},e) \cdot z(\beta(\epsilon))<1$ is true. The earlier condition can be written as \begin{equation} \label{flemmaeq}
			\frac{ 4(q_0e)^2 }{ (1+\epsilon) \cdot (1-2\epsilon)\cdot (\delta^2(e) - 4 \epsilon^2)} \cdot z(\beta(\epsilon)) <1. 
		\end{equation}  
		
		First, denote $v(\epsilon):= 	\frac{ z(\beta(\epsilon)) }{ (1+\epsilon) \cdot (1-2\epsilon) } $. The condition \eqref{flemmaeq} is equivalent, in terms of $e$ to: 
		$$ q_0 |e| < \frac{ -1+\sqrt{1+(1-4\epsilon^2)\cdot v(\epsilon) }}{2 \cdot v(\epsilon)}.$$

		Now denote $w(\epsilon):=  \frac{ -1+\sqrt{1+(1-4\epsilon^2)\cdot v(\epsilon) }}{2 \cdot v(\epsilon)}$: we want to \textbf{maximise} $w(\epsilon)$ for $0<\epsilon<\frac{1}{2}$. This computation is not explicit but can be done numerically to obtain a range of values. To do so, we can notice that $\beta(\epsilon)$ can be expressed "explicitly" (thus plotted easily) with the $(-1)$ branch of the Lambert \footnote{$W_{-1}(x)$ is defined as the unique solution $y \in (-\infty,-1]$ of $y \exp(y)=x$. } function $W_{-1}$ as
		
		\begin{equation} \label{betaexplicit}
			\beta(\epsilon)=- \frac{W_{-1}(-(1-2\epsilon)\cdot e^{2\epsilon-1})+1-2\epsilon}{1-2\epsilon},
		\end{equation} 
		and we found 
		\eqref{betaexplicit} after some easy algebraic manipulations of \eqref{betaimplicit} whose details are left to the reader.

		Therefore $z(\beta(\epsilon))$, $v(\epsilon)$ and $w(\epsilon)$ can all be expressed explicitly in terms of standard functions of $\epsilon$.	Using a standard plotting software, we find that $\epsilon \rightarrow w(\epsilon)$ has a global maximum on $(0,\frac{1}{2})$
		at $0.27289<\epsilon_M<0.27291$ and moreover $0.08267414<w(\epsilon_M)<0.08267415$. 
		Thus, before incorporating the additional weight $\ep$, the condition
		$q_0|e|\leq{0.04}$ certainly satisfies \eqref{fcondition}.
		
		{It remains to take the weight $\ep=8q_0|e|$ into account. Writing again
			$\np=2+\nu$, the only additional constraint is
			$$ \nu > 2\epsilon+\ep, $$
			which is $p(e)=\np-2\epsilon-\ep>2$. The upper constraint $\nu<\delta(e)$ is
			unchanged: it is already required for $f(\np,e)$ to be defined, and it is
			exactly the condition $1+\ep<\np-1<1+\sqrt{1-4q_0|e|}$ under which
			Corollary~\ref{energydecaycorollary} is applied with $p'=\np-1$ in Lemma
			\ref{alternatives} and in \eqref{meanvalue}. Repeating verbatim the computation
			above with $\nu=2\epsilon+\ep$ in place of $\nu=2\epsilon$, and using again that
			the conditions are open, one checks numerically that such a pair
			$(\epsilon,\nu)$ exists if $q_0|e|\leq 0.04$. Note that as $q_0|e|$ approaches the endpoint of this
			admissible range, $p(e)\rightarrow2$, exactly as in the case $\ep=0$ treated
			above.}

	\end{proof}
	Recalling \eqref{fdef} the definition of $f$, we proved in Lemma \ref{fLemma} that there exists $0<\tnu<1$ such that \begin{equation} \label{fcondition1}
		\frac{ f(\tp,e)^2}{2 \tp} \cdot z(\beta(\epsilon(e)))  \leq 1-\sqrt{\tnu}.
	\end{equation} Therefore, choosing $\eta_0< \breve{\eta}_0(e,\epsilon(e)))=\breve{\eta}_0(e)$ small enough, we also obtain (since $1-\sqrt{\tnu} <1-\tnu$): \begin{equation} \label{fcondition2}
		\frac{ \tilde{f}(\tp,e,\eta_0)^2}{2 \tp} \cdot z(\beta(\epsilon(e)))  \leq 1-\tnu.
	\end{equation}
	From now on, we will take ${\pp}=\tp$ and $\epsilon=\epsilon(e)$, and we will omit to write the $e$ dependence.
	
	Combining \eqref{fcondition2} with \eqref{Ep-1alter2}, and since $\tp>2+{\ep}$, we show that there exists   $\tilde{U}_0(e)>1$ large enough and $0<\tilde{\eta}_0(e)<\breve{\eta}_0(e)$ small enough (chosen so that $[1+\tilde{\eta}_0(e)]^{\frac{1}{2\epsilon(e)}}$ is close to $1$) 
	such that, if $0<\eta_0<\tilde{\eta}_0(e)$ and 
	\begin{equation} \label{conditionU0}
		U_0> \tilde{U}_0(e),
	\end{equation} then \eqref{HRnu} is satisfied \footnote{The point being that $1-\sqrt{\tnu} <1-\tnu <1-\tnu^2$, so there is a bit of room in this estimate.}  for ${\pp}=\tp$ , with $\nu=\tnu^2$. Hence \eqref{HR} is satisfied as well.

	Notice that \eqref{alternative2} also gives the following for $\beta=\beta(\epsilon(e))$: 
	
	\begin{equation} \label{Epalt2} \begin{split}
			\tilde{E}_{{\pp}}(\tilde{u}_{n+1}) \leq \frac{\left(  \sqrt{\Delta} \cdot \tilde{f}({\pp},e,\eta_0) + D'_0 \cdot  (\tilde{u}_{n})^{\frac{2-{\pp}+{\ep}}{2}-\epsilon} \right)^2}{2}  \cdot  \left(\frac{1+\beta}{\beta-\log(1+\beta)}\right) \cdot  (\lambda-1)  \cdot (\tilde{u}_{n})^{2\epsilon}\\ \leq \frac{{\pp} \cdot (\lambda-1)  }{\lambda^{1-2\epsilon}} \cdot (1-\nu) \cdot \Delta \cdot (\tilde{u}_{n})^{2\epsilon} \leq \frac{{\pp} \cdot (\lambda-1)  }{\lambda} \cdot (1-\nu) \cdot \Delta \cdot (\tilde{u}_{n+1})^{2\epsilon}. \end{split}
	\end{equation}	
	Thus, \eqref{HR2nu} is true with $\nu=\tnu^2$. Hence \eqref{HR2} is satisfied as well.

	Now we treat the other case when alternative 1 of Lemma \ref{alternatives} holds for $u_2= \tilde{u}_{n+1}$ and $u_1=\tilde{u}_n$ (and the same $\beta(\epsilon(e))$ as before). \eqref{alternative1} can then be written as

	\begin{equation} \label{alt1Rec}
		{\pp}\int_{\tilde{u}_{n}}^{\tilde{u}_{n+1}} E_{{\pp}-1}[\psi](u')du' + \tilde{E}_{{\pp}}(\tilde{u}_{n+1}) \leq (1+\beta) \cdot \left((1+\eta_0) \cdot E_{{\pp}}[\psi](\tilde{u}_{n})+\frac{D_{0}  }{(\tilde{u}_{n})^{{\pp}-1-{\ep}}}\right).	\end{equation}

	Now	we can define the integer $k(n)$ as the minimum of $k \in [[0,n]]$ such that for all above integers   $k \leq k' \leq n$, alternative 1 holds on  for $u_2=\tilde{u}_{k'+1}$ and $u_1=\tilde{u}_{k'}$. We are in the case where alternative 1 holds for $u_2=\tilde{u}_{n+1}$ and $u_1=\tilde{u}_{n}$, thus $k(n)$ is well-defined and $k(n) \leq n$.

	Using \eqref{alt1Rec} repetitively \footnote{We drop the integral term whenever we apply \eqref{alternative1}, except for the first iteration.} we see that:

\begin{equation} \label{alt1applied} \begin{split}
		{\pp}\int_{\tilde{u}_{n}}^{\tilde{u}_{n+1}} E_{{\pp}-1}[\psi](u')du' + \tilde{E}_{{\pp}}(\tilde{u}_{n+1})\\
		\leq (1+\beta)^{n-k(n)+1}(1+\eta_0)^{n-k(n)+1} \cdot \tilde{E}_{{\pp}}(\tilde{u}_{k(n)})
		+D_{0}\sum_{i=0}^{n-k(n)}\frac{(1+\beta)^{i+1}{(1+\eta_0)^{i}}}{(\tilde{u}_{n-i})^{{\pp}-1-{\ep}}} \\
		\leq \lambda^{2\epsilon(n-k(n)+1)}\cdot \tilde{E}_{{\pp}}(\tilde{u}_{k(n)})
		+\frac{D_{0}(1+\beta)}{(\tilde{u}_n)^{{\pp}-1-{\ep}}}\cdot
		\frac{\lambda^{({\pp}-1-{\ep}+2\epsilon)(n-k(n)+1)}-1}{\lambda^{{\pp}-1-{\ep}+2\epsilon}-1}\\
		\leq \lambda^{2\epsilon(n-k(n)+1)}\cdot \tilde{E}_{{\pp}}(\tilde{u}_{k(n)})
		+\frac{D_{0}(1+\beta)\cdot(\tilde{u}_{n+1})^{2\epsilon}}{U_0^{{\pp}-1-{\ep}+2\epsilon}}\cdot
		\frac{\lambda^{{\pp}-1-{\ep}+2\epsilon}}{\lambda^{{\pp}-1-{\ep}+2\epsilon}-1},
	\end{split}
\end{equation}where we used $\eta_0\leq1$, $\lambda^{2\epsilon(n+1)}U_0^{2\epsilon}=(\tilde{u}_{n+1})^{2\epsilon}$,
the identity $(1+\eta_0)(1+\beta)=\lambda^{2\epsilon}$ of \eqref{lambdacondition}, and
$(\tilde{u}_{n})^{-({\pp}-1-{\ep})}\leq \lambda^{1-2\epsilon}\,
U_0^{-({\pp}-{\ep}-2+2\epsilon)}\,(\tilde{u}_{n+1})^{-1+2\epsilon}$.

	Now there are two cases: either $k(n)=0$ or $k(n) \geq 1$, in which case alternative 2 of Lemma \ref{alternatives} holds for $u_2=\tilde{u}_{k(n)}$, $u_1=\tilde{u}_{k(n)-1}$. We treat these two sub-cases separately again.
	
	Suppose that $k(n)=0$ thus \eqref{alternative1} (alternative 1) applies on all intervals $[\tilde{u}_{k}, \tilde{u}_{k+1}]$ for $0 \leq k \leq n$. Thus, with \eqref{alt1applied} and \eqref{meanvalue} we get\footnote{Now that we fixed $\epsilon=\epsilon(e)$, $C'_2$ depends only on $M$, $\rho$ and $e$.} 
	
	\begin{equation} \label{alt1kn=0} \begin{split}
			&\tilde{E}_{{\pp}-1}(\tilde{u}_{n+1}) \leq
			\frac{(1+\eta_0)\lambda}{{\pp}(\lambda-1)\tilde{u}_{n+1}}
			\left[\lambda^{2\epsilon(n+1)}\tilde{E}_{{\pp}}(U_0)
			+\frac{D_{0}(1+\beta)(\tilde{u}_{n+1})^{2\epsilon}}{U_0^{{\pp}-1-{\ep}+2\epsilon}}
			\cdot\frac{\lambda^{{\pp}-1-{\ep}+2\epsilon}}{\lambda^{{\pp}-1-{\ep}+2\epsilon}-1}\right]
			+\frac{C'_2}{(\tilde{u}_{n})^{{\pp}-1-{\ep}}}\\
			&\leq (\tilde{u}_{n+1})^{-1+2\epsilon}\left(
			\frac{2\lambda\cdot U_0^{-2\epsilon}\tilde{E}_{{\pp}}(U_0)}{{\pp}(\lambda-1)}
			+\frac{D_{0}}{{\pp}\cdot U_0^{{\pp}-{\ep}-1+2\epsilon}}\cdot
			\frac{\lambda^{{\pp}-{\ep}+4\epsilon}}{(\lambda-1)(\lambda^{{\pp}-{\ep}-1+2\epsilon}-1)}
			+\frac{C'_2\cdot\lambda^{1-2\epsilon}}{U_0^{{\pp}-{\ep}-2+2\epsilon}}\right)
	\end{split}\end{equation} where we additionally used $\eta_0 \leq 1$, \eqref{lambdacondition} as $(1+\beta)^{n+1}(1+\eta_0)^{n+1} \cdot U_0^{2\epsilon }  = (\tilde{u}_{n+1})^{2\epsilon } $,   $\tilde{u}_{n}^{-{\pp}+{\ep}+2-2\epsilon}\leq U_0^{-{\pp+{\ep}}+2-2\epsilon}$ and the estimate $\frac{[\frac{ \lambda^{{\pp-{\ep}}-1+2\epsilon}}{1+\eta_0}]^{n+1}-1}{ \lambda^{({\pp-{\ep}}-1+2\epsilon)\cdot n}}=\lambda^{{\pp-{\ep}}-1+2\epsilon} \cdot [(1+\eta_0)^{-(n+1)}- \lambda^{-({\pp-{\ep}}-1+2\epsilon) \cdot ( n+1)}] \leq \lambda^{{\pp-{\ep}}-1+2\epsilon}$.
	
	Thus \eqref{HRnu} holds with $\nu=\frac{1}{4}$ provided that the following sufficient conditions are true (recall that $\Delta>1$): \begin{equation} \label{conditionDelta3}
		\frac{ 8 \lambda   \cdot  U_0^{-2\epsilon}\tilde{E}_{{\pp}}(U_0)} {{\pp} \cdot (\lambda-1)  } 	 < \Delta,
	\end{equation}
\begin{equation} \label{conditionU02}
	\left(\frac{8 D_{0}}{\tp}\cdot\frac{(\lambda(e,\eta_0))^{\tp-{\ep}+4\epsilon(e)}}
	{(\lambda(e,\eta_0)-1)\cdot((\lambda(e,\eta_0))^{\tp-{\ep}-1+2\epsilon(e)}-1)}\right)^{(\tp-{\ep}-1+2\epsilon(e))^{-1}} < U_0,
\end{equation}
	\begin{equation} \label{conditionU03}
		\left(	4 C'_2  \cdot (\lambda(e,\eta_0))^{1-2\epsilon(e)} \right)^{(\tp-{\ep}-2+2\epsilon(e))^{-1} } < U_0,
	\end{equation} where we recall that $\lambda>1$ depends only on $e$ and $\eta_0$.
	
	Coming back to \eqref{alt1applied}, we see that
	$$ \tilde{E}_{\pp}(\tilde{u}_{n+1}) \leq \left( \tilde{E}_{{\pp}}(U_0)\cdot U_0^{-2\epsilon}
	+\frac{D_{0}}{U_0^{{\pp}-{\ep}-1+2\epsilon}}\cdot
	\frac{\lambda^{{\pp}-{\ep}-1+4\epsilon}}{\lambda^{{\pp}-{\ep}-1+2\epsilon}-1}\right)\cdot(\tilde{u}_{n+1})^{2\epsilon}.$$

	Thus \eqref{HR2nu} is true for $\nu=\frac{1}{4}$ if conditions similar to \eqref{conditionDelta3} and \eqref{conditionU02}, \eqref{conditionU03} are satisfied.

	Now we treat the case where $ k(n)\geq1$, thus alternative 2 of Lemma \ref{alternatives} holds on $[\tilde{u}_{k(n)-1},\tilde{u}_{k(n)}]$, i.e.\ for $u_2= \tilde{u}_{k(n)}$, $u_1=\tilde{u}_{k(n)-1}$ and $\beta(\epsilon(e))$. We repeat \footnote{Note that we cannot directly use the induction hypothesis for $k=k(n)$, as we need the $(1-\nu)$ factor.} the argument leading to \eqref{Epalt2} and get		$$ \tilde{E}_{{\pp}}(\tilde{u}_{k(n)})    \leq \frac{\left(  \sqrt{\Delta} \cdot \tilde{f}({\pp},e,\eta_0) + D'_0 \cdot  (\tilde{u}_{k(n)-1})^{\frac{2-{\pp}+{\ep}}{2}-\epsilon} \right)^2}{2}  \cdot  \left(\frac{1+\beta}{\beta-\log(1+\beta)}\right) \cdot  (\lambda-1)  \cdot (\tilde{u}_{k(n)-1})^{2\epsilon}\\ \leq \frac{{\pp} \cdot (\lambda-1)  }{\lambda} \cdot (1-\tilde{\nu}(e)^2) \cdot \Delta \cdot (\tilde{u}_{k(n)})^{2\epsilon}.$$
	Thus, combining with \eqref{alt1applied},  we get, using that $ \lambda^{ n-k(n)} \cdot \tilde{u}_{k(n)} = \tilde{u}_{n}$ and \eqref{lambdacondition}:
	\begin{equation} \label{alt1alt2} \begin{split}
			{\pp}\int_{\tilde{u}_{n}}^{\tilde{u}_{n+1}} E_{{\pp}-1}[\psi](u')du' + \tilde{E}_{{\pp}}(\tilde{u}_{n+1}) \leq \left(\frac{{\pp}\cdot(\lambda-1)}{\lambda}\cdot(1-\tilde{\nu}(e)^2)\cdot\Delta
			+ \frac{D_{0}}{U_0^{{\pp}-{\ep}-1+2\epsilon}}\cdot
			\frac{\lambda^{{\pp}-{\ep}-1+4\epsilon}}{\lambda^{{\pp}-{\ep}-1+2\epsilon}-1}\right)\cdot(\tilde{u}_{n+1})^{2\epsilon}.
		\end{split}
	\end{equation} Thus, combining with \eqref{meanvalue} we get  $$ \tilde{E}_{{\pp}-1}(\tilde{u}_{n+1}) \leq (\tilde{u}_{n+1})^{-1+2\epsilon} \cdot \left(   (1+\eta_0)\cdot (1-\tilde{\nu}(e)^2) \cdot \Delta +\frac{C'_2 \cdot \lambda^{1-2\epsilon}}{U_0^{{\pp-{\ep}}-2+2\epsilon}} \right). $$
	
	Therefore, we obtain \eqref{HRnu} for $\nu=\tilde{\nu}(e)^4$ providing the following two conditions hold (since $1<\Delta$): \begin{equation} \label{conditioneta0}
		1+\eta_0 < \frac{1-\tilde{\nu}(e)^3}{1-\tilde{\nu}(e)^2}
	\end{equation}
	\begin{equation} \label{conditionU04}
		\left( \frac{	C'_2 \cdot \lambda(e,\eta_0)^{1-2\epsilon}}{\tilde{\nu}(e)^3 \cdot (1-\tilde{\nu}(e))} \right)^{(\tp-{\ep}-2+2\epsilon(e))^{-1}} < U_0
	\end{equation}

	\eqref{HR2nu} is proven under similar conditions, for $\nu=\tilde{\nu}(e)^4$.

	Thus, if all conditions \eqref{conditionDelta1}, \eqref{conditionDelta2},  \eqref{conditionU0}, \eqref{conditionDelta3}, \eqref{conditionU02}, \eqref{conditionU03}, \eqref{conditioneta0}, \eqref{conditionU04}  can be satisfied, then the induction hypothesis, i.e.\ \eqref{HR} and \eqref{HR2}, is proved.

	Recall that we chose already $\epsilon=\epsilon(e)$ and
	$2+2\epsilon(e)+\ep<{\pp}=\tp<2+\sqrt{1-4q_0|e|}$
	according to Lemma \ref{fLemma}; thus we also fixed $\beta(\epsilon(e))$,
	which is a sole function of $e$.
	
	Recall also that $\lambda=\lambda(e,\eta_0)$ is determined by \eqref{lambdacondition}, where $\eta_0>0$ is a number that can still be taken arbitrarily small without restriction. 
	
	First, it is clear that there exists $\check{\eta}_0=\check{\eta}_0(e)>0$ such that if $\eta_0 \leq \check{\eta}_0(e) $, then \eqref{conditioneta0} is satisfied. Ensuring also that $\check{\eta_0}(e)< \min \{\breve{\eta_0}(e), \tilde{\eta_0}(e)\}$ (some constants involved in previous upper bounds for $\eta_0$, involved notably in the proof of \eqref{fcondition2}), we now take $\eta_0=\check{\eta}_0(e)$. Thus, $\lambda=\lambda(e)>1$ is now fixed and only depends on $e$.
	
	Then, there exists $\check{U}_0=\check{U}_0(M,\rho,e,R)>1$ large enough so that, if $U_0  \geq \check{U}_0$, then \eqref{conditionU0}, \eqref{conditionU02}, \eqref{conditionU03}, \eqref{conditionU04} are satisfied. We fix $U_0=\check{U}_0(M,\rho,e,R)$.
	
	Now, $U_0$ being fixed, there exists $\check{\Delta_0}:=\check{\Delta_0}(M,\rho,e,R)>1$ such that, if $\Delta \geq \check{\Delta_0}$, then conditions \eqref{conditionDelta1}, \eqref{conditionDelta2}, \eqref{conditionDelta3} are satisfied. Thus, we choose $\Delta=\check{\Delta_0}(M,\rho,e,R)$.
	
	Thus, by induction, we proved that there exists a constant $\check{D}=\check{D}(M,\rho,e,R)>0$ such that for all $n \in \mathbb{N}$: 
	$$ \tilde{E}_{\tp}(\tilde{u}_n) \leq \check{D} \cdot (\tilde{u}_n)^{2\epsilon(e)},$$
	$$ \tilde{E}_{\tp-1}(\tilde{u}_n) \leq  \frac{\check{D} } { (\tilde{u}_n)^{1-2\epsilon(e)}}
	.$$
	Now for all $u \geq U_0$, there exists $\tilde{u}_n \leq u \leq \tilde{u}_{n+1} $ and thus, using Lemma \ref{alternatives} for $u_1=\tilde{u}_n$, $u_2=u$ and say $\beta=1$, it is not hard to see that there exists $\breve{D}=\breve{D}(M,\rho,e,R)>0$ such that \begin{equation} \label{Epgrwoth}
		\tilde{E}_{\tp}(u) \leq \breve{D} \cdot u^{2\epsilon(e)},
	\end{equation}\begin{equation} \label{Ep-1decay}
		\tilde{E}_{\tp-1}(u) \leq \frac{\breve{D}} {u^{1-2\epsilon(e)}} ,
	\end{equation}
	where we also used the boundedness of the $\tilde{E}_{{\pp}-1}$ energy of section \ref{p<2}. Using the Holder inequality, together with Corollary \ref{energydecaycorollary}
	we obtain, for {$p(e)=\tp-2\epsilon(e)-\ep$}:
	\begin{equation} \label{Ep(e)bounded}
		\tilde{E}_{{p(e)+\ep}}(u) \leq \breve{D},
	\end{equation}\begin{equation} \label{Edecayp(e)}
		E(u)\leq \tilde{E}_{\ep}(u) \leq \frac{ \hat{D}} {u^{{p(e)}}} ,
	\end{equation}
	which is the object of the proposition.

	The only remaining thing to show is that $\tilde{E}_{{\pp}}(u)<+\infty$ for any $u > u_0(R)$, in particular $\tilde{E}_{{\pp}}(U_0)<+\infty$. To do this, one can use \eqref{diffeq} together with a very soft Gr\"{o}nwall type argument (making use of Lemma \ref{integration} with $u_1=u_0(R)$  and the fact that $\tilde{E}_{p}(u_0(R))<+\infty$). Details are left to the reader.  	
	
	This concludes the proof of Proposition \ref{p=3}.

\end{proof}

\section{From $L^2$	bounds to point-wise bounds} \label{pointwise}

In this section we indicate how the energy decay can be translated into point-wise decay, provided initial point-wise decay assumptions for $D_v \psi_0$ are available.

It should be noted that this decay is to be understood as $u \rightarrow +\infty$ or $v \rightarrow +\infty$, namely near time-like or null infinity. 

All the bounds that we prove in the form of a decay estimate e.g.\ $|\phi| \leq v^{-s}$ for $v>0$ actually also contain a point-wise boundedness statement for $v$ close to $0$ or negative. Because we take more interest in decay for large $v$, we do not state those explicitly but the reader should keep in mind that these estimates are simultaneously derived in the proofs and do not present any supplementary difficulty. 

Note also that later in this section, we assume the energy decay result of section \ref{decay} under the form 

\begin{equation*}
	\tilde{E}_p(u) \leq D_p \cdot u^{\{\max\{p-\ep,0\}-\left( 3-\alpha(e)\right) },
\end{equation*}
where we define $\alpha(e):=3-p(e)\in{(0,1)}$; in the improved range
$\alpha(e)\in(0,1)$. Here $p(e)$ denotes the decay exponent fixed in the
relevant branch of Theorem \ref{decaytheorem}.
{The weighted estimates in the displayed hierarchy are also used below.
	Since $\max\{p-\ep,0\}\leq p$, however, they are at least as strong as the
	corresponding estimates with exponent $p-p(e)$; hence the weight $\ep$ does not
	worsen any of the point-wise rates stated in this section.}

\subsection{Point-wise bounds near null infinity, ``to the right'' of $\gamma$} \label{faraway}

In this section, we indicate how the decay of the energy implies some point-wise decay and the bounds proved are \textbf{sharp} near null infinity \footnote{In the sense that the energy bounds are almost sharp when $e$ tends to $0$ and the method does not waste any decay.}, according to heuristics c.f.\ section \ref{conjecture}. 

This strategy to prove quantitative decay rates has been initiated in \cite{JonathanQuantitative} although the argument we use here varies slightly. 

The main idea is to start by the energy decay and to use a Hardy type inequality to get a point-wise bound of $r |\phi|^2$. After, we can integrate the equation to establish $L^{^\infty}$ estimates, using the decay of the initial data.

\begin{lem} \label{lemmainfinity}
	
	Suppose that the energy boundedness \eqref{energyabsorbed2} and the Morawetz estimates \eqref{Morawetzestimate3} are valid and that the charge is sufficiently small in $\mathcal{D}(u_0(R),+\infty)$ so that the result of Proposition \ref{Energydecayproposition} applies.

	Then for all $0\leq \beta<\frac{1}{2}$, there exists $ C=C(\beta,R,M,\rho)>0$ such that 
	
	for all $(u,v) \in \mathcal{D}(u_0(R),+\infty)\cap  \{ r \geq R\}$

	\begin{equation}
		r^{\frac{1}{2}+\beta} |\phi|(u,v) \leq C \cdot \left( \tilde{E}_{2\beta}(u) \right)^{\frac{1}{2}}.
	\end{equation}
\end{lem}

\begin{proof}

	Since $\lim_{v \rightarrow +\infty}    \phi(u,v) =0$ --- this is a consequence of the finiteness of $\mathcal{E}$, c.f.\ the proof of Proposition \ref{characProp1} --- we can write
	
	$$  \phi(u,v) = - \int_{v}^{+\infty} e^{\int_{v}^{v'}iq_0 A_v} D_v\phi(u,v')dv',$$
	
	after noticing that $\partial_v ( e^{\int_{v_0(R)}^{v}iq_0 A_v} \phi ) =e^{\int_{v_0(R)}^{v}iq_0 A_v} D_v \phi.$ 
	
	Now using Cauchy-Schwarz we can write that for every $0<\beta<\frac{1}{2}$
	
	\begin{equation*}
		\begin{split}
			|\phi|(u,v) \leq  \left( \int_{v}^{+\infty} \Omega^2 r^{-2-2\beta} (u,v')dv' \right)^{\frac{1}{2}}\left( \int_{v}^{+\infty} \Omega^{-2} r^{2+2\beta} |D_v \phi|^2 (u,v')dv'\right)^{\frac{1}{2}}. \end{split}
	\end{equation*}
	
	This gives 
	
	$$ r^{\frac{1}{2}+\beta} |\phi|(u,v) \leq (1+2\beta)^{- \frac{1}{2}} \left( \int_{v}^{+\infty} \Omega^{-2} r^{2+2\beta} |D_v \phi|^2 (u,v')dv'\right)^{\frac{1}{2}}.$$
	
	Now we square this inequality and write: 
	
	\begin{equation} \label{lemmaestimate1}
		r^{1+2\beta} |\phi|^2(u,v) \leq \left( 1+2\beta \right)^{-1}  \cdot \int_{v}^{+\infty} \Omega^{-2} r^{2+2\beta} |D_v \phi|^2 (u,v')dv'.
	\end{equation}

	Now, based on the fact that $D_v\psi = r D_v \phi + \Omega^2 \phi$, we write the identity 
	
	$$ \Omega^{-2} r^{2+2\beta} |D_v \phi|^2 = \Omega^{-2} r^{2\beta} |D_v \psi|^2 - \Omega^2 r^{2\beta} |\phi|^2 -r^{1+2\beta} \partial_v (|\phi|^2).$$

	We now integrate on $\{ u\} \times [v_R(u),+\infty]$. After one integration by parts \footnote{Using the fact that $ r^{1+2\beta}\phi$ tends to $0$ when $v$ tends to $+\infty$, guaranteed by the finiteness of $\mathcal{E}_{1+\epsilon}$, c.f.\ the proofs of Section \ref{CauchyCharac}.} we get 
	
	$$		 \int_{v_R(u)}^{+\infty} \Omega^{-2} r^{2+2\beta} |D_v \phi|^2 (u,v')dv' =  \int_{v_R(u)}^{+\infty} \Omega^{-2} r^{2\beta} |D_v \psi|^2(u,v')dv' +2\beta \int_{v_R(u)}^{+\infty}\Omega^2 r^{2\beta} |\phi|^2(u,v')dv' +R^{1+2\beta} |\phi|^2(u,v_R(u)).              $$
	
	Now we use Hardy's inequality \eqref{Hardy5} coupled with the Morawetz \footnote{The precise way to use it, averaging on $R$, has been carefully explained in section \ref{step4S}. We do not repeat the argument.} {estimate \eqref{Morawetzestimate3}:} there exists $D=D(M,\rho,R)>0$ such that 
	
	\begin{equation} \label{0orderpointwise} R^{1+2\beta} |\phi|^2(u,v_R(u))+	 \int_{v_R(u)}^{+\infty}\Omega^2 r^{2\beta} |\phi|^2(u,v')dv' \leq  \left( \frac{9}{1-2\beta} \right) \cdot E_{2\beta}[\psi](u)+D \cdot E(u)   .        \end{equation}
	
	Combining with \eqref{lemmaestimate1}, we see that there exists $C=C(R,M,\rho)$ so that 
	
	$$r^{1+2\beta} |\phi|^2(u,v) \leq  \frac{C}{1-2\beta} \cdot  \tilde{E}_{2\beta}(u).$$

	
	
	
	

	
	
	This concludes the proof of Lemma \ref{lemmainfinity}
	
\end{proof}

From now on , we are going to assume the energy decay result of section \ref{decay} with  $p(e)=3-\alpha(e)$.

We can now establish the decay of $D_v \psi$ in  $\{ r^* \geq \frac{v}{2} +R^*\}$ --- which is the ``right'' of $\gamma$, c.f.\ Figure \ref{Fig3} --- together with estimates for the radiation field $\psi$, in particular on $\mathcal{I}^+$: 

\begin{cor} \label{corinfinity}

	Make the same assumptions as for Lemma \ref{lemmainfinity}.
	
	Suppose also that there exists $D_0>0$ such that for all $v_0(R) \leq v$: 
	
	$$ |D_v \psi|(u_0(R),v) \leq D_0 \cdot v^{-2+ \frac{\alpha}{2}}.$$
	
	Then there exists $C=C(M,\rho,R,e,D_0)>0$ so that for all $(u,v) \in \mathcal{D}(u_0(R),+\infty)\cap  \{ r^* \geq \frac{v}{2}+R^* \}$ we have for $u>0$:
	
	\begin{equation} \label{corinfinity1}
		|D_v \psi| \leq C \cdot  v^{-2+ \frac{\alpha}{2}},
	\end{equation}
	\begin{equation} \label{corinfinity2}
		|\psi| \leq C \cdot  u^{-1+ \frac{\alpha}{2}},
	\end{equation}
	\begin{equation} \label{corinfinity3}
		|Q-e| \leq C \cdot  u^{-2+\alpha},
	\end{equation} where $\alpha(e):=3-p(e) \in [0,2)$, as introduced in the beginning of the section.
\end{cor}

\begin{proof}
	Choose $0<\epsilon<\min\{\frac12,\frac{\alpha}{2}\}$ and set
	$\beta=\frac12-\epsilon$.
	
	We are working to the ``right'' of $\gamma$ where $r \sim v \gtrsim u$. Therefore using \ref{wavev} and the result of Lemma \ref{lemmainfinity}, there exists $C=C(R,\epsilon,e,M,\rho)>0$ such that
	
	$$ |D_u (D_v \psi)| \leq C \cdot r^{-2+\epsilon}
	u^{-1-\epsilon+\frac{\alpha}{2}}.$$
	
	Integrating on $[u_0(R),u]$ and using $r\sim v\gtrsim u$ gives
	
	$$ |D_v\psi|(u,v) \leq |D_v\psi|(u_0(R),v)
	+C v^{-2+\epsilon}u^{\frac{\alpha}{2}-\epsilon}
	\leq |D_v\psi|(u_0(R),v)+C v^{-2+\frac{\alpha}{2}}.$$
	
	Making use of the decay of the initial data gives \eqref{corinfinity1}.
	
	Now we notice that from Lemma \ref{lemmainfinity}, we have that there exists $C'=C'(R,M,\rho,e)>0$ such that
	
	$$ |\psi|(u, v_{\gamma}(u)) \leq C' u^{-1+ \frac{\alpha}{2}}.$$
	
	Therefore, integrating \eqref{corinfinity1} in $v$ to the right of $\gamma$ we obtain using by Lemma \ref{integrationlemma}:$$  |\psi|(u,v) \leq  |\psi|(u, v_{\gamma}(u)) + C_1 \int_{v_{\gamma}(u)}^{v} (v')^{-2+ \frac{\alpha}{2}} dv'  \leq C_2 u^{-1+ \frac{\alpha}{2}}  ,$$ for constants $C_1(M,\rho,R,e,D_0)>0$, $C_2(M,\rho,R,e,D_0)>0$. This proves \eqref{corinfinity2}.
	\color{black}
	
	Now we turn to the charge. Integrating \eqref{ChargeVEinstein} towards null
	infinity and using Cauchy--Schwarz, we obtain to the right of $\gamma$:
	
	{\begin{equation*}
			\begin{split}
				|Q(u,v)-Q_{|\mathcal I^+}(u)|
				&\le q_0\left(\int_v^{+\infty}r^{1+\epsilon}|D_v\psi|^2(u,v')\,dv'\right)^{\frac12}
				\left(\int_v^{+\infty}\frac{|\psi|^2(u,v')}{r^{1+\epsilon}}\,dv'\right)^{\frac12}\\
				&\le \frac{D}{\sqrt{\epsilon}}
				\left(\tilde E_{1+\epsilon}(u)\right)^{\frac12}
				\left(\tilde E_{1-\epsilon}(u)\right)^{\frac12}
				\le D^2u^{-2+\alpha},
			\end{split}
	\end{equation*}}
	where $D=D(R,M,\rho,e)>0$ and in the last inequality we use the weighted
	energy hierarchy with $\epsilon=\frac12$.
	
	Using \eqref{chargeUEinstein}, and \eqref{corinfinity2}, we
	also obtain
	
	\begin{equation*}
			\begin{split}
				|Q_{|\mathcal I^+}(u)-e|\leq q_0 \int_{u}^{+\infty} |\psi|_{|\mathcal{I}^+}(u') |rD_u\phi|_{|\mathcal{I}^+}(u')du' \leq C   \int_{u}^{+\infty} (u')^{-1+\frac{\alpha}{2}} |rD_u\phi|_{\mathcal{I}^+}(u')du'
			\end{split}
	\end{equation*} Now, define $u_n = 2^{n} u$ (dyadic decomposition) and write, using Cauchy-Schwarz  \begin{equation*}\begin{split}
&	\int_{u}^{+\infty} (u')^{-1+\frac{\alpha}{2}} |rD_u\phi|_{|\mathcal{I}^+}(u')du' = \sum_{n=0}^{+\infty} 	\int_{u_n}^{u_{n+1}} (u')^{-1+\frac{\alpha}{2}} |rD_u\phi|_{|\mathcal{I}^+}(u')du' \\ & \lesssim \sum_{n=0}^{+\infty} 	(u_n)^{-1+\frac{\alpha}{2}}[\int_{u_n}^{u_{n+1}}  |rD_u\phi|^2_{|\mathcal{I}^+}(u')du']^{1/2} u_n^{1/2}\ls  \sum_{n=0}^{+\infty} 	(u_n)^{-1/2+\frac{\alpha}{2}}[E(u_n)]^{1/2}\ls \sum_{n=0}^{+\infty} 	(u_n)^{-2+\alpha} \ls u^{-2+\alpha}.   \end{split}
	\end{equation*}
	Combining the two estimates proves \eqref{corinfinity3}.

	This concludes the proof of Corollary \ref{corinfinity}.

\end{proof}

Notice that these two results prove estimates \eqref{thm3} and \eqref{thm4} of Theorem \ref{maintheorem}, together with \eqref{thm2}, \eqref{thm5} to the right of $\gamma$.

\subsection{Point-wise bounds in the region between $\gamma$ and $\gamma_R$}

We now propagate the bounds already obtained ``to the right'' of $\gamma$ towards the right of $\gamma_R$. Since we already have an estimate from Lemma \ref{lemmainfinity}, the argument is very soft.

\begin{figure}  
	
	\begin{center} 
		
		\includegraphics[width=107 mm, height=65 mm]{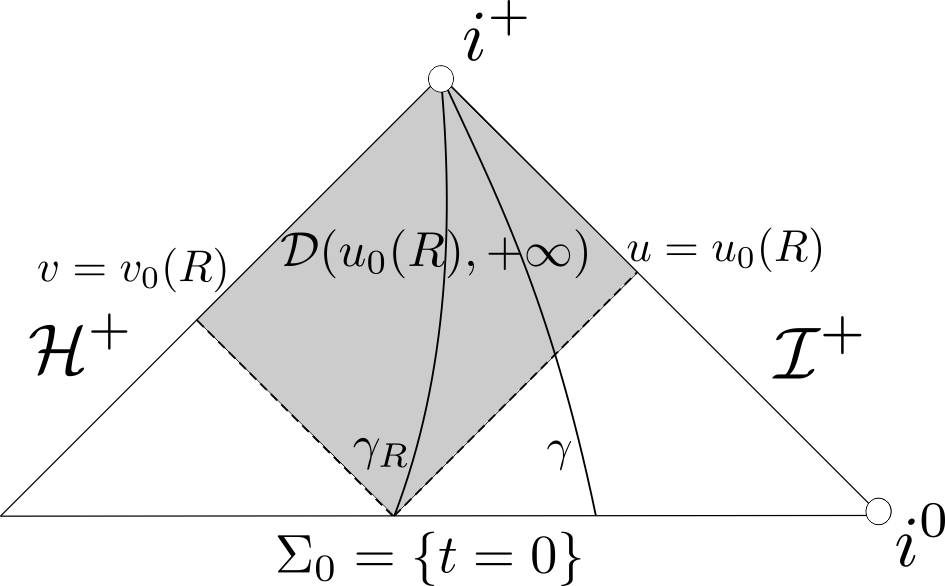}
		
	\end{center}
	
	\caption{Illustration of $\mathcal{D}(u_0(R),+\infty)$ and $\gamma$}
	\label{Fig3}
\end{figure}

\begin{prop} 	Make the same assumptions as for Corollary \ref{corinfinity}.

	Then there exists $C=C(M,\rho,R,e)>0$ so that for all $(u,v) \in \mathcal{D}(u_0(R),+\infty)\cap  \{  R^* \leq r^* \leq  \frac{v}{2}+R^* \}$ we have for $u>0$:
	
	\begin{equation} \label{trans1}
		r^{\frac{1}{2}}| \phi| +	r^{\frac{3}{2}}|D_v \phi| \leq C \cdot  u^{ \frac{-3+\alpha}{2}} ,
	\end{equation}
	\begin{equation} \label{trans3}
		|Q-e| \leq C \cdot  u^{-2+\alpha},
	\end{equation} where $\alpha(e):=3-p(e) \in [0,2)$, as introduced in the beginning of the section.
	
\end{prop}

\begin{proof}
	With the work that we have already done, this proof is easy.
	
	First notice that in this region $u \sim v \sim 2u$ when $v$ tends to $+\infty$.
	
	Then, making use of the boundedness of $Q$ it is enough to integrate \eqref{wavev}, using the estimate of Lemma \ref{lemmainfinity} for $\beta=0$ and the bound on $|D_v \psi|_{ |\gamma}$ of Corollary \ref{corinfinity}.

	We get for all $(u,v) \in \mathcal{D}(u_0(R),+\infty)\cap  \{  R^* \leq r^* \leq  \frac{v}{2} +R^*\}$

	$$ |D_v \psi|(u,v) \leq  C' \cdot  	r^{-\frac{1}{2}}(u,v)  \cdot u^{ \frac{-3+\alpha}{2}},$$
	
	for some $C'=C'(M,\rho,R,e)>0$ since $E(u)$ decays like $u^{ -3+\alpha}$. Combining with the estimate of Lemma \ref{lemmainfinity} for $\beta=0$ gives directly \eqref{trans1}, noting that $D_v \psi = \Omega^2 \phi + r D_v \phi$.
	
	Then using \eqref{ChargeVEinstein}, we see that, by \eqref{trans1}
	
	$$ |\partial_v Q| \leq C \cdot u^{-3+\alpha},$$
	
	which gives \eqref{trans3} easily using the former estimate for $|Q-e|$ , after integrating in this region where $u \sim v \sim 2u$.
\end{proof}
\subsection{Point-wise bounds in the finite $r$ region $\{ r\leq R \}$} \label{finiterregion}

Now, notice that by construction of the domain $\mathcal{D}(u_0(R),+\infty)$, bounds have been proven on the whole curve $\gamma_R=\{ r = R \}$. In what follows, we can completely forget about the foliation $\mathcal{V}_u$ and $\mathcal{D}$ and we consider the full region $\{ r \leq R \}$, bounded by $\mathcal{H}^+$ , $\Sigma_0$ and $\gamma_R$, c.f.\ Figure \ref{Fig3}.

We are going to prove the following proposition, that also includes a red-shift estimate.

\begin{prop} \label{finiter}
	Suppose that Q is bounded in $\{ r_+ \leq r \leq R \}$ and define $Q^+= \|Q\|_{\infty}$. 
	
	By section \ref{CauchyCharac}, $Q^+$ can be taken arbitrarily close to $|e|$ or $|e_0|$. We also have assumed naturally that $\mathcal{E}<\infty$.
	
	Suppose that $\frac{D_u \phi_0}{\Omega} \in L^{\infty}(\Sigma_0)$ and $\phi_0$ in $L^{\infty}(\Sigma_0)$. We denote $N_{\infty}:= \|\frac{D_u \phi_0}{\Omega}  \|_{L^{\infty}(\Sigma_0)}+ \|\phi_0 \|_{L^{\infty}(\Sigma_0)}$.
	
	Then there exists $C=C(M,\rho,R,e,\mathcal{E},N_{\infty})>0$ such that 
	
	for all $(u,v) \in \{ r_+ \leq r \leq R \}$, if $v>0$: 
	
	\begin{equation} \label{finiterest}
		|\phi|(u,v) + |D_v \phi|(u,v) + \frac{|D_u \phi|(u,v)}{\Omega^2} \leq C \cdot v^{ \frac{-3+\alpha}{2}},
	\end{equation}
	where $\alpha(e):=3-p(e) \in [0,2)$, as introduced in the beginning of the section.
\end{prop}

\begin{proof}

	The first step in the proof is to prove a red-shift estimate with no time decay in a region 
	
	$\{ V_0 \leq v\leq v_0(R) \}$. This is the object of the following lemma:

	\begin{lem}\label{Sigma0RSlemma}
		Under the same assumptions as in Proposition \ref{finiter},  then for any $V_0 < v_0(R)$, there exists $D_0=D_0(V_0,M,\rho,R,e,\mathcal{E},N_{\infty})>0$ such that for all $(u,v) \in \{ V_0 \leq v\leq v_0(R) \}$,
		
		\begin{equation}
			\frac{|D_u \psi|(u,v)}{\Omega^2} \leq D_0.
		\end{equation}
	\end{lem}
	\begin{proof}
		Using the red-shift estimate \eqref{Charac2} from Proposition \ref{characProp1} with the help of Cauchy-Schwarz, one finds $C=C(M,\rho)>0$ such that in $\{ v \leq v_0(R)\}$
		
		\begin{equation} \label{phipointwiselemma}
			|\phi| \leq \| \phi_0 \|_{\infty} + C \cdot \sqrt{\mathcal{E}} \leq C \cdot ( N_{\infty}+\sqrt{\mathcal{E}}).
		\end{equation}
		
		Now we write \eqref{waveu} as 	$$  D_v (\frac{D_u \psi}{\Omega^2} )= \frac{-2K}{\Omega^2} D_u \psi + \frac{ \phi}{r}(iq_0Q -r\cdot 2K),$$

		which can also be expressed as 
		$$  D_v ( \exp(\int_{v_0(u)}^v 2K(u,v')dv') \cdot \frac{D_u \psi}{\Omega^2} )=   \exp(\int_{v_0(u)}^v 2K(u,v')dv') \cdot \frac{ \phi}{r}(iq_0Q -r \cdot 2K).$$ 
		
		We can then integrate such an estimate on $\{ u\} \times [v_0(u),v]$ and obtain
		
		$$  \frac{|D_u \psi(u,v)|}{\Omega^2} \leq \exp(-\int_{v_0(u)}^{v} 2K(u,v')dv') \cdot \Omega^{-1}(u, v_0(u)) \cdot N_{\infty} + C' \cdot ( N_{\infty}+\sqrt{\mathcal{E}}) \cdot \int_{v_0(u) }^{v}  \exp(-\int_{v'}^v 2K(u,v'')dv'')dv',$$
		
		for some $C'=C'(M,\rho,R)>0$ where we used the fact that $\frac{ 1}{r}(iq_0Q -\cdot r 2K)$ is bounded and \eqref{phipointwiselemma}.
		
		Now remember from \eqref{regularomega} that  $\Omega(u, v_0(u)) \sim C_+ \cdot e^{-2K_+u}$ as $u\rightarrow+\infty$.
		
		Also it can be proven easily that there exists $D_+=D_+(M,\rho)>0$
		such that in this region, since $(r-r_+)e^{-2K_+ \cdot (v-u)}$ is bounded,
		
		$$	 |2K(u,v)-2K_+| \leq D^+ e^{2K_+ \cdot (v-u)}.$$
		
		If we integrate this as a lower bound, we get that there exists $D'_+=D'_+(M,\rho)>0$ such that 
		
		$$ \exp(-\int_{v_0(u)}^{v} 2K(u,v')dv') \leq  D'_+ \cdot e^{-2K_+ \cdot (v+u)}.$$
		
		Then we get that for some $D''_+=D''_+(M,\rho)>0$: 
		
		$$  \frac{|D_u \psi(u,v)|}{\Omega^2} \leq D''_+\cdot e^{-2K_+ v} \cdot N_{\infty} + C' \cdot ( N_{\infty}+\sqrt{\mathcal{E}}) \cdot \int_{v_0(u) }^{v}  \exp(-\int_{v'}^v 2K(u,v'')dv'')dv',$$
		Now, because we stand in a region where $r \leq R$, one can find a constant $K_0=K_0(M,\rho,R)>0$ such that $K(r)>K_0$. Therefore we can write 
		
		$$ \int_{v_0(u) }^{v}  \exp(-\int_{v'}^v 2K(u,v'')dv'')dv' \leq \int_{v_0(u) }^{v}  \exp(-2K_0 \cdot (v-v'))dv' \leq \frac{1}{2K_0}.$$

		This concludes the proof of the lemma, after controlling $e^{-2K_+ v}$  by  $e^{-2K_+ V_0}$. 
		
	\end{proof}
	
	Then we use the time decay of the non-degenerate energy under the form: 
	
	for all $(u,v) \in \{  v \geq v_0(R)\} \cap \{r \leq R \}$

	$$ |\phi(u,v)- \phi(u_R(v),v)| \leq \left( \int_{u_R(v)}^{u} \Omega^2 du' \right)^{\frac{1}{2}} \cdot \left( \int_{u_R(v)}^{u}\frac{ |D_u \phi|^2(u',v)} {\Omega^2} du' \right)^{\frac{1}{2}} \leq \tilde{C} \cdot v^{ \frac{-3+\alpha}{2}}, $$
	
	where $\tilde{C}=\tilde{C}(M,\rho,R,e)>0$ and we took advantage of the fact that $|u_R(v)| \sim v$.
	
	Therefore, using the bounds of the former section we see that there exists $C'=C'(M,\rho,R,e)>0$ such that for all $(u,v) \in \{  v \geq v_0(R)\} \cap \{r \leq R \}$
	
	\begin{equation}
		|\phi(u,v)| \leq C' \cdot v^{ \frac{-3+\alpha}{2}}.
	\end{equation}
	
	Now integrating \eqref{wavev} we can, allowing $C'$ to be larger, derive the same estimate for $D_v \psi$: 
	
	\begin{equation}
		|D_v\psi(u,v)| \leq C' \cdot v^{ \frac{-3+\alpha}{2}}.
	\end{equation}
	
	Now using a reasoning similar to one of the lemma, we get that for some $C''=C''(M,\rho,R,e)>0$: 
	
	$$  \frac{|D_u \psi(u,v)|}{\Omega^2} \leq \exp(-\int_{v_0(R)}^{v} 2K(u,v')dv') \cdot D_0 + C'' \int_{v_0(R) }^{v} (v')^{ \frac{-3+\alpha}{2}} \exp(-\int_{v'}^v 2K(u,v'')dv'')dv'.$$
	
	Now observe that for all $\beta>0$, $s>0$ ---using an integration by parts ---  when $v \rightarrow +\infty$:
	
	$$ e^{-\beta v} \int_{v_0(R)}^{v} e^{\beta v'} (v')^{-s} dv' \leq \frac{v^{-s}}{\beta}+ O( v^{-s-1}).$$
	
	Now, because there exists $K_0=K_0(M,\rho,R)>0$ such that $K>K_0$, this concludes the proof of the proposition, after noticing that  $\exp(-\int_{v_0(R)}^{v} 2K(u,v')dv') = o( v^{ \frac{-3+\alpha}{2}}) $.

\end{proof}










\begin{rmk} \label{omegaRS}
	Notice that \eqref{finiterest} was stated for $v>0$ only. This is related to a degeneration  \color{black} of the $\Omega^2$ weight towards the bifurcation sphere c.f.\ Remark \ref{RSstatementremark}. Actually, from Lemma \ref{Sigma0RSlemma} it is not hard to infer that $ e^{2K_+u} |D_u \phi|$ is bounded \textbf{on the whole space-time}, in conformity with Hypothesis \ref{regderivRS}. For a discussion of the $L^2$ analogue, c.f.\ also Remark \ref{regularomegaremark}.
\end{rmk}

\subsection{Remaining estimates for $D_u\psi$ near $\mathcal{I}^+$ and $|Q-e|$ near $\mathcal{H}^+$ } \label{remaining}

For this last sub-section of section \ref{pointwise}, we make a little summary of what we did.

First, we derived an estimate for $r^{\frac{1}{2}} \phi$ on the whole region $\{ r \geq R\}$ using energy decay.

Then, with the help of the point-wise decay of $D_v \psi$ on $u=u_0(R)$ --- itself coming from the point-wise decay of $\psi_0$ and $D_v \psi_0$ on $\Sigma_0$, c.f.\ Proposition \ref{propagationdecay} --- we derived $v$ decay of  $D_v \psi$, still on $\{ r \geq R\}$.

This gives directly almost optimal point-wise estimate for $\psi$ and $|Q-e|$ on null infinity and nearby.  \\

Then the $v$ decay of $\phi$ can be propagated from $\gamma_R$ to the event horizon using again the decay of the energy. As a consequence $v$ decay can also be retrieved for $D_v \phi$ on $\{ r \leq R \}$. After we use a red-shift estimate to prove the same $v$ decay for the regular derivative $\Omega^{-2} D_u \phi$. \\

Now compared to the statement of our theorems, we are missing four estimates.

The first is an estimate for $D_u \psi$ in the large $r$ region $\{ r \geq R \}$. It is the object of Proposition \ref{lastprop1}.

The second is the almost optimal $L^2$ flux bounds on $\phi$ and $D_v \phi$ on any constant $r$ curve. We prove them in Proposition \ref{lastprop2}.

The third is the existence of the future asymptotic charge $e$. We proved that $Q$ admits a future limit when $t \rightarrow +\infty$ on constant $r$ curves, the event horizon and null infinity but we never proved that they were all the same. 

The fourth is related to the third: it is the $v$ decay of $|Q-e|$ in the bounded $r$ region. These two are the object of Proposition \ref{lastprop3}

\begin{prop} \label{lastprop1} We make the same assumptions as for the former propositions.
	Then there exists a constant  $C=C(M,\rho,R,e,\mathcal{E},N_{\infty})>0$ such that in $\{ r \geq R\} \cap \{ u \geq u_0(R)\} $ , for $u>0$:
	
	$$ |D_u \psi| \leq C \cdot u^{ \frac{-3+\alpha}{2}},$$
	where $\alpha(e):=3-p(e) \in [0,2)$, as introduced in the beginning of the section.
\end{prop}

\begin{proof}
	Using the estimates for $\phi$ in $\{ r \geq R\} \cap \{ u \geq u_0(R)\} $ and \eqref{waveu} we get that for some $D=D(M,\rho,R)>0$
	
	$$ |D_v D_u \psi| \leq D \cdot r^{ -\frac{3}{2}} \cdot  u^{ \frac{-3+\alpha}{2}} .$$
	
	It is enough to integrate this bound --- given that $|\partial_u r| \geq \Omega^2(R)$ and make use of the estimate in the past of Proposition \ref{finiter}.

\end{proof}

\begin{prop} \label{lastprop2}
	For every $r_+ \leq R_0 \leq R $, there exists a constant $C_0=C_0(R_0,M,\rho,R,e)>0$ such that
	
	\begin{equation} \label{lastclaim}
		\int_{v}^{+\infty} \left[ |\phi|^2(u_{R_0}(v'),v')+ |D_v\phi|^2(u_{R_0}(v'),v') \right] dv' \leq C_0 \cdot v^{ -3+\alpha}.
	\end{equation}
\end{prop}

\begin{proof}

	Take any $r_+ \leq R_0 < R_1 \leq R$.
	
	First, using Cauchy-Schwarz, we find that there exists $D=D(R_0,R_1,M,\rho)>0$ such that
	
	$$ |\phi|^2 (u_{R_1}(v),v) \leq 2| \phi|^2 (u_{R_0}(v),v) + D \cdot \int_{ u_{R_0}(v )}^{ u_{R_1}(v )} \frac{|D_u\phi|^2(u,v)}{\Omega^2} du ,$$
	
	where we squared the Cauchy-Schwarz estimate and bounded  $2\int_{ u_{R_0}(v )}^{ u_{R_1}(v )} \Omega^2 du \leq D$.
	This implies that

	\begin{equation} \label{morawetzrconstant}
		\int_{v}^{+\infty}  |\phi|^2(u_{R_1}(v'),v') dv' \leq 2 \int_{v}^{+\infty}  |\phi|^2(u_{R_0}(v'),v') dv'+ D \cdot \int_{   \{ v' \geq v\} \cap \{ R_0 \leq r \leq R_1\}} \frac{|D_u\phi|^2(u,v')}{\Omega^2} du dv'.
	\end{equation}

	Now since $\{ R_0 \leq r \leq R_1\}$ is included inside $\{ r \leq R\}$, one can use the Morawetz estimate \eqref{Morawetzestimate3}: there exists $C=C(M,\rho,e,R)>0$ such that 
	
	\begin{equation}
		\int_{v}^{+\infty}  |\phi|^2(u_{R_1}(v'),v') dv' \leq 2 \int_{v}^{+\infty}  |\phi|^2(u_{R_0}(v'),v') dv'+ C \cdot E( u_R(v)).
	\end{equation}
	
	Similarly, using Cauchy-Schwarz and \eqref{wavev}, we find that there exists $D'=D'(R_0,R_1,M,\rho)>0$ such that
	
	$$ |D_v\psi|^2 (u_{R_1}(v),v) \leq 2| D_v\psi|^2 (u_{R_0}(v),v) + D' \cdot \int_{ u_{R_0}(v )}^{ u_{R_1}(v )} \Omega^2 |\phi|^2(u,v) du .$$
	
	Therefore one can use again the non-degenerate Morawetz estimate \eqref{Morawetzestimate3}: there exists $C'=C'(M,\rho,e,R)>0$ such that 
	
	\begin{equation}
		\int_{v}^{+\infty}  |D_v\psi|^2(u_{R_1}(v'),v') dv' \leq 2 \int_{v}^{+\infty}  |D_v\psi|^2(u_{R_0}(v'),v') dv'+ C' \cdot E( u_R(v)).
	\end{equation}
	
	Now we use the mean-value theorem with the Morawetz estimate  \eqref{Morawetzestimate3}, like we did several times in section \ref{energyboundedsection}: there exists $r_+<\tilde{R}<R$ such that 
	
	\begin{equation*} 
		\int_{v}^{+\infty} \left[ |\phi|^2(u_{\tilde{R}}(v'),v')+ |D_v\phi|^2(u_{\tilde{R}}(v'),v') \right] dv' \leq  C''	\int_{\mathcal{D}(u_R(v),+\infty) \cap \{ r\leq R \}}	\left( r^2 |D_v \phi|^2 +  |\phi|^2\right)\Omega^2  du dv \leq (C'')^2 \cdot E(u_R(v)),
	\end{equation*}
	
	where $C''=C''(M,\rho,e,R)>0$.
	
	Therefore, taking $R_0=\tilde{R}$, it means that for all $r_+ \leq R_1 \leq R$,
	
	\begin{equation}
		\int_{v}^{+\infty} \left[ |\phi|^2(u_{R_1}(v'),v')+  |D_v\psi|^2(u_{R_1}(v'),v') \right] dv' \leq C_1 \cdot E( u_R(v)) \leq C_1^2 \cdot v^{ -3+\alpha},
	\end{equation} 
	for some $C_1=C_1(M,\rho,e,R)>0$, and we used \eqref{zeroenergydecayest} for the last inequality.\color{black}

	Notice that since $D_v \psi = r D_v \phi + \Omega^2 \phi$, this estimate is equivalent to the claimed \eqref{lastclaim}. This concludes the proof of the proposition.

\end{proof}

\begin{prop} \label{lastprop3}
	The future asymptotic charge $e$ exists --- in the sense explained above --- and the following estimate is true:  for some $C'=C'(M,\rho,R,e)>0$
	
	\begin{equation}  |Q - e| (u,v)	\leq C' \cdot \left(	 u^{-2+\alpha}  1_{ \{ r \geq R\}}+v^{-3+\alpha} 1_{ \{ r \leq R\}}\right).
	\end{equation}
\end{prop}

\begin{proof}	For now on, $e$ is defined to be the asymptotic limit of $Q_{|\mathcal{I}^+}(t)$  when $t \rightarrow +\infty$.
	
	Temporarily we also define $e_{\mathcal{H}^+}$ as the asymptotic limit of $Q_{|\mathcal{H}^+}(t)$  when $t \rightarrow +\infty$, and $e(R_0)$ the asymptotic limit of $Q_{|\gamma_{R_0}}(t)$  when $t \rightarrow +\infty$. One of the goals of the proposition is to prove $e=e_{\mathcal{H}^+}=e(R_0)$.
	
	

	We apply the estimate \eqref{finiterest}  of  former section to get some $\tilde{C}= \tilde{C}(M,\rho,R,e,\mathcal{E},N_{\infty})>0$ such that 
	
	$$ |\partial_u Q| \leq  \tilde{C} \cdot \Omega^2 \cdot v^{ -3+\alpha}.$$
	
	Integrating this gives that for some $\bar{C}= \bar{C}(M,\rho,R,e,\mathcal{E},N_{\infty})>0$ and for all $r_+ \leq R_0 \leq R_1 \leq R$
	
	\begin{equation} \label{chargeUdifference}
		| Q(u_{R_0}(v),v) -Q(u_{R_1}(v),v)| \leq \bar{C}\cdot  v^{ -3+\alpha}.
	\end{equation}
	
	In particular, taking $v$ to $+\infty$, it proves that for all $r_+ \leq R_0 \leq R_1 \leq R$, $e(R_0)=e(R_1)=e_{\mathcal{H}^+}$.
	
	But \eqref{trans3} also gives $e(R)=e$. Hence $e(R_0)=e(R_1)=e_{\mathcal{H}^+}=e$ as requested.
	
	Now apply \eqref{lastclaim} on the event horizon, for $R_0=r_+$ and integrate \eqref{ChargeVEinstein}: we get that for some
	
	$C'_0=C'_0(R_0,M,\rho,R,e)>0$
	
	\begin{equation*}
		| Q_{|\mathcal{H}^+}(v) -e| \leq C'_0 \cdot  v^{ -3+\alpha}.
	\end{equation*} 
	
	Combining with estimate with \eqref{chargeUdifference} taken in $R_0=r_+$ we get that for all $v \geq v_0(R)$ and $u \geq u_R(v)$
	
	\begin{equation}
		| Q(u,v) -e| \leq (C'_0+ \bar{C})\cdot  v^{ -3+\alpha},
	\end{equation} 
	
	which is the claimed estimate on $\{ r \leq R\}$. The second part in $\{ r \geq R\}$ was already proven in 	 section \ref{finiterregion}. This concludes the proof of the proposition.
\end{proof}

\appendix 

\section{A variant of Theorem \ref{decaytheorem} for more decaying initial data} \label{moredecayappendix}

While the point-wise results stated in Theorem \ref{maintheorem} and Theorem \ref{decaytheorem} are established assuming the weakest decay of the initial data that made the proofs work, they do not give the same $r$ decay rate for $D_v \psi$ towards null infinity as we would expect for say compactly supported data, which is $|D_v\psi| \lesssim r^{-2}$.

Notice that by what was done in the proof of Theorem \ref{decaytheorem}, we already know that $|D_v \psi| \lesssim v^{-1-\epsilon}$ for some $\epsilon>0$. Therefore $ \partial_v (e^{iq_0 \int_{0}^{v} A_v(u,v') dv'}\psi(u,v) )$ is integrable in $v$ so $e^{iq_0 \int_{0}^{v} A_v(u,v') dv'}\psi(u,v)$ admits a finite limit $\tilde{\varphi_0}(u)$ when $v \rightarrow +\infty$.

In this section, we want to prove that for data decaying slightly more than in the assumption of Theorem \ref{decaytheorem}, $e^{iq_0 \int_{0}^{v} A_v(u,v') dv'} r^2 D_v \psi $ admits a bounded limit towards null infinity, like for the wave equation \underline{without any symmetry assumption}. 

It is equally interesting to notice that --- similarly to the wave equation c.f.\ \cite{Latetime} --- $e^{iq_0 \int_{0}^{v} A_v(u,v') dv'} r^2 D_v (r^2 D_v \psi)$ also admits a bounded limit towards null infinity, providing the data decay even more. In fact, one can apply $r^2 D_v$ to the radiation field infinitely many times and still find a finite limit towards null infinity \footnote{After multiplying the last term by $e^{iq_0 \int_{0}^{v} A_v(u,v') dv'}$.}, providing the data decay rapidly.

This feature is reminiscent of what happens for the uncharged wave equation, although the analogous of the Newman-Penrose quantity is no longer conserved in the charged case.

\begin{thm} \label{decaytheoremmod}

	We make the same assumptions as for Theorem \ref{decaytheorem}, assume in
	addition that the improved branch $q_0|e|<{0.04}$ applies, and make the
	following additional assumption: there exist $\epsilon_0>0$ and $C_0>0$ such that
	
	\begin{equation} \label{hypappendix1} |D_v(r^2 D_v \psi_0)|+ r |D_v \psi_0|(r)+ |\psi_0|(r) \leq C_0 \cdot  r^{-1-\frac{q_0|e_0|}{4}-\epsilon_0}, \end{equation}

	then --- in addition to all the conclusions of Theorem \ref{decaytheorem} being true --- we can also conclude that for all $u \in \mathbb{R}$, $ e^{iq_0 \int_{0}^{v} A_v(u,v') dv'} v^2 D_v \psi (u,v)$ admits a finite limit $\tilde{\varphi}_1(u)$ as $v \rightarrow +\infty$. Moreover in the whole region $ \mathcal{D}(u_0(R),+\infty) \cap \{ r^* \geq  \frac{v}{2} +R^* \}$, there exists $C>0$ such that for all $u>0$: 
	
	\begin{equation} \label{appendix1} r^2 |D_v \psi| \leq C \cdot  u^{\frac{3-p(e)}{2}}. \end{equation}
	
	In particular $|\tilde{\varphi}_1(u)| \leq  C \cdot  u^{\frac{3-p(e)}{2}}$.
	
	If now we assume that we have rapidly decaying data, i.e.\ that for all $\omega>0$, there exists $C_0=C_0(\omega)>0$ such that on $\Sigma_0$: 
	
	\begin{equation} \label{hypappendix2} r |D_v \psi_0|(r)+ |\psi_0|(r) \leq C_0 \cdot r^{-\omega}, \end{equation}
	and for all $n \in \mathbb{N}$, on $\Sigma_0$: 
	\begin{equation} \label{hypappendix3} |\varphi_n|(r) \leq C_0, \end{equation}
	where we defined $\varphi_0:= \psi$ and $\varphi_{n+1}= r^2 D_v \varphi_n$.
	
	Then one can prove that for all $n \in \mathbb{N}$, for all $u \in \mathbb{R}$, $ e^{iq_0 \int_{0}^{v} A_v(u,v') dv'}\varphi_n (u,v)$ admits a finite limit $\tilde{\varphi}_n(u)$ as $v \rightarrow +\infty$.

\end{thm}

\begin{proof} First we only assume \eqref{hypappendix1}. We start by the proof of \eqref{appendix1}, the easiest claim. The hypothesis of Proposition \ref{propagationdecay} are satisfied for $\omega=1+ \frac{q_0|e_0|}{4}+\epsilon_0$ (or equivalently $\theta=\frac{q_0|e_0|}{4}+\epsilon_0>\frac{q_0|e_0|}{4}$) \color{black} so \eqref{ext1} and \eqref{ext2} are true: in particular there exists $D_0>0$ such that for all $v \geq v_0(R)$,
	
	\begin{equation} \label{extappendix1}
		|D_v \psi(u_0(R),v)| \leq D_0 \cdot  v^{-2}.
	\end{equation}

	We write \eqref{wavev} as: 
	
	$$ D_u (r^2 D_v \psi) = -2 r \Omega^2 D_v \psi + \Omega^2 \psi \cdot (iq_0 Q -\frac{2M}{r}+\frac{2\rho^2}{r^2}).$$
	
	We are going to apply Corollary \ref{corinfinity}: we recall that we defined $\alpha(e):=3-p(e)$. Applying \eqref{corinfinity1} and \eqref{corinfinity2} and given that in this region $v \gtrsim |u|$, there exists $D>0$ such that for $u>0$: 
	
	$$ |D_u(r^2 D_v \psi)| \leq D \cdot u^{-1+ \frac{\alpha}{2}}.$$
	
	Integrating this in $u$ and using \eqref{extappendix1} gives \eqref{appendix1} in this region.
	
	Now, to prove that $e^{iq_0 \int_{0}^{v} A_v(u,v') dv'}r^2 D_v \psi (u,v)$ admits a limit when $v \rightarrow +\infty$, we prove that 
	
	$\partial_v (e^{iq_0 \int_{0}^{v} A_v(u,v') dv'}r^2 D_v \psi)=e^{iq_0 \int_{0}^{v} A_v(u,v') dv'} D_v (r^2 D_v \psi)$ is integrable. From now on, we denote 
	
	$r^2 D_v \psi= \varphi_1$. Applying $D_v$ on \eqref{wavev} it is possible to prove that 
	
	$$ D_u D_v \varphi_1  = \frac{-2\Omega^2}{r} D_v \varphi_1 + \frac{\Omega^2}{r^2} \varphi_1 \cdot \left[2+3iq_0Q -\frac{10M}{r}-\frac{8 \rho^2}{r^2}\right] + \frac{\Omega^2 \psi}{r^2} \left[ (iq_0Q-r \cdot 2K)\cdot (2K \cdot r^2)+iq_0^2 \Im(\varphi_1 \bar{\psi})+ (2M-\frac{4\rho^2}{r})\cdot \Omega^2\right],$$
	
	where we recall that $2K(r):= \frac{2M}{r^2} -\frac{2\rho^2}{r^3}$: hence $r^2 \cdot 2K$ is bounded. This can also be written as:

	\begin{equation} \label{Dvwavev} D_u(r^{-2} D_v \varphi_1)  =  \frac{\Omega^2}{r^4} \varphi_1 \cdot \left[2+3iq_0Q -\frac{10M}{r}-\frac{8 \rho^2}{r^2}\right] + \frac{\Omega^2 \varphi_0}{r^4} \left[ (iq_0Q-r \cdot 2K)\cdot (2K \cdot r^2)+iq_0^2 \Im(\varphi_1 \bar{\varphi_0})+ (2M-\frac{4\rho^2}{r})\cdot \Omega^2\right]. \end{equation}
	
	First, in the region $\{ u \leq u_0(R)\}$, we come back to the proof of Proposition \ref{propagationdecay}. If we examine more closely the estimate \eqref{propdecay}, we see that we actually proved that for some $D'>0$ and for all $u \leq u_0(R)$: 
	
	\begin{equation}  \label{extomega} |\varphi_1|=|r^2 D_v \psi| \leq D' \cdot |u|^{1-\omega}, \end{equation}
	
	where $\omega = 1+\frac{q_0|e_0|}{4}+\epsilon_0$. From Proposition \ref{propagationdecay}, we also have the estimate
	
	$$ |\psi| \leq D' \cdot |u|^{-\omega}.$$  
	
	Using \eqref{Dvwavev} in the region $\{ u \leq u_0(R)\}$ and the boundedness of relevant quantities, we see that there exists $D''>0$ such that
	
	$$|D_u(r^{-2} D_v \varphi_1)| \leq D'' \cdot  r^{-4} \cdot |u|^{1-\omega}.$$ 
	
	Integrating in $u$ towards $\Sigma_0$, since $1< \omega <2$ , we get that there exists $C''>0$ such that 
	
	$$ |r^{-2} D_v \varphi_1| \leq D'' \cdot r^{-2-\omega},$$
	where we also used the first term of the \eqref{hypappendix1}. This means in particular that for all $v \geq v_0(R)$:
	
	\begin{equation} \label{constantuappendix}
		|D_v \varphi_1| (u_0(R),v) \leq D'' \cdot r^{-\omega}.
	\end{equation}
	
	Now we can use the bound on $\varphi_1$ \eqref{appendix1} with \eqref{corinfinity2} and integrate \eqref{Dvwavev} in $u$ towards $\{ u=u_0(R)\}$ to get that there exists $C'>0$ such that for all $(u,v) \in  \mathcal{D}(u_0(R),+\infty) \cap \{ r^* \geq  \frac{v}{2} +R^* \}$ 
	
	\begin{equation}
		|D_v \varphi_1| (u,v) \leq D'' \cdot r^{-\omega},
	\end{equation}

	where we used the fact that $0 < \alpha < 1$ and \eqref{constantuappendix}. Now because  $\omega>1$, it is clear that 	$\partial_v (e^{iq_0 \int_{0}^{v} A_v(u,v') dv'}r^2 D_v \psi)=e^{iq_0 \int_{0}^{v} A_v(u,v') dv'} D_v (r^2 D_v \psi)$ is integrable. Hence $e^{iq_0 \int_{0}^{v} A_v(u,v') dv'}r^2 D_v \psi (u,v)$ admits a limit when $v \rightarrow +\infty$, as claimed.
	
	We now assume \eqref{hypappendix2}, \eqref{hypappendix3} and we want to prove that for all $n \geq 2$, $e^{iq_0 \int_{0}^{v} A_v(u,v') dv'}\varphi_n (u,v)$ admits a limit when $v \rightarrow +\infty$. 
	
	For this, we need to commute the equation \eqref{wavev} with the operator $X=r^2 D_v$ n times. While the precise formula is very complicated one can write the following: 
	
	\begin{equation} \label{derivntimes}
		D_u D_v \varphi_n = -\frac{2n\Omega^2}{r} D_v \varphi_n + \sum_{k=0}^{n} \left[ (B_k + \sum_{q=k}^{n} i C_{k q}X^{n-q}Q) \cdot \frac{\varphi_k}{r^2}   \right],
	\end{equation}
	
	where $B_k(r)$, $C_{k q}$ are real valued entire functions of $r^{-1}$ (therefore bounded on the space-time) whose coefficients depend only on $M$, $\rho$ and $q_0$.
	
	We check this by a quick induction: by what we derived earlier, the formula is obviously true for $n=0$ and $n=1$, since $X(Q) = q_0\Im(\varphi_1 \bar{\varphi_0})$. Now if \eqref{derivntimes} is true, we can rewrite it using $\varphi_{n+1}=r^2 D_v \varphi_n$ as 
	
	$$ D_u ( \varphi_{n+1}) = -\frac{2(n+1)\Omega^2}{r} \varphi_{n+1} + \sum_{k=0}^{n} \left[ (B_k + \sum_{q=k}^{n}  C_{k q}X^{n-q}Q) \cdot \varphi_k   \right] .$$
	
	Then we apply $D_v$ to get 
	
	$$ D_v D_u ( \varphi_{n+1}) = -\frac{2(n+1)\Omega^2}{r} D_v\varphi_{n+1} - 2(n+1)\partial_v (\frac{\Omega^2}{r}) \varphi_{n+1} + \sum_{k=0}^{n} \left[ (B_k + \sum_{q=k}^{n}  C_{k q}X^{n-q}Q) \cdot \frac{\varphi_{k+1}}{r^2}   \right] + \sum_{k=0}^{n}   \sum_{q=k}^{n}  C_{k q}\partial_v(X^{n-q}Q) \cdot \varphi_{k}   .$$
	
	Now, because $r^2 \partial_v (\frac{\Omega^2}{r})$ is an entire function in $r^{-1}$, $\partial_v(X^{n-q}Q)= r^{-2} X^{n+1-q}Q$ and 
	
	$D_v D_u ( \varphi_{n+1}) = D_u D_v ( \varphi_{n+1}) -\frac{2\Omega^2}{r^2}\varphi_{n+1}$, it is clear that \eqref{derivntimes} is true at the level $(n+1)$, therefore for all $n$.
	
	Now that the equation is written, let us assume first we have in the region $\{ u \leq u_0(R)\}$ that for all $n$, there exists $D_n>0$ such that for all $v \geq v_0(R)$
	
	\begin{equation} \label{extlastestimate}
		|D_v \psi_n|(u_0(R),v) \leq D_n \cdot r^{-2}.\end{equation} 
	
	Then we can prove by induction that $|\varphi_n|(u,v) \lesssim u^{n-1+\frac{\alpha}{2}}$ for $u>0$. This is indeed the case for $n=0$ and $n=1$. Now if we assume the result for all integer $0 \leq k \leq n$ we want to prove it at the level $(n+1)$. For this we first write \eqref{derivntimes} as

	\begin{equation} \label{derivntimes2}	D_u ( r^{-2n}D_v \varphi_n) = r^{-2n-2} \sum_{k=0}^{n} \left[ (B_k + \sum_{q=k}^{n}  C_{k q}X^{n-q}Q) \cdot \varphi_k \right]. \end{equation}		
	
	Now notice that $X^q Q $ can be written as a linear combination of the form $X^q Q = \sum_{i=0}^{q} a_i \Im(\varphi_{q-i} \bar{\varphi_{i}})$, with $a_i \in \mathbb{R}$. Using the induction hypothesis we get that there exits $C_n>0$ such that for $u>0$
	
	$$ 	|D_u ( r^{-2n}D_v \varphi_n)|(u,v) \leq C_n \cdot  r^{-2n-2} u^{n-1+\frac{\alpha}{2}}.$$	
	
	Then integrating in $u$ towards $\{ u=u_0(R)\}$ and using \eqref{extlastestimate} we get that for some $C'_n >0$
	
	$$ |\varphi_{n+1}|=|r^2 D_v \varphi_n| \leq C'_n \cdot u^{n+\frac{\alpha}{2}},$$
	
	which finishes the induction. Now, this proves that for all $n$, $|D_v \varphi_n| \lesssim r^{-2}u^{n+\frac{\alpha}{2}}$ hence $ \partial_v (e^{iq_0 \int_{0}^{v} A_v(u,v') dv'}\varphi_n(u,v) )$ is integrable in $v$ so $e^{iq_0 \int_{0}^{v} A_v(u,v') dv'}\psi(u,v)$ admits a finite limit $\tilde{\psi_0}(u)$ when $v \rightarrow +\infty$, as claimed.
	
	Now, the last part of the proof is to establish \eqref{extlastestimate}. As noticed in \eqref{extomega}, we can prove, from Proposition \ref{propagationdecay} that for all $\omega>1$, there exists $\tilde{D}>0$ such that in $\{ u \leq u_0(R)\}$: 
	
	$$ |\psi| = |\varphi_0| \leq \tilde{D} \cdot |u|^{-\omega},$$
	$$ |\varphi_1| \leq \tilde{D} \cdot  |u|^{1-\omega}.$$
	
	Then this time we want to prove by induction that $|\varphi_n| \lesssim |u|^{n-\omega}$ for all $0<\omega<n$. It is enough to use \eqref{derivntimes2} with an argument similar to the one developed in the region $\{ u \geq u_0(R)\}$. We also need the hypothesis \eqref{hypappendix3} to deal with the term on $\Sigma_0$ when integrating in $u$.
	
	This finally gives \eqref{extlastestimate} and concludes the proof.

\end{proof}

\section{Proofs of section \ref{CauchyCharac} }\label{appendixCC}

In this section, we carry out the proofs of the results claimed in section \ref{CauchyCharac}. We are going to state once more the propositions, for the convenience of the reader. 

We start by the	proof of Proposition \ref{characProp1}.

\begin{prop}  Suppose that there exists $p>1$ such that $\tilde{\mathcal{E}}_p < \infty$ and that $Q_0 \in L^{\infty}(\Sigma_0)$. 
	
	Assume also that $\lim_{r \rightarrow +\infty} \phi_0(r)=0$.
	
	We denote $Q_0^{\infty} = \| Q_0 \|_{ L^{\infty}(\Sigma_0)}$.
	
	{There exist $r_+<R_0=R_0(M,\rho)$, $\delta=\delta(M,\rho)>0$, a number $0<\ep\leq C|q_0|\delta<\frac14$, and $C=C(M,\rho)>0$ such that, for all $R_1>R_0$, if} 
	
	$$ Q_0^{\infty}+  \tilde{\mathcal{E}}_p < \delta,$$
	
	then for all $u \geq u_{0}(R_1)$:
	
	\begin{equation}
		E_{R_1}(u) \leq C \cdot \tilde{\mathcal{E}}_{{\ep}}.
	\end{equation}
	
	Also for all $v \leq v_{R_1}(u)$: 
	\begin{equation}  
		\int_{\tilde{u}(v)}^{+\infty}  \frac{ r^2|D_u \phi|^2}{\Omega^2} (u',v)du'  \leq  C \cdot \tilde{\mathcal{E}}_{{\ep}},	\end{equation}
	
	where $\tilde{u}(v)=u_0(v)$ if $v \leq  v_0(R_1)$ and $\tilde{u}(v)=u_{R_1}(v)$ if $v \geq v_0(R_1)$.
	
	Moreover for all  $(u,v)$ in the space-time: 
	\begin{equation} 
		|Q|(u,v)	\leq  C \cdot  \left( Q_0^{\infty} + \tilde{\mathcal{E}}_p \right).
	\end{equation}

	Finally there exists $C_1=C_1(R_1,M,\rho)$ such that for all $u \geq u_{0}(R_1)$:
	\begin{equation}  
		\int_{         \{ v \leq v_{R_1}(u)\}      \cap \{ r \leq R_1 \}   } \left( |D_{t}\phi|^2(u',v)+|D_{r^*}\phi|^2(u',v)+|\phi|^2(u',v) \right) \Omega^2 du dv \leq C_1 \cdot \tilde{\mathcal{E}}_{{\ep}},
	\end{equation} 
	
\end{prop}

\begin{proof} Let $R_0>r_+$ large to be chosen later and $R_1 \geq R_0$.
	
	We are going to apply various identities to the domain 
	
	$\Theta_u(R_1) := \{   v\leq v_{R_1}(u) , \hskip 2 mm r \leq R_1\} \cup \{ u'\leq u , \hskip 2 mm r \geq R_1 \}$  , defined for all $u \geq u_0(R_1)$, having boundaries $ \mathcal{H}^+ \cap \{  v\leq v_{R_1}(u)\}$,  $[u,+\infty] \times \{ v_{R_1}(u)\} $, $ \{ u\} \times [v_{R_1}(u),+\infty]$, $\mathcal{I}^+ \cap \{  u'\leq u\}$ and $\Sigma_0$, c.f.\ Figure \ref{Fig4}
	
	Notice that for all $u \geq u_0(R_1)$, $\Theta_u(R_1)$ is the complement of the interior of $\mathcal{D}_{R_1}(u,+\infty)$.
	
	We apply the divergence identity in $\Theta_u(R_1)$ for the Killing vector field $\partial_t$, c.f.\ Appendix \ref{appendix} for more details. We get that {for all $u\geq u_0(R_1)$:} 
	
	\begin{equation} \label{energiIDappendix} \begin{split} \int_{-\infty}^{v_{R_1}(u)} r^2 |D_v \phi|^2_{|\mathcal{H}^+}(v)dv 
			+\int_{-\infty}^{u} r^2 |D_u \phi|^2_{|\mathcal{I}^+}(u)du+	\int_{u}^{+\infty} r^2 |D_u \phi|^2(u',v_{R_1}(u))du'+\int_{v_{R_1}(u)}^{+\infty} r^2 |D_v \phi|^2(u,v)dv\\+  \int_{u}^{+\infty}	  \frac{2\Omega^2 Q^2(u',v_{R_1}(u))}{r^2} du'+ \int_{v_{R_1}(u)}^{+\infty}	  \frac{2\Omega^2 Q^2(u,v)}{r^2} dv =  \int_{\Sigma_0} r^2\left( \frac{|D_u \phi|^2 + |D_v \phi|^2}{2} + \frac{2\Omega^2 Q^2}{r^4}\right)  dr^*, \end{split} \end{equation}
	
		The initial Hardy estimates used in Section~\ref{CauchyCharac}, together with {$0\leq\ep<p$}, control the scalar part of the right-hand side by $C\tilde{\mathcal E}_{\ep}$. Hence
		\begin{equation} \label{degenappendix}
			E_{deg,R_1}^+(u)\leq C\left(\tilde{\mathcal E}_{\ep}+(Q_0^{\infty})^2\right).
		\end{equation}
		Here $C=C(M,\rho)>0$, and $E_{deg,R_1}^+(u)$ is defined by

	\begin{equation} \label{defEdegR1} E_{deg,R_1}^+(u):= \int_{-\infty}^{v_{R_1}(u)} r^2 |D_v \phi|^2_{|\mathcal{H}^+}(v)dv 
		+\int_{-\infty}^{u} r^2 |D_u \phi|^2_{|\mathcal{I}^+}(u)du+	\int_{u}^{+\infty} r^2 |D_u \phi|^2(u',v_{R_1}(u))du'+\int_{v_{R_1}(u)}^{+\infty} r^2 |D_v \phi|^2(u,v)dv. \end{equation}
	
	Now we want to prove a similar estimate for the non-degenerate energy $E_{R_1}(u)$. We are going to use a slight variation of the red-shift effect as demonstrated in section \ref{RS}.

		We now repeat, on $\Theta_u(R_1)$, the combined Morawetz--red-shift--$r^p$ argument of Section~\ref{Morawetz}.
		
		\begin{lem}\label{combinedext}
			There exist $R_0=R_0(M,\rho)>r_+$, $\delta_*=\delta_*(M,\rho)>0$, $\alpha>4$, and $C=C(M,\rho)>0$ with the following property. For every $0\leq\delta_Q\leq\delta_*$ satisfying $8|q_0|\delta_Q<\frac14$, set $\ep:=8|q_0|\delta_Q$. For every $R_1\geq R_0$ and $u\geq u_0(R_1)$, if
			\[
			\|Q\|_{L^\infty(\Theta_u(R_1))}\leq\delta_Q,
			\]
			then
			\begin{equation}\label{combinedextestimate}
				\begin{split}
					&\iint_{\Theta_u(R_1)}r^{1-\alpha}\left(\frac{|D_u\phi|^2}{\Omega^2}+\Omega^2|D_v\phi|^2\right)\,du'\,dv
					+\iint_{\Theta_u(R_1)}\Omega^2r^{-1-\alpha}|\phi|^2\,du'\,dv\\
					&\quad+\int_u^{+\infty}\frac{r^{2-\alpha}|D_u\phi|^2}{\Omega^2}(u',v_{R_1}(u))\,du'
					+\int_{v_{R_1}(u)}^{+\infty}r^{\ep}|D_v\psi|^2(u,v)\,dv\\
					&\leq C\left(E_{deg,R_1}^+(u)+\mathcal E_{\ep}\right).
				\end{split}
			\end{equation}
			In particular, for some $\bar\sigma=\bar\sigma(M,\rho)>1$,
			\begin{equation}\label{Morawetzext}
				\iint_{\Theta_u(R_1)}\frac{|D_t\phi|^2+|D_{r^*}\phi|^2+|\phi|^2}{r^{\bar\sigma}}\Omega^2\,du'\,dv
				\leq C\left(E_{deg,R_1}^+(u)+\mathcal E_{\ep}\right),
			\end{equation}
			and, for a fixed $r_+<\widetilde R_0<R_0$ and every $v\leq v_{R_1}(u)$,
			\begin{equation}\label{RSext}
				\int_{u_{\widetilde R_0}(v)}^{+\infty}\frac{r^2|D_u\phi|^2}{\Omega^2}(u',v)\,du'
				+\iint_{\Theta_u(R_1)\cap\{r\leq\widetilde R_0\}}\frac{r^2|D_u\phi|^2}{\Omega^2}\,du'\,dv
				\leq C\left(E_{deg,R_1}^+(u)+\mathcal E_{\ep}\right).
			\end{equation}
		\end{lem}
		
		\begin{proof}
			Apply, on a truncated version of $\Theta_u(R_1)$, the same three multipliers as in the proof of Proposition~\ref{PropositionM}: $X_\alpha=-r^{-\alpha}\partial_{r^*}$, $X_{RS}=r^{-\alpha}\Omega^{-2}\partial_u$, and $X_{\ep}=r^{\ep}\partial_v$. The bulk identities, Hardy inequalities, and absorptions are unchanged. The only additional boundary is $\Sigma_0$; its scalar terms are bounded by $\mathcal E_{\ep}$, while its charge term is already contained in \eqref{degenappendix}. Combining the three estimates exactly as in the conclusion of Proposition~\ref{PropositionM}, and then removing the truncation, proves \eqref{combinedextestimate}. Equations \eqref{Morawetzext} and \eqref{RSext} are its respective Morawetz and red-shift consequences; for \eqref{RSext}, apply the same estimate to the subdomain ending at the chosen constant-$v$ boundary.
		\end{proof}

	\begin{figure}  
		
		\begin{center} 
			\includegraphics[width=107 mm, height=65 mm]{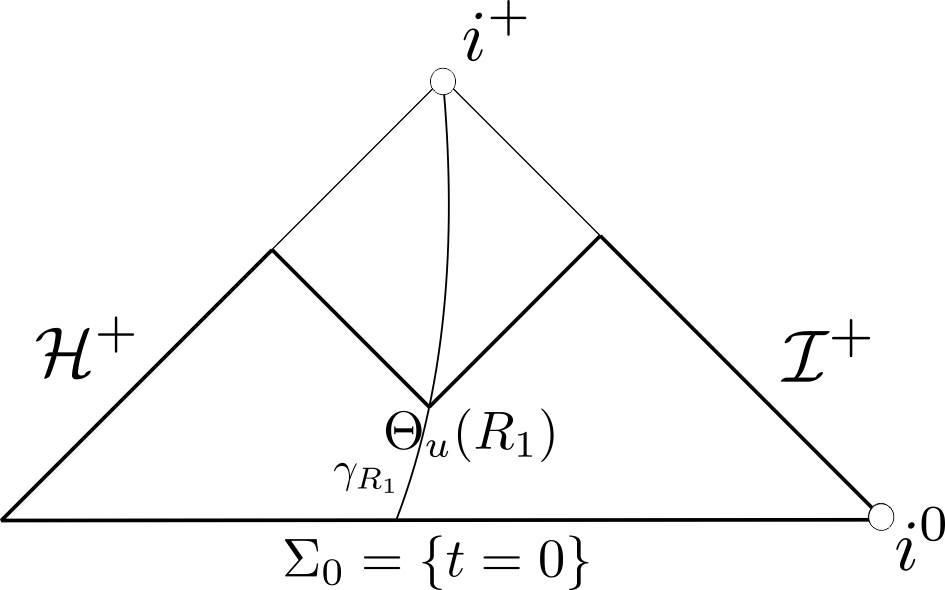}
			
		\end{center}
		
		\caption{Illustration of the domain $\Theta_u(R_1)$, $u \geq u_0(R_1)$, $R_0(M,\rho) \leq R_1 \leq R$}
		
		\label{Fig4}

	\end{figure}

		Let $C_*=C_*(M,\rho)>1$ be chosen at the end and set
		\[
		e^+:=C_*\bigl(Q_0^{\infty}+\tilde{\mathcal E}_p\bigr),\qquad \ep:=8|q_0|e^+.
		\]
		By reducing the data smallness constant $\delta$, we may assume that $e^+\leq\delta_*$ and $\ep<\frac14$, where $\delta_*$ is the threshold in Lemma~\ref{combinedext}. For fixed $u\geq u_0(R_1)$, we bootstrap on $\Theta_u(R_1)$:
		\begin{equation}\label{boostrapQ}
			|Q|\leq e^+.
		\end{equation}
	
	{Since $C_*>1$, the bootstrap set is non-empty on the initial boundary.}
	We now want to establish preliminary ``weak'' estimates. First we write, on the constant $v=V$ surface

	$$ |e^{\int_{-V}^{u}iq_0 A_u}\phi(u,V)-\phi_0(-V,V)| \leq 	|\int_{-V}^{u} e^{\int_{-V}^{u'}iq_0 A_u} D_u\phi(u',V)du'|	\leq   \left( \int_{-V}^{u} \Omega^2 r^{-2} (u',V)du' \right)^{\frac{1}{2}}\left( \int_{-V}^{u} \Omega^{-2} r^{2} |D_u \phi|^2 (u',V)du'\right)^{\frac{1}{2}},  $$
	
	where we used the fact that {$\partial_u\!\left(e^{\int_{-V}^{u}iq_0 A_u}\phi\right)=e^{\int_{-V}^{u}iq_0 A_u}D_u\phi$.}
	
	Then we take the limit $V \rightarrow +\infty$, we use the fact from the hypothesis  that $\phi_0$ tends to $0$ towards spatial infinity, together with \eqref{degenappendix} to get, for some $C'=C'(M,\rho)>0$:

		\[
		r^{\frac{1}{2}} |\phi|_{\mathcal{I}^+}(u) \leq \left( \int_{-\infty}^{u}  r^{2} |D_u \phi|^2 (u',V)\,du'\right)^{\frac{1}{2}} \leq C'\left[(\tilde{\mathcal{E}}_{\ep})^{\frac12}+Q_0^{\infty}\right].
		\]

	Hence we get that $\lim_{v \rightarrow +\infty} \phi(u,v) =0$.
	
	Now because we know that $\lim_{v \rightarrow +\infty} \phi(u,v) =0$, one can use Hardy's inequality under the form \eqref{Hardy2} (the proof is the same although the statement differs slightly) and \eqref{degenappendix} to get that for all $(u,v)  \in \{ r \geq R_0  \}$

		\begin{equation} \label{characfirst}
			r |\phi|^{2}(u,v) \leq \int_{v}^{+\infty} \frac{r^2 |D_v \phi|^2(u,v')}{\Omega^2}\,dv' \leq C(M,\rho)\,\tilde{\mathcal{E}}_{\ep}.
		\end{equation}

	We were able to claim the estimate in the whole region $\{ r \geq R_0  \}$ because it is true in  $\{ r \geq R_1  \}$ and $R_1$ is allowed to vary in the range $[R_0,+\infty]$

	Now using a method that has been made explicit already several times in section \ref{energyboundedsection}, we use the mean-value theorem with the Morawetz estimate \eqref{Morawetzext}: we find that there exists $R'_0 \in (0.8 R_0, 0.9 R_0)$ and $\bar{C}_0=\bar{C}_0(M,\rho)$ such that
	
	$$	\int_{0}^{t_{R_1,R_0'}(u)} \left(|\phi|^2 (R'_0,t)+ |D_{r^*}\phi|^2(R'_0,t) \right)dt \leq {\bar{C}_0 \left( \tilde{\mathcal{E}}_{\ep}+(Q_0^{\infty})^2 \right)},$$ this provided $r_+<0.8R_0$ and where $t_{R_1,R_0'}(u):=2v_{R_1}(u)-(R'_0)^*=2u+2R_1^*-(R'_0)^*$ is defined such that $(t_{R_1,R_0'}(u),(R'_0)^*) = \gamma_{R'_0} \cap \{ v= v_{R_1}(u)\}$.
	
	Therefore, using \eqref{chargeUEinstein}, \eqref{ChargeVEinstein} as $|\partial_t Q| \leq q_0 r^2 |\phi| |D_{r^*}\phi|$ we integrate and find that there exists 
	
	$C'_0=C'_0(M,\rho)>0$ so that for all $ 0 \leq t \leq t_{R_1,R_0'}(u)$
	
	\begin{equation*}
		|Q(t,R'_0)| \leq {Q_0^{\infty} + C'_0 \left( \tilde{\mathcal{E}}_{\ep}+(Q_0^{\infty})^2 \right)}
	\end{equation*}
	
	Then we integrate $\partial_uQ$ in $u$ towards $\gamma_{R'_0}$, using the red-shift estimate \eqref{RSext}. {The remaining integration is identical to the one in the proof of the red-shift estimate.}
	
	We find that in $\{ v_0(R'_0)  \leq  v \leq v_{R_1}(u)  \} \cap \{ r \leq R'_0 \}$, there exists $C''_0= C''_0(M,\rho)>0$ such that
	
	\begin{equation*}
		|Q(u,v)| \leq {Q_0^{\infty} + C''_0 \left( \tilde{\mathcal{E}}_{\ep}+(Q_0^{\infty})^2 \right)}.
	\end{equation*} 
	
	Using the same technique with \eqref{RSext} in $\{    v \leq v_0(R'_0) \}$ and then in the bounded region $\{ R'_0 \leq  r \leq  R_0\}$ we then get  that there exists $C=C(M,\rho)>0$ such that in the \underline{whole} region $\{ v \leq v_{R_1}(u)\} \cap \{ r \leq R_0    \}$
	
	\begin{equation} \label{localcharge}
		|Q(u,v)| \leq {Q_0^{\infty} + C \left( \mathcal{E}_{\ep}+(Q_0^{\infty})^2 \right)},
	\end{equation} 
	
	where we used the fact that $R_0$ depends only on $M$ and $\rho$.
	
	This should be thought of as a local in \underline{space} smallness propagation of the charge, for potentially large times. 
	
	We now want to ``globalise'' this result in space. For this, we need to control higher $r^p$ weighted energies, for any $p>1$. We will need to use the $r^p$ method, as developed in section \ref{decay}.

	We establish the $r^p$ estimate, in the very same way as in section \ref{decay} but this time on the domain  $ \Theta_u \cap \{r \geq R_0\}$.
	
	We are going to write two different $r^p$ estimates, according to $ u <u_0(R_0)$ or $u > u_0(R_0)$. The proof is however the same. 
	
	\begin{figure}  
		
		\begin{center}
			
			\includegraphics[width=107 mm, height=65 mm]{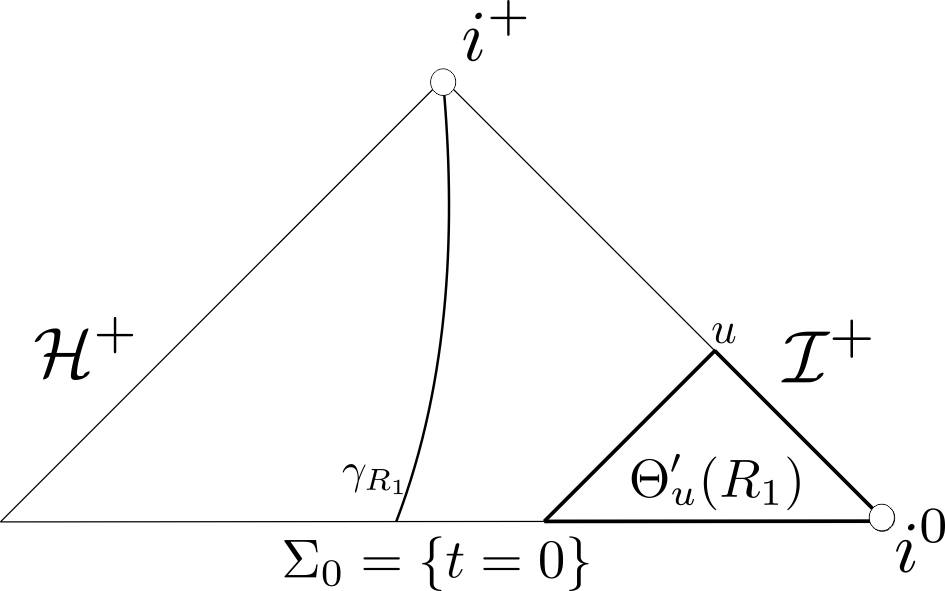}
			
		\end{center}
		
		\caption{Illustration of the domain $\Theta'_u(R_1)$ for $u <u_0(R_1)$, $R_0(M,\rho) \leq R_1 \leq R$}
		
		\label{Fig5}
		
	\end{figure}
	
	For this we define $\Theta'_{u}(R_0)=\Theta_u(R_0) \cap \{ r \geq R_0 \}$ if $u \geq u_0(R_0)$ and   $\Theta'_{u}(R_0)=\{ u' \leq u\}$ if $u \leq u_0(R_0)$, c.f.\ Figure \ref{Fig5}. We also define $\tilde{v}(u)=v_{R_0}(u)$ if $u \geq u_0(R_0)$ and  $\tilde{v}(u)=v_{0}(u)$ if $u \leq u_0(R_0)$. We can then write 
	
	\begin{equation*} \begin{split}
			\int_{\Theta'_{u}(R_0)} \left( pr^{p-1} \Omega^2 |D_v \psi|^2 + {\Omega^2(r)}\left(2M(3-p)-(4-p)\frac{2\rho^2}{r}\right)r^{p-4}|\psi|^2 \right) du dv  + \int_{\tilde{v}(u)}^{+\infty} r^{p} \ |D_v \psi|^2  (u,v)dv\\+ \int_{-\infty}^{u}\Omega^2\left(2M-\frac{2\rho^2}{r}\right)r^{p-3}|\psi|^2_{\mathcal{I}^+} (u)du = \int_{\Theta'_{u}(R_0)} 2q_0 Q \Omega^2 r^{p-2} \Im(\psi \overline{D_v \psi}) du dv + \int_{ \Sigma_0\cap \{ v \geq \tilde{v}(u)\}   }r^{p} \ |D_v \psi_0|^2(r^*)dr^*.  \end{split} \end{equation*} {To control the zero-order bulk term, we proceed as in the proof of Lemma~\ref{lemma.rp.first} and} we get 
	
	\begin{equation} \begin{split}
			\int_{\Theta'_{u}(R_0)}  pr^{p-1} \Omega^2 |D_v \psi|^2  du dv  + \int_{\tilde{v}(u)}^{+\infty} r^{p} \ |D_v \psi|^2  (u,v)dv \leq \int_{\Theta'_{u}(R_0)} [2q_0 |Q|+{P_0(r)]} \Omega^2 r^{p-2} {|\psi| |D_v \psi|}) du dv  + \mathcal{E}_p,\end{split} \end{equation}	{where $P_0(r) = O(r^{-1})$ as $r \rightarrow \infty$.}
	
	Now we use a variant of Hardy's inequality \eqref{Hardy5} for any $0\leq q<2$ under the form 
	
	\begin{equation*} 
		\left( \int_{\Theta'_{u}(R_0)}r^{q-3} \Omega^2  | \psi|^2 du dv\right)^{\frac{1}{2}} \leq               \frac{2}{(2-q)\Omega(R_0)} \left(\int_{\Theta'_{u}(R_0)}r^{q-1}  |D_v \psi|^2 du dv  \right)^{\frac{1}{2}}  +                       {\left( \frac{1}{2-q} \int_{-\infty}^{u} r^{q-2}|\psi|^2 (u',\tilde{v}(u'))\,du' \right)^{\frac{1}{2}}.}  
	\end{equation*}
	
	Now we are free to choose $1<p<2$ without loss of generality, and with $p$ as close to $1$ as needed. Taking $q=p$, this implies, by the hypothesis and using the Morawetz estimate \eqref{Morawetzext} that there exists a constant
	
	$C_0=C_0(M,\rho)>0$ such that 
	
	$$ \frac{1}{2-p} \int_{-\infty}^{u} r^{p-2}|\psi|^2 (u,\tilde{v}(u)) du \leq C_0 \cdot \tilde{ \mathcal{E}}_p.$$
	
	Combining this with the Cauchy-Schwarz inequality and the bootstrap assumption \eqref{boostrapQ}  we see ---like in section \ref{p<2} --- that for all $\eta>0$ small enough we have 
	
	\begin{equation} \begin{split}
			\left(p- \frac{4q_0e^+}{(2-p)\Omega(R_0)} -\frac{\eta}{2}\right)\int_{\Theta'_{u}(R_0)}  r^{p-1} \Omega^2 |D_v \psi|^2  du dv  + \int_{\tilde{v}(u)}^{+\infty} r^{p} \ |D_v \psi|^2  (u,v)dv \leq \left[ 1+ \frac{C_0}{2\eta} \right] \cdot \tilde{\mathcal{E}}_p, \end{split} \end{equation}		
	
	so if say $4q_0e^+ < \frac{ (2-p)p \cdot \Omega(R_0)}{2}$ --- which can be assumed, taking $\delta$ small enough --- we proved that for some $C'_0=C'_0(M,\rho)>0$ and for all $u \in \mathbb{R}$: 
	
	\begin{equation}  \label{rpannex}\begin{split}
			\int_{\tilde{v}(u)}^{+\infty} r^{p} \ |D_v \psi|^2  (u,v)dv \leq C'_0 \cdot \tilde{\mathcal{E}}_p.\end{split} \end{equation}
	
	Now we re-write \eqref{ChargeVEinstein} as $|\partial_v Q| \leq q_0 |\psi| |D_v \psi|$ and we integrate towards $\gamma_{R_0}$ if $u \geq u_{0}(R_0)$ and towards $ \Sigma_0$ if $u \leq u_{0}(R_0)$ , using Cauchy-Schwarz: 
	
	\begin{equation}  \label{chargeannex}|Q(u,v)-Q(u,\tilde{v}(u))| \leq q_0 \left( \int_{\tilde{v}(u)}^{+\infty} r^{p} \ |D_v \psi|^2  (u,v)dv \right)^{ \frac{1}{2}} \left(\int_{\tilde{v}(u)}^{+\infty} r^{-p} \ |\psi|^2  (u,v)dv  \right)^{ \frac{1}{2}}. \end{equation}
	
	We now use the ``boundary term'' version of Hardy's inequality \eqref{Hardy5} and we get that there exists $C_1=C_1(M,\rho)>0$ such that 
	
	\begin{equation} \label{Hardyannex}\left(\int_{\tilde{v}(u)}^{+\infty} r^{-p} \ |\psi|^2  (u,v)dv  \right)^{ \frac{1}{2}} \leq C_1 \cdot \left[  \left( \int_{\tilde{v}(u)}^{+\infty} r^{p} \ |D_v \psi|^2  (u,v)dv \right)^{ \frac{1}{2}} + [ r(u,\tilde{v}(u))]^{ \frac{1-p}{2}} |\psi|(u,\tilde{v}(u)) \right],\end{equation}
	
	where we used the fact that $2-p<p$ because $p>1$.
	
	Now, to estimate $\psi_0$ we can use the Cauchy-Schwarz under the following form, taking advantage of the fact that $p>1$: for all $r \geq R_0$
	
	$$ | \psi_0(r)- \psi_0(R_0) | \leq \int_{ (R_0)^*}^{r^*} |D_{r^*} \psi_0|( (r')^*) d(r')^*\leq \sqrt{2} \left(  \int_{ (R_0)^*}^{r^*}  (r')^{-p} d(r')^* \right)^{\frac{1}{2}}  (\mathcal{E}_p)^{\frac{1}{2}},$$
	
	where we used the fact that$ \int_{\Sigma_0} r^p |D_{r^*} \psi_0|^2 dr^* \leq 2 \mathcal{E}_p$.
	
	Using also \eqref{characfirst} and remembering that $R_0=R_0(M,\rho)$ depends only on $M$ and $\rho$ we can then write, for some $\tilde{C}=\tilde{C}(M,\rho)>0$: 
	
	\begin{equation*}
		|\psi_0(r)| \leq \tilde{C} \cdot (\tilde{ \mathcal{E} }_p)^{\frac{1}{2}}.
	\end{equation*}
	
	Coupled with  \eqref{characfirst} applied on $\gamma_{R_0} $, there exists also $\bar{C}=\bar{C}(M,\rho)>0$ such that for all $u \in \mathbb{R}$: 
	
	\begin{equation} \label{psiannex}
		|\psi|(u,\tilde{v}(u)) \leq \bar{C} \cdot (\tilde{ \mathcal{E} }_p)^{\frac{1}{2}}.
	\end{equation}
	
	Now combining \eqref{rpannex}, \eqref{chargeannex},	\eqref{Hardyannex} and \eqref{psiannex} we get that there exists $C_2=C_2(M,\rho)>0$ such that

	\begin{equation*} |Q(u,v)-Q(u,\tilde{v}(u))| \leq C_2 \cdot   \tilde{ \mathcal{E} }_p. \end{equation*}

	Now combining with 	\eqref{localcharge}, we see that for all $ u \in \mathbb{R}$: 
	
	\begin{equation} |Q(u,v)| \leq Q_0^{\infty} + C \cdot (Q_0^{\infty})^2 + C'_2 \cdot   \tilde{ \mathcal{E} }_p < \delta \cdot (1+C'_2+ C\cdot \delta), \end{equation}
	
	where we defined $C'_2=C_2+C$.

	{Choose $C_*>2\max\{1,C'_2\}$. Since $Q_0^{\infty}+\tilde{\mathcal E}_p<\delta$, shrinking $\delta$ further makes $C(Q_0^{\infty})^2\leq\frac12 C_*\bigl(Q_0^{\infty}+\tilde{\mathcal E}_p\bigr)$. The preceding bound is then strictly smaller than $e^+$, so the bootstrap \eqref{boostrapQ} closes. In particular, $e^+\leq C_*(M,\rho)\delta$ and hence $\ep\leq C(M,\rho)|q_0|\delta$.}
	
	This proves the charge bound claimed in the statement of the Proposition.

		It remains to prove the non-degenerate energy bound with the revised right-hand side. We use the energy identity \eqref{energiIDappendix}. The charge terms on the future boundary of $\Theta_u(R_1)$ are decomposed exactly as in Steps~\ref{step1S}--\ref{step4S}. The charge terms on $\Sigma_0$ are treated by the same integrations by parts; every resulting scalar boundary term is controlled by $\tilde{\mathcal{E}}_{\ep}$, and the constant-$r$ charge difference is controlled by \eqref{Morawetzext}.
		
		Combining these estimates with the red-shift component \eqref{RSext} of Lemma~\ref{combinedext} gives
		\[
		E_{R_1}(u)\leq C\tilde{\mathcal{E}}_{\ep}+C|q_0|e^+E_{R_1}(u).
		\]
		{Since $e^+=C_*\bigl(Q_0^{\infty}+\tilde{\mathcal E}_p\bigr)$, the data smallness constant is chosen so that $C|q_0|e^+<\frac12$; the last term is therefore absorbed into the left-hand side.} Thus
		\[
		E_{R_1}(u)\leq C\tilde{\mathcal E}_{\ep}.
		\]
		Estimate \eqref{Charac4} follows from \eqref{RSext}, while \eqref{Charac2} follows from \eqref{Morawetzext}. This completes the proof of the proposition.

\end{proof}

We now turn to the proof of Proposition \ref{CharacProp2}: this time we assume already energy boundedness and the Morawetz estimate but \textbf{not arbitrary charge smallness}. The method of proof is very similar to that of Proposition \ref{characProp1}.

\begin{prop} Suppose that there exists $1<p<2$ such that $\tilde{\mathcal{E}}_p < \infty$.
	
	It follows that there exists $e_0 \in \mathbb{R}$ such that 
	
	$$ \lim_{ r \rightarrow + \infty} Q_0(r) = e_0.$$
	
	Without loss of generality one can assume that $1<p < 1+\sqrt{1-4q_0|e_0|}$.

	Assume also  that $\lim_{r \rightarrow +\infty} \phi_0(r)=0$.
	
	Now assume \eqref{Charac1}, \eqref{Charac4} for $R_1=R$: there exists $\bar{C}=\bar{C}(M,\rho)>0$ such that for all $u \geq u_0(R)$ and for all $ v \leq v_{R}(u)$: 
	
	\begin{equation}  \label{Charac2.1}  {E(u)=E_R(u) \leq \bar{C}\,\tilde{\mathcal{E}}_{\ep}.}	\end{equation}
	
	\begin{equation} \label{Charac2.2}
		{\int_{\tilde{u}(v)}^{+\infty}  \frac{ r^2|D_u \phi|^2}{\Omega^2} (u',v)\,du'  \leq \bar{C}\,\tilde{\mathcal{E}}_{\ep},}	\end{equation}
	
	where $\tilde{u}(v)=u_0(v)$ if $v \leq  v_0(R)$ and $\tilde{u}(v)=u_R(v)$ if $ \ v_0(R) \leq v  \leq v_R(u)$.
	
	Assume also \eqref{Charac2}: there exists $\bar{R}_0=\bar{R}_0(M,\rho)>r_+$ such that for all $\bar{R}_1>\bar{R}_0$, 
	
	there exists $\bar{C}_1=  \bar{C}_1(\bar{R}_1, M,\rho)>0$ such that 
	
	\begin{equation} \label{CharacMora}
		{\int_{\{u'\leq u\}\cap\{r\leq\bar R_1\}}\left(|D_t\phi|^2+|D_{r^*}\phi|^2+|\phi|^2\right)\Omega^2\,du\,dv\leq\bar C_1\,\tilde{\mathcal E}_{\ep}.}
	\end{equation}

	Make also the following smallness hypothesis: for some $\delta>0$: 	
	
	$$  \tilde{\mathcal{E}}_p < \delta,$$
	$$q_0 |e_0| < \frac{1}{4}.$$
	
	There exists $\delta_0=\delta_0(e_0,M,\rho)>0$ and $C=C(M,\rho)>0$ such that if $\delta<\delta_0$ 	then for all $(u,v)$ in the space-time: 
	\begin{equation}
		|Q(u,v)-e_0|	\leq  C \cdot   \tilde{\mathcal{E}}_p ,
	\end{equation}
	\begin{equation}
		q_0|Q|(u,v)	<  \frac{1}{4}.
	\end{equation}
	
	Moreover,  there 
	exists       
	$C'=C'(e_0,p,M,\rho)>0$  such that 
	for all $u \geq u_0(R)$:
	
	\begin{equation} \label{rpboundednesscharac}
		E_{p}[\psi](u) \leq C' \cdot \tilde{\mathcal{E}}_p.
	\end{equation}

\end{prop}

\begin{proof} First take $R_0=R_0(M,\rho)$. Using Cauchy-Schwarz and the Maxwell equation under the form 
	
	$|\partial_{r^*}Q| \leq q_0 |\psi| |D_t \psi|$, we can show that there exists $\bar{C}_0=\bar{C}_0(M,\rho)>0$ such that for all $r \geq R_0$,
	
	\begin{equation*}
		|Q_0(r)-e_0| \leq \bar{C}_0 \cdot  \tilde{\mathcal{E}}_p.
	\end{equation*}
	
	Similarly on $\{ r_+ \leq r \leq R_0\}$ one can prove that there exists $\tilde{C}_0=\tilde{C}_0(M,\rho)>0$ such that for all $r \leq R_0$,
	
	\begin{equation*}
		{|Q_0(r)-Q_0(R_0)| \leq \tilde{C}_0\,\tilde{\mathcal{E}}_p,}
	\end{equation*}
	where we used Cauchy-Schwarz and the Maxwell equation under the form $|\partial_{r^*}Q| \leq \frac{q_0  R_0^2}{2} \left( |\phi| |D_u \phi|+ |\phi| |D_v \phi| \right)$.

	This gives, on the whole $\Sigma_0$: 
	
	\begin{equation} \label{initchargesigma}
		|Q_0(r)-e_0 | \leq (\tilde{C}_0+\bar{C}_0 ) \cdot  \tilde{\mathcal{E}}_p.
	\end{equation}		
	
	Now assume that $R_0>\bar{R_0}(M,\rho)$ so that estimate \eqref{CharacMora} is valid.
	
	Now, using the Morawetz estimate \eqref{CharacMora} in a way that was explained numerous times in this paper, one can find $r_+<\tilde{R_0}=\tilde{R_0}(M,\rho)<\bar{R_0}$ and $D_0=D_0(M,\rho)>0$ such that for all $t \geq 0$
	
	\begin{equation} \label{constantrQ}
		{|Q(t,\tilde{R}_0)-Q_0(\tilde{R}_0)|\leq D_0\,\mathcal{E}_{\ep}\leq C\,\tilde{\mathcal{E}}_p,}
	\end{equation}
	
	where we used Cauchy-Schwarz and the Maxwell equation under the form $|\partial_{t}Q| \leq q_0  R_0^2  |\phi| |D_{r^*} \phi|.$
	
	Then using \eqref{Charac2.2} and \eqref{chargeUEinstein}, one can prove that there exists $D'_0=D'_0(M,\rho)>0$ such that for all $(u,v) \in \{r \leq R \}$,
	
	\begin{equation*} 
		|Q(u,v)-Q( \hat{u}(v),v)| \leq D'_0 \cdot  \tilde{\mathcal{E}}_p,
	\end{equation*}
	where $\hat{u}(v) = u_0(v)$ if $v \leq v_0(R)$ and $\hat{u}(v) = u_{\tilde{R_0}}(v)$ if $v \geq v_0(R)$.
	
	{Combining this with \eqref{initchargesigma} and \eqref{constantrQ} proves} that there exists $C_0=C_0(M,\rho)>0$ such that for all 
	
	$(u,v) \in \{r \leq R \}$,
	
	\begin{equation} \label{initcharge}
		|Q(u,v)-e_0| \leq C_0 \cdot  \tilde{\mathcal{E}}_p.
	\end{equation}
	
	In the same way as in Proposition \ref{characProp1}, we can derive the estimate for $R_0=R_0(M,\rho)$ large enough, for all $R_1 \geq R_0$ and for all $u \in \mathbb{R}$: 
	
	\begin{equation} \begin{split}
			\int_{\Theta'_{u}(R_1)}  pr^{p-1} \Omega^2 |D_v \psi|^2  du dv  + \int_{\tilde{v}(u)}^{+\infty} r^{p} \ |D_v \psi|^2  (u,v)dv \leq \int_{\Theta'_{u}(R_1)} 2q_0 Q \Omega^2 r^{p-2} \Im(\psi \overline{D_v \psi}) du dv  + \tilde{\mathcal{E}}_p, \end{split} \end{equation}
	where all the notations are introduced in the proof of Proposition \ref{characProp1}. Note that this time we consider $\Theta'_{u}(R_1)$  for any $R_1 \geq R_0$ to be chosen later, in contrast to the proof of Proposition \ref{characProp1} where only $\Theta'_{u}(R_0)$ was considered.
	
	Then we bootstrap for some $\bar{e}=|e_0| +2 \epsilon$, $\epsilon>0$ and in $\Theta'_u(R_1)$: 
	
	\begin{equation} \label{chargeboostrap2}
		|Q| < \bar{e}.
	\end{equation}
	
	Because $Q_0 \rightarrow e_0$ towards spatial infinity, it is clear that the set of points for which this bootstrap is verified is non-empty.
	
	For $\epsilon$ small enough, one can assume that $q_0 \cdot ( |e_0| +2 \epsilon)< \frac{1}{4} $.
	
	By assumption , {one can choose $1<p<1+\sqrt{1-4q_0\bar e}$} such that $\tilde{\mathcal{E}}_p < \infty$.
	
	Using this, we find constants $C_{R_1}=C_{R_1}(R_1,M,\rho)>0$ and $D=D(M,\rho)$ such that for all $\eta>0$  small enough and for all  $u \in \mathbb{R}$:
	
	\begin{equation} \label{rpcharac}\begin{split}
			\int_{\tilde{v}(u)}^{+\infty} r^{p} \ |D_v \psi|^2  (u,v) dv \leq \frac{D \cdot \tilde{\mathcal{E}}_p}{\eta \cdot \left(p- \frac{4q_0\bar{e}}{(2-p)\Omega(R_1)} -\eta\right)} , \end{split} \end{equation} 
	\begin{equation} \begin{split}
			{\int_{\tilde{v}(u)}^{+\infty} r^{p-2}|\psi|^2(u,v)\,dv\leq} \frac{C_{R_1} \cdot \tilde{\mathcal{E}}_p}{\eta \cdot \left(p- \frac{4q_0\bar{e}}{(2-p)\Omega(R_1)} -\eta\right)} , \end{split} \end{equation}
	where for the second estimate, we used a Hardy inequality coupled with the Morawetz estimate \eqref{CharacMora}.
	
	Now we take $R_1=R_1(e_0,M,\rho)$  large enough so that $  4q_0\bar{e} < \Omega(R_1)$.
	Then we can take temporarily $p>1$ sufficiently close to $1$  and $\eta$ small enough so that $	p- \frac{4q_0\bar{e}}{(2-p)\Omega(R_1)} -\eta >0$.
	
	We then find using \eqref{ChargeVEinstein} that there exists $C=C(e_0,M,\rho)>0$ such that on $\{ r \geq R_1 \}$
	
	\begin{equation}
		|Q(u,v) -Q(u, \bar{v}(u))| \leq C  \cdot \tilde{ \mathcal{E}}_p,
	\end{equation}
	
	where $\bar{v}(u)=v_{R_1}(u)$ if $u \geq u_0(R_1) $ and $ \bar{v}(u)=v_0(u)$ if $u \leq u_0(R_1) $. Then with \eqref{initcharge} and provided that $R_1 < R$ --- we actually proved that for some $C'=C'(e_0,M,\rho)>0$ and on the whole $\Theta'_u(R_1)$:
	
	\begin{equation}
		|Q(u,v) -e_0| \leq C'  \cdot \tilde{ \mathcal{E}}_p < C' \cdot \delta.
	\end{equation}
	
	Then it suffices to take $C' \cdot \delta < \epsilon$ to retrieve bootstrap \eqref{chargeboostrap2}. This evidently gives the first two claims on the whole space-time.

	{Now return to a general $1<p<1+\sqrt{1-4q_0|e_0|}$} and notice that \eqref{rpcharac} can be written as, for all $\eta>0$:

	\begin{equation} \begin{split}
			\int_{\tilde{v}(u)}^{+\infty} r^{p} \ |D_v \psi|^2  (u,v) dv \leq \frac{D \cdot \tilde{\mathcal{E}}_p}{\eta \cdot \left(p- \frac{4q_0|e_0|+\eta}{(2-p)} -\eta\right)} , \end{split} \end{equation} 
	
	for $\delta$ small enough and for the choice $R_1=R$, for $R$ large enough. This gives directly \eqref{rpboundednesscharac} and concludes the proof of the proposition.

\end{proof}

Now we turn to the	proof of Proposition \ref{propagationdecay}.

\begin{prop} In the conditions of Proposition \ref{CharacProp2}, assume moreover that there exists $\omega > 0$ and $ C_0>0$ such that
	
	$$ r|D_v \psi_0| + |\psi_0| \leq C_0 \cdot r^{-\omega}.$$
	
	Then in the following cases \begin{itemize}
		\item $\omega=1+\theta$ with $\theta>\frac{q_0|e_0|}{4}$
		\item $\omega=\frac{1}{2}+\beta$ with $\beta \in ( - \frac{\sqrt{1-4q_0|e_0|}}{2},\frac{\sqrt{1-4q_0|e_0|}}{2})$, if $q_0|e_0| < \frac{1}{4}$,
	\end{itemize}
	
	there exists $\delta=\delta(e_0,\omega,M,\rho)>0$ and $R_{low}=R_{low}(\omega,e_0,M,\rho)>r_+$ such that if $\tilde{\mathcal{E}}_p<\delta$ and $R>R_{low}$ then the decay is propagated: there exists $C'_0=C'_0(C_0,\omega,R,M,\rho,e_0)>0$ such that for all $u \leq u_0(R)$: 
	
	$$ |D_v \psi|(u,v) \leq C'_0 \cdot r^{ -1-\omega'},$$
	$$ | \psi|(u,v) \leq C'_0 \cdot |u|^{-\omega},$$
	
	where $\omega' = \min \{\omega,1\}$.

	In that case, for every $0<p<2\omega'+1$, we have the finiteness of the $r^p$ weighted energy on $\mathcal{V}_{u_0(R)}$:
	
	$$ E_p[\psi](u_0(R)) < \infty.$$

\end{prop}

\begin{proof}
	
	We start the proof in a region $ \{u \leq U_0 \}	$ for $|U_0|$ large enough to be chosen later.
	
	We are going to use the notations $u_0(v), v_0(u), r_0(u), r_0(v)$, c.f.\ section \ref{foliations} for a definition.
	
	For some $B>0$ large enough to be chosen appropriately later, we bootstrap the following in  $ \{u \leq U_0 \}	$
	
	$$ |\psi| \leq B |u|^{-\omega}.$$
	
	Notice that with the assumptions and the fact that $r_0(u) \sim 2|u|$ when $u \rightarrow -\infty$, the set of points for which the bootstrap is verified is non-empty, for $B$ large enough.
	
	We also denote $Q^+ = \sup_{  \{u \leq U_0 \}} |Q|$. By Proposition \ref{CharacProp2}, $Q^+$ can be taken arbitrarily close to $|e_0|$ or $|e|$ for $\delta$ appropriately small.
	
	Then we use \eqref{wavev} under the form
	
	$$ |D_u D_v \psi| \leq \frac{B|u|^{-\omega} }{r^2} (q_0 Q^+ + |2r \cdot K(r)|).$$
	
	Now take $\epsilon>0$.  $ \{u \leq U_0 \} \subset \{ r \geq r_0(U_0)\}	$ so by taking $|U_0|$ large enough, one can assume that $|2r \cdot K(r)|<\epsilon$ since this quantity tends to $0$ as $r$ tends to $+\infty$.
	
	There are two cases: either $\omega >1$ or $\omega \leq 1$. We start by $\omega >  1$: we can integrate in $u$ the inequality above and get 
	
	\begin{equation}  \label{propdecay} |D_v \psi|(u,v) \leq C_0 \cdot (r_0(v) )^{-1-\omega} +\frac{B|u|^{1-\omega} }{(\omega-1)r^2} (q_0 Q^+ + \epsilon). \end{equation}
	
	Now we want to integrate this in $v$: first notice that $\frac{d (r_0(v))}{dv} = 2 \Omega^2(-v,v) \geq 2 \Omega^2 (r_0(U_0)) \geq \frac{2}{(1+\epsilon)}$ if $|U_0|$ is large enough. Therefore

		\[
		\int_{v_0(u)}^{v} (r_0(v'))^{-1-\omega}\,dv' \leq \frac{1+\epsilon}{2\omega}\left[(r_0(u))^{-\omega}-(r_0(v))^{-\omega}\right] \leq \frac{1+\epsilon}{2\omega}(r_0(u))^{-\omega}.
		\]
	
	$$ \int_{v_0(u)}^{v} (r_0(v') )^{-2} dv'  \leq \frac{(1+\epsilon)}{2} (r_0(u))^{-1}, $$
	
	after noticing that $r_0(v_0(u))=r_0(u)$. Hence we have
	
	$$ |\psi(u,v)| \leq C_0 \cdot (1+\frac{(1+\epsilon)}{2\omega}) \cdot  ( r_0(u))^{-\omega} + \frac{B(q_0 Q^+ + \epsilon) (1+\epsilon) }{2(\omega-1)}  (r_0(u))^{-1}|u|^{1-\omega} .$$
	
	Since $r_0(u) \sim 2 |u|$ when $u \rightarrow -\infty$, it is clear that for $|U_0|$ large enough, $\epsilon$ small enough and $B$ large enough, the bootstrap is retrieved if 
	
	$$ \omega > 1+ \frac{q_0 Q^+}{4},$$
	
	or equivalently if $\delta$ is small enough,
	
	$$ \omega > 1+ \frac{q_0 |e_0|}{4}.$$
	
	Now we turn to $\omega \leq1$. We ignore the case $\omega=1$ and consider $\omega <1$. Integrating in $u$ we get this time
	
	$$ |D_v \psi|(u,v) \leq C_0 \cdot (r_0(v) )^{-1-\omega} +\frac{B \cdot v^{1-\omega} }{(1-\omega)r^2} (q_0 Q^+ + \epsilon),$$
	
	where we used that $|u_0(v)|=v$.
	
	Now, taking $|U_0|$ large enough, one can assume that everywhere on $\{ u \leq U_0 \}$: $r^{-2} \leq \frac{(r^*)^{-2}}{1-\epsilon} \leq \frac{v^{-2}}{1-\epsilon}$, since $u \leq0$. Using this, we get
	
	$$ \int_{v_0(u)}^v \frac{ (v')^{1-\omega}}{r^2(u,v')} dv' \leq \frac{ |u|^{-\omega}}{(1-\epsilon) \omega},$$
	
	hence we have

		\[
		|\psi(u,v)| \leq C_0 \cdot \left(1+\frac{1+\epsilon}{2\omega}\right) (r_0(u))^{-\omega} + \frac{B(q_0 Q^+ + \epsilon)}{(1-\omega)\omega(1-\epsilon)} |u|^{-\omega}.
		\]

	We now see that the bootstrap is retrieved on the condition: 
	
	$$ q_0 Q^+ < (1-\omega) \omega ,$$
	
	Otherwise said
	
	$$ \frac{ 1 - \sqrt{1-4q_0Q^+}}{2} < \omega < \frac{ 1 + \sqrt{1-4q_0Q^+}}{2} ,$$
	
	or equivalently if $\delta$ is small enough: 
	
	$$ \frac{ 1 - \sqrt{1-4q_0|e_0|}}{2} < \omega < \frac{ 1 + \sqrt{1-4q_0|e_0|}}{2} .$$
	
	This proves the proposition in the region $\{ u \leq U_0 \}$, for $|U_0|$ large enough with respect to $\omega$ and $e_0$ and \textbf{independently of} $R$.
	
	Now it is enough to take $R$ large enough so that $|u_0(R)|>|U_0|$ and the result is proven, if we accept that the constants now depend on $R$.

\end{proof}

Then we turn to our last step, the proof of Proposition \ref{lastCharac}.

\begin{prop} 
	Suppose that for all $0 \leq p<2+ \sqrt{1-4q_0|e_0|}$, $\tilde{\mathcal{E}}_{p}<\infty$.
	
	We also assume the other hypotheses of Theorem \ref{decaytheorem}.

	Then for all $0 \leq p<2+ \sqrt{1-4q_0|e_0|}$, there  exists $\delta_p=\delta_p(e_0,p,M,\rho)>0$,  such that if $\delta<\delta_p$ then  for all $u \leq u_0(R)$: 
	
	$$ E_{p}[\psi](u)<\infty.$$
\end{prop}	

\begin{proof} We generalize the $r^p$ weighted estimate and use a $r^p u^{s}$  weighted estimate, for $s>0$.
	
	We are going to state this identity on the neighbourhood of spatial infinity $\{ u \leq u_0(R)\}$ where $R$ is large enough so that $u_0(R)<0$. As a consequence, $|u|=-u$ in this region. 
	
	We multiply \eqref{wavev} by $(-u)^s r^p \overline{D_v \psi}$, take the real part, integrate by parts and get, for all $u \leq u_0(R)$
	
	\begin{equation*} \begin{split}
			\int_{  \{u' \leq u\}}  \left[p r^{p-1} |u'|^s + s r^p |u'|^{s-1} \right] |D_v\psi|^2 du' dv + |u|^s \int_{v_0(u)}^{+\infty} r^p |D_v\psi|^2(u,v) dv \\ \leq  \left[1+P_0(r) \right]\left[ \int_{\Sigma_0 \cap  \{v \geq v_0(u)\} } |u_0(v)|^s (r_0(v))^p |D_v\psi|^2(u_0(v),v) dv+\int_{  \{u' \leq u\}}  2q_0 Q r^{p-2} |u'|^s  \Im(\psi \overline{D_v \psi} )du' dv \right], \end{split}
	\end{equation*} where $P_0(r)$ is a polynomial in $r$ that behaves like $O(r^{-1})$ as $r$ tends to $+\infty$ and with coefficients depending only on $M$ and $\rho$.
	
	Now because $|u_0(v)| \sim \frac{r_0(v)}{2}$ as $v \rightarrow +\infty$, for $R$ large enough we can say that $$\int_{\Sigma_0 \cap  \{v \geq v_0(u)\} } |u_0(v)|^s (r_0(v))^p |D_v\psi|^2(u_0(v),v) dv \leq \mathcal{E}_{p+s}.$$
	
	Denoting $Q^+ = \sup_{ u' \leq u_0(R)} |Q|$: note that $|Q^+ -|e_0|| \lesssim \delta$ by the last proposition.
	
	With this, we find that for all $\eta>0$, taking $R$ large enough so that $|1+P_0(r)| < (1+\eta)$: 
	
	\begin{equation} \label{rpgenerallastappendix} \begin{split}
			\int_{  \{u' \leq u\}}  \left[p r^{p-1} |u'|^s + s r^p |u'|^{s-1} \right] |D_v\psi|^2 du' dv + |u|^s \int_{v_0(u)}^{+\infty} r^p |D_v\psi|^2(u,v) dv  \\ \leq (1+\eta) \cdot 	 \left[\mathcal{E}_{p+s}+ 2q_0 Q^+ \int_{  \{u' \leq u\}}   r^{p-2} |u'|^s  |\psi| |D_v \psi|  du' dv \right].\end{split}
	\end{equation} 
	
	As it was seen in section \ref{p<2}, if $1-\sqrt{1-4q_0Q^+}<p'< 1+\sqrt{1-4q_0Q^+}$ the interaction term can be absorbed inside the bulk term using Hardy's inequality because the presence of the $u$ weight does not change anything. Therefore for all $s' \geq0$, there exists a \footnote{Since we just need a finiteness statement, the parameters on which $D$ depends do not matter so we do not specify them.} constant $D>0$ such that 
	
	\begin{equation}  \label{rplowerorderappendix}
		\int_{  \{u' \leq u\}}  \left[ r^{p'-1} |u'|^{s'} +  r^{p'} |u'|^{s'-1} \right] |D_v\psi|^2 du' dv + |u|^{s'}\int_{v_0(u)}^{+\infty} r^{p'} |D_v\psi|^2(u,v) dv   \leq D  \cdot 	 \tilde{\mathcal{E}}_{p'+s'}.
	\end{equation} 
	
	Now we take  $2-\sqrt{1-4q_0|e_0|}<p<2+\sqrt{1-4q_0 |e_0|}$ and come back to \eqref{rpgenerallastappendix}. If $\delta$ is small enough ---depending on $p$ --- one can assume that actually $2-\sqrt{1-4q_0Q^+}<p<2+\sqrt{1-4q_0Q^+}$.
	
	Now applying Cauchy-Schwarz we find that 
	
	$$	\int_{  \{u' \leq u\}}   r^{p-2} |u'|^s  |\psi| |D_v \psi|  du' dv \leq \left(\int_{  \{u' \leq u\}}   r^{p-4} |u'|^{s+1} |\psi|^2   du' dv \right)^{\frac{1}{2}} \left(\int_{  \{u' \leq u\}}   r^{p} |u'|^{s-1}  |D_v \psi|^2  du' dv \right)^{\frac{1}{2}}. $$
	
	{Then we use the following Hardy inequality, analogous to \eqref{Hardy5}} under the form: for all $u' \leq u$: 
	
	$$\int_{v_0(u')}^{+\infty} r^{p-4} |\psi|^2(u',v) dv \leq  \frac{2}{ (3-p) \Omega^2(R)} (r_0(u'))^{p-3} |\psi_0|^2(u',v_0(u'))+\frac{4}{ (3-p)^2 \Omega^4(R)}  \int_{v_0(u')}^{+\infty} r^{p-2} |D_v\psi|^2(u',v) dv.$$
	
	Using Fubini's theorem, one can now prove that for \footnote{Since we just need a finiteness statement, the parameters on which $D'$ depends do not matter so we do not specify them.} some $D'>0$: 
	
	\begin{equation*}
		\int_{  \{u' \leq u\}}   r^{p-4} |u'|^{s+1} |\psi|^2   du' dv \leq D' \cdot \left[ \mathcal{E}_{p+s}+\int_{  \{u' \leq u\}}   r^{p-2} |u'|^{s+1} |D_v\psi|^2   du' dv \right],
	\end{equation*}
	where we controlled the integral of $|u|^{s+1}(r_0(u'))^{p-3} |\psi_0|^2(u',v_0(u'))$ by the zero order term of $\mathcal{E}_{p+s}$.

	Combining with \eqref{rplowerorderappendix} for $p'=p-1$ and $s'=s+1$ we finally get 
	
	\begin{equation}
		\int_{  \{u' \leq u\}}   r^{p-4} |u'|^{s+1} |\psi|^2   du' dv \leq D' \cdot (1+D) \cdot  \tilde{\mathcal{E}}_{p+s}.
	\end{equation}

		Let $D''=2|q_0|Q^+(1+\eta)\sqrt{D'(1+D)}$. Returning to \eqref{rpgenerallastappendix}, we obtain
		\begin{equation*}
			\begin{split}
				&\int_{\{u'\leq u\}}\left[pr^{p-1}|u'|^s+s r^p|u'|^{s-1}\right]|D_v\psi|^2\,du'\,dv
				+|u|^s\int_{v_0(u)}^{+\infty}r^p|D_v\psi|^2(u,v)\,dv\\
				&\leq C\,\tilde{\mathcal E}_{p+s}
				+D''\bigl(\tilde{\mathcal E}_{p+s}\bigr)^{1/2}
				\left(\int_{\{u'\leq u\}}r^p|u'|^{s-1}|D_v\psi|^2\,du'\,dv\right)^{1/2}.
			\end{split}
		\end{equation*}
		For $s>0$, Young's inequality $ab\leq\varepsilon a^2+(4\varepsilon)^{-1}b^2$, with $\varepsilon<s/2$, absorbs the last factor into the $s r^p|u'|^{s-1}|D_v\psi|^2$ term on the left. Consequently $E_p[\psi](u)<\infty$ whenever $\tilde{\mathcal E}_{p+s}<\infty$, which proves the proposition.
	
\end{proof}

\section{Useful computations for the vector field method} \label{appendix}

This appendix carries out explicitly a few computations to apply the vector field method to the case of Maxwell-Charged-Scalar-Field. For the linear wave equation, the traditional vector field method proceeds as follows: we can construct a quadratic quantity $\mathbb{T}_{\mu \nu}^{Wave}(\phi)= \partial_{\mu}\phi \partial_{\nu}\phi -\frac{1}{2}(g^{\alpha \beta} \partial_{\alpha}\phi \partial_{\beta}\phi )g_{\mu \nu}$. We then see that the wave equation corresponds to  a conservation law:

$$ \nabla^{\mu}  \mathbb{T}_{\mu \nu}^{Wave} = 0 \Longleftrightarrow \Box \phi=0.$$

Then for any solution of the wave equation $\phi$ and for a well chosen vector field $X$ we construct the current $J_{\mu}^{X} = \mathbb{T}_{\mu \nu}^{Wave}(\phi) X^{\nu}$ and we integrate $\nabla^{\mu}J_{\mu}^{X} = \mathbb{T}_{\mu \nu}^{Wave} \nabla^{ (\mu} X^{\nu)} $ on a space-time domain.

Making use of the divergence theorem, we see that a bulk term, namely the integral on a space-time domain of $\mathbb{T}_{\mu \nu}^{Wave} \nabla^{ (\mu} X^{\nu)}$, equals some boundary terms involving the current $J_{\mu}^{X}$.

Notice that if $X$ is a Killing vector field, then $\nabla^{ (\mu} X^{\nu)}=0$ and the identity only includes boundary terms. \\

Compared to the classical vector field method, a major difference --- in the case we consider in this paper --- is the presence of a Maxwell stress-energy tensor that is coupled to the scalar field's. This still gives rise of a conservation law that couples the scalar field and the charge.
While compactly supported scalar fields on asymptotically flat space-time decay, this is not the case for charges on black hole space-times. 

Therefore, in many cases we do not apply this conservation law directly, except in subsection \ref{energysection}. Instead, we treat the Maxwell term as an error term, that can be controlled by the energy of the scalar field. To put it otherwise, instead of looking at charges ---that do not decay--- we look at the fluctuation of these charges, that do enjoy time decay estimates. To see the main related computations, c.f.\ subsection \ref{appendixvectorfield}.

\subsection{Stress-Energy momentum tensor} \label{T}

For a spherically symmetric scalar field $\phi$ and 2-form $F$ we define the stress energy momentum tensor of the scalar field $\mathbb{T}_{ \mu \nu}^{SF}(\phi) $ and the one of the Maxwell field $\mathbb{T}_{ \mu \nu}^{EM}(F) $. Notice that the Maxwell-Charged-Scalar-Field equations \eqref{Maxwell}, \eqref{CSF} imply the following conservation law: 

$$ \nabla^{\mu} ( \mathbb{T}_{\mu \nu}^{SF} + \mathbb{T}_{\mu \nu}^{EM}) = 0,$$ where $\mathbb{T}_{ \mu \nu}^{SF}(\phi) $ and $\mathbb{T}_{ \mu \nu}^{EM}(F) $ are defined as: 

\begin{equation*}\label{3}  \mathbb{T}^{SF}_{\mu \nu}=  \Re(D _{\mu}\phi \overline{D _{\nu}\phi}) -\frac{1}{2}(g^{\alpha \beta} D _{\alpha}\phi \overline{D _{\beta}\phi} )g_{\mu \nu}, \end{equation*} 

\begin{equation*} \label{2} \mathbb{T}^{EM}_{\mu \nu}=g^{\alpha \beta}F _{\alpha \nu}F_{\beta \mu }-\frac{1}{4}F^{\alpha \beta}F_{\alpha \beta}g_{\mu \nu}.
\end{equation*}

In the $(u,v,\theta,\varphi)$ coordinate system of section \ref{Coordinates} this gives: 	
\begin{equation} \label{Tvv}
	\mathbb{T}_{v v}^{SF} = |D_v \phi| ^2,
\end{equation}
\begin{equation} \label{Tuu}
	\mathbb{T}_{u u}^{SF} = |D_u \phi| ^2,
\end{equation}
\begin{equation} \label{Tangular}
	\mathbb{T}_{\theta \theta}^{SF} = \mathbb{T}_{\varphi \varphi}^{SF}\sin^{-2}(\theta) =  \frac{r^2 \Re (D_u \phi \overline{D_v \phi})}{2\Omega^2}.
\end{equation}

\begin{equation} \label{TEMuv}
	\mathbb{T}_{u v}^{EM} = \frac{2\Omega^2 Q^2}{r^4} ,
\end{equation}
\begin{equation} \label{TangEM}
	\mathbb{T}_{\theta \theta}^{EM} = \mathbb{T}_{\varphi \varphi}^{EM}\sin^{-2}(\theta) =   \frac{ Q^2}{r^2} .
\end{equation}

\begin{equation} \label{Tuv}
	\mathbb{T}_{u v}^{SF} = \mathbb{T}_{v u}^{SF} = \mathbb{T}_{v v}^{EM}=\mathbb{T}_{u u}^{EM}= 0.
\end{equation}

While \eqref{Tvv}, \eqref{Tuu}, \eqref{Tangular} are used everywhere in the paper, particularly in subsection \ref{Morawetz} and \ref{RS}, \eqref{TEMuv} is only useful in subsection \ref{energysection} while \eqref{TangEM} is not used.

Notice that in section \ref{decay}, we do not make use of the divergence theorem but directly of equations \eqref{wavev}, \eqref{waveu}. Hence this appendix is mainly useful for section \ref{energyboundedsection}.

\subsection{Deformation tensors} \label{Pi}

As we saw in the beginning of this section, to use the divergence theorem we need to compute the derivative of vector fields.	More precisely for any vector field $X$ we define the deformation tensor as 

$$ \Pi_{X}^{\mu \nu }:= \nabla^{ (\mu} X^{\nu)}.$$

In the $(u,v,\theta,\varphi)$ coordinate system of section \ref{Coordinates} we compute: 	

\begin{equation} \label{deformationvv}	\Pi_{X}^{v v} = \frac{-2}{\Omega^2} \partial_u X^v, \end{equation}

\begin{equation} \label{deformationuu}	\Pi_{X}^{u u} = \frac{-2}{\Omega^2} \partial_v X^u,\end{equation}

\begin{equation} \label{deformationuv}	\Pi_{X}^{u v} = \frac{-1}{\Omega^2} \left( \partial_v X^v + \partial_v \log(\Omega^2) X^v+ \partial_u X^u + \partial_u \log(\Omega^2) X^u \right) , \end{equation}

\begin{equation} \label{deformationang}	\Pi_{X}^{\theta \theta}= \Pi_{X}^{\phi \phi} \sin^2(\theta) = \frac{1}{r^2} \left( \frac{\partial_v r}{r}X^v + \frac{\partial_u r}{r}X^u\right).\end{equation}

Notice that for $\partial_t = \frac{\partial_v + \partial_u}{2}$, 

$$ 		 \Pi_{\partial_t}^{\mu \nu}=0.$$

This is because $\partial_t$ is a Killing vector field, corresponding to the $t$ invariance of Reissner--Nordstr\"{o}m.

\subsection{Computations of bulk terms in the divergence formula} \label{appendixvectorfield}

	Let $(\phi,F)$ be a solution to the Maxwell--Charged--Scalar--Field equations \eqref{Maxwell}, \eqref{CSF}, and let $X$ be a vector field. Define the scalar-field current
	\[
	J_\mu^X(\phi):=\mathbb T^{SF}_{\mu\nu}X^\nu.
	\]
	Its divergence is
	\begin{equation}\label{current}
		\nabla^\mu J_\mu^X=\mathbb T^{SF}_{\mu\nu}\Pi_X^{\mu\nu}+F_{\mu\nu}X^\mu\mathcal J^\nu(\phi),
	\end{equation}
	where
	\[
	\mathcal J_\mu(\phi)=q_0\Im\!\left(\phi\,\overline{D_\mu\phi}\right).
	\]
	For $X=-f(r)\partial_{r^*}$, the deformation-tensor term in \eqref{current} is
	\begin{equation}\label{current2}
		\begin{split}
			\mathbb T^{SF}_{\mu\nu}\Pi_X^{\mu\nu}
			&=-\frac{f'(r)}{2}\left(|D_u\phi|^2+|D_v\phi|^2\right)
			-\frac{f(r)}{2r}\Re\!\left(D_v\phi\,\overline{D_u\phi}\right)\\
			&=-f'(r)\left(|D_{r^*}\phi|^2+|D_t\phi|^2\right)
			+\frac{f(r)}{2r}\left(|D_{r^*}\phi|^2-|D_t\phi|^2\right).
		\end{split}
	\end{equation}
	Finally, since $F_{uv}=2\Omega^2Q/r^2$, the interaction term is
	\begin{equation}\label{Maxwellcomputation}
		F_{\mu\nu}X^\mu\mathcal J^\nu(\phi)=\frac{Q}{r^2}\left[X^v\mathcal J_v(\phi)-X^u\mathcal J_u(\phi)\right].
	\end{equation}

\subsection{Volume forms, normals and current fluxes }

To conclude this appendix, we include a few computations of space-time volume forms and volume forms induced on the curves we use in this paper, together with exterior unit normals.

The volume form corresponding to Reissner--Nordstr\"{o}m metric is defined by 

$$ dvol = 2\Omega^2 r^2 du dv d\sigma_{\mathbb{S}^2},$$ where $d\sigma_{\mathbb{S}^2}$ is the standard volume form on the unit sphere.

We now give consecutively the normals, the induced volumes forms and the current flux for the following: constant $u$ hyper-surfaces, constant $v$ hyper-surfaces, constant $r$ hyper-surfaces $\gamma_{R_1}$ and $\Sigma_0 = \{ t=0\}$:

$$ n_{u=cst}^{\mu}  =\frac{1}{\Omega^2}  \partial_v ,$$
$$ dvol_{u=cst}  = \Omega^2  r^2 dv d\sigma_{\mathbb{S}^2},$$
$$ J_{\mu}n_{u=cst}^{\mu}dvol_{u=cst}  =J_v  r^2dv d\sigma_{\mathbb{S}^2}. $$




$$ n_{v=cst}^{\mu}  =\frac{1}{\Omega^2}  \partial_u, $$
$$ dvol_{v=cst} = \Omega^2  r^2 du  d\sigma_{\mathbb{S}^2},$$
$$ J_{\mu}n_{v=cst}^{\mu}dvol_{v=cst}  =J_u  r^2 du d\sigma_{\mathbb{S}^2}. $$

For these constant $u$ and constant $v$ hyper-surfaces, we expressed the \textbf{future-directed} normal. 

Notice that if such a null hyper-surface appears as the \textbf{future} boundary of a space-time domain, than the \textbf{future-directed} is also the \textbf{exterior} normal. 

Then we carry out the same computations for $\gamma_{R_1}= \{ r=R_1\}$, for any $r_+<R_1$: 

$$ n_{\gamma_{R_1}}^{\mu}  =\frac{1}{2\Omega(R_1)} (\partial_u -\partial_v) ,$$
$$ dvol_{\gamma_{R_1}} = \Omega(R_1) (dv+du) r^2d\sigma_{\mathbb{S}^2},$$
$$ J_{\mu}n_{\gamma_{R_1}}^{\mu}dvol_{\gamma_{R_1}}  = \frac{ J_u - J_v}{2}(dv+du)  r^2d\sigma_{\mathbb{S}^2}. $$

Notice that if $\gamma_{R_1} $ appears as the \textbf{right-most} boundary of a space-time domain, e.g.\ for $\{ r \leq R_1 \}$, than the \textbf{exterior} normal is the one we wrote.

Symmetrically, if $\gamma_{R_1} $ appears as the \textbf{left-most} boundary of a space-time domain, e.g.\ for $\{ r \geq R_1 \}$, than the \textbf{exterior} normal is the \textbf{opposite} of the one we wrote. \\

Finally we write  the same computations for $\Sigma_0 = \{ t=0\}$: 

$$ n_{\Sigma_0}^{\mu}  =\frac{1}{2\Omega} (\partial_v +\partial_u) ,$$
$$ dvol_{\Sigma_0} = \Omega \cdot (dv-du) r^2d\sigma_{\mathbb{S}^2},$$
$$ J_{\mu}n_{\Sigma_0}^{\mu}dvol_{\Sigma_0}  = \frac{ J_u + J_v}{2}(dv-du)  r^2d\sigma_{\mathbb{S}^2}. $$

We wrote the future directed normal, which is in fact the opposite of the exterior unit normal, because the space-time is included inside $\{ t \geq 0 \}$.





\begin{thebibliography}{9}
	
	
	\bibitem{Newrp} Yannis Angelopoulos, Stefanos Aretakis, Dejan Gajic \emph{A vector field approach to almost-sharp decay for the wave equation on spherically symmetric, stationary spacetimes}.
	Annals of PDE,  4:15, 2018.
	\bibitem{Latetime} Yannis Angelopoulos, Stefanos Aretakis, Dejan Gajic \emph{ Late-time asymptotics for the wave equation on spherically symmetric, stationary spacetimes}.
	Advances in Mathematics, Vol. 323, 529–621, 2018.
	
	\bibitem{Bieri}  Lydia Bieri, Shuang Miao, and Sohrab Shahshahani \emph{ Asymptotic properties of solutions of the Maxwell Klein Gordon
		equation with small data}. Communications in Analysis and Geometry 25(1), 2014.
	
	\bibitem{conical} Nicolas Burq, Fabrice Planchon, John Stalker, A. Shadi Tahvildar-Zadeh \emph{Strichartz estimates for the wave and Schroedinger equations with the inverse-square potential}, Journal of Functional Analysis Vol. 203, 519--549, 2003.
	
	
	
	
	
	
	
	
	
	
	\bibitem{Christo1}
	Demetrios Christodoulou \emph{The Formation of Black Holes and Singularities in
		Spherically Symmetric Gravitational Collapse
	}.
	Comm. Pure Appl. Math. 44, no. 3, 339–373, 1991.
	
	\bibitem{Christo2}
	Demetrios Christodoulou \emph{Bounded Variation Solutions of the Spherically Symmetric Einstein-Scalar Field
		Equations
	}.
	Comm. Pure Appl. Math. 46, 1131–1220, 1993.
	
	\bibitem{Christo3}
	Demetrios Christodoulou \emph{The instability of naked singularities in the gravitational collapse of a scalar
		field. 
	}
	Ann. of Math., (2) 149,, no. 1, 183–217, 1999.
	
	\bibitem{Shatah}
	Piotr T. Chru\'{s}ciel, Jalal Shatah \emph{Global existence of solutions of the Yang-Mills equations on globally hyperbolic four dimensional Lorentzian manifolds.
	}
	Asian journal of Mathematics, 	Vol.1, No.3,pp.
	530-548,
	1997.
	
	
	
	
	
	
	
	
	
	\bibitem{MihalisPHD}
	Mihalis Dafermos
	\emph{Stability and instability of the Cauchy
		horizon for the spherically symmetric
		Einstein--Maxwell-scalar field equations}.
	Ann. of Math. 158 , 875–928, 2003.
	
	
	
	\bibitem{Mihalis1}
	Mihalis Dafermos
	\emph{The Interior of Charged Black Holes
		and the Problem of Uniqueness in General Relativity}.
	Comm. Pure Appl. Math. 58, 0445–0504, 2005.
	
	\bibitem{MihalisStabExt}		Mihalis Dafermos, Gustav Holzegel, Igor Rodnianski		\emph{The linear stability of the Schwarzschild solution to gravitational perturbations.}		 Acta Math., 222, 1–214, 2019.
	
	\bibitem{KerrStab}
	Mihalis Dafermos, Jonathan Luk \emph{The interior of dynamical vacuum black holes I: The $C^0$-stability of the Kerr Cauchy horizon}.
	Preprint, arXiv:1710.01772, 2017. 
	
	\bibitem{PriceLaw}
	Mihalis Dafermos, Igor Rodnianski
	\emph{A proof of Price's law for the collapse of
		a self-gravitating scalar field}.
	Invent. math. 162, 381–457, 2005.
	
	\bibitem{BH}
	Mihalis Dafermos, Igor Rodnianski  \emph{Lectures on black holes and linear waves}. 2008.
	
	
	
	
	\bibitem{Redshift}
	Mihalis Dafermos, Igor Rodnianski
	\emph{The red-shift effect and radiation decay on black hole spacetimes.}
	Comm. Pure Appl.
	Math., 62(7):859:919, 2009.
	
	\bibitem{RP}
	Mihalis Dafermos, Igor Rodnianski  \emph{A new physical-space approach to decay for the wave
		equation with applications to black hole spacetimes}. pp. 421–432 in XVIth International Congress on Mathematical Physics,
	edited by P. Exner, World Sci. Publ., Hackensack, NJ, 2010. 
	
	\bibitem{SlowKerr}		Mihalis Dafermos, Igor Rodnianski		\emph{A proof of the uniform boundedness of solutions to the wave equation on slowly rotating			Kerr backgrounds.}							Invent. Math., 185, 2011.
	
	
	
	\bibitem{KerrDaf} 	Mihalis Dafermos, Igor Rodnianski, Yakov Shlapentokh-Rothman  \emph{Decay for solutions to the wave equation on Kerr exterior spacetimes III : the full sub-extremal case $ |a|<M$}.		Ann. of Math., 183, 787–913, 2016.
	
	
	\bibitem{Harvey}
	Oscar J. C. Dias, Harvey S. Reall, Jorge E. Santos
	\emph {Strong cosmic censorship for charged de Sitter black holes with a charged scalar field}. Classical and Quantum Gravity, Volume 36, Number 4, 2019.
	
	\bibitem{catenoid} Roland Donninger, Joachim Krieger, J\'{e}r\'{e}mie Szeftel, Willie Wong
	\emph{Codimension one stability of the catenoid under the vanishing mean curvature flowin Minkowski space}.  Duke Math.  J.165(4), 723-791, 2016
	
	
	\bibitem{Schlag}
	Roland Donninger, Wilhelm Schlag, Avy Soffer  \emph{A proof of Price's Law on Schwarzschild black hole manifolds for all angular momenta
	}.
	Advances in Mathematics	Volume 226, Issue 1, Pages 484-540, 2011.
	
	\bibitem{Eardley1}
	Douglas M. Eardley, Vincent Moncrief  \emph{The global existence of Yang–Mills–Higgs fields in 4-
		dimensional Minkowski space, I: Local existence and smoothness properties		}.
	Comm. Math. Phys. , 83:2, 171–191, 1982.
	
	\bibitem{Eardley2}
	Douglas M. Eardley, Vincent Moncrief  \emph{The global existence of Yang–Mills–Higgs fields in 4-
		dimensional Minkowski space, II: Completion of proof	}.
	Comm. Math. Phys. , 83:2, 193–212, 1982.
	
	\bibitem{GregJan}
	Grigorios Fournodavlos, Jan Sbierski  \emph{Generic blow-up results for the wave equation in the interior of a
		Schwarzschild black hole}, Arch. Ration. Mech. Anal. 235, 927-971, 2020.
	
	
	
	
	
	
	
	
	\bibitem{Dejaninverse}
	Dejan Gajic \emph{Late-time asymptotics for geometric wave equations with
		inverse-square potentials.}
	Journal of Functional Analysis 285 (7), 110058, 2023.
	
	\bibitem{extremeJonathan}
	Dejan Gajic, Jonathan Luk \emph{The interior of dynamical extremal black holes in
		spherical symmetry.}
	Pure and Appl. Anal., 1(2):263-326, 2019. 
	
	
	\bibitem{Andras}		Peter Hintz, Andr\'{a}s Vasy \emph{The global non-linear stability of the Kerr-de Sitter family			of black holes.}	Acta Math., 220:1--206, 2018.
	
	
	
	\bibitem{HodPiran1}
	Shahar Hod, Tsvi Piran
	\emph{Late-Time Evolution of Charged Gravitational Collapse and Decay of Charged Scalar Hair- I}.
	Phys.Rev D58 024017, 1998.
	
	
	
	
	
	
	
	\bibitem{Kauffman}
	Christopher Kauffman
	\emph{Global Stability for Charged Scalar Fields in an Asymptotically Flat Metric in Harmonic Gauge}.
	Preprint, arXiv:1801.09648, 2018.	
	
	
	\bibitem{KeelRoyTao}	Markus Keel, Tristan Roy, Terence Tao 	\emph{Global well-posedness of the Maxwell–Klein–Gordon equation below the
		energy norm}.		Discrete Contin. Dyn. Syst. 30:3, 573–621, 2011.
	
	
	
	\bibitem{KlainermanMachedon}
	Sergiu Klainerman, Matei Machedon
	\emph{On the Maxwell–Klein–Gordon equation with finite
		energy}.
	Duke Math. J. 74:1 , 19–44, 1994.
	
	\bibitem{Klainerman}
	Sergiu Klainerman, J\'{e}r\'{e}mie Szeftel
	\emph{Global Nonlinear Stability of Schwarzschild Spacetime under Polarized Perturbations}.
	Preprint, arXiv:1711.07597, 2017. 
	
	
	\bibitem{Kommemi}
	Jonathan Kommemi \emph{The Global Structure of Spherically Symmetric Charged Scalar Field Spacetimes}.
	Comm. Math. Phys., 323: 35. doi:10.1007/s00220-013-1759-1, 2013.
	
	\bibitem{Krieger}
	Joachim Krieger, Jonas Lührmann \emph{Concentration compactness for the critical Maxwell–Klein–Gordon
		equation}.
	Annals of PDE 1:1, Art. 5, 208, 2015.
	
	\bibitem{Krieger2}
	Joachim Krieger, Jacob Sterbenz, Daniel Tataru \emph{Global well-posedness for the Maxwell–Klein–Gordon equation in
		(4+1) dimensions: small energy}.
	Duke Math. J. 164:6, 973–1040, 2015.
	
	
	\bibitem{LindbladKG}
	Hans Lindblad, Jacob Sterbenz
	\emph{Global stability for charged-scalar fields on Minkowski space}.
	IMRP Int. Math. Res. Pap., Art. ID 52976, 109, 2006.
	
	
	\bibitem{JonathanQuantitative}
	Jonathan Luk, Sung-Jin Oh
	\emph{Quantitative decay rates for dispersive solutions to the Einstein-scalar field system in spherical symmetry.} Analysis and PDE, 8(7):1603:1674, 2015.	
	
	
	
	\bibitem{JonathanInstab}
	Jonathan Luk, Sung-Jin Oh
	\emph{Proof of Linear Instability of the Reissner--Nordstr\"{o}m Cauchy Horizon under Scalar Perturbations}.
	Duke Math. J. Volume 166, Number 3, 437-493, 2017.
	
	\bibitem{JonathanStab}
	Jonathan Luk, Sung-Jin Oh
	\emph{Strong Cosmic Censorship in Spherical Symmetry for two-ended Asymptotically Flat Initial Data I. The Interior of the Black Hole Region}, Annals of Math., 190(1):1-111, 2019. 
	
	\bibitem{JonathanStabExt}
	Jonathan Luk, Sung-Jin Oh
	\emph{Strong Cosmic Censorship in Spherical Symmetry for two-ended Asymptotically Flat Initial Data II. The Exterior of the Black Hole Region}, Annals of PDE, 5(6), 2019. 
	
	\bibitem{KerrInstab}
	Jonathan Luk, Jan Sbierski
	\emph{Instability results for the wave equation in the interior of Kerr black holes}. J. Funct. Anal., 271(7):1948-1995, 2016. 
	
	
	
	
	
	
	\bibitem{Machedon}	Matei Machedon, Jacob Sterbenz		\emph{Almost optimal local well-posedness for the (3+1)-dimensional
		Maxwell–Klein–Gordon equations}.  Journal of the American Mathematical Society 17:2, 297–359,  2004.
	
	\bibitem{Tataru2} Jason Metcalfe, Daniel Tataru, Mihai Tohaneanu \emph{Price's Law on Nonstationary Spacetimes}.		
	Adv. Math. 230, no. 3,  995-1028, 2012.
	
	
	\bibitem{Morawetz} Cathleen Synge Morawetz	\emph{Time decay for the nonlinear Klein-Gordon equations.}
	Proc. Roy. Soc. Ser. A, 306:291:296, 1968.
	
	\bibitem{Moschidisrp} Georgios Moschidis 	\emph{The $r^{p}$-weighted energy method of Dafermos and Rodnianski in general asymptotically flat spacetimes and applications.}
	Annals of PDE 2:6 , 1-194, 2016.
	
	\bibitem{OhTataru} Sung-Jin Oh, Daniel Tataru	\emph{Global well-posedness and scattering of the (4+1)-dimensional Maxwell–Klein–
		Gordon equation.}
	Invent. Math. 205:3 , 781–877, 2016.
	
	
	
	
	
	
	\bibitem{Pricepaper}
	Richard Price
	\emph{Nonspherical perturbations of relativistic gravitational collapse. I. Scalar and gravitational perturbations}
	Phys. Rev. D (3)
	5, 2419-2438, 1972.
	
	
	
	\bibitem{RodTao}  Igor Rodnianski, Terence Tao \emph{Global regularity for the Maxwell–Klein–Gordon equation with small
		critical Sobolev norm in high dimensions
		.}		Comm. Math. Phys.,  251:2, 377–426, 2004.
	
	
	
	
	
	
	
	
	\bibitem{Volker}  Volker Schlue \emph{Decay of linear waves on higher dimensional Schwarzschild black holes
		.}		Analysis and PDE 6-3, 515--600, 2013.
	
	\bibitem{Shu} Wei-Tong Shu \emph{Asymptotic properties of the solutions of linear and nonlinear spin field equations in Minkowski space
		.}		
	Comm. Math. Phys. 140:3, 449–480, 1991.
	
	\bibitem{Tataru} Daniel Tataru \emph{ Local decay of waves on asymptotically flat stationary space-times 
		.}		
	Amer. J. Math. 130, no. 3, 571--634, 2008.
	
	
	
	
	
	\bibitem{ShiwuKG}  Shiwu Yang \emph{Decay of solutions of Maxwell-Klein-Gordon equations with arbitrary Maxwell fields.}		Analysis and PDE, Vol. 9, No. 8, 2016.
	
	\bibitem{YangYu}  Shiwu Yang, Pin Yu \emph{On global dynamics of the Maxwell-Klein-Gordon equations.}, Cambridge Journal of Mathematics, 7 no.4, 365-467, 2019.
	
	\bibitem{Moi}  Maxime Van de Moortel \emph{Stability and instability of the sub-extremal Reissner--Nordström black hole interior for
		the Einstein--Maxwell-Klein-Gordon equations in spherical symmetry.}		Comm. Math. Phys., Volume 360, Issue 1, pp 103–168, 2018.
	
	
	
	
\end{thebibliography}
\end{document}